\documentclass[10pt]{article}
\usepackage[margin=.8in,left=.8in]{geometry}

\usepackage{amsthm}
\usepackage{amssymb}
\usepackage{amsmath}
\usepackage{mathtools}
\usepackage{adjustbox}

\usepackage[table]{xcolor}

\usepackage{stmaryrd}    
\usepackage{xfrac}
\usepackage{enumitem}\setlist{nosep}

\usepackage[safe]{tipa} 

\usepackage{hhline}

\usepackage{tensor}

\usepackage{tikz}
\usetikzlibrary{
  backgrounds, 
  decorations.pathmorphing, 
  decorations.markings,
  decorations.text
}
\usetikzlibrary{cd}

\usepackage{mathrsfs} 
\DeclareMathAlphabet{\mathpzc}{OT1}{pzc}{m}{it} 

\usepackage{hyperref}
\hypersetup{
  colorlinks = true,
  linkcolor= darkblue, citecolor=darkblue            
}

\hypersetup{
        colorlinks=true,
        linkcolor=darkblue,
        filecolor=darkblue,      
        urlcolor=black,
        citecolor=darkblue,
        linktoc=page
    }

\usepackage{cleveref}

\DeclareRobustCommand{\rchi}{{\mathpalette\irchi\relax}}
\newcommand{\irchi}[2]{\raisebox{\depth}{$#1\chi$}} 

\definecolor{darkblue}{rgb}{0.05,0.25,0.65}
\definecolor{darkgreen}{RGB}{20,140,10}
\definecolor{lightgray}{rgb}{0.9,0.9,0.9}
\definecolor{darkorange}{RGB}{200,100,5}
\definecolor{darkyellow}{rgb}{.91,.91,0}
\definecolor{orangeii}{RGB}{200,100,5}
\definecolor{lightblue}{RGB}{243, 250, 255}

\newtheorem{theorem}{Theorem}[section]
\newtheorem{claim}{Claim}[section]
\newtheorem{lemma}[theorem]{Lemma}

\newtheorem{proposition}[theorem]{Proposition}
\newtheorem{corollary}[theorem]{Corollary}
\theoremstyle{definition}
\newtheorem{definition}[theorem]{Definition}
\newtheorem{notation}[theorem]{Notation}
\newtheorem{example}[theorem]{Example}

\newtheorem{examples}[theorem]{Examples}
\newtheorem{remark}[theorem]{Remark}

\usepackage{accents}
\newlength{\dhatheight}
\newcommand{\doublehat}[1]{%
    \hspace{1.3pt}
    \settoheight{\dhatheight}{\ensuremath{\widehat{#1}}}%
    \addtolength{\dhatheight}{-0.23ex}%
    \widehat{\vphantom{\rule{1pt}{\dhatheight}}%
    \hspace{-1.3pt}
    \smash{\widehat{#1}}}}

\let\PLAINthebibliography\thebibliography
\renewcommand\thebibliography[1]{
  \PLAINthebibliography{#1}
  \setlength{\parskip}{0.5pt}
  \setlength{\itemsep}{0.5pt plus .3ex}
}

\newcommand{\vspaceabove}{$\phantom{\mathclap{\vert^{\vert^{\vert^{\vert}}}}}$}

\newcommand{\proofstep}[1]{\scalebox{.85}{#1}}

\newcommand{\yields}{\Rightarrow}

\newcommand{\shape}{
  \raisebox{1pt}{\rm\normalfont\textesh}
}

\newcommand{\CurvatureAtPsiOne}{J}

\newcommand{\bfOmega}{\mathbf{\Omega}\hspace{-8.8pt}\mathbf{\Omega}}

\newcommand{\evencoordinateindex}{r}
\newcommand{\oddcoordinateindex}{\rho}

\newcommand{\ZTwo}{\mathbb{Z}_2}

\newcommand{\defneq}{\equiv}

\newcommand{\differential}{\mathrm{d}}

\newcommand{\covariantderivative}{\nabla}

\newcommand\bos[1]{\mathstrut\mkern2.5mu#1\mkern-14mu\raise1.7ex%
  \hbox{$\scriptstyle\rightsquigarrow$}}

\newcommand\bosonic[1]{\mathstrut\mkern2.5mu#1\mkern-14mu\raise1.7ex%
  \hbox{$\scriptstyle\rightsquigarrow$}}

\usepackage[cmtip,all]{xy}
\newcommand{\longsquiggly}{\xymatrix{{}\ar@{~>}[r]&{}}}

\usepackage{bbm} 

\usepackage[new]{old-arrows}   

\newcommand{\FR}{\mathbb{R}}

\newcommand{\SmoothManifolds}{\mathrm{SmthMfd}}

\newcommand{\sManifolds}{\mathrm{sSmthMfd}}

\newcommand{\sSmoothSets}{\mathrm{sSmthSet}}

\newcommand{\dd}{\mathrm{d}}

\newcommand{\odd}{\mathrm{odd}}

\begin{document}

\setlength{\abovedisplayskip}{3pt}
\setlength{\belowdisplayskip}{3pt}
\setlength{\abovedisplayshortskip}{-3pt}
\setlength{\belowdisplayshortskip}{3pt}

\title{Flux Quantization on 11-dimensional Superspace}

\author{
  Grigorios Giotopoulos${}^{\ast}$,
  \;\;
  Hisham Sati${}^{\ast \dagger}$,
  \;\;
  Urs Schreiber${}^{\ast}$
}

\maketitle

\thispagestyle{empty}

\begin{abstract}
  Flux quantization of the C-field in 11d supergravity is arguably necessary for the (UV-)completion of the theory, in that it determines the torsion charges carried by small numbers $N \ll \infty$ of M-branes.
  However, hypotheses about C-field flux-quantization (``models of the C-field'') have previously been discussed only in the bosonic sector of 11d supergravity and ignoring the supergravity equations of motion. Here we highlight a duality-symmetric formulation of on-shell 11d supergravity on superspace, observe that this naturally lends itself to completion of the theory by flux quantization, and indeed that 11d super-spacetimes are put on-shell by carrying quantizable duality-symmetric super-C-field flux; the proof of which we present in detail.   
\end{abstract}

\vspace{.1cm}

\begin{center}
\begin{minipage}{11.5cm}
\tableofcontents
\end{minipage}
\end{center}

\medskip

\vfill

\hrule
\vspace{5pt}

{
\footnotesize
\noindent
\def\arraystretch{1}
\tabcolsep=0pt
\begin{tabular}{ll}
${}^*$\,
&
Mathematics, Division of Science; and
\\
&
Center for Quantum and Topological Systems,
\\
&
NYUAD Research Institute,
\\
&
New York University Abu Dhabi, UAE.  
\end{tabular}
\hfill
\href{https://ncatlab.org/nlab/show/Center+for+Quantum+and+Topological+Systems}{
\adjustbox{raise=-15pt}{
\includegraphics[width=3cm]{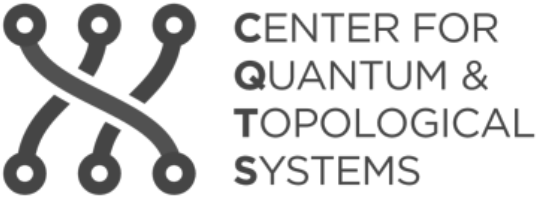}
}
}

\vspace{1mm} 
\noindent ${}^\dagger$The Courant Institute for Mathematical Sciences, NYU, NY

\vspace{.2cm}

\noindent
The authors acknowledge the support by {\it Tamkeen} under the 
{\it NYU Abu Dhabi Research Institute grant} {\tt CG008}.
}

\newpage


\noindent
{\bf Overview and Results.}
In \S\ref{SuperFluxQuantization} we address the open problem of flux- and charge-quantization (see \cite{SS24Flux}) of $D=11$, $\mathcal{N}=1$ supergravity (11d SuGra \cite{CJS78}, review in \cite{DNP86}\cite{MiemiecSchnakenburg06}), explaining how this is a necessary step towards the completion of 11d SuGra to the conjectured ``M-theory'' \cite{Duff99}.
Our main observation is that a solution proceeds naturally via {\it duality-symmetric} formulation
of the C-field \cite{BBS98}
but for its
{\it super-flux density} \eqref{SuperFluxDensitiesInIntroduction} on super-spacetime manifolds (``superspace supergravity'', going back to \cite{WZ77}\cite{SG79}\cite{CF80}\cite{BrinkHowe80}\cite{Howe82}\cite{DF82}).

\medskip 
Up to some mild but consequential change in perspective --- for us the whole theory is driven by the construction of (quantizable) super-C-field flux, which traditionally has instead been an afterthought --- the required duality-symmetric formulation of 11d superspace supergravity 
has previously been indicated in \cite[\S III.8.5]{CDF91} and (apparently independently in) \cite[\S 6]{CL94} and as such is known to experts (cf. \cite[\S 5]{HT03}\cite[\S 2]{Tsimpis04}). However, since the rather non-trivial proof has never been spelled out in print (and also seems not to exist in the proverbial drawers) --- 
while our new perspective clarifies its impact which may not have been fully appreciated before --- we take the occasion, in \S\ref{Supergravity}, to
demonstrate in detail this construction of on-shell 11d SuGra, with computer algebra checks recorded in \cite{AncillaryFiles}.

\medskip 
With this result in hand, the flux quantization of 11d SuGra follows by lifting the (rational-)homotopy theoretic formulation of flux quantization
\cite{FSS23Char}\cite{SS24Flux} 
to supergeometric homotopy theory \cite[\S 3.1.3]{SS20Orb}. Our brief survey of the required higher supergeometry in \S\ref{SuperCartanGeometry}
(with more details to be presented in \cite{GSS24}\cite{GSS25}) should thus make (flux quantized) 11d supergravity broadly accessible to both physicists and mathematicians. We also pause to carefully connect this more abstract formulation to the traditional notion of C-field gauge potentials (Prop. \ref{RecoveringTraditionalSuperCFieldGaugePotentials}).

\medskip 
While this article is therefore to some extent a unified and modernized review of (\S\ref{SuperFluxQuantization}) flux quantization, (\S\ref{SuperCartanGeometry}) higher supergeometry, and (\S\ref{Supergravity}) on-shell 11d SuGra on superspace,
we suggest that the new perspective opens the door to further progress towards the elusive M-theory: In follow-up articles \cite{GSS-M5Brane}\cite{GSS-Exceptional},
we use the approach to discuss flux-quantized {\it super-exceptional}-geometric supergravity compatible both with the super-exceptional embedding construction of the M5-brane \cite{FSS20Exc} as well as with the (level-)quantization of its Hopf-WZ/Page-charge term \cite{FSSHopf} for {\it nonabelian} worldvolume (higher) gauge fields \cite{FSS20TwistedString}, a major open issue in M-theory.

\section{Super-Flux Quantization of 11d SuGra}
\label{SuperFluxQuantization}

\noindent
{\bf The role of 11d Supergravity.}
It is known, but may remain underappreciated, that any field theory with fermions (such as the standard model) 
by necessity lives ``in superspace'' (as per the terminology of \cite{GGRS83}\cite{BK95}), in that its phase space is an object of super-geometry (\S\ref{SuperCartanGeometry}), regardless of whether or not 
the theory is super-symmetric. This is discussed in detail in the companion article \cite{GSS24} (quick exposition in \cite{Schreiber24}).

\medskip

This being so and contrary to common perception, on platonic grounds it would be surprising if the world were {\it not fundamentally} supersymmetric, with its observed bosonic symmetries just being broken fundamental super-symmetries. The only reason that this evident conclusion contradicts contemporary perception is the widespread focus on {\it global} supersymmetry, which however, like all global symmetries, cannot be expected to be more than accidental. Instead, fundamental supersymmetry is {\it local} supersymmetry, which in a relativistic world means \cite[p. 3]{Duff05}:\footnote{To recall the famous quote from \cite[p. 318]{Weinberg00}: ``Gravity exists, so if there is any truth to supersymmetry then any realistic supersymmetry theory must eventually be enlarged to a supersymmetric 
theory of matter and gravitation, known as supergravity. Supersymmetry  without supergravity is not an option, though it may be a good approximation at energies far below the Planck scale.''} {\it super-gravity} (henceforth SuGra; reviews include \cite{vN81}\cite{DNP86}\cite{CDF91}\cite[\S 31]{Weinberg00}\cite{VF12}\cite{Sezgin23}).

\medskip

Remarkably, theories of supergravity are both highly constrained as well as tightly interrelated, with the result that they seem to all revolve around the central instance in super-dimension $(11\vert \mathbf{32})$ (aka $D=11$, $\mathcal{N}=1$ supergravity \cite{CJS78}, review in \cite{MiemiecSchnakenburg06}, streamlined re-derivation in \S\ref{Supergravity}). The 
still elusive but plausibly existing (UV-)completion of 11d supergravity to a quantum theory has famously been conjectured (working title: ``M-Theory'' cf. \cite{Duff96}\cite{Duff99}\cite[\S 12]{Moore14}) to be the (similarly elusive) ``grand-unified theory of everything'' which ought to complete the standard model of particle physics coupled to gravity at sizable energies $1/\ell$, sizable coupling $g$,
and sizable `quantumness' $\hbar$.

\medskip

\noindent
{\bf The open problem of flux \& charge quantization of the C-field in 11d SuGra.}
While such extraordinary conjectures require extraordinary evidence, there has previously been little work on the completion of 11d supergravity even as a classical field theory. The traditional formulations of 11d SuGra  on local charts of spacetime (only) are incomplete because of the presence of the higher gauge field, namely the ``3-index photon'' or {\it C-field}, in the theory (\cite[p. 409]{CJS78}\cite[p. 1.15]{DF82}, review in \cite[pp. 31]{MiemiecSchnakenburg06}\cite[Ex. 2.12]{SS24Flux}). 

\smallskip 
As familiar from the ordinary photon field (the {\it A-field}) of Maxwell theory, the consistent definition of such higher gauge fields requires a specification of their {\it flux-quantization} laws (for review and pointers to the literature see \cite{SS24Flux})
which encodes in particular the torsion charges that may be reflected in the field flux on topologically non-trivial spacetimes or with non-trivial boundary conditions. 
Here ``torsion'' (cf. Rem. \ref{OnTheTermTorsion} for disambiguation) refers to charges which are torsion elements of their cohomology group in that some multiple $k$ of them {\it vanishes}. This means that the solitons (branes) carrying such $k$-torsion charges do not exist in large $N \gg  k$-numbers, and hence constitute small-$N$, hence large-$1/\ell$ information of the field theory.

But flux quantization in general, and of the C-field in particular, requires as input datum higher Bianchi identities in {\it duality-symmetric form} \cite{BBS98}\cite{CJLP98}\cite{Nu03}\cite{Sati06}\cite[\S 5.3]{Sati10}. Before we state the aim and conclusion of this article in \cref{DualitySymmetryOnSuperSpacetime}, we briefly recall in \S\ref{DualitySymmetryForFluxQuantization} why this is the case (following \cite{SS23FQ}\cite{SS24Flux}, full details in \cite{FSS23Char}).

\subsection{Duality-Symmetry for Flux Quantization}
\label{DualitySymmetryForFluxQuantization}

\noindent
{\bf The role of duality-symmetric Bianchi identities.}
The flux quantization law $\mathcal{A}$ of a higher gauge theory is a further choice of non-perturbative field content beyond that encoded by differential forms alone. Remarkably, the available choices of flux quantization laws are controlled by the form of the Bianchi identities 
\begin{equation}
  \label{PremetricBianchiIdentities}
  \overset{
    \mathclap{
      \raisebox{2pt}{
        \scalebox{.7}{
          \color{darkblue}
          \bf
          \def\arraystretch{.9}
          \begin{tabular}{c}
            de Rham
            \\
            differential
          \end{tabular}
        }
      }
    }
  }{
    \mathrm{d}
    \,
    \vec F
  }
  \;=\;
  \vec P\big(
    \vec F
\,  \big)
  \;:=\;
      \overset{
      \raisebox{2pt}{
        \scalebox{.7}{
          \color{darkblue}
          \bf
          polynomial \;\;\;
        }
      }
    }{\big(
    P^i\big(
      \vec F
    \, \big)
  \big)_{i \in I}
  }
\end{equation}
on the flux densities
\begin{equation}
  \label{TheFluxDensities}
  \vec F
  \;:=\;
  \Big(\!
    \overset{
      \raisebox{2pt}{
        \scalebox{.7}{
          \color{darkblue}
          \bf
          differential forms
        }
      }
    }{
    F^i
    \;\in\;
    \Omega
      _{\mathrm{dR}}
      ^{\mathrm{deg}_i}
    (X)
    }
  \!\Big)_{i \in I}
\end{equation}
(on spacetime $X$ indexed by some set $I$)
in their ``duality-symmetric'' or ``pre-metric'' guise, where the duality relation between magnetic and electric flux densities (the ``constitutive equation'')
\begin{equation}
  \label{TheHodgeDualityConstraint}
  \overset{
    \mathclap{
      \raisebox{2pt}{
        \scalebox{.7}{
          \color{darkblue}
          \bf
          \def\arraystretch{.9}
          \begin{tabular}{c}
            Hodge
            \\
            star-operator
          \end{tabular}
        }
      }
    }
  }{
  \star
  \,
  \vec F
  }
  \;\;\;\;
    =
  \;\;\;\;
  \overset{
    \mathclap{
      \raisebox{2pt}{
        \scalebox{.7}{
          \color{darkblue}
          \bf
          linear map
        }
      }
    }
  }{
  \vec \mu
  \big(
    \vec F
  \, \big)
  }
\end{equation}
is not imposed (yet) \cite[\S 2.4]{SS24Flux}.
For example (cf. \cite[Ex. 2.12]{SS24Flux}), the duality-symmetric form of the flux densities and their Bianchi identities in 11d supergravity is
\begin{equation}
  \label{PremetricCFieldFlux}
  \mathllap{
    \scalebox{.7}{
      \color{darkblue}
      \bf
      \def\arraystretch{.9}
      \begin{tabular}{c}
        Pre-metric/duality-symmetric
        \\
        C-field flux densities
        \\
        in 11d supergravity
      \end{tabular}
    }
  }
  \vec F
  \;\;
    =
  \;\;
  \left(\!\!
    \def\arraystretch{1.3}
    \begin{array}{l}
      G_4 
        \,\in\,
      \Omega^4_{\mathrm{dR}}(X)
      \\
      G_7 
        \,\in\,
      \Omega^7_{\mathrm{dR}}(X)
    \end{array}
\!\!  \right)
  \,,
  \hspace{.7cm}
  \def\arraystretch{1.1}
  \begin{array}{l}
    \differential
    \, 
    G_4
    \;=\; 0
    \,,
    \\
    \differential
    \, 
    G_7
    \;=\;
    \tfrac{1}{2}
    G_4 \wedge G_4
    \,,
  \end{array}
\end{equation}

\smallskip 
\noindent on which the electromagnetic duality relation of the C-field 
\begin{equation}
  \label{DualityConstraintOnCField}
  \def\arraystretch{1}
  \begin{array}{rcr}
    \star\, G_4 &=&  G_7
    \,,
    \\
    \star\, G_7 &=& - G_4
    \,.
  \end{array}
\end{equation}
is still to be imposed.
Now, under mild conditions the pre-metric Bianchi identities \eqref{PremetricBianchiIdentities} 
are equivalent \cite[\S 3.1]{SS24Flux} to the closure condition on an $\mathfrak{a}$-valued differential form
\vspace{-1mm} 
$$
  \mathrm{d}
  \, 
  \vec F
  \;=\;
  \vec P\big(
    \vec F
  \, \big)
  \hspace{.5cm}
  \Leftrightarrow
  \hspace{.5cm}
  \vec F
  \;\in\;
  \overset{
    \mathclap{
      \raisebox{2pt}{
      \scalebox{.7}{
        \color{darkblue} \bf
        \begin{tabular}{c}
          closed
          $\mathfrak{a}$-valued
          \\
          differential forms
        \end{tabular}
      }
      }
    }
  }{
  \Omega^1_{\mathrm{dR}}
  \big(
    X
    ;\,
    \mathfrak{a}
  \big)_{\mathrm{clsd}}\,.
  }
$$
Here $\mathfrak{a}$ is the $L_\infty$-algebra whose underlying graded vector space $\mathfrak{a}_\bullet$ is spanned by elements
$$
  \vec v
  \;\;
  :=
  \;\;
  \big(
    v_i \,\in\,
    \mathfrak{a}_{\mathrm{deg}_i-1}
  \big)_{i \in I}
$$
with $n$-ary graded-skew symmetric brackets 
$$
  \big[
    -,\cdots, -
  \big]
  \;:\;
  \mathfrak{a}^{\otimes^n}
  \xrightarrow{\phantom{--}}
  \mathfrak{a}
$$
given by the coefficients of the graded-symmetric polynomial appearing in \eqref{PremetricBianchiIdentities}:
$$
  \big[
    v_{j_1},
    \cdots,
    v_{j_n}
  \big]
  \;=\;
  P^i_{j_1 \cdots j_n}
  \,
  v_i
  \,,
  \;\;\;\;
  \mbox{where}
  \;\;\;\;
  P^i\big((F^j)_{j \in I}\big)
  \;\;
  =
  \;\;
  \underset{
    n \in \mathbb{N}
  }{\sum}
  P^i_{
    j_{1}
    \cdots 
    j_n
  }
  F^{j_1}
  \cdots
  F^{j_n}
  \,.
$$

\medskip

\noindent
{\bf Flux \& charge quantization.}
With this characteristic $L_\infty$-algebra $\mathfrak{a}$ of the higher gauge theory identified, a compatible {\it flux quantization law} is given  \cite[\S 3.2]{SS24Flux} by a classifying space $\mathcal{A}$ whose rational Whitehead $L_\infty$-algebra $\mathfrak{l}\mathcal{A}$ (the ``Quillen model'' of $\mathcal{A}$) coincides with $\mathfrak{a}$. 
Such a space comes with a generalized {\it character map} \cite{FSS23Char} assigning to total charges $[\rchi]$ quantized in $\mathcal{A}$-cohomology
the total fluxes $\big[\vec F_{\rchi}\big]$ sourced by these charges:
$$
  \begin{tikzcd}[row sep=-6pt]
     \scalebox{.7}{
      \color{darkblue}
      \bf
      \def\arraystretch{1}
      \begin{tabular}{c}
        Nonabelian
        \\
        generalized cohomology
        \\
        with coefficients in 
        $\mathcal{A}$
      \end{tabular}
    }
    &[-30pt]
    H^1\big(
      X
      ;\,
      \Omega \mathcal{A}
    \big) 
   \ar[
      rr,
      shorten=-4pt,
      "{
        \mathrm{ch}^{\mathcal{A}}_X
      }",
      "{
        \scalebox{.7}{
          \color{darkgreen}
          \bf
          character map
        }
      }"{swap}
    ]
    &&
    H^1_{\mathrm{dR}}\big(
      X
      ;\,
      \overbrace{
          \mathfrak{l}\mathcal{A}
      }^{
        \cong 
        \,
        \mathfrak{a}
      }
    \big)
    &[-30pt]
    \scalebox{.7}{
      \def\arraystretch{1}
      \begin{tabular}{c}
        \color{darkblue}
        \bf
        Nonabelian
        \\
        \color{darkblue}
        \bf
        de Rham cohomology
        \\
        \color{darkblue}
        \bf
        with coefficients in 
        $\mathfrak{l}\mathcal{A}$
        \\
        (Def. \ref{NonabelianDeRhamCohomology})
      \end{tabular}
    }
    \\[-3pt]
    &
    {[\rchi]}
    &\longmapsto&
    \big[
      \vec F_{\rchi}
    \big]
    \\
    &
    \scalebox{.7}{
      \color{darkblue}
      \bf
      \def\arraystretch{.9}
      \begin{tabular}{c}
        quantized 
        \\
        total charge
      \end{tabular}
    }
    &&
    \scalebox{.7}{
      \color{darkblue}
      \bf
      \def\arraystretch{.9}
      \begin{tabular}{c}
        sourced
        \\
        total flux
      \end{tabular}
    }
  \end{tikzcd}
$$

\vspace{-1mm} 
\noindent Hence $\mathcal{A}$-quantization of flux means first of all that flux densities $\vec F$ are
to be accompanied by total charges $[\rchi]$ such that their $\mathfrak{a}$-valued de Rham class 
$[\vec F\, ]$ coincides with the character of the total charge:
\begin{equation}
  \label{TotalChargeQuantization}
  \begin{tikzcd}[row sep=large, column sep=huge]
    &&
    &
    H^1(
      X
      ;\,
      \Omega \mathcal{A}
    )
    \ar[
      dd,
      "{
        \mathrm{ch}
          ^{\mathcal{A}}
          _X
      }"{swap},
      "{
        \scalebox{.7}{
          {
          \color{darkgreen}
          \bf
          \def\arraystretch{.9}
          \begin{tabular}{c}
            character
            \\
            map
          \end{tabular}
          }
          \color{gray}
          \cite[Def. IV.2]{FSS23Char}
        }
      }"{xshift=-7pt}
    ]
    \\
    \\
    \ast
    \ar[
      rr,
      "{
        \vec F
      }",
      "{
        \scalebox{.7}{
          \color{darkgreen}
          \bf
          flux density
        }
      }"{swap}
    ]
    \ar[
      uurrr,
      dashed,
      bend left=15,
      "{
        [\rchi]
      }",
      "{
        \scalebox{.7}{
          \color{darkgreen}
          \bf
          total charge
        }
      }"{swap, sloped}
    ]
    \ar[
      rrr,
      rounded corners,
      to path={
           ([yshift=+00pt]\tikztostart.south)  
        -- ([yshift=-12pt]\tikztostart.south)  
        -- node[xshift=30pt, yshift=5pt]{
             \scalebox{.7}{
               \color{darkgreen}
               \bf
               total flux
             }
           }
           ([yshift=-10pt]\tikztotarget.south)  
        -- ([yshift=-00pt]\tikztotarget.south)  
      }
    ]
    &&
    \Omega^1_{\mathrm{dR}}(
      X
      ;\,
      \mathfrak{a}
    )
    \ar[
      r,
      ->>
    ]
    &
    H^1_{\mathrm{dR}}(
      X
      ;\,
      \mathfrak{a}
    )
    \,.
  \end{tikzcd}
\end{equation}

\vspace{2mm} 
\noindent Stated in more detail \cite[\S 3.3]{SS24Flux}, the character map lifts from 
cohomology classes to {\it moduli stacks} and
$\mathcal{A}$-flux quantization means that non-perturbative gauge field configurations are triples consisting of:
\begin{itemize}
\item[\bf (i)] {\it flux densities} $\vec F \in \Omega^1_{\mathrm{dR}}\big(X;\, \mathfrak{a}\big)_{\mathrm{clsd}}$ satisfying their pre-metric Bianchi identities;

\item[\bf (ii)] {\it local charges} $\rchi : X \to \mathcal{A}$ representing classes in $\mathcal{A}$-cohomology;

\item[\bf (iii)] 
{\it deformations} $\widehat{A} : \mathrm{ch}(\rchi) \Rightarrow \eta^{\scalebox{.6}{\,$\shape$}}(\vec F\,)$ of the flux densities into the character fluxes sourced by the local charges.
\end{itemize}
The last component $\widehat{A}$ turns out to be equivalently the global form of the {\it gauge potential} which constitutes the actual flux-quantized higher gauge field:
\vspace{-1mm} 
\begin{equation}
  \label{TheGaugePotentials}
  \hspace{-1.5cm}
  \begin{tikzcd}[row sep=large, 
    column sep=65pt
  ]
    & &
    &
    \overset{
      \mathclap{
        \scalebox{.7}{
          \begin{tabular}{c}
            \color{darkblue}
            \bf
            Classifying space
            \\
            \color{darkblue}
            \bf
            of quantized charges
            \\
            \rm
            (Ex. \ref{PathInfinityGroupoids})
          \end{tabular}
        }
        \mathrlap{
          \scalebox{.7}{
            \color{gray}
            \begin{tabular}{c}
              subject to 
              \\
              $\mathfrak{l}\mathcal{A} \,\cong\, \mathfrak{a}$
            \end{tabular}
          }
        }
      }
    }{
      \mathcal{A}
    }
    \ar[
      dd,
      "{
        \mathbf{ch}^{\mathcal{A}}_X
      }"{swap},
      "{
        \scalebox{.7}{
          {
          \color{darkgreen}
          \bf
          \def\arraystretch{.9}
          \begin{tabular}{c}
            differential
            \\
            character
          \end{tabular}
          }
          \cite[Def. 9.2]{FSS23Char}
        }
      }"{xshift=-7pt}
    ]
    \\
    \\
    \underset{
      \scalebox{.7}{
        \def\arraystretch{.9}
        \begin{tabular}{c}
          \color{darkblue}
          \bf
          spacetime
          \\
          \color{darkblue}
          \bf
          manifold
          \\
          (Ex. \ref{CechResolutionOfSupermanifold})
        \end{tabular}
      }
    }{
      X
    }
    \ar[
      rr,
      "{ 
        \vec F 
      }"{pos=.65, name=t},
      "{
        \scalebox{.7}{
          \def\arraystretch{.9}
          \begin{tabular}{c}
            \color{darkgreen}
            \bf
            flux density
            \\
            \rm
            (Def. \ref{ClosedLInfinityValuedDifferentialForms})
          \end{tabular}
        }
      }"{swap, pos=.6}
    ]
    \ar[
      uurrr,
      bend left=15,
      dashed,
      "{
        \rchi
      }"{swap, pos=.4, name=s},
      "{
        \scalebox{.7}{
          \begin{tabular}{c}
            (Ex. \ref{NonabelianCohomology})
            \\
            \color{darkgreen}
            \bf
            local charge
          \end{tabular}
        }
      }"{sloped, pos=.4}
    ]
    &&
    \underset{
      \mathclap{
        \raisebox{-11pt}{
        \scalebox{.7}{
          \def\arraystretch{.85}
          \begin{tabular}{c}
            \color{darkblue}
            \bf
            Smooth set of
            \\
            \color{darkblue}
            \bf
            duality-symmetric
            \\
            \color{darkblue}
            \bf
            flux densities
            \\
            \eqref{SmoothSuperSetOfClosedForms}
          \end{tabular}
        }
        }
      }    
    }{
      \Omega^1_{\mathrm{dR}}(
        -;
        \mathfrak{a}
      )_{\mathrm{clsd}}
    }
   \quad
   \ar[
      r,
      shorten <=-12pt,
      "{
        \eta^{\scalebox{.6}{\,$\shape$}}
      }"{pos=.4},
      "{
        \scalebox{.7}{
          \def\arraystretch{.9}
          \begin{tabular}{c}
            \color{darkgreen}
            \bf
            up to
            \\
            \color{darkgreen}
            \bf
            deformations
            \\            \eqref{ShapeUnitForDeformationModuli}
          \end{tabular}
        }
      }"{swap, pos=.4}
    ]
    &
    \underset{
      \mathclap{
        \raisebox{-8pt}{
        \scalebox{.7}{
          \def\arraystretch{.85}
          \begin{tabular}{c}
            \color{darkblue}
            \bf
            their deformation
            \\
            \color{darkblue}
            \bf
            $\infty$-groupoid
            \\
            (Ex. \ref{InfinityGroupoidOfFluxDeformations})
          \end{tabular}
        }
        }
      }
    }{
    \shape
    \,
    \Omega^1_{\mathrm{dR}}(
      -;
      \mathfrak{a}
    )_{\mathrm{clsd}}
    }
    \ar[
      from=s,
      to=t,
      Rightarrow,
      "{
        \widehat A
        \mathrlap{
          \scalebox{.7}{
            \color{orangeii}
            \bf
            global gauge potential
          }
        }
      }"{description}
    ]
  \end{tikzcd}
\end{equation}

{
\noindent
For example, this procedure \eqref{TheGaugePotentials} recovers \cite[Ex. 3.10 \& \S 4.1]{SS24Flux} the following familiar examples of globally well-defined flux-quantized higher gauge fields:
\smallskip 
\begin{itemize}[leftmargin=.4cm]
\setlength\itemsep{2pt}
\item
{\it Maxwell field}: global gauge potentials are connections on $\mathrm{U}(1)$-principal bundles,
for the choice $\mathcal{A} \defneq B^2\mathbb{Z} \times B^2 \mathbb{Q}$ 

(as proposed by \cite{Dirac1931}\cite{Schwinger1966}\cite{Zwanziger1968} 

and recast in modern language by \cite{Alvarez85b}\cite[\S 7.1]{Brylinski93}\cite[\S 16.4e]{Frankel97})

or rather on electro-magnetic {\it pairs} of $\mathrm{U}(1)$ principal bundles,
for the choice $\mathcal{A} = B^2 \mathbb{Z} \times B^2 \mathbb{Z}$

(as considered in \cite{FreedMooreSegal07}\cite[Rem. 2.3]{BBSS17}\cite[Def. 1.16]{LS22}\cite[(3)]{LS23})

\item
{\it B-field} in 10d: global gauge potentials are connections on $\mathrm{U}(1)$-bundle gerbes,
for the choice $\mathcal{A} \defneq B^3 \mathbb{Z} \times B^3 \mathbb{Q}$

(as proposed by \cite{Gawedzki86}\cite{FW99}\cite{CJM04}, review in \cite{FNSW09}),

\item 
{\it RR-field}: global gauge potentials are cocycles in twisted differential K-theory,
for the choice $\mathcal{A} \defneq \mathrm{KU}_0 \sslash B^2 \mathbb{Z}$ 

(as proposed in various forms by \cite{MinasianMoore97}\cite{Witten98}\cite{MooreWitten99}\cite{FreedHopkins00}\cite{BouwknegtMathai01} and established in full form in \cite{GS-RR});
\end{itemize}

\smallskip 
and it seamlessly generalizes further to the case of interest here:
\smallskip 
\begin{itemize}[leftmargin=.4cm]
\item \label{HypothesisH}
{\it C-field} in 11d: global gauge potentials are cocycles in (twisted) differential Cohomotopy,
for the choice $\mathcal{A} \defneq S^4$ 

(``Hypothesis H'', proposed in \cite[\S 2.5]{Sati13}, checked in \cite{FSS20-H}\cite{FSSHopf}\cite{FSS22Twistorial} to reproduce the expectations from the M-theory literature, reviewed in \cite[\S 12]{FSS23Char}\cite[\S 4.3]{SS24Flux}).
\end{itemize}
}

\smallskip

\noindent
We recall a few more details of how this works in \cref{BackgroundOnFluxQuantization}.

\begin{center}
\small 
\def\arraystretch{1.2}
\begin{tabular}{cccc}
\hline 
{\bf Field} 
& 
{\bf Choice of $\mathcal{A}$ }
&
{\bf Induced global gauge potentials} 
& 
{\bf Details}
\\
\hline 
\rowcolor{lightgray} Maxwell field 
& 
$B^2\mathbb{Z} \times B^2 \mathbb{Z}$ 
&
cocycles in differential integral 2-cohomology 

&
\footnotesize  \cite[Prop. 9.5]{FSS23Char}
\\
 B-field in 10d 
&
$ B^3 \mathbb{Z} \times B^7 \mathbb{Z}$
& 
cocycles in higher differential integral cohomology
&
{\footnotesize \cite[Prop. 9.5]{FSS23Char}}
\\
\rowcolor{lightgray} RR field in 10d 
&
$\mathrm{KU}_0 \sslash B^2 \mathbb{Z}$ 
& 
cocycles in twisted differential K-theory 
&
\footnotesize  \cite[Ex. 11.2]{FSS23Char}
\\
C-field in 11d 
& 
$S^4 \sslash \mathrm{Spin}(5)$
& 
cocycles in (twisted) differential Cohomotopy 
&
\footnotesize  \cite[\S 12]{FSS23Char}
\\
\hline 
\end{tabular} 
\end{center}

This shows that flux quantization of a higher gauge field theory is the step where the actual global non-perturbative gauge field content of the theory is determined. As such, flux quantization is not an afterthought but the core of any higher gauge theory, non-perturbatively.

\medskip

\noindent
{\bf The duality issue after flux quantization.} However, this may seem to leave a puzzle. Since flux quantization applies to the pre-metric flux densities \eqref{PremetricBianchiIdentities}, providing their global higher gauge potentials, it may be unclear how to understand the duality-constraint \eqref{TheHodgeDualityConstraint} after flux quantization: Should it remain a constraint on just the underlying flux densities, or should it somehow be lifted to the gauge potentials, hence to the higher differential cocycles (as has been proposed in the case of RR-field fluxes quantized in K-theory)?

\medskip

\noindent
{\bf Flux quantization on phase space.}
In \cite{SS23FQ} we have observed that this issue goes away {\it on phase space} (cf. \cite[\S 2.5]{SS24Flux}). After pulling back the flux densities $\vec F$ \eqref{TheFluxDensities} to any Cauchy hypersurface (a ``spatial slice'') of spacetime (assumed to be globally hyperbolic), they become initial value data $\vec B$ with half of them playing the role of {\it independent canonical momenta}, while the duality constraint \eqref{TheHodgeDualityConstraint} ceases to be a relation among the $\vec B$ and instead controls their evolution away from the Cauchy surface.
This naturally suggests that in the ``canonical'' phase space perspective on the higher gauge theory, flux quantization \eqref{TheGaugePotentials} of the pre-metric flux densities gives the complete specification of the higher gauge fields.
That is, it would not be subjected to further duality constraints; and inspection shows \cite[\S 3.1, 3.2]{SS23FQ} that this is tacitly how basic examples are handled in the literature.

\smallskip

Nevertheless, here we intend to go one step further and resolve the flux-quantized duality issue on spacetime  itself, or rather on super-spacetime.

\subsection{Duality-Symmetry on Super-Spacetime}
\label{DualitySymmetryOnSuperSpacetime}

\noindent
{\bf Passage to super-spacetime.}
Most of the higher gauge theories \eqref{PremetricBianchiIdentities} of interest (notably of the  B-field, RR-field,
and C-field, for pointers see \cite[\S 2.4]{SS24Flux}) arise in the bosonic sector of higher-dimensional supergravity theories. 
Yet their global properties are traditionally discussed with disregard for the fermionic content of these theories, 
treating it as if just a tedious afterthought (however, cf. \cite{EvslinSati03}). This perspective seems convenient but
goes against the conceptual grain of supergravity theory, which is arguably all controlled by phenomena in the 
fermionic sector (cf. Rem. \ref{FormOfTheSuperTorsionConstraint} below). In fact, it is well-known (reviewed in \S\ref{Supergravity}) that supergravity theories have a slick formulation ``in superspace'', namely as phenomena of
super Cartan geometry (\S\ref{SuperCartanGeometry}), where the entire physics is essentially a consequence of just adjoining fermions to the basic rules of Cartan geometry.

\medskip

\noindent
{\bf Duality-symmetric super-flux and 11d SuGra.}
In this vein, the basic observation to be highlighted here is:

\smallskip 
\begin{itemize}[leftmargin=.75cm]

\item[\bf (i)] the pre-metric duality-symmetric formulation of the C-field flux in 11d supergravity exists on super-spacetime,

\item[\bf (ii)] where, remarkably, it implies/absorbs the duality constraint \eqref{TheHodgeDualityConstraint}, so that on super-spacetime the  higher gauge theory of the C-field is purely of the pre-metric form \eqref{PremetricBianchiIdentities},

\item[\bf (iii)] to which flux quantization \eqref{TheGaugePotentials} may be applied, yielding, for the first time, candidates for the full field content of 11d supergravity (\S\ref{SuperFluxQuantization}).
\end{itemize}

\smallskip

The first two points involve observing that demanding
the super C-field flux-densities $(G^s_4, G^s_7)$ to have an expansion in terms of super-coframe fields (Def. \ref{SuperSpacetime}) of the following form:
\begin{equation}
  \label{SuperFluxDensitiesInIntroduction}
  \def\arraystretch{1.5}
  \begin{array}{l}
    G^s_4
    \;\;:=\;\;
    \tfrac{1}{4!} (G_4)_{a_1 \cdots a_4}
    e^{a_1} \cdots e^{a_4}
    \;\,+\;\,
    \tfrac{1}{2}
    \big(\,
    \overline{\psi}
    \Gamma_{a_1 a_2}
    \psi
    \big)
    e^{a_1} \, e^{a_2}
    \\
    G^s_7
    \;\;:=\;\; 
    \tfrac{1}{7!}
    (G_7)_{a_1 \cdots a_7}
    e^{a_1} \cdots e^{a_7}
    \;\,+\;\,
    \tfrac{1}{5!}
    \big(\,
    \overline{\psi}
    \Gamma_{a_1 \cdots a_5}
    \psi
    \big)
    e^{a_1} \cdots e^{a_5}
  \end{array}
\end{equation}
while satisfying the pre-metric form \eqref{PremetricCFieldFlux} of the C-field Bianchi identities, now on super-spacetime
\begin{equation}
  \label{SuperCFieldBianchiInIntro}
  \def\arraystretch{1.2}
  \begin{array}{l}
    \differential
    \, 
    G^s_4
    \;=\; 0
    \\
    \differential
    \, 
    G^s_7
    \;=\;
    \tfrac{1}{2}
    G^s_4 \wedge G^s_4
    \,,
  \end{array}
\end{equation}
already {\it implies}  (\cite[p. 878]{CDF91}, cf. Lem. \ref{SuperBianchiIdentityForG7InComponents} below) the Hodge duality constraint on the bosonic component:
\begin{equation}
  \label{HodgeDualityImplied}
  \mbox{
    \eqref{SuperCFieldBianchiInIntro}
  }
  \;\;\;\;
  \Rightarrow
  \;\;\;\;
  (G_7)_{a_1 \cdots a_7}
  \;\;
  =
  \;\;
  \tfrac{1}{4!}
  \epsilon_{a_1 \cdots a_7 b_1 \cdots b_4}
  (G_4)^{b_1 \cdots b_4}.
\end{equation}

\newpage 
\noindent In fact, that is equivalent to the super-spacetime solving the equations of motion of 11d SuGra with the given $G_4$-flux source:
\medskip 
\begin{equation}
  \label{MainTheoremInIntroduction}
 \colorbox{lightgray}
  {
     \def\arraystretch{1.2}
    \begin{tabular}{c}
      $(11\vert\mathbf{32})$-dimensional 
      super-spacetimes
      $\big(
        X, (e, \psi, \omega)
      \big)$
      \\
      carrying super-flux 
      $(G_4^s, G_7^s)$
      (from \eqref{SuperFluxDensitiesInIntroduction}),
      \\
      satisfying its Bianchi identity
      (from \eqref{SuperCFieldBianchiInIntro}).
    \end{tabular}
  }
  \hspace{.9cm}
  \Leftrightarrow
  \hspace{.9cm}
  \colorbox{lightgray}{
    \begin{tabular}{c}
      Solutions of 11d SuGra
      \\
      with flux source $G_4$
      \\
      and dual flux $G_7$.
    \end{tabular}
    }
\end{equation}

\medskip 
\noindent This is our Thm. \ref{11dSugraEoMFromSuperFluxBianchiIdentity} below.\footnote{We re-amplify that \eqref{MainTheoremInIntroduction} holds only when demanding that the super-fluxes $(G_4^s, G_7^s)$, defined on super-spacetime, satisfy the Bianchi identities on the \textit{full super-spacetime}. It follows that their restrictions to the underlying spacetime again satisfy the Bianchi identities, but this purely spacetime condition is not sufficiently strong to imply the Hodge duality constraint and the rest of the 11d SuGra equations of motion.} It is in spirit and in computational detail close to the result of \cite{Howe97}\cite[\S 2]{CGNT05} that for diligent choices of spin connection $\omega$ the super-torsion constraint (which we assume as part of the definition of super-spacetimes \eqref{TorsionConstraint}) already enforces the equations of motion of 11d SuGra, with the flux density being a derived quantity from this perspective (cf. \cite[(7)]{Howe97} following \cite[(11)]{CF80}). In our formulation instead the flux density is the primary datum, so as to prepare the stage for flux-quantization and hence for the completion of the theory.

\medskip 
In particular, the implication in \eqref{HodgeDualityImplied} means that on super-spacetime the Hodge duality constraint \eqref{TheHodgeDualityConstraint} is entirely absorbed into the pre-metric Bianchi identities \eqref{PremetricBianchiIdentities}!  
Hence the presence of the Hodge duality constraint in higher gauge theory, and the above-mentioned problems that it brings with it appear just as an artifact of disregarding the natural superspace context of the theory (at least for the case of the C-field in 11d SuGra): 
\begin{quote}
\colorbox{lightblue}{ \emph{ Flux quantization of higher gauge theory fundamentally ought to be applied on super-spacetime. }}
\end{quote}

\noindent
{\bf Relation to the literature.}
The computations (in \S\ref{Supergravity}) behind the statement \eqref{MainTheoremInIntroduction} will not be new to experts (though no substantial details seem to have previously been published). In particular the  super-form of $G_4^s$ in \S\ref{SuperFluxDensitiesInIntroduction} is classical (\cite[p. 411]{CJS78}),
and that a super-form $G_7^s$ \eqref{SuperFluxDensitiesInIntroduction} satisfying \eqref{SuperCFieldBianchiInIntro} exists on on-shell 
super-spacetime is almost explicit in \cite[\S III.8.3-5 \& p. 878]{CDF91} and \cite[(6.9)]{CL94}. 
Nevertheless, since the tradition in the supergravity literature to indicate computations only in broad outline can make it hard to see which precise claims are being made on which precise assumptions, we use the occasion in \S\ref{Supergravity} to present a complete derivation of 11d supergravity from imposing just the C-field super-flux Bianchi identities \eqref{SuperCFieldBianchiInIntro} on an 11d super-spacetime. Besides serving as an exposition of the required computations, this presents 11d SuGra in a somewhat novel form adapted to its completion by flux quantization, which has previously received little to no attention.

\smallskip 
However,  the primacy we assign to the flux Bianchi identities has a technical impact also on these classical computations.
We find that imposing the duality-symmetric super-flux Bianchi identity \eqref{SuperCFieldBianchiInIntro} as 
a constraint {\it implies} the vanishing of the $(\psi^2)$-component of the gravitino field strength $\rho$ (cf. Rem. \ref{RoleOfDualitySymmetricBianchiIdentitiesInEnforcingTheSuperTorsionConstraint} below), a crucial constraint which previously has been motivated differently:
\begin{itemize}[leftmargin=.4cm]
\item In \cite[\S III.8.5]{CDF91}, this constraint is motivated as necessary for rheonomic parametrization.

\item In \cite{Howe97}\cite{CGNT05}, the satisfaction of this constraint is shown to be achievable by a careful (gauge) fixing of the frame super-field and spin connection. 
\end{itemize}
\noindent
Here we see instead that both of these moves are consequences of the duality-symmetric formulation of 11d SuGra, not needing to be implemented ``by hand''.

\medskip

\noindent
{\bf Flux-quantization on super-spacetime.} The key impact of our result \eqref{MainTheoremInIntroduction} is that it makes immediate how to proceed with flux quantization of (the C-field in) 11d supergravity, along the lines of \cite{FSS23Char}.
Namely in super homotopy theory \cite[\S 3.1.3]{SS20Orb} (reviewed in \S\ref{HigherSuperGeometry}) the structures on the right of \eqref{TheGaugePotentials} exist verbatim, with the moduli sheaves $\Omega^1_{\mathrm{dR}}(-;\mathfrak{a})_{\mathrm{clsd}}$ of closed $L_\infty$-algebra valued differential forms generalized to super-differential forms (see \S\ref{SuperDifferentialForms}), and  \eqref{MainTheoremInIntroduction} means that the characteristic $L_\infty$-algebra $\mathfrak{l}S^4$ of the C-field  (Ex. \ref{Rational4Sphere} \& \ref{ClosedlS4ValuedDifferentialForms}) still classifies the super-flux densities $(G_4^s, G_7^s)$, so that admissible flux-quantization laws on super-spacetime are still given by classifying spaces $\mathcal{A}$ whose  $\mathbb{R}$-Whitehead $L_\infty$-algebra (Ex. \ref{WhiteheadLInfinityAlgebras}) is that of the 4-sphere.

\newpage

\noindent
{\bf Completion of $D=11$ supergravity.}
The remarkable upshot of all this is the following:
\begin{claim}[{\bf Flux-quantized super-fields of 11d SuGra}]
\label{FluxQuantizedSuperFieldsOf11dSugra}
For 
\begin{itemize}[leftmargin=.5cm]
\item[\bf --]
$\big(X, (e, \psi, \omega)\big)$ an $(11\vert \mathbf{32})$-dimensional super-spacetime (Def. \ref{SuperSpacetime}),

\item[\bf --] $\mathcal{A}$, a choice of flux-quantization law as in \cite[\S 3.2]{SS24Flux} embodied by a classifying 
space with rational Whitehead $L_\infty$-algebra (Ex. \ref{WhiteheadLInfinityAlgebras}) that of the 4-sphere,
\end{itemize}
the full flux-quantized super-C-field histories on $X$ are diagrams in
super-homotopy theory (Def. \ref{SmoothSuperInfinityGroupoids}) of the following form:

\vspace{-4mm}
\begin{equation}
  \label{FluxQuantizedSugraFields}
  \hspace{-5mm} 
  \begin{tikzcd}[
    row sep=65pt, 
    column sep=60pt
  ]
    &[-20pt]
    &&
    &[+5pt]
    \overset{
      \mathclap{
      \scalebox{.7}{
        \def\arraystretch{.9}
        \begin{tabular}{c}
          \color{darkblue}
          \bf
          classifying space
          \\
          \rm
          (Ex. \ref{PathInfinityGroupoids})
          \\
          \color{darkblue}
          \bf
          of quantized
          \\
          \color{blue}
          \bf
          M-brane charges
        \end{tabular}
      }
      \mathrlap{
        \scalebox{.7}{
          \color{gray}
          \begin{tabular}{c}
            subject to
            \\
            $\mathfrak{l}\mathcal{A} \,\cong\, \mathfrak{l}S^4$
          \end{tabular}
        }
      }
      }    
    }{
      \mathcal{A}
    }
    \ar[
      d,
      "{
        \mathbf{ch}^{\mathcal{A}}
      }"{swap},
      "{
        \scalebox{.7}{
          \def\arraystretch{.9}
          \begin{tabular}{c}
            \color{darkgreen}
            \bf
            differential
            \\
            \color{darkgreen}
            \bf
            character
            \\
            \rm
            \cite[Def. 9.2]{FSS23Char}
          \end{tabular}
        }
      }"{xshift=-6pt}
    ]
    \\
    \underset{
      \mathclap{
        \raisebox{-2pt}{
      \scalebox{.7}{
        \def\tabcolsep{-2pt}
        \def\arraystretch{.9}
        \begin{tabular}{c}
          \color{gray}
          \bf
          ordinary
          \\
          \color{gray}
          \bf
          spacetime
          \\
          \rm
          (Ex. \ref{BosonicBodyOfSupermanifold})
        \end{tabular}
      }            
        }
      }
    }{
      \mathcolor{gray}
      {\bos{X}}
    }
    \ar[
      r,
      hook,
      gray,
      "{
        \eta^{\rightsquigarrow}_X
      }"
    ]
    &
    \underset{
      \mathclap{
        \raisebox{-5pt}{
      \scalebox{.7}{
        \def\tabcolsep{-2pt}
        \def\arraystretch{.9}
        \begin{tabular}{c}
          \color{darkblue}
          \bf
          super-
          \\
          \color{darkblue}
          \bf
          spacetime
          \\
          \rm
          (Def. \ref{SuperSpacetime}
          \\
          \rm
          Ex. \ref{CechResolutionOfSupermanifold})
        \end{tabular}
      }            
        }
      }
    }{
      X
    }
    \ar[
      rr,
      "{
        (G_4^s,G_7^s)
      }"{name=t},
      "{
        \scalebox{.7}{
          \begin{tabular}{c}
            \color{darkgreen}
            \bf
            super-C-field flux            \eqref{SuperFluxDensitiesInIntroduction}
            \\
            \rm
            (Ex. \ref{ClosedLInfinityValuedDifferentialForms})
          \end{tabular}
        }
      }"{swap}
    ]
    \ar[
      urrr,
      bend left=17,
      "{
        \scalebox{.7}{
          \begin{tabular}{c}
            \rm
            (Ex. \ref{NonabelianCohomology})
            \\
            \color{darkgreen}
            \bf
            local C-field charge
          \end{tabular}
        }
      }"{sloped},
      "{
        \rchi
      }"{sloped, pos=.48,swap, name=s}
    ]
    \ar[
      from=s,
      to=t,
      Rightarrow,
      "{
        (
          \widehat{C}{}^s_3, 
          \widehat{C}{}^s_6
        )
        \mathrlap{
          \scalebox{.7}{
            \color{darkorange}
            \bf
            global super C-field gauge potentials
          }
        }
      }"{description}
    ]
    &&
    \underset{
      \mathclap{
        \raisebox{-6pt}{
      \scalebox{.7}{
        \def\tabcolsep{-2pt}
        \def\arraystretch{.9}
        \begin{tabular}{c}
          \color{darkblue}
          \bf
          Smooth super-set of
          \\
          \color{darkblue}
          \bf
          duality-symmetric
          \\
          \color{blue}
          \bf
          C-field flux densities
          \\
          \rm
          (Ex. \ref{ClosedlS4ValuedDifferentialForms})
        \end{tabular}
      }            
        }
      }
    }{
    \Omega^1_{\mathrm{dR}}(
      -;
      \mathfrak{l}S^4
    )_{\mathrm{clsd}}
    }
    \ar[
      r,
      "{
        \eta^{\,\scalebox{.5}{$\shape$}}
      }",
      "{
        \scalebox{.7}{
          \def\arraystretch{.9}
          \begin{tabular}{c}
            \color{darkgreen}
            \bf
            up to 
            \\
            \color{darkgreen}
            \bf
            deformations
            \\
            \rm
            \eqref{ShapeUnitForDeformationModuli}
          \end{tabular}
        }
      }"{swap}
    ]
    &
    \underset{
      \mathclap{
        \raisebox{-6pt}{
      \scalebox{.7}{
        \def\tabcolsep{-2pt}
        \def\arraystretch{.9}
        \begin{tabular}{c}
          \color{darkblue}
          \bf
          their deformation
          \\
          \color{darkblue}
          \bf
          $\infty$-groupoid
          \\
          \rm
          (Ex. \ref{InfinityGroupoidOfFluxDeformations})
        \end{tabular}
      }            
        }
      }
    }{
    \shape
    \,
    \Omega^1_{\mathrm{dR}}(
      -;
      \mathfrak{l}S^4
    )_{\mathrm{clsd}}
    }
  \end{tikzcd}
\end{equation}
where the bottom part exists, by Thm. \ref{11dSugraEoMFromSuperFluxBianchiIdentity}, if and only if $\big(X,(e,\psi,\omega)\big)$ solves the equations of motion of 11d SuGra for the given flux density $G_4$ with $G_7$ its dual.
\end{claim}

Before expanding on the implications of Claim \ref{FluxQuantizedSuperFieldsOf11dSugra}, we notice that it is backward-compatible with the traditional notion of C-field gauge potential, where it applies:

\begin{proposition}[\bf Recovering traditional super-C-field gauge potentials]
\label{RecoveringTraditionalSuperCFieldGaugePotentials}
 If the total C-field charge in Diagram \eqref{TotalChargeQuantization}
 vanishes, $[\rchi] = 0$ (as happens over any coordinate chart), such that the local charge equivalently factors through the point 
 \vspace{-3mm} 
 \begin{equation}
  \label{DiagramForTrivialCharge}
  \begin{tikzcd}[
    row sep=40pt, column sep=huge
  ]
    &[+10pt]
    \ast
    \ar[r, hook]
    &[-20pt]
    \mathcal{A}
    \ar[
      d,
      "{
        \mathbf{ch}  ^{\mathcal{A}}
      }"
    ]
    \\
    X
    \ar[
      r,
      "{
        (G_4^s, G_7^s)
      }"{name=t}
    ]
    \ar[
      ur,
      bend left=30,
      "{
        \scalebox{.7}{
          \color{darkgreen}
          \bf
          \def\arraystretch{.9}
          \begin{tabular}{c}
          trivial 
          \\
          charge
          \end{tabular}
        }
      }"{sloped},
      "{
        0
      }"{swap, name=s}
    ]
    \ar[
      from=s,
      to=t,
      Rightarrow,
      "{\color{orangeii} \bf 
        (C^s_3, C^s_6)
      }"{pos=.2}
    ]
    &
    \Omega^1_{\mathrm{dR}}(-;\mathfrak{l}S^4)_{\mathrm{clsd}}
    \ar[
      r,
      "{
        \eta^{\,\scalebox{.7}{$\shape$}}
      }"
    ]
    &
    \shape
    \,
    \Omega^1_{\mathrm{dR}}(-;\mathfrak{l}S^4)_{\mathrm{clsd}}
  \end{tikzcd}
 \end{equation}
 then the C-field gauge potentials according to Claim \ref{FluxQuantizedSuperFieldsOf11dSugra}
 correspond to super-differential forms
 \begin{equation}
   \label{GlobalGaugePotentials}
   \mathllap{
   \scalebox{.7}{
    \color{darkblue}
    \bf
    \begin{tabular}{c}
      ordinary 
      gauge potentials
      \\
      for the C-field
    \end{tabular}
    }
  }
  \left.
  \def\arraystretch{1.3}
  \begin{array}{ll}
    C_3^s 
    \in
    \Omega^3_{\mathrm{dR}}(X)
    \\
    C_6^s 
    \in
    \Omega^6_{\mathrm{dR}}(X)
  \end{array}
 \!\! \right\}
  \;\;
  \mbox{s.t.}
  \;\;
  \left\{\!\!\!
  \def\arraystretch{1.3}
  \begin{array}{ll}
  \mathrm{d}\, C^s_3 =
  G_4^s
  \,,
  \\
  \mathrm{d}\, C_6^s =
  G_7^s - \tfrac{1}{2}
  C_3^s \, G_4^s
  \,,
  \end{array}
  \right.
 \end{equation}
 as traditionally considered in the literature {\rm(e.g. \cite[(6.7), (6.11)]{CL94}\cite[(III.8.32d,e)]{CDF91})}.
\end{proposition}
\begin{proof}
  By Ex. \ref{InfinityGroupoidOfFluxDeformations} below, the homotopies in \eqref{DiagramForTrivialCharge} corresponds to coboundaries for $\big(G^s_4, G^s_7\big) \,\in\, \Omega^1_{\mathrm{dR}}(X;\, \mathfrak{l}S^4)$ in $\mathfrak{l}S^4$-valued de Rham cohomology (Def. \ref{NonabelianDeRhamCohomology}), and by Prop. \ref{CoboundariesOfClosedLS4ValuedForms} below (see there for details) these correspond to gauge potentials \eqref{GlobalGaugePotentials} as claimed.
\end{proof}
The statement of Prop. \ref{RecoveringTraditionalSuperCFieldGaugePotentials} motivates and justifies the notation $\big(\widehat{C}{}^s_3,\, \widehat{C}{}^s_6\big)$ in \eqref{FluxQuantizedSugraFields} for globally defined gauge potentials (following traditional such hat-notation for lifts of differential forms to cocycles in differential cohomology).

\smallskip 
Indeed, on topologically non-trivial (super-)spacetimes, diagram \eqref{FluxQuantizedSugraFields} gives 
{\it new global field content}, which enhances the chart-wise data \eqref{GlobalGaugePotentials} by topological structure reflecting individual solitonic (brane-)charges that may source the flux-quantized C-field --- in generalization of how individual Dirac monopoles and Abrikosov vortex strings are imprinted in the electromagnetic flux density (cf. \cite[\S 2]{SS24Flux}) once Dirac charge quantization is imposed. Some implications of such C-field flux quantization are surveyed in \cite[\S 4]{SS24Flux}; for more see  \cite{SS20Tadpole}\cite{FSS20TwistedString}\cite{SS21M5Anomaly} \cite{SS23MF}. 

\newpage 
In these previous discussions, however, it was left open whether the conclusions drawn are all subject to a pending imposition of a duality constraint \eqref{DualityConstraintOnCField}, lifted somehow from flux densities to global gauge potentials on spacetime. By Claim \ref{FluxQuantizedSuperFieldsOf11dSugra} this is not the case, and duality-symmetric flux-quantized super-fields as in \eqref{SuperFluxDensitiesInIntroduction} constitute the full global field content of 11d SuGra on (super-)spacetime, and thereby also on ordinary spacetime:

\begin{proposition}[\bf Bosonic spacetime flux quantization implied by super-flux quantization]
  \label{BosonicSpacetimFluxQuantizationImplied}
  $\,$
  
\noindent   Given $\mathcal{A}$-flux-quantized super-C-fields 
  $\big(\widehat{C}{}^s_3,\, \widehat{C}{}^s_6\big)$ on super-spacetime \eqref{SuperFluxDensitiesInIntroduction},
  their restriction 
  $
    \big(\eta^\rightsquigarrow_X\big)^{\! \ast} 
    \big(\widehat{C}{}^s_3,\, \widehat{C}{}^s_6\big)
  $
  to ordinary bosonic spacetime 
  (Ex. \ref{BosonicBodyOfSupermanifold}) is then itself flux-quantized in differential $\mathcal{A}$-cohomology.
\end{proposition}
\begin{proof}
  This is immediate from the diagrammatic definition \eqref{SuperFluxDensitiesInIntroduction} of the flux-quantized fields, since the pullback operation $\big(\eta^{\rightsquigarrow}_X\big)^\ast$ corresponds
  (Remark \ref{PullbackOfDifferentialFormsViaClassifyingSuperSets})
  just to the precomposition of the diagram with $\bos X \xrightarrow{ \eta^\rightsquigarrow_X } X$, as indicated on the left of \eqref{SuperFluxDensitiesInIntroduction}.
\end{proof}
  Hence Prop. \ref{BosonicSpacetimFluxQuantizationImplied} solves the problem of flux-quantizing C-field histories on all of spacetime (instead of just on a Cauchy surface as in \cite{SS23FQ}), by detour through super-spacetime.

\begin{remark}[\bf Super-fields restricted to ordinary spacetime]
$\,$

\noindent {\bf (i)} 
The pullback of super-fields to ordinary spacetime $\bos X$, as invoked in Prop. \ref{BosonicSpacetimFluxQuantizationImplied}, is the operation which on a super-coordinate chart $U$ looks as follows (cf. Rem. \ref{FrameAndCoordinateIndices} for our coordinate-index notation):
\vspace{-2mm} 
\begin{equation}
  \label{PurelyBosonicFluxDensitiesInIntroduction}
  \def\arraystretch{1.6}
  \begin{array}{l}
  \big(\eta^{\rightsquigarrow}_U\big)^\ast
  G^s_4
  \;
  =
  \;
  \tfrac{1}{4!}(G_4)_{\evencoordinateindex_1 \cdots \evencoordinateindex_4}\big\vert_{\theta^\rho=0} \cdot  \mathrm{d}x^{\evencoordinateindex_1} \cdots \mathrm{d}x^{\evencoordinateindex_4}
  ,
  \\
  \big(\eta^{\rightsquigarrow}_U\big)^\ast
  G^s_7
  \;
  =
  \;
  \tfrac{1}{7!}(G_7)_{\evencoordinateindex_1 \cdots \evencoordinateindex_7} \big\vert_{\theta^\rho=0} \cdot \mathrm{d}x^{\evencoordinateindex_1} \cdots \mathrm{d}x^{\evencoordinateindex_7},
  \end{array}
\end{equation}
hence which discards all fermionic contributions to the super-flux density.

\noindent {\bf (ii)}  More generally, there are such restrictions ``at non-trivial Grassmann stage''
(see \S\ref{SuperSmoothFieldSpaces} and specifically Ex.  \ref{FermionicGravitinoFieldSpace} below for how this works),
where the classical gravitino field $\psi$ is retained on ordinary spacetime, 
whose contribution to the flux density  
is then chart- and plot-wise of the following form (reproducing the usual formula in the literature, e.g. \cite[(4)]{CF80}):
\vspace{-1mm} 
\begin{equation}
  \label{OddFluxOnBosonicSupermanifold}
  \def\arraystretch{1.9}
  \begin{array}{l}
  \Big(
  \tfrac{1}{4!}(G_4)_{\evencoordinateindex_1 \cdots \evencoordinateindex_4} \big\vert_{\theta^\rho=0} 
  \;+\;
  \tfrac{1}{2}\big(\,
   \overline{\psi}_{\evencoordinateindex_1}
   \,\Gamma_{\evencoordinateindex_2 \evencoordinateindex_3}\,
   \psi_{\evencoordinateindex_4}
  \big) \big\vert_{\theta^\rho=0}
  \Big)
  \mathrm{d}x^{\evencoordinateindex_1} \cdots \mathrm{d}x^{\evencoordinateindex_4},
  \\
  \Big(
  \tfrac{1}{7!}(G_7)_{\evencoordinateindex_1 \cdots \evencoordinateindex_7} \big\vert_{\theta^\rho=0}
  \;+\;
  \tfrac{1}{5!}\big(\,
   \overline{\psi}_{\evencoordinateindex_1}
   \,\Gamma_{\evencoordinateindex_2 \cdots \evencoordinateindex_6}\,
   \psi_{\evencoordinateindex_7}
  \big) \big\vert_{\theta^\rho=0}
  \Big)
  \mathrm{d}x^{\evencoordinateindex_1} \cdots \mathrm{d}x^{\evencoordinateindex_7}.
  \end{array}
\end{equation}
\end{remark}

\noindent
{\bf The Role of Super-Flux in Supergravity.}
We may notice a curious principle, apparently underlying the form of the super-flux densities $(G_4^s, G_7^s)$ (see \eqref{SuperFluxDensitiesInIntroduction}):

\smallskip 
\begin{itemize}[leftmargin=.4cm]
\item On the one hand we have purely bosonic flux densities $(G_4, G_7)$ on topologically non-trivial manifolds subject to a Bianchi identity \eqref{PremetricCFieldFlux}.
\item 
On the other hand, we have purely super-geometric objects 
$\big( \tfrac{1}{2}\big(\,\overline{\psi}\,\Gamma_{a_1 a_2}\, \psi\big), \, \tfrac{1}{5!}\big(\,\overline{\psi}\ \Gamma_{a_1 \cdots a_5}\, \psi\big)  \big)$ on super-Minkowski spacetime, satisfying {\it the same} form Bianchi identity \eqref{TheSupercocycles}.
\end{itemize}

\vspace{1mm} 
\noindent Faced with this situation, a natural question is how these two algebraically similar structures from different sectors of mathematics relate. One way of reading Thm. \ref{11dSugraEoMFromSuperFluxBianchiIdentity} is as answering this by saying that on-shell 11d supergravity is precisely the result of {\it unifying} these two structures:
\smallskip 
$$
  \hspace{-1.5mm}
  \adjustbox{
    scale=.83
  }{
  \begin{tikzcd}[
    column sep=-7.5cm,
    row sep=.7cm
  ]
    \adjustbox{fbox}{
    \def\tabcolsep{-5pt}
    \begin{tabular}{l}
      Flux densities
      on bosonic but
      {\color{darkgreen} \bf curved} manifolds,
      \\
      satisfying:
    \\
    $
    \def\arraystretch{1.4}
    \begin{array}{l}
      \mathrm{d}
      \,
      \Big({
        \color{darkblue}
        \tfrac{1}{4!}
        (G_4)_{a_1 \cdots a_4}
        \,
        e^{a_1} \cdots e^{a_4}
      }\Big)
      \;=\; 0
      \\
      \mathrm{d}
      \Big({
        \color{darkblue}
        \tfrac{1}{7!}
        (G_7)_{a_1 \cdots a_7}
        \,
        e^{a_1} \cdots e^{a_7}
      }\Big)
      \;=\;
      \tfrac{1}{2}
      \Big({
        \color{darkblue}
        \tfrac{1}{4!}
        (G_4)_{a_1 \cdots a_4}
        \,
        e^{a_1} \cdots e^{a_4}
      }\Big)^2
    \end{array}
    $
    \end{tabular}
    }
    \ar[
      dr,
      hook
    ]
    &&
    \adjustbox{fbox}{
    \def\tabcolsep{-6pt}
    \begin{tabular}{l}
      Supersymmetric forms
      on flat but {\color{olive} \bf super} spacetime,
      \\
      satisfying: \,
    \\
    $
    \def\arraystretch{1.4}
    \begin{array}{l}
      \mathrm{d}
      \,
      \Big({
       \color{darkorange}
      \tfrac{1}{2}
      \big(\,
        \overline{\psi}
        \,\Gamma_{a_1 a_2}\,
      \big)
      e^{a_1} e^{a_2}
      }\Big)
      \;=\; 0
      \\
      \mathrm{d}
      \Big({
      \color{darkorange}
      \tfrac{1}{5!}
      \big(\,
        \overline{\psi}
        \,\Gamma_{a_1 \cdots a_5}\,
      \big)
      e^{a_1} \cdots e^{a_5}
      }\Big)
      \;=\;
      \tfrac{1}{2}
      \Big({
      \color{darkorange}
      \tfrac{1}{2}
      \big(\,
        \overline{\psi}
        \,\Gamma_{a_1 a_2}\,
      \big)
      e^{a_1} e^{a_2}
      }\Big)^2
    \end{array}
    $
    \end{tabular}
    }
    \ar[
     dl,
     hook'
    ]
    \\
    &
    \adjustbox{fbox}{
    \def\tabcolsep{-5pt}
    \begin{tabular}{l}
      Locally supersymmetric
      flux densities on
       {\color{darkgreen} \bf curved} 
        {\color{olive} \bf super}manifolds,
      \\
      satisfying:
    \\
    $
    \def\arraystretch{1.7}
    \begin{array}{l}
      \mathrm{d}
      \,
      \Big(
      {
      \color{darkblue}
      \tfrac{1}{4!}
      (G_4)_{a_1 \cdots a_4}
      e^{a_1} \cdots e^{a_4}
      }
      \;+\;
      {
      \color{darkorange}
      \tfrac{1}{2}
      \big(\,
        \overline{\psi}
        \,\Gamma_{a_1 a_2}\,
      \big)
      e^{a_1} e^{a_2}
      }
      \Big)
      \;=\; 0
      \\
      \mathrm{d}
      \,
      \Big(
      {
      \color{darkblue}
      \tfrac{1}{7!}
      (G_7)_{a_1 \cdots a_7}
      \,
      e^{a_1} \cdots e^{a_7}
      }
      \;+\;
      {
      \color{darkorange}
      \tfrac{1}{5!}
      \big(\,
        \overline{\psi}
        \,\Gamma_{a_1 \cdots a_5}\,
      \big)
      e^{a_1} \cdots e^{a_5}
      }
      \Big)
      \;=\;
      \tfrac{1}{2}
      \Big(
      {
      \color{darkblue}
      \tfrac{1}{4!}
      (G_4)_{a_1 \cdots a_4}
      e^{a_1} \cdots e^{a_4}
      }
      \;+\;
      {
      \color{darkorange}
      \tfrac{1}{2}
      \big(\,
        \overline{\psi}
        \,\Gamma_{a_1 a_2}\,
      \big)
      e^{a_1} e^{a_2}
      }
      \Big)^2
    \end{array}
    $
    \\
    \vspaceabove
    thereby enforcing the 
    equations of motion of 11d supergravity.
    \end{tabular}
    }
  \end{tikzcd}
  }
$$

\newpage 
In this unification, the two summands (indicated in blue and in orange) separately still satisfy their Bianchi identities, but in addition now a plethora of mixed terms potentially appear (from non-trivial curvature/torsion, but also from the non-linearity of the Bianchi identity) 
whose vanishing, remarkably, is equivalently the equations of motion of 11d SuGra.

\medskip 
It is amusing to consider this in the special case of vanishing $G_4$ flux: Here it says that demanding the residual superforms 
$\tfrac{1}{2}\big(\,\overline{\psi} \,\Gamma_{a_1 a_2}\, \psi\big)\, e^{a_1} \, e^{a_2}$ and 
$\tfrac{1}{5!}\big(\,\overline{\psi} \,\Gamma_{a_1 \cdots a_5}\, \psi\big)\, e^{a_1} \cdots e^{a_5}$ on a 
curved super-spacetime to still satisfy their Bianchi identity as on super-Minkowski spacetime is equivalent to the curved super-spacetime satisfying the source-free Einstein-Rarita-Schwinger equation.

\medskip

\noindent
{\bf Outlook: Super-Exceptional Geometric Supergravity.}
This suggests that exotic forms of supergravity may be discovered by similarly generalizing supersymmetric relations 
found on generalized super-Minkowski spacetimes to curved generalized super-spacetimes.
Notably there are ``super-exceptional geometric'' enhancements of 11d super-Minkowski spacetime (\cite[\S 3]{FSS20Exc},
the ``hidden supergroup'' of \cite[\S 6]{DF82}\cite{BDIPV04}\cite{ADR16}):
\begin{equation}
  \label{SuperExceptionalMinkowskiSpacetime}
  \begin{tikzcd}
    \mathbb{R}^{1,10\vert\mathbf{32}}_{\mathrm{ex}, s}
    \ar[
      r,
      ->>,
      "{
        \phi^0
      }"
    ]
    &
    \mathbb{R}^{1,10\vert\mathbf{32}}
  \end{tikzcd}
\end{equation}
(indexed by a parameter $s \in \mathbb{R} \setminus \{0\}$) whose bosonic body is (independent of $s$) the exceptional ``generalized tangent bundle'' expected in M-theory \cite{Hull07}
$$
\overset{\longsquiggly}{
    \mathbb{R}^{1,10\vert\mathbf{10}}_{\mathrm{ex}, s}
}
  \;\;\cong\;\;
  \mathbb{R}^{1,10}
  \times
  \Lambda^2\big(
    \mathbb{R}^{1,10\vert \mathbf{32}}
  \big)^\ast
  \times
  \Lambda^5\big(
    \mathbb{R}^{1,10\vert \mathbf{32}}
  \big)^\ast
  ,
$$
while its further fermionic structure has the curious property that it admits the construction of a supersymmetric form $H_3^{0} \in \Omega^3_{\mathrm{dR}}\big( \mathbb{R}^{1,10\vert \mathbf{32}}_{\mathrm{ex},s}\big)$ which trivializes the pullback of the above 
supersymmetric 4-form on super-Minkowski spacetime:
$$
  \mathrm{d}
  \,
  \Big(
  \underbrace{
    \color{darkorange}
    \alpha_0(s)
    \,
    e_{a_1 a_2}
    \, 
    e^{a_1}\, e^{a_2}
    +
    \cdots
    }_{H_3^0}
  \Big)
  \,=\,
  (\phi^0)^\ast
  \Big(
   \underbrace{
     \color{darkorange}
      \tfrac{1}{2}
      \big(\,
        \overline{\psi}
        \,\Gamma_{a_1 a_2}\,
        \psi
      \big)
      e^{a_1}\, e^{a_2}
    }_{
    \mathclap{
      G_4^0
    }
  }
  \Big)
  \,.
$$
But this Bianchi identity of supersymmetric forms on (generalized) super-Minkowski spacetimes is of just the same algebraic structure as the Bianchi identity
of ordinary flux densities on ordinary but curved spacetimes for the case of the B-field flux $H_3$ 
(cf. \cite[\S 4.3]{SS24Flux})
on the extended worldvolume $\phi \,:\, \Sigma^{6+1} \to X$ of M5-branes, which is, chartwise:
$$
  \mathrm{d}
  \, 
  \Big({
    \color{darkblue}
    (H_3)_{a_1 a_2 a_3}
    \,
    e^{a_1} \, e^{a_2}\, e^{a_3}
  }\Big)
  \;=\;
  \phi^\ast 
  \Big({
    \color{darkblue}
    (G_4)_{a_1 \cdots a_4}
    \,
    e^{a_1} \cdots e^{a_4}
  }\Big)
  \,.
$$
The evident analogy with the above situation for plain 11d SuGra suggests that its super-exceptional variant on curved superspacetimes $X_{\mathrm{ex},s}$ locally modelled not on ordinary but on the exceptional super-Minkowski spacetime \eqref{SuperExceptionalMinkowskiSpacetime} ought to be controlled (if not defined) by the following super-flux Bianchi identity:
\smallskip 
\begin{equation}
  \label{SuperH3FLuxBianchiIdentity}
  \hspace{-2mm}
  \mathrm{d}\Big(
    {\color{darkblue}
    \tfrac{1}{3!}
    (H_3)_{a_1 a_2 a_3}
    e^{a_1}\, e^{a_2}\, e^{a_3}
    }
    +
    {
      \color{darkorange}
      \alpha_0(s)
      \,
      e_{a_1 a_2}
      \, e^{a_1} \, e^{a_2}
      +
      \cdots
    }
  \Big)
  \!=
  \big(\phi^s\big)^\ast\Big(
    {
      \color{darkblue}
      \tfrac{1}{4!}
      \, 
      (G_4)_{a_1 \cdots a_4}
      \, 
      e^{a_1} \cdots e^{a_4}
    }
    +
    {
      \color{darkorange}
      \tfrac{1}{2!}
      \big(\,
        \overline{\psi}
        \,\Gamma_{a_1 a_2}\,
        \psi
      \big)
      \, e^{a_1} \, e^{a_2}
    }
  \Big)
\end{equation}

\smallskip 
\noindent where 
$
  \phi^s 
  :
  \Sigma_{\mathrm{ex},s}
  \xrightarrow{\;}
  X_{\mathrm{ex},s}
$
is the super-exceptional embedding 
of an extended super-exceptional M5-brane worldvolume into the super-exceptional spacetime (as considered in the flat and fluxless case in \cite{FSS20Exc}\cite{FSS21-SU2} and hereby generalized to the curved and fluxed case). 
Now as before, the structure of this super-Bianchi identity \eqref{SuperH3FLuxBianchiIdentity} allows to apply super $H_3$-flux quantization and hence impose (level-)quantization of the M5-branes Hopf-WZ/Page-charge term as previously considered on bosonic spacetimes \cite{FSSHopf}\cite[\S 4.3]{SS24Flux}. 

\smallskip 
We hope to discuss this flux-quantized super-exceptional geometric supergravity elsewhere \cite{GSS-Exceptional}, based on the results presented here and extending the computations in \S\ref{Supergravity}.

\vspace{1cm}


\smallskip

\smallskip

\noindent

\newpage

\noindent
{\bf Conventions.} Our conventions are standard in the differential geometry and (super-)gravity literature, but since the computations in \S\ref{Supergravity} depend delicately on a plethora of combinatorial signs and prefactors to conspire appropriately, we make them fully explicit, for the record:

\noindent
\begin{notation}[\bf Algebra conventions]
$\,$
\noindent
\begin{itemize}[leftmargin=.4cm]
\item Our ground field is the real numbers $\mathbb{R}$.

\item We write
$
  \ZTwo 
  \,:=\,
  \mathbb{Z}/2\mathbb{Z}
$
for the prime field of order two, thought of as consisting of the set of elements $\{0,1\}$ equipped with 
\begin{itemize}[leftmargin=.5cm]
\item 
the abelian group operation given by addition in $\mathbb{Z}$ modulo 2,
\item the commutative ring structure given by multiplication in $\mathbb{Z}$ modulo 2.
\end{itemize}
In the context of superalgebra, the elements $0, 1 \in \ZTwo$ indicate ``even'' and ``odd'' degrees, respectively.
\end{itemize}
\end{notation}

\begin{notation}[\bf Tensor conventions]
$\,$
\begin{itemize}[leftmargin=.4cm]
\item
  The Einstein summation convention applies throughout: Given a product of terms indexed by some $i \in I$, with the index of one factor in superscript and the other in subscript, then a sum over $I$ is implied:
  $
    x_i \, y^i
    :=
    \sum_{i \in I} 
    x_i \, y^i
  $.
\item We name super-coordinate/frame indices as follows

\begin{equation}
\adjustbox{scale=.95}{
\begin{tabular}{ll}
\def\arraystretch{1.3}
\def\tabcolsep{5pt}
\begin{tabular}{c|cc}
  \hline
  & 
  {\bf Even}
  &
  {\bf Odd}
  \\
  \hline
\rowcolor{lightgray}  
{\bf Frame}
  & 
  $a \in \{0, \cdots, 10\}$
  &
  $\alpha \in \{1, \cdots, 32\}$ 
  \\
  {\bf Coord}
  & 
  $\evencoordinateindex \in \{0, \cdots, 10\}$
  &
  $\oddcoordinateindex \in \{1, \cdots, 32\}$ 
  \\
  \hline
\end{tabular}
&
\quad so that\quad
$
  \def\arraystretch{1.3}
  \begin{array}{ccccc}
  \mbox{\bf frame-}
  &&
  &
  \mathclap{
    \mbox{\bf coord-differentials}
  }
  \\
  e^a 
  &=&
  e^a_{\evencoordinateindex}
  \,
  \mathrm{d} x^\evencoordinateindex
  &+&
  e^a_{\oddcoordinateindex}
  \,
  \mathrm{d} \theta^\oddcoordinateindex
  \,
  \\
  \psi^\alpha 
  &=&
  \psi^\alpha_{\evencoordinateindex}
  \,
  \mathrm{d} x^\evencoordinateindex
  &+&
  \psi^\alpha_{\oddcoordinateindex}
  \,
  \mathrm{d} \theta^\oddcoordinateindex
  \end{array}
$
\end{tabular}
}
\end{equation}
\item
Our Minkowski metric is the matrix
\begin{equation}
  \label{MinkowskiMetric}
  \big(\eta_{ab}\big)
    _{a,b = 0}
    ^{ D }
  \;\;
    =
  \;\;
  \big(\eta^{ab}\big)
    _{a,b = 0}
    ^{ D }
  \;\;
    :=
  \;\;
  \Big(
    \mathrm{diag}
      (-1, +1, +1, \cdots, +1)
  \Big)_{a,b = 0}^{D}
\end{equation}
\item
  Shifting position of frame indices always refers to contraction with the  Minkowski metric \eqref{MinkowskiMetric}:
  $$
    V^a 
      \;:=\;
    V_b \, \eta^{a b}
    \,,
    \;\;\;\;
    V_a \;=\; V^b \eta_{a b}
    \,.
  $$
\item Skew-symmetrization of indices is denoted by square brackets ($(-1)^{\vert\sigma\vert}$ is sign of the permutation $\sigma$):
$$
  V_{[a_1 \cdots a_p]}
  \;:=\;
  \tfrac{1}{p!}
  \;
  \sum_{
    \mathclap{
      \sigma \in \mathrm{Sym}(n)
    }
  }
  \;
  (-1)^{\vert \sigma \vert}
  V_{ a_{\sigma(1)} \cdots a_{\sigma(p)} }\,.
$$
\item
We normalize the Levi-Civita symbol to \begin{equation}
  \label{NormalizationOfLeviCivitaSymbol}
  \epsilon_{0 1 2 \cdots} 
    \,:=\, 
  +1
  \;\;\;\;\mbox{hence}\;\;\;\;
  \epsilon^{0 1 2 \cdots} 
    \,:=\, 
  -1
  \,.
\end{equation}
\item
We normalize the Kronecker symbol to
$$
  \delta
    ^{a_1 \cdots a_p}
    _{b_1 \cdots b_p}
  \;:=\;
  \delta^{[a_1}_{[b_1}
  \cdots
  \delta^{a_p]}_{b_p]}
  \;=\;
  \delta^{a_1}_{[b_1}
  \cdots
  \delta^{a_p}_{b_p]}
  \;=\;
  \delta^{[a_1}_{b_1}
  \cdots
  \delta^{a_p]}_{b_p}
$$
so that
\begin{equation}
  \label{ContractingKroneckerWithSkewSymmetricTensor}
  V_{
    \color{darkblue}
    a_1 \cdots a_p
  }
  \tensor*
   { \delta }
   {
    ^{ 
       \color{darkblue}
       a_1 \cdots a_p 
    }
    _{b_1 \cdots b_p}
   }
  \;\;
  =
  \;\;
  V_{[b_1 \cdots b_p]}  
  \;\;\;\;
  \mbox{and}
  \;\;\;\;
  \epsilon^{
    {\color{darkblue}  
      c_1 \cdots c_p
    }
    a_1 \cdots a_q
  }
  \,
  \epsilon_{
    {\color{darkblue}
    c_1 \cdots c_p 
    }
    b_1 \cdots b_q
  }
  \;\;
  =
  \;\;
  -
  \,
  p! \cdot q!
  \;
  \delta
    ^{a_1 \cdots a_q}
    _{b_1 \cdots b_q}
  \,.
\end{equation}
\end{itemize}
\end{notation}

\begin{notation}[\bf Clifford algebra conventions]
$\,$
\begin{itemize}[leftmargin=.4cm]
\item Clifford algebra generators $\big(\Gamma_a\big)_{a = 0}^{10}$ are taken to square to the Minkowski metric \eqref{MinkowskiMetric}:
\begin{equation}
  \label{TheCliffordAlgebra}
    \Gamma_a 
    \Gamma_b
    +
    \Gamma_b
    \Gamma_a
    \;=\;
    +
    2 \, \eta_{ab}
    \,.
\end{equation}
\item The linear basis spanning the Clifford algebra is denoted:
\begin{equation}
  \label{CliffordBasisElements}
  \Gamma_{a_1 \cdots a_p}
  \;:=\;
  \Gamma_{[a_1} 
    \cdots 
  \Gamma_{a_p]}
  \;:=\;
  \tfrac{1}{p!}
  \underset{
    \sigma 
  }{\sum}
  (-1)^{\vert\sigma\vert}
  \,
  \Gamma_{a_{\sigma(1)}}
    \cdots
  \Gamma_{a_{\sigma(p)}}
  \,,
\end{equation}
which just means that
$$
  \Gamma_{a_1 \cdots a_p}
  \;=\;
  \left\{\!\!
  \begin{array}{ll}
    \Gamma_{a_1}
    \cdots
    \Gamma_{a_p}
    &
    \mbox{\small if the $a_i$ are pairwise distinct},
    \\
    0 & \mbox{\small otherwise}.
  \end{array}
  \right.
$$
\end{itemize}
For more on spinor algebra see \S\ref{TheSpinRep}.
\end{notation}

Our Clifford conventions agree for instance with \cite[\S 2.5]{MiemiecSchnakenburg06}\cite[p. ii]{VF12}\cite[\S A]{Sezgin23}, but:

\begin{remark}[\bf Alternative conventions]
\label{CliffordConventionsInCDF}
Beware that the Clifford algebra conventions used by \cite[(II.7.1-2)]{CDF91} and other supergravity authors are related to our convention by changing the sign of the metric and multiplying, under the Majorana embedding $\mathbb{R}^{32} \hookrightarrow \mathbb{C}^{32}$, the Clifford generators with the imaginary unit (cf. \cite[\S A.1]{HSS19}):
\begin{equation}
  \eta \;\;=\;
  {\color{purple} - }
  \eta^{\mathrm{CDF}}
  \,,
  \hspace{.5cm}
  \Gamma_a
  \;=\;
  {\color{purple}\mathrm{i}}
  \,
  \Gamma^{\mathrm{CDF}}_a
  .
\end{equation}
Conversely, using this translation all factors of $\mathrm{i}$ appearing in formulas shown in \cite{CDF91} are absorbed and hence do not appear in our formulas. The form of the crucial Fierz identities \eqref{TheQuarticFierzIdentities} is invariant under this transformation.
\end{remark}

\newpage

\section{Super Cartan Geometry}
\label{SuperCartanGeometry}

\noindent
{\bf Gravity as Cartan Geometry.}
Due to quirks of history, what mathematicians call {\it Cartan geometry} (see \cite{Sharpe} \cite[ch. 1]{CapSlovak09}\cite{McKay}) is (see \cite{Catren15} for translation) what physicists refer to by words like ``vielbeins'', ``moving frames'', ``spin connection'', and ``first-order formulation of gravity'' (e.g. \cite[\S I.2-4]{CDF91}\footnote{Beware that the authors of \cite{CDF91} refer to Cartan geometries as ``soft group manifolds'', following \cite{NeemanRegge78}\cite{DFR79}. This terminology is non-standard but well in the spirit of Cartan geometry as the curved generalization of the Kleinian geometry of group coset spaces, cf. \cite{Sharpe}.}). The basic concepts, going back to \cite{Cartan23} (see \cite{Scholz19}), are simple, elegant, and powerful and yet arguably remain underappreciated,
\footnote{
\label{CartanConnection}
For example, there recur conflicting claims in the literature on whether gravity ``is a gauge theory'' or not. But both the similarity and the distinction are clearly brought out by the concept of {\it Cartan connection} for formulating the field content of gravity:
This is indeed like that of a gauge connection (for the Poincar{\'e} group in the case of ordinary gravity), {\it but} crucially subject to the soldering constraint \eqref{CartanProperty} not present for Yang-Mills- or Chern-Simons-type gauge theory (as highlighted for instance in \cite[\S 10]{Krasnov20}). Even though this has eventually been realized \cite[p. 3]{Witten07}, much (if not most) of the literature on 3d gravity still claims its equivalence to  Chern-Simons theory by identifying the connection data on both sides --- thereby ignoring the fact that the Cartan connection on the gravity side (but not the gauge connection on the Chern-Simons side) is constrained, in that its frame form field must be non-degenerate. For a more careful discussion of this point see for instance \cite{CGRS20}.} though this may be changing at the moment.

\smallskip

Cartan geometry had explicitly been introduced in \cite{Cartan23} as an alternative to Riemannian metric geometry for discussing general relativistic gravity,
and these days 
Cartan geometry is {\it de facto} what underlies the familiar ``first-order formulation of gravity'' (e.g. \cite[\S I.4]{CDF91}\cite[\S 4, 5]{Zanelli01}\cite[\S 5]{Fre13a}\cite[\S 3]{Krasnov20}) in terms of Cartan's {\it moving frames} (Def. \ref{SuperSpacetime}), {\it structural equations} (Def. \ref{SuperGravitationalFieldStrengths}) and their {\it Bianchi identities} \eqref{SuperGravitationalBianchiIdentities}, which had been used for this purpose in superspace supergravity since  \cite{Baaklini77-1}\cite{Baaklini77-2}\cite{WZ77}\cite{NeemanRegge78}\cite{GWZ79}\cite{DFR79}\cite{CF80}\cite{BrinkHowe80}\cite{Howe82}\cite{DF82}. Nevertheless the method has often been overlooked, for example \cite{Wise10} pointed out that the seminal work \cite{MM77} is (more) naturally re-cast in terms of Cartan geometry.

\medskip

\noindent
{\bf Super Gravity as Super Cartan Geometry.}
Since it is well-known that the ``first-order'' formulation of gravity via Cartan geometry is inevitable once one considers coupling to {\it fermionic} fields on spacetime (e.g. \cite[\S 5.4.1]{Fre13a}\cite[p. 6]{Krasnov20}), it should not be surprising that the natural conceptual home of {\it super-}gravity is a form of {\it super-Cartan geometry}. More recently, this notion, well-known to supergravity theorists since the 1970s (cf. again \cite{Baaklini77-1}\cite{Baaklini77-2} and the other references from the previous paragraph),
is being appreciated more widely, cf. \cite{Lott90}\cite{Lott01}\cite{EEC12}\cite[p. 7-8]{HuertaSchreiber18}\cite[p. 6-7]{HSS19}\cite[\S 3]{Ratcliffe22}\cite{Eder23}\cite{EHN23}\cite{FR24}.

\medskip 

\noindent
{\bf Fluxed Super Gravity as Higher Super Cartan Geometry.}
Most of these discussions have previously ignored the fact that the (higher-degree) flux densities intrinsic to (higher-dimensional) supergravity theories, hence their (higher) gauge fields -- even though these appear as further superpartners of the gravitino field and as such are intrinsically part of the super-geometry -- are {\it globally not} subject to ordinary (super-)geometry.
Indeed, their (higher) gauge symmetries instead make these be objects in {\it higher} geometry (exposition in \cite{FSS15-Stacky}\cite{Alfonsi24}\cite{Borsten24}\cite{Schreiber24}\cite[\S 3.3]{SS24Flux}), where (super-)manifolds are generalized to smooth (super-){\it $\infty$-groupoids} (``$\infty$-stacks'', cf. \cite[\S 1]{FSS23Char}\cite[\S 2]{SS20Orb} for details and further pointers).

\smallskip 
This issue can be ignored, to some extent, (only) if one focuses on the local description of supergravity on a single contractible (super-)coordinate chart, where global topological effects are invisible. This is the situation tacitly considered in most of the existing literature, but this is not sufficient for discussing the complete flux-quantized field content (as of \S\ref{SuperFluxQuantization}).

\smallskip 
Concretely, diagrams \eqref{TheGaugePotentials} and 
\eqref{SuperFluxDensitiesInIntroduction}
defining (the global gauge potentials) of flux-quantized higher gauge fields crucially involve (higher) homotopies (meaning: higher gauge transformations baked into the geometry) which are not available in ordinary geometry. (This concerns the manifest homotopy filling these diagrams and reflecting the gauge potentials, but it also concerns a plethora of implicit higher homotopies that enter the construction of the charge classifying map $\rchi : X \xrightarrow{\;} \mathcal{A}$ via ``cofibrant resolution'' of spacetime $X$ (Ex. \ref{CechResolutionOfSupermanifold}).

\medskip

Therefore, here we give a quick account of the {\it higher} super Cartan geometry \cite{SS20Orb} (advocated earlier in  \cite{Schreiber15}\cite{Schreiber16}\cite{Cherubini18}) that underlies higher-dimensional supergravity theories, in a way that allows to apply flux quantization of superspacetime (in \S\ref{SuperFluxQuantization}), which is not possible with machinery available elsewhere in the literature. More extensive discussion will appear in \cite{GSS24}\cite{GSS25}.

\medskip

\medskip

-- \S\ref{HigherSuperGeometry}: {Higher Super Geometry}.

-- \S\ref{SuperSpaceTime}: {Super Space Time Geometry}.

\newpage

\subsection{Higher Super Geometry}
\label{HigherSuperGeometry}

We give a quick account of higher supergeometry, along the lines previously indicated in \cite[\S 3.1.3]{SS20Orb} (also \cite{HSS19}\cite{Schreiber19}); for more details see the companion article \cite{GSS25}, for more exposition see \cite{Schreiber24}, for more technical details on the higher geometric aspect see \cite[\S 1]{FSS23Char}.
It is this higher version of supergeometry that we need for super-flux quantization in \S\ref{SuperFluxQuantization} (since the C-field is a higher gauge field) and which is not found elsewhere in the literature. 
But for related discussion in the literature of non-higher mathematical supergeometry in view of supersymmetric field theory see also: \cite{Manin88}\cite{Lott90}\cite{KS98}\cite[\S II.2]{CDF91}\cite{Schmitt97}\cite{DM99a}\cite{DF99a}\cite{Mirkovic04}\cite{CCF11}\cite{Eder21}.

\begin{remark}[{\bf Category theory in the background}]
\label{CategoryTheoryInTheBackground}
$\,$

\begin{itemize}[leftmargin=.65cm]

\item[\bf (i)]
Super-algebra (\S\ref{SuperAlgebra}) and its enhancement (\S \ref{HomologicalSuperAlgebra}) to homological algebra (cf. \cite{Weibel94}) is just (homological) algebra {\it internal} to (cf. \cite{Boardman95}) the symmetric braided monoidal (cf. \cite[\S III]{EilenbergKelly65}) tensor category (cf. \cite{EGNO15}) of super-vector spaces (Def. \ref{SuperVectorSpaces} below, cf. \cite[\S 3.1]{Varadarajan04}), where the heart of the subject -- the super sign-rule -- is encoded in the non-trivial braiding isomorphism \eqref{BraidingOfSuperVectorSpaces}, cf. also Rem. \ref{SignsInHomotopicalSuperAlgebra} below. 

\item[\bf (ii)] 
Super geometry (\S \ref{SuperGeometry}) and its enhancement to higher super geometry (\S \ref{HigherSuperGeometry}) is just (higher) topos theory (Def. \ref{SmoothSuperInfinityGroupoids}) over the site of Cartesian super-spaces (Def. \ref{SuperCartesianSpace}) \cite[\S 3.1.3]{SS20Orb} (in generalization of \cite{KS98}\cite{Sachse08} \cite{Schmitt97}), cf. \cite{Schreiber24} for exposition and pointers and \cite{GSS24}\cite{GSS25} for more details.
\end{itemize}

\smallskip

\noindent
However, for the record, we spell out the definitions explicitly and do not assume that the reader is familiar with category theory -- though for stating definitions we do assume that the reader knows at least what a category {\it is}. A basic introduction to categories aimed at mathematical physicists is in \cite{Geroch85}; for further introduction, we recommend \cite{Awodey06}.
\end{remark}

\noindent We proceed as follows:

\S\ref{SuperAlgebra} -- Super Algebra

\S\ref{SuperGeometry} -- Super Geometry

\S\ref{HomologicalSuperAlgebra} -- Homological Super Algebra

\S\ref{SuperDifferentialForms} -- Super Differential Forms

\S\ref{SuperSmoothFieldSpaces} -- Super Field Spaces

\S\ref{SuperModuliStacks} -- Super Moduli Stacks

\S\ref{BackgroundOnFluxQuantization} -- Super Flux Quantization

\subsubsection{Super Algebra}
\label{SuperAlgebra}

\begin{definition}[{\bf Super vector spaces}]
\label{SuperVectorSpaces}
We write $\mathrm{sMod}$ for the symmetric monoidal category of {\it super vector spaces} whose 
\begin{itemize}
\item objects are $\ZTwo$-graded vector space $V \,:=\, V_0 \oplus V_1$, 
\item morphisms are linear maps preserving the grading,
\item tensor product is that of the underlying vector spaces with grading given by
$$
  \big(
    V \otimes V'
  \big)_\sigma
  \;\;:=\;\;
  V_0 \otimes V'_{0 + \sigma}
  \;
  \oplus
  \;
  V_1 \otimes V'_{1 + \sigma}
$$
\item braiding 
\begin{equation}
  \label{BraidingOfSuperVectorSpaces}
  \begin{tikzcd}
    V
    \otimes
    V'
    \ar[
      rr,
      "{
        \mathrm{br}_{V, V'}
      }",
      "{ \sim }"{swap}
    ]
    &&
    V' \otimes V
  \end{tikzcd}
\end{equation}
is that of the underlying vector spaces times a sign when two odd-graded factors are swapped:
$$
  v \in V_\sigma
  ,\;
  v' \in V'_{\sigma'}
  \;\;\;\;\;
  \Rightarrow
  \;\;\;\;\;
  \mathrm{brd}_{V, V'}
  \big(
    v \otimes v'
  \big)
  \;:=\;
  (-1)^{\sigma \cdot \sigma'}
  \,
  v' \otimes v
$$
\end{itemize}
\end{definition}
\begin{example}[{\bf Purely odd vector spaces}]
\label{PurelyOddVectorSpaces}
  For $V$ an  ordinary vector space, we write $V_{\mathrm{odd}}$ for the super-vector space which is concentrated on $V$ in odd degree:
  $$
    V \in
    \mathrm{Mod}
    \;\;\;\;\;\;\;
    \yields
    \;\;\;\;\;\;\;
    \left\{\!\!
    \def\arraystretch{1.5}
    \begin{array}{l}
    V_{\mathrm{odd}}
    \;\in\;
    \mathrm{sMod}
    ,
    \\
    (V_{\mathrm{odd}})_0
    \;=\;
    0
    ,
    \\
    (V_{\mathrm{odd}})_1
    \;=\;
    V
    \,.
    \end{array}
    \right.
  $$
\end{example}
\begin{remark}[{\bf Dual super vector spaces}]
\label{DualSuperVectorSpace}
Given $V \in \mathrm{sMod}$ its dual object is degreewise the ordinary dual vector space
$$
  (V^\ast)_\sigma
  \;\cong\;
  (V_\sigma)^\ast
  \,.
$$
\end{remark}

\smallskip

The following category $\mathrm{sCAlg}$ is that of commutative monoid objects {\it internal} to $\mathrm{sMod}$, but we spell 
out explicitly what this means:
\begin{definition}[{\bf Super-commutative algebras}]
\label{SupercommutativeAlgebras}
By $\mathrm{sCAlg}$ we denote the category of {\it super-commutative $\mathbb{R}$-algebras}, whose objects are $\ZTwo$-graded $\mathbb{R}$-vector spaces $A \defneq A_{0} \oplus \mathrm{A}_{1}$ equipped with unital associative algebra structure on the underlying vector space
$$
  (-)\cdot (-)
  \;:\;
  A \otimes A
  \longrightarrow
  A
$$
such that this respects the grading
and
is graded-commutative:
$$
  a \in A_\sigma
  ,\,
  \;
  a' \in A_{\sigma'}
  \;\;\;\;
  \Rightarrow
  \;\;\;\;
  \left\{\!\!
  \def\arraystretch{1.4}
  \begin{array}{l}
    a \cdot a' 
    \;\in\;
    A_{\sigma + \sigma'}
    \,,
    \\
    a \cdot a'
    \;=\;
    (-1)^{ \sigma \cdot \sigma' }
    a' \cdot a
    \,.
  \end{array}
  \right.
$$
A morphism of supercommutative algebras $A \longrightarrow A'$ is a linear map on the underlying vector spaces which is a homomorphism of underlying associative algebras and respects the $\ZTwo$-grading. 

The tensor product on this category
$$
  (-)\otimes (-)
  \;:\;
  \mathrm{sCAlg}
  \times 
  \mathrm{sCAlg}
  \longrightarrow
  \mathrm{sCAlg}
$$
is the usual tensor product on the underlying $\mathbb{R}$-algebras with grading given by
$$
  \big(
    A \otimes A'
  \big)_{\sigma}
  \;:=\;
    A_0 \otimes A'_{\sigma}
    \;\oplus\;
    A_1 \otimes A'_{ 1 + \sigma }
  \,.
$$
\end{definition}

\begin{example}[{\bf Smooth manifolds as duals of super-commutative algebras}]
\label{SmoothManifoldsAsDualyOfSCAlgs}
  Every commutative $\mathbb{R}$-algebra $A$ becomes a super-commutative $\mathbb{R}$-algebra by setting $A_{0} \,:=\, A$
  and $A_1 := 0$.
  Here we are particularly interested in ordinary algebras of smooth functions $C^\infty(X)$ on a smooth manifold $X$. A fundamental (if maybe underappreciated) theorem of differential geometry implies that the assignment $C^\infty(-)$ is a {\it fully faithful} embedding of smooth manifolds into (the opposite of the category of commutative algebras $\mathrm{CAlg}^{\mathrm{op}}$, and hence into) the opposite of the category of super-commutative algebras:
  \vspace{-2mm} 
  $$
    \begin{tikzcd}[sep=0pt]
      \mathrm{SmthMfd}
      \ar[
        rr,
        hook
      ]
      &&
      \mathrm{sCAlg}^{\mathrm{op}}
      \\[-3pt]
      X &\longmapsto&
      C^\infty(X)
      \,.
    \end{tikzcd}
  $$
\end{example}

In the spirit of algebraic geometry, this example allows us to regard objects of $\mathrm{sCAlg}^{\mathrm{op}}$ as generalized smooth manifolds, namely as affine ``super-schemes'',  of sorts. 
In fact,  we just need (in Def. \ref{SuperCartesianSpace} below) rather mild such generality, namely such as to locally include the following Ex. \ref{GrassmannAlgebra}:

\begin{example}[{\bf Grassmann algebra}]
  \label{GrassmannAlgebra}
  For $q \in \mathbb{N}$, the  {\it Grassmann algebra} $\Lambda^\bullet (\mathbb{R}^q)^\ast$ is the super-commutative algebra freely generated by $q$ elements $\theta^1, \cdots \theta^q$ of odd degree. Hence
  $$
    \theta^{\oddcoordinateindex_1}
    \,
    \theta^{\oddcoordinateindex_2}
    \;=\;
    -
    \theta^{\oddcoordinateindex_2}
    \,
    \theta^{\oddcoordinateindex_1},
  \quad 
  \mbox{\small  in particular}
  \;\;
    \theta^\oddcoordinateindex 
    \theta^\oddcoordinateindex 
    \;=\; 
    0
    \,.
  $$
  Hence a general element
  $$
    a
    +
    \sum_{\oddcoordinateindex = 1}^q
    a_\oddcoordinateindex 
    \,\theta^\oddcoordinateindex
    \; +
    \sum_{\oddcoordinateindex_1,\oddcoordinateindex_2 = 1}^q
    \tfrac{1}{2}
    a_{\oddcoordinateindex_1 \oddcoordinateindex_2 }
    \,
    \theta^{\oddcoordinateindex_1} \theta^{\oddcoordinateindex_2}
    +
    \cdots
    +
    a_{1 \cdots q}
    \,
    \theta^1 
    \cdots
    \theta^q
    \,,
    \;\;\;\;
    a_{\cdots} \in \mathbb{R}
  $$
  may be thought of as a kind of polynomial function on a space that is in some sense like a $q$-dimensional Cartesian space, but 
   {\bf (i)} of such tiny (infinitesimal) extension that the square of any of its canonical coordinate functions identically vanishes, 
   {\bf (ii)} in fact which is ``odd'' in that the coordinate functions anti-commute with each other. 
  While such a space does not exist ``classically'', we may think of it as dually {\it defined} as whatever it is that has $\Lambda^\bullet (\mathbb{R}^{q})^\ast$ as its algebra of functions. As such we denote this space as $\mathbb{R}^{0\vert q}$ in the following Def. \ref{SuperCartesianSpace}.
\end{example}

\subsubsection{Super Geometry}
\label{SuperGeometry}

We may now speak of differential supergeometry embodied by {\it smooth super sets} in direct analogy with the {\it smooth sets} discussed in \cite{GS23} (to which we refer the reader for more motivation), just with the role of plain Cartesian spaces replaced by Cartesian super-spaces:

\begin{definition}[{\bf Cartesian super spaces}]
  \label{SuperCartesianSpace}
  The category $\mathrm{sCartSp}$ of {\it super Cartesian spaces} is the full subcategory of the opposite of
  super-commutative $\mathbb{R}$-algebras (Def. \ref{SupercommutativeAlgebras}) on those which are tensor 
  products of the $\mathbb{R}$-algebra of smooth functions on a Cartesian space $\mathbb{R}^n$ (Ex. \ref{SmoothManifoldsAsDualyOfSCAlgs})
  with the Grassmann algebra on finitely many generators (Ex. \ref{GrassmannAlgebra}):
  \begin{equation}
    \label{AlgebraOfFunctionsOnSuperCartesianSpace}
    \begin{tikzcd}[sep=0pt]
      \mathrm{sCartSp}
      \ar[
        rr,
        hook,
        "{
          C^\infty(-)
        }"
      ]
      &&
      \mathrm{sCAlg}^{\mathrm{op}}
      \\[-2pt]
      \mathbb{R}^{n\vert q}
      &\longmapsto&
      C^\infty(\mathbb{R}^n)
      \otimes_{{}}
      \wedge^\bullet
      (\mathbb{R}^q)^\ast.
    \end{tikzcd}
  \end{equation}
\end{definition}

\begin{examples}[\bf Purely bosonic/fermionic]
  By construction, it follows that ordinary cartesian spaces $\mathbb{R}^n$ (with smooth maps between them) 
  are fully faithfully embedded inside super Cartesian spaces, as are the {\it superpoints} $\mathbb{R}^{0\vert q}$.
\end{examples}

\begin{lemma}[{\bf Site of Cartesian super spaces}]
\label{SiteOfCartesianSuperSpaces}
The category $\mathrm{sCartSp}$ of Cartesian superspaces (Def. \ref{SuperCartesianSpace}) carries a coverage
(Grothendieck pre-topology) where the coverings of any $\mathbb{R}^{n \vert q}$ are of the form
$$
  \big\{
    U_i 
      \cong 
  \mathbb{R}^{n\vert q}
  \xrightarrow{ \iota_i }
   \mathbb{R}^{n \vert q}
  \big\}_{i \in I}
$$
such that 
\begin{itemize}
\item[\bf (i)]
each $\iota_i$ is the product 
$\iota_i \,\cong\, \bosonic{\iota}{}_i \times \mathrm{id}_{\mathbb{R}^{0\vert q}}$
of its  bosonic body with the identity on super point factor

\item[\bf (ii)] the bosonic maps
constitute a differentiably good open cover of smooth manifolds
$$
  \Big\{
    \bos{U}_i 
  \xrightarrow{ \bosonic{\iota}_i }
   \mathbb{R}^{n}
  \Big\}_{i \in I}
$$
meaning that every finite intersection $\bos{U}_{i_1} \cap \cdots \cap \bos{U}_{i_n}$ is either empty or diffeomorphic to $\mathbb{R}^n$.
\end{itemize}
\end{lemma}

\begin{definition}[{\bf Smooth super sets}]
\label{SuperSmoothSets}
The category of {\it smooth super sets} is the {\it sheaf topos} over the site of super Cartesian spaces (Lem. \ref{SiteOfCartesianSuperSpaces}), which (for the purpose of higher generalization in \S\ref{SuperModuliStacks}) we think of as the localization of the presheaves at the local isomorphisms (liso)
$$
  \mathrm{sSmthSet}
  \;:=\;
  L^{\mathrm{liso}}
  \,
  \mathrm{Func}\big(
    \mathrm{sCartSp}^{\mathrm{op}}
    ,\,
    \mathrm{Set}
  \big)
  \,.
$$
This concretely means:
\begin{itemize}[leftmargin=.7cm]
\item[\bf (i)] smooth super sets $X$ are (represented by) functors
$$
  \begin{tikzcd}[row sep=-3pt, column sep=0]
    \mathrm{sCartSp}^{\mathrm{op}}
    \ar[
      rr
    ]
    &&
    \mathrm{Set}
    \\
    \mathbb{R}^{n\vert q}
    &\longmapsto&
    \mathrm{Plt}(
      \mathbb{R}^{n\vert q}
      ,\,
      X
    )
  \end{tikzcd}
$$
which we think of as assigning to a Cartesian super space $\mathbb{R}^{n \vert q}$ the set of ways of mapping it into the would-be smooth super-set $X$, hence of {\it plotting out} Cartesian super-spaces inside $X$;
\item[\bf (ii)] maps $X \to Y$ between smooth super-sets are natural transformations between these plot-assigning functors of the form
$$
  \begin{tikzcd}
    X
    \ar[
      from=r,
      "{
        \mathrm{liso}
      }"{},
      "{
        p
      }"{swap}
    ]
    &
    \widehat X
    \ar[
      r,
      "{ f }"
    ]
    &
    Y
  \end{tikzcd}
$$
where the left one is a {\it local isomorphism} in that for all $n,q \in \mathbb{N}$ it restricts to a bijection 
\begin{equation}
  \label{LocalIsomorphisms}
  \begin{tikzcd}
    \widehat{X}
    \ar[
      r,
      "{ \mathrm{liso} }"
    ]
    &
    X
  \end{tikzcd}
  \hspace{.7cm}
    \Leftrightarrow
  \hspace{.7cm}
  \underset{
    n,q \in \mathbb{N}
  }{\forall}
  \hspace{.4cm}
  \begin{tikzcd}
    \mathrm{PltGrm}\big(
      \mathbb{R}^{n \vert q}
      ,\,
      \widehat X
    \big)
    \ar[
      r,
      "{ \sim }"
    ]
    &
    \mathrm{PltGrm}(
      \mathbb{R}^{n \vert q}
      ,\,
      X
    )
  \end{tikzcd}
\end{equation}
on the {\it stalks} of {\it germs of plots}
\begin{equation}
  \label{GermsOfPlots}
  \mathrm{PltGrm}(
    \mathbb{R}^{n \vert q}
    \,,
    X
  )
  \;:=\;
  \mathrm{PltGrm}\big(
    \mathbb{R}^{n \vert q}
    \,,
    X
  \big)\big/\sim
  \,,
\end{equation}
where plots $\phi \,\sim\, \phi'$ iff they agree on some open super-neighborhood of the origin.
\end{itemize}
\end{definition}
\begin{remark}[\bf Super vs. super smooth]
In differential geometry, it is tradition to understand by default that the underlying manifolds of {\it supermanifolds} are smooth. 
In this tradition, it may make sense to refer to {\it super smooth sets},  {\it super smooth $\infty$-groupoids} and their {\it super smooth homotopy theory} for short as just {\it super set}, {\it super $\infty$-groupoids} and their {\it super-homotopy theory}, respectively,  at least when the differential-geometric context is understood. 
However, beware that this is ambiguous, as there are other notions of geometry (such as algebraic and derived geometry) that have super-versions. In particular, there is a super version already of {\it discrete} geometry, embodied by the presheaf topos on super-points.
\end{remark}

\begin{example}[{\bf Supermanifolds as smooth super sets}]\label{SmoothManifoldsAsSmoothSuperSets}
A {\it smooth super-manifold} (e.g. \cite[\S 4.1]{Manin88}\cite[\S 2]{DM99a}) becomes a smooth super set (Def. \ref{SuperSmoothSets}) by declaring its plots to be the ordinary maps of supermanifolds:
$$
  X \in \mathrm{sSmthMfd}
  \;\;\;\;\;\;\;\;
  \yields
  \;\;\;\;\;\;\;\;
  \mathrm{Plt}\big(
    \mathbb{R}^{n\vert q}
    ,\,
    X
  \big)
  \;\;
  :=
  \;\;
  \mathrm{Hom}_{\mathrm{sSmthMfd}}
  \big(
    \mathbb{R}^{n\vert q}
    ,\,
    X
  \big)
  \,.
$$
This construction constitutes a fully faithful embedding of smooth supermanifolds into smooth supersets.
\begin{equation}
  \label{SupermanifoldsAmongSmoothSuperSets}
  \begin{tikzcd}
    \mathrm{sSmthMfd}
    \ar[
      r,
      hook
    ]
    &
    \mathrm{sSmthSet}
    \,.
  \end{tikzcd}
\end{equation}
Without even recalling any definition of supermanifolds, we can make this fully concrete by appeal to Batchelor's theorem \cite{Batchelor79}\cite[\S 1.1.3]{Batchelor84}:
For $V \!\xrightarrow{\;}\! \bosonic{X}$ a smooth real vector bundle of finite rank over an ordinary smooth manifold $\bosonic{X}$, consider the super-commutative algebra (Def. \ref{SupercommutativeAlgebras}) which is the Grassmann algebra {\it over $C^\infty\big(\bosonic{X}\big)$}
\begin{equation}
  \label{AlgebraOfFunctionsOfOddVectorBundle}
  C^\infty\big(
    \bosonic{X} \vert V_{\mathrm{odd}}
  \big)
  \;\;
  :=
  \;\;
  \wedge^\bullet_{
    \scalebox{.5}{
      $C^{{}^\infty}\!\!\big(\bosonic{X}\big)$
    }
  }
  \Gamma_X\big(
    V^\ast
  \big)
  \;\;
  =
  \Gamma_X\big(
    \wedge^\bullet
    V^\ast
  \big)
  \,.
\end{equation}
From this we obtain a smooth super-set
by declaring its plots to be given by the evident dual super-algebra homomorphisms out of \eqref{AlgebraOfFunctionsOfOddVectorBundle} into the algebra of function \eqref{AlgebraOfFunctionsOnSuperCartesianSpace} on the given probe space: 
\begin{equation}
  \label{PlotsOfOddVctorBundle}
  \bosonic{X}\vert V_{\mathrm{odd}}
  \,\in\,
  \mathrm{sSmthSet}
  \,,
  \;\;\;\;
  \mbox{with}
  \;\;\;\;
  \mathrm{Plt}\big(
    \mathbb{R}^{n\vert q}
    ,\,
    \bosonic{X}\vert V_{\mathrm{odd}}
  \big)
  \;\;
  :=
  \;\;
  \mathrm{Hom}_{\mathrm{sCAlg}}
  \Big(
    C^\infty\big(
      \bosonic{X}\vert V_{\mathrm{odd}}
    \big)
    ,\,
    C^\infty\big(
      \mathbb{R}^{n\vert q}
    \big)
  \Big)
  \,.
\end{equation}
By Batchelor's theorem \cite{Batchelor79}, a smooth super-set is a supermanifold seen under the embedding \eqref{SupermanifoldsAmongSmoothSuperSets} iff it is isomorphic to one of the form \eqref{PlotsOfOddVctorBundle}.

\end{example}
\begin{example}[\bf Open covers of supermanifolds as local resolutions]
\label{OperCoverOfSupermanifoldAsLocalResolution}
Given an open cover $\big\{ U_i \xhookrightarrow{\iota_i} X \big\}_{i \in I}$  of a super-manifold, consider the smooth super-set whose plots are only those maps into $X$ that land in one of the charts $U_i$, hence whose plot-assigning functor is
$$
  \begin{tikzcd}[row sep=-5pt]
    \mathllap{
      \mathrm{Plt}\big(
        -;\,
        \widehat X
      \big)
      \;\;
      :
      \;
    }
    \mathrm{sCartSp}^{\mathrm{op}}
    \ar[
      rr
    ]
    &&
    \mathrm{Set}
    \\
    \mathbb{R}^{n \vert q}
    &\longmapsto&
    \mathrm{Hom}_{\mathrm{sSmthMfd}}
    \Big(
      \mathbb{R}^{n \vert q}
      ,\,
      \underset{i \in I}{\coprod}
      U_i
    \Big)
    \Big/\sim
  \end{tikzcd}
$$

\vspace{-2mm} 
\noindent where on the right $\big(\phi_i : \mathbb{R}^{n \vert q} \xrightarrow{\;} U_i\big) \sim \big(\phi'_{i'} : \mathbb{R}^{n \vert q} \xrightarrow{\;} U_{i'}\big)$ iff they agree as maps to $X$, hence iff $\iota_i \circ \phi_i \,=\, \iota_{i'} \circ \phi'_{i'}$.
Then the evident natural transformation
$$
  \begin{tikzcd}[row sep=-4pt, column sep=0pt]
    \mathrm{Plt}\big(
      \mathbb{R}^{n \vert q}
      ,\,
      \widehat X
    \big)
    \ar[
      rr
    ]
    &&
    \mathrm{Plt}\big(
      \mathbb{R}^{n \vert q}
      ,\,
      X
    \big)
    \\
    \phi_i &\longmapsto&
    \iota_i \circ \phi_i
  \end{tikzcd}
$$
is a local isomorpism of smooth super-sets 
$\begin{tikzcd}\widehat X \ar[r, "{\mathrm{liso}}"] & X \end{tikzcd}$.
\end{example}

\begin{example}[\bf{Bosonic body of super manifold}]
\label{BosonicBodyOfSupermanifold}
By definition, a smooth super-manifold $X\in \sManifolds$ has an underlying ordinary manifold $\bos{X} \in \SmoothManifolds\hookrightarrow \sManifolds$, viewed canonically as an (even) super-manifold accompanied with a canonical embedding
$$
\eta_{X}: \bos{X} \longhookrightarrow X \, .
$$
This is given dually, in any local chart $\mathbb{R}^{n|q}$, by the projection of function super-algebras\footnote{In terms of globally defined function algebras this is the canonical projection $C^\infty(X) \longrightarrow C^\infty(X) / J \cong C^\infty\big(\!\bos{X}\big)$, where $J$ is the ideal generated by \textit{odd} elements.}
\begin{align*} 
f(x)
    +
    \sum_{\oddcoordinateindex = 1}^q
    f_\oddcoordinateindex(x) 
    \,\theta^\oddcoordinateindex
    +
    \sum_{\oddcoordinateindex_1,\oddcoordinateindex_2 = 1}^q
    \tfrac{1}{2}
    f(x)_{\oddcoordinateindex_1 \oddcoordinateindex_2 }
    \,
    \theta^{\oddcoordinateindex_1} \theta^{\oddcoordinateindex_2}
    +
    \cdots
    +
    f(x)_{1 \cdots q}
    \,
    \theta^1 
    \cdots
    \theta^q  
    \qquad \longmapsto \;\; f(x) \, .
\end{align*}
The embeddings $\eta_X:\bos{X}\hookrightarrow X$ define an endofunctor 
$$\eta: \sManifolds \longrightarrow \sManifolds$$
which `forgets the odd structure' of any supermanifold. We shall use the same symbol $\bos{X}$ for the bosonic body of $X$ considered as a smooth manifold, an (even) smooth super manifold, or a smooth super set (Ex. \ref{SmoothManifoldsAsSmoothSuperSets}).
\end{example}

\subsubsection{Homological Super Algebra}
\label{HomologicalSuperAlgebra}

\begin{definition}[{\bf $\mathbb{Z}$-Graded super vector spaces}]
\label{GradedSuperVectorSpaces}
We write $\mathrm{sgMod}$ for the symmetric monoidal category of {\it graded super vector spaces} whose 
\begin{itemize}[leftmargin=.5cm]
\item 
objects are $\mathbb{Z}$-graded super vector spaces, hence $(\mathbb{Z} \times \ZTwo)$-bigraded vector spaces $V = \bigoplus_{ {n \in \mathbb{Z}} \atop { \sigma \in \ZTwo }  }\, V_{n, \sigma} $,
\item morphisms are linear maps preserving the bigrading,
\item
tensor product is that of the underlying vector spaces with bi-grading given by
$$
  \big(
    V
    \,\otimes\,
    V'
  \big)_{n, \sigma}
  \;:=\;
  \underset{
    { k \in \mathbb{Z} },
    \;
    { \rho \in \ZTwo }
  }{\bigoplus}
  V_{k, \sigma}
  \otimes
  V_{n - k, \sigma - \rho}
  \,,
$$
\item braiding is that of the underlying super-vector spaces \eqref{BraidingOfSuperVectorSpaces} times an {\it additional} sign  (cf. Rem. \ref{SignsInHomotopicalSuperAlgebra}) when a pair of $\mathbb{Z}$-graded factors is swapped
$$
  v 
    \,\in\,
  V_{n, \sigma}
  ,\;
  v'
    \,\in\,
  V_{n', \sigma'}
  \qquad 
    \Rightarrow
  \qquad 
  \mathrm{brd}_{V,V'}(v \otimes v')
  \;=\;
  (-1)^{ n \cdot n' }
  (-1)^{ 
    \sigma \cdot \sigma' 
  }
  v' \otimes v
  \,.
$$

Moreover, we write
$$
  \begin{tikzcd}
    \mathrm{sgMod}^{\mathrm{ft}}
    \ar[
      r,
      hook
    ]
    &
    \mathrm{sgMod}
  \end{tikzcd}
$$
for the full subcategory of graded super vector spaces of {\it finite type}, i.e., those that are degree-wise finite-dimensional.
\end{itemize}

\end{definition}
\begin{notation}[{\bf Shifted and dual graded super-vector spaces}]
\label{ShiftedAndDual}
For $V \in \mathrm{sgMod}$ (Def. \ref{GradedSuperVectorSpaces})
we write
\begin{itemize}[leftmargin=.5cm]
\item $V^\ast$ for the dual object, which is bi-degreewise the dual vector space but with the $\mathbb{Z}$-grading {\it reversed} \footnote{ One may say that also the super-grading is reversed under dualization, but this is not visible since $- \sigma = \sigma \in \ZTwo$; cf. Rem. \ref{DualSuperVectorSpace}. }:
\begin{equation}
  \label{DualgsVect}
  \big(
    V^\ast
  \big)_{n, \sigma}
  \;=\;
  \big(
    V_{-n, \sigma}
  \big)^\ast
  \,.
\end{equation}
\item $V^\vee$ for the degree-wise dual vector spaces:
\begin{equation}
  \label{DegreewiseDualgsVect}
  \big(V^\vee\big)_{n, \sigma}
  \;:=\;
  \big(
    V_{n, \sigma}
  \big)^\ast
  \,.
\end{equation}
\item
$b V$ for the result of shifting up in $\mathbb{Z}$-degree:
\begin{equation}
  \label{DegreeShift}
    \big(
      b V
    \big)_{n, \sigma}
    \;\;
    :=
    \;\;
    V_{n-1, \sigma}
    \,.
\end{equation}
\end{itemize}
\end{notation}

The following category $\mathrm{sgcAlg}$ is just that of $\mathbb{Z}$-graded-commutative algebras {\it internal} to $\mathrm{sMod}$ (Def. \ref{SuperVectorSpaces})
and equivalently just that of commutative algebra {\it internal} to $\mathrm{sgMod}$ (Def. \ref{GradedSuperVectorSpaces}), but we spell it out explicitly:

\begin{definition}[{\bf Super graded-commutative algebras}]
\label{SupergCAlgebras}
The category $\mathrm{sgCAlg}$ of {\it super graded-commutative algebras}, has as objects $(\mathbb{Z} \times \ZTwo)$-(bi)graded vector spaces
$$
  A 
    \;\defneq\;
  \underset{
    n \in \mathbb{Z}
  }{\oplus}
  \big(
    A_{n, 0}
    \,\oplus\,
    A_{n, 1}
  \big)
$$
equipped with 
an associative and unital multiplication
$$
  (-)
  \cdot
  (-)
  \;:\;
  A \otimes A \longrightarrow A
$$
which respects the bigrading and is bigraded-commutative, in the following sense:
\begin{equation}
  \label{RespectForBigrading}
  a \in A_{n,\sigma}
  ,\,
  a' \in A'_{n', \sigma'}
  \,,
  \;\;\;\;\;
  \Rightarrow
  \;\;\;\;\;
  \left\{\!\!
  \def\arraystretch{1.4}
  \begin{array}{l}
    a \cdot a' 
    \;\in\;
    A_{n + n',\, \sigma + \sigma' }
    \\
    a \cdot a'
    \;=\;
    (-1)^{
      n \cdot n'
      \,+\,
      \sigma \cdot \sigma'
    }
    \,
    a' \cdot a
    \,.
  \end{array}
  \right.
\end{equation}
A homomorphism of such SGC-algebras is a homomorphism of the underlying associative algebras which respects the bigrading.
\end{definition}

\begin{example}[{\bf Free super graded-commutative algebras}]
\label{FreeSuperGradedCommutativeAlgebras}
  For $V \in \mathrm{sgMod}$,
  its {\it free} super graded-commutative algebra
  $$
    \mathbb{R}[V]
    \;:=\;
    \mathrm{Sym}\big(
      V
    \big)
    \;\;
    \in
    \;
    \mathrm{sgCAlg}
  $$
  is the symmetric tensor algebra on $V$ {\it internal} to $\mathrm{sgMod}$. This means that if $(v_i)_{i \in I}$ 
  is a linear basis of $V$ with homogeneous basis elements $v_i \in V_{n_i, \sigma_i}$ then $\mathbb{R}[V]$ 
  is the associative algebra freely generated by this basis subject to the relation
  $$
    v_i \cdot v'_{i'}
    \;=\;
    (-1)^{
      n_i \cdot n'_{i'}
      +
      \sigma_i \cdot \sigma'_{i'}
    }
    \;
    v'_{i'} \cdot v_i\;.
  $$
\end{example}

The following category $\mathrm{sdgcAlg}$ is just that of dg-algebras {\it internal} to $\mathrm{sMod}$ (Def. \ref{SuperVectorSpaces}), but we spell it out:

\begin{definition}[{\bf Super differential-graded-commutative algebras}]
\label{SuperdgCAlgebras}
The category $\mathrm{sdgcAlg}$ of {\it super differential-graded-commutative algebras} has as objects super graded-commutative algebras $A$ (Def. \ref{SupergCAlgebras})
equipped with a linear map (the {\it differential})
$$
  \mathrm{d}
  \;:\; A \longrightarrow A
$$
which is a graded derivation of bidegree $(+1,0)$ squaring to zero:
\begin{equation}
  \label{BigradedDerivation}
  a \in A_{n,\sigma}
  ,\,
  a' \in A'_{n', \sigma'}
  \;\;\;\;\;
  \Rightarrow
  \;\;\;\;\;
  \left\{\!\!
  \def\arraystretch{1.2}
  \begin{array}{l}
    \mathrm{d}a \,\in\, A_{n+1, \sigma}
    \,,
    \\
    \mathrm{d}(a \cdot a')
    \;=\;
    (\mathrm{d} a) \cdot a'
    \;+\;
    (-1)^{n} 
    \,
    a \cdot \mathrm{d} a'
    \,,
    \\
    \mathrm{d} \mathrm{d} a \;=\; 0
    \,.
  \end{array}
  \right.
\end{equation}

A morphism of such SDGC-algebras is a homomorphism of the underlying super graded-commutative algebras which respects the differential.

The tensor product on this category is the usual tensor product on the underlying dg-algebras, with bigrading given by
$$
  (A \otimes A')_{n, \sigma}
  \;
  :=
  \;
  \bigoplus_{
    { k \in \mathbb{Z}, \, \rho \in \ZTwo }    
  }
  A_{k, \sigma}
  \otimes
  A_{n-k, \sigma - \rho}
  \,.
$$
\end{definition}

\begin{remark}[{\bf Signs in homological super-algebra}]
\label{SignsInHomotopicalSuperAlgebra}
Note the sign rule in \eqref{RespectForBigrading}: 
\begin{itemize} 
\item[{\bf (i)}] This is evidently the rule obtained by internalizing the notion of dg-algebras into
the symmetric monoidal category of super-vector spaces, and it is the sign rule used in the supergravity 
literature \cite[p. 880]{BBLPT88}\cite[(II.2.109)]{CDF91}. Further physics-oriented discussion indicating the mathematical motivation via internalization is in \cite[\S 1.2]{DM99a}\cite[\S 1]{DM99b}\cite[\S A.6]{DF99b}.

\item[{\bf (ii)}]  Nevertheless, a sizeable part of the mathematical physics literature (mostly authors 
who say ``Q-manifold'' for certain dg-algebras) use a {\it different} sign rule, with sign 
$(-1)^{ (n + \sigma) \cdot (n' + \sigma') \,\mathrm{mod}\, 2 }$. This defines a nominally different 
but equivalent symmetric braiding on the monoidal category $\mathrm{Ch}_\bullet(\mathrm{sMod})$ 
(comparison of the two rules is in \cite[pp. 62-64]{DM99a}\cite[p. 8]{DM99b}).

\item[{\bf (iii)}]  While one needs to carefully stick to one of the two rules for global consistency 
(we use the natural sign rule \eqref{RespectForBigrading} throughout), notice that the crucial commutativity 
of gravitino fields among themselves \eqref{GravitinoFormsCommutingAmongEachOther} holds with both sign
rules (cf. Rem. \ref{CommutingSpinors}).
\end{itemize} 
\end{remark}

The following identification follows Ref. \cite[\S 3]{FSS19}, in evident super-generalization of \cite[Def. 13]{SSS09}\cite[\S 4]{FSS23Char}\cite[\S 3.1]{SS24Flux}. 

\begin{definition}[{\bf Super $L_\infty$-algebras}]
\label{SuperLInfinityAlgebras}
We may identify  the category $\mathrm{shLAlg}^{\mathrm{ft}}$
of degreewise finite-dimensional {\it super $L_\infty$-algebras} as the full subcategory of the opposite of SDGC-algebras (Def. \ref{SuperdgCAlgebras})
on those whose underlying SGC-algebra (Def. \ref{SupergCAlgebras}) is free (Ex. \ref{FreeSuperGradedCommutativeAlgebras}) on a degreewise finite-dimensional super vector space (Def. \ref{SuperVectorSpaces}), namely on the shifted \eqref{DegreeShift}  degreewise dual \eqref{DegreewiseDualgsVect} of the super $L_\infty$-algebra space $\mathfrak{\mathfrak{a}}$. This embedding assigns to a super $L_\infty$-algebra $\mathfrak{a}$ its {\it Chevalley-Eilenberg algebra}
\begin{equation}
  \label{shLAgAsSubcatOfsgcdAlgop}
  \begin{tikzcd}[row sep=0pt, column sep=small]
    \mathrm{shLAlg}^{\mathrm{ft}}
    \ar[
      rr,
      hook,
      "{
        \mathrm{CE}(-)
      }"
    ]
    &&
    \mathrm{sdgcAlg}^{\mathrm{op}}
    \\
  \scalebox{0.85}{$  \big(
      \mathfrak{a}
      ,\,
      [-],
      [-,-],
      [-,-,-], 
      \cdots
    \big)
    $}
    &\longmapsto&
  \scalebox{0.85}{$  \big(
      \mathbb{R}\big[
        b\mathfrak{a}^{\vee}
      \big]
      ,\,
      \mathrm{d}_{\vert b \mathfrak{a}^\vee}
      \;=\;
      [-]^\ast
      +
      [-,-]^\ast
      +
      [-,-,-]^\ast
      + 
      \cdots
    \big)
    $}.
  \end{tikzcd}
\end{equation}
\end{definition}
\begin{remark}[{\bf ``FDA'' terminology in supergravity}]
$\,$

  \noindent {\bf (i)} The SDGC-algebras arising as Chevalley-Eilenberg-algebras of super $L_\infty$-algebras in \eqref{shLAgAsSubcatOfsgcdAlgop} are 
  (this was first pointed out in \cite[\S 6.5.1]{SSS09}\cite{FSS15-WZW}, reviewed in \cite{FSS19})
  what in \cite[\S III.6]{CDF91}\cite[\S 6.3]{Fre13b}\cite[\S 6]{Castellani18} are called ``free differential algebras'' or ``FDA''s, for short, following \cite{vN83}. 

\noindent {\bf (ii)} Note that this is a bit of a misnomer: It is only their {\it underlying} super graded-commutative algebras which are free (Ex. \ref{FreeSuperGradedCommutativeAlgebras}), while as super {\it differential} graded-commutative algebras these CE-algebras are crucially {\it not free} in general. The free {\it differential} algebra on $b \mathfrak{a}^\vee$ is contractible and isomorphic to the {\it Weil algebra} $\mathrm{W}(\mathfrak{a})$.
\end{remark}
\begin{example}[{\bf Ordinary super Lie algebras}]
  \label{OrdinarySuperLieUnderCE}
    Consider a finite-dimensional super Lie algebra $\big(\mathfrak{a},\, [-,-]\big)$ with linear basis $\big\{ v^i \big\}_{i \in I}$ of homogeneous super-degree $v^i \in \mathfrak{a}_{\sigma_i}$ and with structure constants
    $$
      [v^i, v^j]
      \;=\;
      f^{i j}_k
      \;
      v^k
      \,.
    $$
    Then $\mathrm{CE}\big(\mathfrak{a}\big)$ (Def. \ref{SuperLInfinityAlgebras}) is the associative algebra freely generated from elements $\omega_i$ in bidegree $(1, \sigma_i)$ subject to the relations
    $$
      \omega_i \, \omega_j
      \;=\;
      - (-1)^{  \sigma_i \cdot \sigma_j }
      \omega_j \, \omega_i
    $$
    and equipped with differential $\mathrm{d}$ satisfying
    $$ 
      \mathrm{d}
      \,
      \omega_k
      \;=\;
      \tfrac{1}{2}
      \,
      f^{i j}_k
      \,
      \omega_i \, \omega_j
      \,.
    $$
    Using the graded derivation property of $\mathrm{d}$ one checks that the condition $\mathrm{d} \circ \mathrm{d} \,=\, 0$ is equivalently the Jacobi identity condition on $[-,-]$.  
\end{example}
\begin{example}[\bf Line Lie $n$-algebras]
  \label{LineLieNAlgebras}
  For $k \in \mathbb{N}$, the {\it line Lie $(1+k)$-algebra} $b^k \mathbb{R}$ has $\mathrm{CE}\big(b^k \mathbb{K}\big)$ (Def. \ref{SuperLInfinityAlgebras}) being the graded-commutative algebra on a single closed generators $\omega$ in degree $(k+1, 0)$, $\mathrm{d} \, \omega = 0$.
\end{example}

\begin{example}[\bf Super-Poincar{\'e} and super-Minkowski Lie algebra]
\label{SuperMinkowskiLieAlgebra}
The {\it super-Poincar{\'e} Lie algebra} (or {\it supersymmetry algebra}, for short) in $11\vert\mathbf{32}$-dimensions
$$
  \mathfrak{iso}\big(\mathbb{R}^{1,10\vert\mathbf{32}}\big)
  \;\in\;
  \mathrm{shLAlg}^{\mathrm{fr}}
$$ 
has (and is defined thereby) CE-algebra \eqref{shLAgAsSubcatOfsgcdAlgop} of this form:
\medskip 
\begin{equation}
  \label{CEOfSuperPoincareAlgebra}
  \mathrm{CE}\Big(
    \mathfrak{iso}\big(\mathbb{R}^{1,D-1\vert\mathbf{32}}\big)
  \Big)
  \;\;
  =
  \;\;
  \mathbb{R}
  \left.
  \left[
  \def\arraystretch{1.2}
  \def\arraycolsep{2pt}
  \begin{array}{lrl}
    \big(
      e^a
    \big)_{a = 0}^{10}\,,
    &
   \scalebox{0.85}{$ \mathrm{deg}(e^a)
    \;=\;
    (1,0)$}
    \\
    \big(
      \omega^{ab} = - \omega^{b a}
    \big)_{a, b = 0}^{10}\,,
    &
    \scalebox{0.85}{$  \mathrm{deg}\big(
      \omega^{ab}
    \big)
    \;=\;
    (1,0)
    $}
    \\
    \big(
      \psi
    \big)_{\alpha=1}^{32}\,,
    &
\;\; \scalebox{0.85}{$     \mathrm{deg}\big(\psi^\alpha\big)
    \;=\;
    (1,1)
    $}
  \end{array}
  \right]
  \right/
  \!  \left(\!\!
  \def\arraystretch{1}
  \begin{array}{l}
    \mathrm{d}
    \, e^a
    \;=\;
    -\, \omega^a{}_b \, 
    e^b
    +
    \big(\,
      \overline{\psi}
      \,\Gamma^a\,
      \psi
    \big)
    \\
    \mathrm{d}
    \,
    \omega^{a b}
    \;=\; -\, 
    \omega^{a}{}_c
    \, 
    \omega^{c b}
    \\
    \mathrm{d}
    \,
    \psi^\alpha 
      \;=\; 
    0
    \,.
  \end{array}
 \!\!\! \right)
  \!,
\end{equation}
  where the pairing 
  $\big(\,\overline{\psi}\,\Gamma^a\,\psi\big)$ is from Lem. \ref{TheSpinorPairing} and Lem. \ref{SpinorQuadraticForms} (using Rem. \ref{CommutingSpinors}).
This contains the ordinary Lorentz Lie algebra as a subalgebra
\begin{equation}
  \label{LorentzInsideSuperPoincareAlgebra}
  \begin{tikzcd}[row sep=-3, column sep=small]
    \mathllap{
      \mathfrak{so}(1,10)
      \,=\;\;
    }
    \mathfrak{so}
    \big(
      \mathbb{R}^{1,10}
    \big)
    \ar[
      rr,
      hook
    ]
    &&
    \mathfrak{iso}\big(
      \mathbb{R}^{1,10\vert \mathbf{32}}
    \big)
    \\
    \mathrm{CE}\big(
      \mathfrak{so}(1,10)
    \big)
    \ar[
      rr,
      <<-
    ]
    &&
    \mathrm{CE}\big(
      \mathfrak{iso}\big(
        \mathbb{R}^{1,10\vert \mathbf{32}}
      \big)
    \big)
    \\
    0 &\longmapsfrom& e^a
    \\
    0 &\longmapsfrom& \psi^\alpha
    \\
    \omega^{a}{}_b
    &\longmapsfrom&
    \omega^{a}{}_b
    \,,
  \end{tikzcd}
\end{equation}
  whose quotient is the {\it super-Minkowski Lie algebra} (the {\it super-translation} part of the supersymmetry algebra):
\begin{equation}
  \label{TheSuperMinkowskiLieAlgebra}
    \def\arraystretch{1.7}
    \begin{array}{l}
    \mathbb{R}^{1,10\vert \mathbf{32}}
    \;\;
    :=
    \;\;
    \mathfrak{iso}\big(
      \mathbb{R}^{1,10\vert \mathbf{32}}
    \big)
    \big/
    \mathfrak{so}\big(
      \mathbb{R}^{1,10\vert \mathbf{32}}
    \big)
    \\
    \mathrm{CE}\big(
      \mathbb{R}^{1,10\vert\mathbf{32}}
    \big)
    \;=\;
  \mathbb{R}
  \left.
  \left[
  \def\arraystretch{1.8}
  \def\arraycolsep{2pt}
  \begin{array}{lrl}
    \big(
      e^a
    \big)_{a = 0}^{10}\,,
    &
    \;\;  
    \scalebox{0.85}{$    \mathrm{deg}(e^a)
    \;=\;
    (1,0)
    $}
    \\
    \big(
      \psi
    \big)_{\alpha=1}^{32}\,,
    &
 \;\; \scalebox{0.85}{$     \mathrm{deg}\big(\psi^\alpha\big)
    \;=\;
    (1,1)
    $}
  \end{array}
  \right]
  \right/
  \!  
  \left(\!\!
  \def\arraystretch{1}
  \begin{array}{l}
    \mathrm{d}
    \, 
    e^a
    \;=\;
    +
    \big(\,
      \overline{\psi}
      \,\Gamma^a\,
      \psi
    \big)
    \\
    \mathrm{d}
    \,
    \psi^\alpha 
      \;=\; 
    0
  \end{array}
  \!\!\! 
  \right)
  .
  \end{array}
\end{equation}
This is the local model geometry for $11\vert\mathbf{32}$-dimensional super-spacetime; 
see Def. \ref{SuperSpacetime} below.
\end{example}

\begin{remark}[{\bf Supersymmetry}]
\label{TheSupersymmetryBracket}
The crucial term in \eqref{CEOfSuperPoincareAlgebra} is the summand $\mathrm{d}\, e^a \,=\, \cdots + \big(\,\overline{\psi}\,\Gamma^a\,\psi\big)$. This is the linear dual to the super Lie bracket of the form
\begin{equation}
  \label{TheSupersymmetryBracketInComponents}
  \big[
    \,
    \overline{Q}_\alpha
    ,\,
    Q_\beta
    \,
  \big]
  \;=\;
  \Gamma^a_{\alpha\beta}
  P_a
  \,,
\end{equation}
which is the hallmark of supersymmetry (the supersymmetry generators $Q$ ``square'' to translation generators $P$).
In some sense, this term controls all of 11d supergravity; see also Rem. \ref{FormOfTheSuperTorsionConstraint}.
\end{remark}

\begin{example}[{\bf Whitehead $L_\infty$-algebras}]
\label{WhiteheadLInfinityAlgebras}
For $X$ a simply-connected topological space with $\mathrm{dim}\big(H^n(X; \mathbb{Q})\big) < \infty$ for all $n \in \mathbb{N}$, it has a {\it minimal Sullivan model} dgc-algebra $\mathrm{CE}\big(\mathfrak{l}X\big)$, which encodes its $\FR$-rational homotopy type (reviewed in \cite[Prop. 4.23]{FSS23Char}). This is the CE-algebra of the $\FR$-rational {\it Whitehead $L_\infty$-algebra} $\mathfrak{l}X$ (\cite[Rem. 5.4]{FSS23Char}, essentially the ``Quillen model'' of $X$).
\end{example}
Specifically:
\begin{example}[{\bf Rational Whitehead $L_\infty$-algebra of 4-sphere is M-theory gauge algebra}]
\label{Rational4Sphere}
    The minimal Sullivan model of the 4-sphere $X \defneq S^4$ is (a standard fact of rational homotopy theory, for review in our context  \cite[p. 21]{SS24Flux}\cite[Ex. 5.3]{FSS23Char}):
    \begin{equation}
      \label{SullivanModelOf4Sphere}
      \mathrm{CE}\big(
        \mathfrak{l}S^4
      \big)
      \;\;
      \cong
      \;\;
      \FR\big[
        G_4
        ,\,
        G_7
      \big]
      \Big/
      \! \left(
      \def\arraycolsep{0pt}
      \def\arraystretch{1}
      \begin{array}{l}
        \mathrm{d}\, G_4 \;=\; 0
        \\
        \mathrm{d}\, G_7 
          \;=\;
        \tfrac{1}{2}
        \,
        G_4 \, G_4
      \end{array}
      \right)
      .
    \end{equation}
    Curiously, the corresponding Whitehead $L_\infty$-algebra (via Ex. \ref{WhiteheadLInfinityAlgebras}) 
    coincides (as highlighted in \cite[\S 4]{Sati10}\cite[p. 20]{SS24Flux}\cite[(12)]{SatiVoronov22}) with the {\it M-theory gauge algebra} (first identified in \cite[(2.6)]{CJLP98}, see also \cite[(3.4)]{LLPS99}\cite[(75)]{KS03}\cite[(86)]{BNS04}):
    \begin{equation}
      \label{WhiteheadAlgebraOf4Sphere}
      \mathfrak{l}S^4
      \;\cong\;
      \FR\langle
        v_3, \, v_6
      \rangle
      \quad 
      \mbox{\small with only non-vanishing bracket 
      of generators being}
      \quad 
      {[v_3, v_3]} \;=\; v_6
      \,.
    \end{equation}
    Notice how the identification works, in direct analogy to the case of ordinary Lie algebras (Ex. \ref{OrdinarySuperLieUnderCE}). The structure constants of the differential of the Sullivan model are identified with those of the Whitehead $L_\infty$-algebra (the Quillen model):
    $$
      \def\arraystretch{1.2}
      \begin{array}{rcr}
        \mathrm{d}
        \,
        G_7 
          &=& 
        \tfrac{1}{2}
        G_4\, G_4\;,
        \\
        v_6 &=& {[v_3,\, v_3]}\;.
      \end{array}
    $$
\end{example}
\begin{example}[\bf The $\mathfrak{l}S^4$-valued super-cocycle on $11\vert\mathbf{32}$-dim super-Minkowski spacetime]
\label{4SphereValuedSuperCocycle}
  There is a non-trivial (even homotopically) morphism of super $L_\infty$-algebras from the super-Minkowski Lie algebra \eqref{TheSuperMinkowskiLieAlgebra}
  to the Whitehead $L_\infty$-algebra of the 4-sphere \eqref{WhiteheadAlgebraOf4Sphere}, as follows:
  \begin{equation}
    \label{TheMinkowskiSupercocycle}
    \begin{tikzcd}[row sep=0 pt, column sep=small]
      \mathbb{R}^{1,10\vert\mathbf{32}}
      \ar[
        rr,
        "{
          (
            G_4^0
            ,\,
            G_7^0
          )
        }"
      ]
      &&
      \mathfrak{l}S^4
      \\
      \mathrm{CE}\big(
        \mathbb{R}^{1,10\vert\mathbf{32}}
      \big)
      \ar[
        rr,
        <-
      ]
      &&
      \mathrm{CE}\big(
        \mathfrak{l}S^4
      \big)
      \\
      \tfrac{1}{2}
      \big(\,
        \overline{\psi}
        \,\Gamma_{a_1 a_2}\,
        \psi
      \big)
      e^{a_1}\, e^{a_2}
      &\longmapsfrom&
      G_4
      \\
      \tfrac{1}{5!}
      \big(\,
        \overline{\psi}
        \,\Gamma_{
          a_1 \cdots a_5
        }\,
        \psi
      \big)
      e^{a_1} \cdots e^{a_5}
      &\longmapsfrom&
      G_7
      \mathrlap{\,.}
    \end{tikzcd}
  \end{equation}
  (The expressions on the left constitute the WZW-terms of the $\kappa$-symmetric Green-Schwarz-type sigma-models 
  for the M2-brane and the M5-brane on super-Minkowski spacetime, cf. \cite[\S 2.1]{HSS19}\cite[\S 5]{BPS}).

  That \eqref{TheMinkowskiSupercocycle} is indeed a homomorphism of super $L_\infty$-algebras, in that 
  its dual map on CE-algebras respects the differential relation \eqref{SullivanModelOf4Sphere}, is equivalent 
  to the fundamental Fierz identities that govern 11d supergravity (Prop. \ref{TheFierzIdentitiesOf11dSupergravity}):
  \vspace{-2mm} 
  \begin{equation}
    \label{TheSupercocycles}
    \left.
    \def\arraystretch{1.8}
    \begin{array}{l}
      \differential
      \,
      \Big(
      \tfrac{1}{2}
      \big(\,
      \overline{\psi}
      \,\Gamma_{a_1 a_2}\,
      \psi
      \big)
      e^{a_1} \, e^{a_2}
      \Big)
      \;=\;
      0
      \,,
      \\
      \differential
      \Big(
        \tfrac{1}{5!}
        \big(\,
        \overline{\psi}
        \,\Gamma_{a_1 \cdots a_5}\,
        \psi
        \big)
        e^{a_1} \cdots e^{a_5}
      \Big)
      \;=\;
      \tfrac{1}{2}
      \Big(
        \tfrac{1}{2}
        \big(\,
        \overline{\psi}
        \,\Gamma_{a_1 a_2}\,
        \psi
        \big)
        e^{a_1} \, e^{a_2}
      \Big)
      \Big(
        \tfrac{1}{2}
        \big(\,
        \overline{\psi}
        \,\Gamma_{a_1 a_2}\,
        \psi
        \big)
        e^{a_1} \, e^{a_2}
      \Big)
    \end{array}
   \!\! \right\}
    \;\;
    \in
    \mathrm{CE}\big(
      \mathbb{R}^{1,10\vert \mathbf{32}}
    \big)
    \,.
  \end{equation}
  This observation is due to
  \cite[\S 3,4]{FSS15-M5WZW}\cite[\S 2]{FSS17}\cite[Prop. 3.43]{HSS19}; it suggests that the {\it Hypothesis H} 
  (cf. p. \pageref{HypothesisH}) -- that C-field flux is quantized in Cohomotopy theory -- lifts to super-space, which is the main claim in \S\ref{SuperFluxQuantization}.

  Our Thm. \ref{11dSugraEoMFromSuperFluxBianchiIdentity} below may be understood as saying that the above $\mathfrak{l}S^4$-valued cocycle relation governs all of 11d supergravity. Namely, just requiring that this homomorphism generalizes from super-Minkowski spacetime to non-flat $11\vert\mathbf{32}$-dimensional super-spacetimes $\big(X, (e,\psi,\omega)\big)$, as a map of $L_\infty$-algebroids over $X$ (see Def. \ref{ClosedLInfinityValuedDifferentialForms} and \eqref{ClosedFormsAsMapsOutOfTangentLieAlgebroid}), in the form
  \begin{equation}
    \label{SuperFluxAsSuperCocycle}
    \begin{tikzcd}[row sep=0 pt, column sep=small]
      T X 
      \ar[
        rr,
        "{
          (
            G_4^s
            ,\,
            G_7^s
          )
        }"
      ]
      &&
      \mathfrak{l}S^4
      \\
      \Omega^\bullet_{\mathrm{dR}}\big(
        X
      \big)
      \ar[
        rr,
        <-
      ]
      &&
      \mathrm{CE}\big(
        \mathfrak{l}S^4
      \big)
      \\
      (G_4)_{a_1 \cdots a_4}
      \,
      e^{a_1} \cdots e^{a_4}
      \,+\,
      \tfrac{1}{2}
      \big(\,
        \overline{\psi}
        \,\Gamma_{a_1 a_2}\,
        \psi
      \big)
      e^{a_1}\, e^{a_2}
      &\longmapsfrom&
      G_4
      \\
      (G_7)_{a_1 \cdots a_7}
      \,
      e^{a_1} \cdots e^{a_7}
      \,+\,
      \tfrac{1}{5!}
      \big(\,
        \overline{\psi}
        \,\Gamma_{a_1 \cdots a_7}\,
        \psi
      \big)
      e^{a_1} \cdots e^{a_7}
      &\longmapsfrom&
      G_7
    \end{tikzcd}
  \end{equation}
  turns out to be equivalent to the super-spacetime $\big(X, (e,\psi,\omega)\big)$ satisfying the 11d SuGra equations of motion with flux source $G_4$ (Thm. \ref{11dSugraEoMFromSuperFluxBianchiIdentity}). 
\end{example}

\begin{remark}[{\bf Relative factors in flux Bianchi identity}]
$\,$

\noindent {\bf (i)} Since minimal Whitehead $L_\infty$-algebras/Sullivan models are unique only {\it up to isomorphism}
of dgc-algebras, the relative factor of $\tfrac{1}{2}$ shown in \eqref{SullivanModelOf4Sphere} is not a characteristic 
of the {\it rational} homotopy type of $S^4$, as it can be scaled away by an algebra isomorphism:
$$
      \FR\big[
        G_4
        ,\,
        G_7
      \big]
      \Big/
      \! \left(
      \def\arraycolsep{0pt}
      \def\arraystretch{1}
      \begin{array}{l}
        \mathrm{d}\, G_4 \;=\; 0
        \\
        \mathrm{d}\, G_7 
          \;=\;
        \tfrac{1}{2}
        \,
        G_4 \, G_4
      \end{array}
      \right)
      \qquad 
  \begin{tikzcd}[row sep=-1 pt, column sep=small]
    \mathfrak{l}S^4
    \ar[
      rr,
      "{ \sim }"
    ]
    &&
    \mathfrak{l}'S^4
    \\
    \mathrm{CE}\big(
      \mathfrak{l}S^4
    \big)
    \ar[
      rr,
      <-,
      "{
        \sim
      }"
    ]
    &&
    \mathrm{CE}\big(
      \mathfrak{l}'S^4
    \big)
    \\
    G_4 &\longmapsfrom& G_4
    \\
    2\, G_7 &\longmapsfrom& G_7
  \end{tikzcd}
$$
this prefactor disappears (and under similar rescalings it can take any non-zero value).
\footnote{Such rescalings have been utilized crucially in the context 
of Mysterious Triality \cite{SatiVoronov21}\cite{SatiVoronov22}.}

\smallskip 
\noindent {\bf (ii)} Note that this also means that the prefactor {\it is} fixed once the scale of the generators
$G_4$, $G_7$ is fixed by some further condition.
For example, in super spacetime geometry it is suggestive to normalize the bifermionic forms by their natural combinatorial prefactors as $\tfrac{1}{p!} \big(\, \overline{\psi}\,\Gamma_{a_1 \cdots a_p}\, \psi \big) e^{a_1} \cdots e^{a_p}$
and with Ex. \ref{4SphereValuedSuperCocycle} this fixes the factor to be $1/2$, as shown and is usual in much of the string theory literature (though not universally).

\smallskip 
\noindent {\bf (iii)} However, a more intrinsic normalization of the generators is given by first imposing a flux-quantization condition law (as discussed in \S\ref{SuperFluxQuantization}) and then asking the generators to be rational images of {\it integral} cohomology classes. For the case of flux quantization in Cohomotopy ({\it Hypothesis H}, p. \pageref{HypothesisH}), this does yield the factor of $1/2$ \cite[Thm. 4.8]{FSSHopf} (in fact it is $2 G_7 + C_3 G_4$ that becomes integral, see further discussion in
\cite[(3)]{FSSHopf}\cite[p. 3]{FSS20TwistedString} and the exposition in \cite[\S 4.3]{SS24Flux}).
\end{remark}

\subsubsection{Super Differential Forms}
\label{SuperDifferentialForms}

\begin{example}[{\bf Ordinary differential forms}]
\label{sdgcAlgebraOfordinaryDifferentialForms}
For $X$ a smooth manifold, the ordinary de Rham algebra $\Omega^\bullet_{\mathrm{dR}}(X)$ 
of differential forms on $X$
is an SDGC-algebra (Ex. \ref{SuperdgCAlgebras}) when regarded as concentrated in bidegree $(\mathbb{N} \times \{0\}) \hookrightarrow (\mathbb{Z} \times \ZTwo)$.
\end{example}
\begin{example}[{\bf Differential forms on a super-point}]
\label{DifferentialFormsOnSuperpoint}
  For $q \in \mathbb{N}$, the de Rham algebra of super-differential forms on a super-point $\FR^{0\vert q}$ is the SDGC-algebra (Ex. \ref{SuperdgCAlgebras}) 
  $$
    \Omega^\bullet_{\mathrm{dR}}\big(
      \FR^{0\vert q}    
    \big)
    \;\in\;
    \mathrm{sdgCAlg}
  $$
  which is freely generated by $C^\infty(\FR^{0\vert q}) = \Lambda^\bullet (\FR^q)^\ast$ (Ex. \ref{GrassmannAlgebra}) in bidegree $(\{0\} \times \ZTwo)$, hence whose underlying bigraded vector space is spanned over $C^\infty(\FR^{0\vert q})$ by new generators $\mathrm{d}\theta^{\rho_1} \cdots \mathrm{d}\theta^{\rho_p}$ in bidegree $(p,\, p \,\mathrm{mod}\, 2)$
  \begin{equation}
    \label{SpaceOfDifferentialFormsOnSuperpoint}
    \Omega^\bullet_{\mathrm{dR}}\big(
      \FR^{0\vert q}
    \big)
    \;\cong\;
    \bigoplus_{p = 0}^q
    \;
    \bigoplus_{1 \leq  \rho_1 < \cdots < \rho_p \leq q}
    \;
    C^\infty(\FR^{0\vert q})
    \big\langle
      \mathrm{d}\theta^{\rho_1}
      \cdots
      \mathrm{d}\theta^{\rho_p}
    \big\rangle
    \,.
  \end{equation}
  The corresponding product is given by
  multiplication of $C^\infty(\FR^{0\vert q})$-coefficients followed by
  ``shuffle'' composition of the generator symbols, and with differential given on generators by the evident
  $
    \mathrm{d}
    :
    \theta^\rho \;\mapsto\;
    \mathrm{d}\theta^\rho
  $.
\end{example}

\newpage 
\begin{definition}[{\bf Sets of Differential forms on super-Cartesian spaces}]
\label{DifferentialFormsOnSuperCartesianSpaces}
  For $n, q \in \mathbb{N}$, the SDGC-algebra (Def. \ref{SuperdgCAlgebras}) 
  $$
    \Omega^\bullet_{\mathrm{dR}}
    \big(
      \FR^{n\vert q}
    \big)
    \;\in\;
    \mathrm{sdgCAlg}
  $$
  of {\it super-differential forms} on the super-Cartesian space $\FR^{n\vert q}$ (Def. \ref{SuperCartesianSpace})
  is the tensor product of the de Rham algebra on $\FR^n$ (Ex. \ref{sdgcAlgebraOfordinaryDifferentialForms}) with the de Rham algebra on $\FR^{0\vert q}$ (Ex. \ref{DifferentialFormsOnSuperpoint}):
  \smallskip 
  $$
    \Omega^\bullet_{\mathrm{dR}}\big(
      \FR^{n\vert q}
    \big)
    \;\;:=\;\;
    \Omega^\bullet_{\mathrm{dR}}\big(\FR^n\big)
    \otimes
    \Omega^\bullet_{\mathrm{dR}}\big(\FR^{0\vert q}\big)
    \;\;
    \in
    \;\;
    \mathrm{sdgCAlg}
    \,.
  $$
\end{definition}

\smallskip 
Equivalently, and perhaps more geometrically, super-differential 1-forms may be identified \cite{GSS24} as the fiber-wise linear maps of supermanifolds
\begin{align}\label{FormsAsBundleMaps}
\Omega^1_{\mathrm{dR}}
\big( \FR^{n|q}\big) \cong \mathrm{Hom}_{\sManifolds}^{\mathrm{fib.lin.}}\big(T(\FR^{n|q}), \, \FR\times \FR_\mathrm{odd} \big) \, ,
\end{align}
where $T(\FR^{n|q})$ is the super-tangent bundle defined dually by its algebra of functions $C^\infty(x^\evencoordinateindex, \dot{x}^\evencoordinateindex)[\dot{\theta}^\rho, \theta^\rho]$. This follows immediately by the suggestive identification of coordinates as $\dot{x}^r \equiv \differential x^r$ and $\dot{\theta}^\rho \equiv \differential \theta^\rho$. 

\smallskip 
Similarly, super-differential $k$-forms are identified as fiber-wise antisymmetric multilinear maps 
$$
T^{\times k}(\FR^{n|q}) \longrightarrow \FR\times \FR_{\mathrm{odd}}\;.
$$ 
%

\begin{definition}[{\bf Sets of Super differential forms with coefficients}]
\label{SuperDifferentialFormsWithCoefficients}
$\,$

\noindent {\bf (i) } 
  For $n, q \in \mathbb{N}$ and 
  $V \,\in\, \mathrm{sgMod}^{\mathrm{ft}}$, the {\it $V$-valued super-differential forms} on $\FR^{n \vert q}$ 
  $$
    \Omega_{\mathrm{dR}}^1
    \big(
      \FR^{n\vert q}
      ;\,
      V
    \big)
    \;\;:=\;\;
    \mathrm{Hom}_{\mathrm{sgCAlg}}
    \big(
      \mathrm{Sym}(
        b V^\vee
      )
      ,\,      \Omega^\bullet_{\mathrm{dR}}(
        \FR^{n\vert q}
      )
    \big)
  $$
  are the homomorphisms of super-graded commutative algebras
  (Def. \ref{SupergCAlgebras})
  from the free super graded-commutative algebra 
  (Ex. \ref{FreeSuperGradedCommutativeAlgebras})
  of the shifted degreewise dual of $V$ 
  \eqref{ShiftedAndDual}
  to the underlying de Rham SGC-algebra of Def. \ref{DifferentialFormsOnSuperCartesianSpaces}.

\noindent {\bf (ii)}  Equivalently, these are the elements of bidegree $(1, 0)$ in the tensor product with the $\mathbb{Z}$-degree reversal of $V$:
  $$
    \Omega^1_{\mathrm{dR}}\big(
      \FR^{n \vert q};
      V
    \big)
    \;\cong\;
    \big(
      \Omega^\bullet_{\mathrm{dR}}
      (
        \FR^{n \vert q}
      )
      \otimes 
      (V^\vee)^\ast
    \big)_{(1,0)}\, .
  $$
For $V\in \mathrm{sMod}^{\mathrm{ft}}$ a super-vector space concentrated in degree 0, these may be also expressed as in $\eqref{FormsAsBundleMaps}$ via 
\begin{align}\label{VectorValuedFormsAsBundleMaps}
\Omega^1_{\mathrm{dR}}
\big( \FR^{n|q}; \, V\big) \cong \mathrm{Hom}_{\sManifolds}^{\mathrm{fib.lin.}}\big(T(\FR^{n|q}), \, V \big) \, .
\end{align}
\end{definition}

\begin{examples}[\bf Super 1-forms with coefficients]
\label{Super1FormsWithCoefficients}
For all $n,q, k \in \mathbb{N}$, we have the following identifications of $V$-valued super-differential forms (Def. \ref{SuperDifferentialFormsWithCoefficients}):
\begin{equation}
  \label{ExamplesOfDiffFormsOnSuperCartesian}
  \def\arraystretch{1.3}
  \def\arraycolsep{4pt}
  \begin{array}{lcl}
    \Omega^1_{\mathrm{dR}}\big(
      \FR^{n \vert q}
      ;\,
      \FR
    \big)
    &\cong&
    \Omega^1_{\mathrm{dR}}\big(
      \FR^{n \vert q}
    \big)_{\mathrm{0}}
    \\
    \Omega^1_{\mathrm{dR}}\big(
      \FR^{n \vert q}
      ;\,
      b^k\FR
    \big)
    &\cong&
    \Omega^{1+k}_{\mathrm{dR}}\big(
      \FR^{n \vert q}
    \big)_0
    \\
    \Omega^1_{\mathrm{dR}}\big(
      \FR^{n\vert 0}
      ;\,
      b^k \FR
    \big)
    &\cong&
    \Omega^{1+k}_{\mathrm{dR}}\big(
      \FR
    \big)
    \\
    \Omega^1_{\mathrm{dR}}\big(
      \FR^{n \vert 0}
      ;\,
      b^k
      \FR_{\mathrm{odd}}
    \big)
    &\cong&
    0
    \\
    \Omega^1_{\mathrm{dR}}\big(
      \FR^{0 \vert 1}
      ;\,
      \FR
    \big)
    &\cong&
    \FR\big\langle
      \theta 
      \, \mathrm{d}\theta
    \big\rangle
    \\
    \Omega^1_{\mathrm{dR}}\big(
      \FR^{0 \vert 1}
      ;\,
      \FR_{\mathrm{odd}}
    \big)
    &\cong&
    \FR\big\langle
      \mathrm{d}\theta
    \big\rangle
    \\
    \Omega^1_{\mathrm{dR}}\big(
      \FR^{1 \vert 1}
      ;\,
      \FR
    \big)
    &\cong&
    C^\infty\big(
      \FR
    \big)
    \big\langle
      \mathrm{d}x
    \big\rangle
    \,\oplus\,
    C^\infty\big(
      \FR
    \big)
    \big\langle
      \theta
      \mathrm{d}\theta
    \big\rangle
    \\
    \Omega^1_{\mathrm{dR}}\big(
      \FR^{1 \vert 1}
      ;\,
      \FR_{\mathrm{odd}}
    \big)
    &\cong&
    C^\infty(\FR)
    \big\langle
      \theta \mathrm{d}x
    \big\rangle
    \,\oplus\,
    C^\infty(\FR)
    \big\langle
      \mathrm{d}\theta
    \big\rangle
    \\
    \Omega^1_{\mathrm{dR}}\big(
      \FR^{1 \vert 2}
      ;\,
      \FR
    \big)
    &\cong&
    C^\infty\big(
      \FR
    \big)
    \langle
      \mathrm{d}x
      ,\,
      \theta^1 \theta^2 \mathrm{dx}
    \rangle
    \,\oplus\,
    C^\infty\big(
      \FR
    \big)
    \langle
      \theta^1 \mathrm{d}\theta^1
      ,\,
      \theta^2 \mathrm{d}\theta^1
      ,\,
      \theta^1 \mathrm{d}\theta^2
      ,\,
      \theta^2 \mathrm{d}\theta^2
    \rangle
    \\
    \Omega^1_{\mathrm{dR}}\big(
      \FR^{1\vert 2}
      ;\,
      \FR_{\mathrm{odd}}
    \big)
    &\cong&
    C^\infty\big(
      \FR
    \big)
    \big\langle
      \theta^1 \mathrm{d}x
      ,\,
      \theta^2 \mathrm{d}x
    \big\rangle
    \,\oplus\,
    C^\infty\big(
      \FR
    \big)
    \big\langle
      \mathrm{d}\theta^1
      ,\,
      \mathrm{d}\theta^2
    \big\rangle \;.
  \end{array}
\end{equation}
\end{examples}

\begin{example}[{\bf Classifying super set of differential forms}]
\label{ClassifyingSuperSetOfDifferentialForms}
For $V \in \mathrm{sgMod}^{\mathrm{ft}}$, the system of $V$-valued differential forms (Def. \ref{SuperDifferentialFormsWithCoefficients}) 
is clearly a sheaf on the site of super Cartesian spaces, and as such constitutes a smooth super-set (Def. \ref{SuperSmoothSets}):
\begin{equation}
  \label{SheafOfVValuedDifferentialForms}
  \begin{tikzcd}[row sep=-5pt, column sep=small]
  \Omega^1_{\mathrm{dR}}\big(
    -;\,
    V
  \big)
  \ar[
    r,
    phantom,
    "{ : }"
  ]
  &[+4pt]
  \mathrm{sCartSp}^{\mathrm{op}}
  \ar[
    rr
  ]
  &&
  \mathrm{Set}
  \\
  &
  \FR^{n\vert q}
  &\longmapsto&
  \Omega^1_{\mathrm{dR}}\big(
    \FR^{n\vert q}
    ;\,
    V
  \big)
  \,.
  \end{tikzcd}
\end{equation}
\end{example}

\begin{definition}[{\bf Differential forms on smooth super sets}]
\label{DifferentialFormsOnSmoothSuperSets}
Given $X \in \mathrm{sSmthSet}$ (Def. \ref{SuperSmoothSets}) and $V \in \mathrm{sgMod}^{\mathrm{ft}}$, a {\it $V$-valued differential 1-form} on $X$ is a morphism (of smooth super sets) from $X$ to the classifying super set of such forms (Ex. \ref{ClassifyingSuperSetOfDifferentialForms}):
$$
  F
  \;:\;
  X 
    \longrightarrow
  \Omega^1_{\mathrm{dR}}\big(
    -;\, V
  \big)
  \,.
$$
Hence the set of all such $V$-valued differential forms on $X$ is the hom-set
\begin{equation}
  \label{VValuedFormsOnSmoothSuperSet}
  \Omega^1_{\mathrm{dR}}\big(
    X;\, V
  \big)
  \;:=\;
  \mathrm{Hom}_{\mathrm{sSmthSet}}\left(
    X
    ,\, 
    \Omega^1_{\mathrm{dR}}
    \big(
      -;
      \, 
      V
    \big)
  \right)
  \,.
\end{equation}
\end{definition}
\begin{remark}[\bf Pullback of differential forms via classifying super-sets]
\label{PullbackOfDifferentialFormsViaClassifyingSuperSets}
With the native characterization of differential forms in Def. \ref{DifferentialFormsOnSmoothSuperSets} as maps to their classifying super-set, the operation of pullback of differential forms along a map $f \,:\, X \xrightarrow{\;} Y$ of smooth super-sets corresponds just to the operation of precomposing the classifying maps with $f$:
\begin{equation}
  \begin{tikzcd}[sep=0pt]
    \Omega^1_{\mathrm{dR}}\big(
      Y
      ;\,
      V
    \big)
    \ar[
      rr,
      "{
        f^\ast
      }"
    ]
    \ar[
      d,
      shorten=-2pt,
      equals
    ]
    &&
    \Omega^1_{\mathrm{dR}}\big(
      X
      ;\,
      V
    \big)
    \ar[
      d,
      shorten=-2pt,
      equals
    ]
    \\[+4pt]
    \mathrm{Hom}\big(
      Y
      ;\,
      \Omega^1_{\mathrm{dR}}(
        -
        ;\,
        V
      )
    \big)
    \ar[rr]
    &&
    \mathrm{Hom}\big(
      X
      ;\,
      \Omega^1_{\mathrm{dR}}(
        -
        ;\,
        V
      )
    \big)
    \\[-2pt]
    \phi 
      &\longmapsto&
    \phi \circ f
    \mathrlap{\,.}
  \end{tikzcd}
\end{equation}
\end{remark}

\begin{example}[{\bf Differential forms on smooth supermanifolds}]
  \label{DifferentialFormsOnSmoothSupermanifolds}
  By the Yoneda Lemma, Def. \ref{DifferentialFormsOnSmoothSuperSets} reduces on Cartesian super-spaces $X \defneq \FR^{n \vert q}$ to the defining construction (Def. \ref{SuperDifferentialFormsWithCoefficients}), so that the notation $\Omega^1_{\mathrm{dR}}\big(\FR^{n \vert q}; V\big)$ is unambiguous.
    In particular, given a super-manifold $X$ and a super-chart $\iota : \FR^{n\vert q} \hookrightarrow X$, then pullback along this inclusion restricts the abstractly defined differential forms of Def. \ref{DifferentialFormsOnSmoothSuperSets} to the concrete differential forms on this chart.
  This way, Examples \ref{ExamplesOfDiffFormsOnSuperCartesian} apply chart-wise to any super-manifold.
\end{example}

\begin{example}[{\bf Odd forms on an even manifold}]
\label{OddFormsOnEvenManifold}
 For $X \in \mathrm{sSmthSet}$, we have
 $$
   \Omega^1_{\mathrm{dR}}\big(\!
     \bos{X}
     ,\,
     \FR_{\mathrm{odd}}
   \big)
   \;\cong\;
   0
   \,.
 $$
The same holds for forms valued in any \textit{odd} vector space $V_\mathrm{odd}$, i.e., the set of odd-vector valued differential $1$-forms, and in turn $k$-forms, is trivial on any bosonic manifold. We come back to this odd state of affairs in \S\ref{SuperSmoothFieldSpaces}.
\end{example}

\begin{definition}[{\bf Closed $L_\infty$-valued differential forms}]
\label{ClosedLInfinityValuedDifferentialForms} Given $X \in \mathrm{sSmthSet}$ (Def. \ref{SuperSmoothSets}) and
$\mathfrak{a} \,\in\, \mathrm{shLAlg}^{\mathrm{ft}}$ (Def. \ref{SuperLInfinityAlgebras}),

\begin{itemize}[leftmargin=.75cm] 
\item[{\bf (i)}] We say (\cite[\S 6.5]{SSS09}\cite[\S 4.1]{FSS12}\cite[Def. 6.1]{FSS23Char})
that the {\it closed} (or {\it flat})  $\mathfrak{a}$-valued differential forms on $X$ are differential forms with coefficients in $\mathfrak{a}$ (Def. \ref{SuperDifferentialFormsWithCoefficients}) which are not just homomorphism of super graded-algebras out of $\FR[b \mathfrak{a}^\vee]$
but of super graded-{\it differential} algebras out of the Chevalley-Eilenberg algebra $\mathrm{CE}(\mathfrak{a})$:
$$
  \Omega^1_{\mathrm{dR}}\big(
    \FR^{n \vert q}
    ;\,
    \mathfrak{a}
  \big)_{\mathrm{clsd}}
  \;:=\;
  \mathrm{Hom}_{\mathrm{s{\color{purple}d}gcAlg}}\big(
    \mathrm{CE}(\mathfrak{a})
    ,\,
    \Omega^\bullet_{\mathrm{dR}}(
      \FR^{n \vert q}
    )
  \big)
  \,.
$$
\item[{\bf (ii)}] Equivalently, these are the (even) {\it Maurer-Cartan elements} in the tensor product of $\Omega^\bullet_{\mathrm{dR}}\big(\FR^{n \vert q}\big)$ with the $\mathbb{Z}$-degree-reversed $\mathfrak{a}$. 
 $$
    \Omega^1_{\mathrm{dR}}\big(
      \FR^{n \vert q};
      \mathfrak{a}
    \big)_{\mathrm{clsd}}
    \;\cong\;
    \mathrm{MC}\Big(
      \Omega^\bullet_{\mathrm{dR}}
      \big(
        \FR^{n \vert q}
      \big)
      \otimes 
      (\mathfrak{a}^\vee)^\ast
    \Big)_{0}\, .
  $$
\item[{\bf (iii)}]  Yet equivalently, along the lines of \eqref{VectorValuedFormsAsBundleMaps}, these are morphisms of (super smooth) 
$L_\infty$-algebroids 
\begin{equation}
  \label{ClosedFormsAsMapsOutOfTangentLieAlgebroid}
  T(\FR^{n|q}) 
  \; \longrightarrow \;  \mathfrak{a}
\end{equation}
out of its tangent Lie algebroid, with the target $\mathfrak{a}$ considered as a trivial $L_\infty$-algebroid over the point (cf. \cite[Ex. A.3]{SSS12}).
\end{itemize}
\end{definition}

These evidently form a sub-sheaf on $\mathrm{sCartSp}$, of the sheaf of all $\mathfrak{a}$-valued forms \eqref{SheafOfVValuedDifferentialForms},
hence a smooth super sub-set
\begin{equation}
  \label{SmoothSuperSetOfClosedForms}
  \begin{tikzcd}
    \Omega^1_{\mathrm{dR}}(
      -
      ;\,
      \mathfrak{a}
    )_{\mathrm{clsd}}  
    \ar[
      r,
      hook
    ]
    &
    \Omega^1_{\mathrm{dR}}(
      -
      ;\,
      \mathfrak{a}
    )
    \;\;
    \in
    \;
    \mathrm{sSmthSet}
    \,.
  \end{tikzcd}
\end{equation}
This way, in generalization of \eqref{VValuedFormsOnSmoothSuperSet}, the closed $\mathfrak{a}$-valued differential forms on $X \in \mathrm{sSmthSet}$ are the maps to this classifying super set:
\begin{equation}
  \label{ClosedVValuedDifferentialFormsOnSmoothSuperSet}
  \Omega^1_{\mathrm{dR}}(
    X
    ;\,
    \mathfrak{a}
  )_{\mathrm{clsd}}
  \;:=\;
  \mathrm{Hom}_{\mathrm{sSmthSet}}
  \big(
    X
    ;\,
    \Omega^1_{\mathrm{dR}}(
      -
      ;\,
      \mathfrak{a}
    )_{\mathrm{clsd}}
  \big).
\end{equation}

\begin{examples}[{\bf Ordinary closed differential forms}]
  The closed differential 1-forms 
  (Def. \ref{ClosedLInfinityValuedDifferentialForms}) with coefficients in $b^k \FR$ (Ex. \ref{LineLieNAlgebras})
  are equivalently (cf. \cite[Ex. 6.2]{FSS23Char}) the ordinary closed $k+1$-forms:
  $$
    \Omega^1_{\mathrm{dR}}\big(
      X
      ;\,
      b^k \FR
    \big)_{\mathrm{clsd}}
    \;\;\;
      \cong
    \;\;\;
    \big\{
      F
      \;\in\;
      \Omega^{1+k}_{\mathrm{dR}}(X)
      \;\big\vert\;
      \mathrm{d}
      \, 
      F \;=\; 0
    \big\}
    \,.
  $$
\end{examples}

\begin{example}[{\bf Closed $\mathfrak{l}S^4$-valued forms are the solutions to super-C-field flux Bianchi identity}]
\label{ClosedlS4ValuedDifferentialForms}
Closed $L_\infty$-valued differential forms (Def. \ref{ClosedLInfinityValuedDifferentialForms}) with coefficients in the Whitehead $L_\infty$-algebra $\mathfrak{l}S^4$ of the 4-sphere
(Ex. \ref{Rational4Sphere}) are precisely (cf. \cite[\S 6.4]{FSS23Char}) those pairs of forms which satisfy the Bianchi identity of the duality-symmetric C-field flux densities in 11d supergravity \cite[\S 2.5]{Sati13}\cite[Rem. 3.9]{FSS17}:
\begin{equation}
  \label{ClosedlS4ValuedForms}
  \Omega^1_{\mathrm{dR}}\big(
    X
    ;\,
    \mathfrak{l}S^4
  \big)_{\mathrm{clsd}}
  \;\;
  \cong
  \;\;
  \left\{\!\!
  \def\arraystretch{1.3}
  \begin{array}{l}
    G_4 \;\in\; \Omega^1_{\mathrm{dR}}(X;\, b^3 \FR)
    \\
    G_7 \;\in\; \Omega^1_{\mathrm{dR}}(X;\, b^6 \FR)
  \end{array}
  \Bigg\vert
  \def\arraystretch{1.3}
  \begin{array}{l}
    \mathrm{d}\, G_4 \;=\; 0
    \\
    \mathrm{d}\, G_7 \;=\;
    \tfrac{1}{2}
    G_4 \wedge G_4
  \end{array}
 \!\!\! \right\}
  .
\end{equation}
This means, with \eqref{ClosedVValuedDifferentialFormsOnSmoothSuperSet}, that the smooth super-set $\Omega^1_{\mathrm{dR}}\big(-; \mathfrak{l}S^4\big)_{\mathrm{clsd}}$ \eqref{SmoothSuperSetOfClosedForms} plays the role of the {\it moduli space} of duality-symmetric super-C-field flux densities, in that
$$
  \mathrm{Hom}_{\mathrm{sSmthSet}}
  \big(
    X
    ,\;
    \Omega^1_{\mathrm{dR}}(
      -
      ;\,
      \mathfrak{l}S^4
    )
  \big)
  \;\;
  \cong
  \;\;
    \Omega^1_{\mathrm{dR}}\big(
      X
      ;\,
      \mathfrak{l}S^4
    \big)
  \,.
$$
\end{example}

\begin{definition}[\bf Nonabelian de Rham cohomology]
\label{NonabelianDeRhamCohomology}
In the situation of Def. \ref{ClosedLInfinityValuedDifferentialForms}, given a pair of closed $\mathfrak{a}$-valued differential forms
$$
  F^{(0)}, F^{(1)}
  \,\in\,
  \Omega^1_{\mathrm{dR}}(
    X
    ;\,
    \mathfrak{a}
  )_{\mathrm{clsd}}
$$
we say \cite[Def. 6.2]{FSS23Char} that a {\it coboundary} between them is a {\it concordance} (deformation) between them, namely
a closed $\mathfrak{a}$-valued differential form on the cylinder $X \times [0,1]$ which restricts to $F^{(i)}$ on $X \times \{i\}$: 
$\iota_i^\ast \widehat F \,=\, F^{(i)}$.
\vspace{1mm} 
\begin{equation}
  \label{CoboundaryInNonabelianDeRhamCohomology}
  F^{(0)}
  \;\sim\;
  F^{(1)}
  \hspace{.9cm}
    \Leftrightarrow
  \hspace{.9cm}
  \begin{tikzcd}[
    row sep=14pt
  ]
    X \times \{0\}
    \ar[
      d, 
      hook,
      "{
        \iota_0
      }"{swap}
    ]
    \ar[
      drr,
      "{
        F^{(0)}
      }"
    ]
    \\
    X \times [0,1]
    \ar[
      rr,
      dashed,
      "{
        \exists
        \,
        \widehat{F}
      }"
    ]
    &&
    \Omega^1_{\mathrm{dR}}(
      -;
      \mathfrak{a}
  )_{\mathrm{clsd}}
    \\
    X \times \{1\}
    \ar[
      u, 
      hook',
      "{
        \iota_1
      }"
    ]
    \ar[
      urr,
      "{
        F^{(1)}
      }"{swap}
    ]
  \end{tikzcd}
\end{equation}
This is an equivalence relation \cite[Prop. 5.10]{FSS23Char}
whose equivalence classes we may call
\cite[Def. 6.3]{FSS23Char}\cite[Def. 3.3]{SS24Flux} the {\it $\mathfrak{a}$-valued nonabelian de Rham cohomology} of $X$:
\begin{equation}
  \label{NonabelianDeRhamCohomologyAsQuotient}
  H^1_{\mathrm{dR}}\big(
    X
    ;\,
    \mathfrak{a}
  \big)
  \;
    :=
  \;
  \Omega^1_{\mathrm{dR}}(
    X;\, \mathfrak{a}
  )_{\mathrm{clsd}}
  \big/{\sim}
  \,.
\end{equation}
\end{definition}
In the case when $\mathfrak{a} = b^n \mathbb{R}$ is a line Lie $(n+1)$-algebra (e.g. \cite[Ex. 4.12]{FSS23Char}), the nonabelian de Rham cohomology \eqref{NonabelianDeRhamCohomologyAsQuotient} reduces to ordinary (abelian) de Rham cohomology \cite[Prop. 6.4]{FSS23Char}:
$$
  H^1_{\mathrm{dR}}(X;\, b^n \mathbb{R})
 \;\cong\;
 H^{n+1}_{\mathrm{dR}}(X)
 \,,
$$
which reflects the charges imprinted on abelian higher gauge fields, but the nonabelian generalization reflects (more in \cite[\S 3.1]{SS24Flux}) also total charges 
\eqref{TotalChargeQuantization}
of nonabelian higher gauge fields, such as the 11d SuGra C-field.

\medskip 
Moreover, just as an ordinary gauge potential is locally a null-coboundary for its flux density, the local null-coboundaries \eqref{CoboundaryInNonabelianDeRhamCohomology} in nonabelian de Rham cohomology play the role of gauge potentials for nonabelian higher gauge fields (to be exemplified in Prop. \ref{CoboundariesOfClosedLS4ValuedForms} below). Accordingly, the gauge transformations between such gauge potentials correspond to concordances-of-concordances of closed $\mathfrak{a}$-valued differential forms:

\begin{definition}[\bf Coboundary-of-coboundaries in nonabelian de Rham cohomology]
\label{CoboundariesOfCoboundaries}
$\,$

\noindent Given a pair of coboundaries \eqref{CoboundaryInNonabelianDeRhamCohomology},
$
{\color{darkgreen}   \widehat{F} }
  ,\,
{\color{darkgreen}   \widehat{F}'}
  \;:\;
  X \times [0,1]
  \xrightarrow{\quad}
  \Omega^1_{\mathrm{dR}}(-;\, \mathfrak{a}
  )_{\mathrm{clsd}}
$,
between the {\it same} pair of closed $\mathfrak{a}$-valued diffrential forms,
$
{\color{olive}   F^{(0)} }
  ,\,
 {\color{olive}  F^{(1)}}
  \;:\;
  X \xrightarrow{\quad}
  \Omega^1_{\mathrm{dR}}(-;\mathfrak{a})_{\mathrm{clsd}}
  \,,
$
we say that a {\it coboundary-of-coboundaries} between them is a closed form on the 2-dimensional cylinder over $X$,
$
 {\color{orangeii} \doublehat{F}}
  \;:\;
  X \times [0,1] \times [0,1]
  \xrightarrow{\quad}
  \Omega^1_{\mathrm{dR}}(-;\mathfrak{a})_{\mathrm{clsd}}
$,

\noindent {\bf (i)} such that 
\vspace{-1mm} 
$$
  \def\arraystretch{1.3}
  \begin{array}{c}
  \big(
    \mathrm{id} \times \{0\}
  \big)^\ast
 {\color{orangeii}  \doublehat{F}}
  \;\;=\;\;
{\color{darkgreen}   \widehat{F}}
  \\
  \big(
    \mathrm{id} \times \{1\}
  \big)^\ast
   {\color{orangeii}  \doublehat{F}}
  \;\;=\;\;
  {\color{darkgreen}  \widehat{F}'}
  \end{array}
  \hspace{1cm}
  \begin{tikzcd}[
    column sep=21pt, 
    row sep=40pt
  ]
    &&&
    X \times \{0\}
    \ar[
      d, 
      hook,
      "{ 
        \mathrm{id}
        \times \{0\}
      }"{swap}
    ]
    \\
    &&& X \times [0,1]
    \ar[
      ddddddr,
      "{
      \color{darkgreen}   \widehat{F}'
      }"
    ]
    \\
    &&&
    X \times \{1\}
    \ar[
      u,
      hook',
      "{
        \mathrm{id} \times \{1\}
      }"{description}
    ]
    \\[-70pt]
    &&
    X \times [0,1] \times [0,1]
    \ar[
      ddddrr,
      crossing over,
      "{
         {\color{orangeii}  \doublehat{F}}
      }"{swap, pos=.2}
    ]
    \ar[
      from=uur,
      "{
        \mathrm{id}
        \times \{1\}
      }"{sloped}
    ]
    \ar[
      from=ddll,
      "{
        \mathrm{id}
        \times \{0\}
      }"{sloped}
    ]
    \\[-70pt]
    X \times \{0\}
    \ar[
      d,
      hook,
      "{
        \mathrm{id} \times \{0\}
      }"{swap}
    ]
    \\
    X \times [0,1]
    \ar[
      ddrrrr,
      crossing over,
      end anchor={[yshift=-4pt]},
      "{
      \color{darkgreen}   \widehat{F}
      }"{swap}
    ]
    \\
    X \times \{1\}
    \ar[
      u,
      hook',
      "{ 
        \mathrm{id}
        \times
        \{1\}
      }"
    ]
    \\[-65pt]
    &&&&
    \Omega^1_{\mathrm{dR}}(-;\mathfrak{a})_{\mathrm{clsd}}
  \end{tikzcd}
$$
\vspace{.3cm}

\noindent
\noindent {\bf (ii)} and such that the deformation is constant on the original form data\footnote{ The condition \eqref{2ndConcordanceConstantAtBoundary} ``removes the corners'' in a concordance-of-concordances, so as to yield the usual ``homotopy relative endpoints''. On the other hand, for purposes other than modeling higher 
gauge-transformations the information supported on corners is relevant, see
\cite{Sati11}\cite{Sati13}\cite{Sati14} for further discussion.}
\begin{equation}
  \label{2ndConcordanceConstantAtBoundary}
  \hspace{-1cm} 
  \begin{tikzcd}[row sep=15pt]
    &
    X \times \{0\} \times [0,1]
    \ar[
      d,
      hook
    ]
    \ar[
      r,
      ->>
    ]
    &
    X \times \{0\}
    \ar[
      dddr,
      "{
      \color{olive}   F^{(0)}
      }"
    ]
    \\
    &
    X \times [0,1] \times [0,1]
    \ar[
      ddrr,
      "{
        {\color{orangeii}  \doublehat{F}}
      }"
    ]
    \\
    X \times \{1\}
    \ar[
      from=r, 
      ->>
    ]
    \ar[
      drrr,
      shift right=2pt,
      "{
      \color{olive}   F^{(1)}
      }"{swap}
    ]
    &
    X \times \{1\} \times [0,1]
    \ar[
      u, hook'
    ]
    \\
    &
    &&
    \Omega^1_{\mathrm{dR}}(
      -;\,
      \mathfrak{a})_{\mathrm{clsd}}
  \end{tikzcd}
  \hspace{0cm}
  \mbox{\footnotesize
    \begin{tabular}{l}
      hence 
      \\
      schematically
      \\
      of this form:
    \end{tabular}
  }
  \quad 
  \begin{tikzcd}
 {\color{olive}    F^{(0)} }
    \ar[
      dd,
      bend left=35,
      "{ \color{darkgreen} \widehat{F}' }",
      "{\ }"{name=t, swap}
    ]
    \ar[
      dd,
      bend right=35,
      "{\color{darkgreen}  \widehat{F} }"{swap},
      "{\ }"{name=s}
    ]
    \ar[
      from=s,
      to=t,
      Rightarrow,
      "{\color{orangeii} 
        \doublehat{F}_{\phantom{.}}
      }"
    ]
    \\
    \\
  {\color{olive}   F^{(1)} }
  \end{tikzcd}
\end{equation}
\end{definition}

\smallskip

The following Lem. \ref{TrivializationOfClosedFormOnCylinder} is an elementary fact (a version of the Poincar{\'e} Lemma) 
but we make it fully explicit since it plays such a crucial role in 
Prop. \ref{CoboundariesOfClosedLS4ValuedForms}, where it governs the existence of (deformations of) potentials for the supergravity C-field
(cf. also discussion in the context of 11d supergravity on manifolds with boundary \cite{Sati12}):

\begin{lemma}[\bf Trivialization of closed forms on a cylinder]
  \label{TrivializationOfClosedFormOnCylinder}
$\,$
  
 \noindent  A closed differential form 
  on a cylinder,
  $\widehat G \,\in\, \Omega^4_{\mathrm{dR}}\big(X \times [0,1]\big)_{\mathrm{clsd}}$,
  which vanishes on one end 
  \begin{equation}
    \label{DifferentialFormVanishingOnOneEndOfCyliner}
    \iota_0^\ast 
    \,
    \widehat{G}_4
    \;=\;
    0
    \,,
  \end{equation}
  is trivialized by the differential form $\widehat{C}_3$ 
  characterized by
  \vspace{-1mm} 
  $$
    \iota_{\partial_t} \widehat{C}_3
    \;=\;
    0
    \;\;\;\;\;\;\;
    \mbox{and}
    \;\;\;\;\;\;\;
    \underset{t\in [0,1]}{\forall}
    \;\;\;
    \widehat{C}_3(-,t)
    \;=\;
    {\int_{[0,t]}} 
    \widehat{G}_4
    \,.
  $$
\end{lemma}
\begin{proof}
Uniquely decomposing the form as
\begin{equation}
  \label{DecomposingAFormOnACylinder}
  \widehat {G}_4
  \;=\;
  A_4 \,+\, \mathrm{d}t \, B_3
  \,,
  \hspace{.4cm}
  \mbox{with}
  \hspace{.3cm}
  \iota_{\partial_t} A_4 = 0
  \hspace{.3cm}
  \mbox{and}
  \hspace{.3cm}
  \iota_{\partial_t} B_3 = 0
  \,,
\end{equation}
its closure, $\mathrm{d}\widehat{G}_4 = 0$, means equivalently that
\begin{equation}
  \label{ClosureConditionOfClosedFormOnCylinder}
  \def\arraystretch{1.5}
    \mathrm{d}_X A_4 \,=\, 0
    \hspace{.5cm}
    \mbox{and}
    \hspace{.5cm}
    \mathrm{d}_{[0,1]} A_4
    \;=\;
    \mathrm{d}t
    \, 
    \mathrm{d}_X B_3
    \,.
\end{equation}

\vspace{-1mm} 
\noindent With this, we compute as follows:
$$
  \def\arraystretch{1.7}
  \begin{array}{lll}
    \mathrm{d}
 {\displaystyle \int}_{\!\!\![0,-]}
    \widehat{G}_4
    &
    \;=\;
    \big(
      \mathrm{d}_{[0,1]}
      +
      \mathrm{d}_X
    \big)
    {\displaystyle \int}_{\!\!\![0,-]}
    \mathrm{d}t'\,
    B_3(t')
    &
    \proofstep{
      by \eqref{DecomposingAFormOnACylinder}
    }
    \\
   & \;=\;
    \mathrm{d}t\, B_3
    +
   {\displaystyle \int}_{\!\!\![0,-]}
    \mathrm{d}t'
    \,
    \mathrm{d}_X B_3(t')
    &
    \proofstep{
      basic prop of integrals
    }
    \\[10pt]
 &   \;=\;
    \mathrm{d}t \, B_3
    +
   {\displaystyle \int}_{\!\!\![0,-]}
    \mathrm{d}_{[0,1]} A_4
    &
    \proofstep{
      by \eqref{ClosureConditionOfClosedFormOnCylinder}
    }
    \\[-2pt]
  &  \;=\;
    \mathrm{d}t \, B_3
    +
    A_4
    &
    \proofstep{
      Stokes
      with \eqref{DifferentialFormVanishingOnOneEndOfCyliner}
    }
    \\[-3pt]
  &  \;=\;
    \widehat{G}_4
    &
    \proofstep{
      by
      \eqref{DecomposingAFormOnACylinder}.
    }
  \end{array}
$$

\vspace{-5mm}
\end{proof}

\begin{proposition}[\bf Coboundaries for closed $\mathfrak{l}S^4$-valued forms give local C-field gauge potentials]
\label{CoboundariesOfClosedLS4ValuedForms}
$\,$

\noindent Given $(G_4, G_7) \,\in\, \Omega^1_{\mathrm{dR}}(X;\, \mathfrak{l}S^4)_{\mathrm{clsd}}$ as in \eqref{ClosedlS4ValuedForms},

\noindent {\bf (i)} there is a natural surjection 

\begin{itemize}[leftmargin=.9cm]
\item
from null-coboundaries \eqref{CoboundaryInNonabelianDeRhamCohomology} for $(G_4, G_7)$ in $\mathfrak{l}S^4$-valued de Rham cohomology, 
\item
to pairs of ordinary differential
forms
\begin{equation}
  \label{PairsOfPotentialForms}
  \left.
  \def\arraystretch{1.3}
  \begin{array}{l}
    C_3 \,\in\,
    \Omega^3_{\mathrm{dR}}(X)
    \\
    C_6 \,\in\,
    \Omega^6_{\mathrm{dR}}(X)
  \end{array}
\!\!  \right\}
  \hspace{.4cm}
  \mbox{\rm \small such that}
  \hspace{.4cm}
  \left\{\!\!
  \def\arraystretch{1.2}
  \begin{array}{l}
    \mathrm{d}
    \, 
    C_3 
      \;=\; 
    G_4
    \,,
    \\
    \mathrm{d}\, 
    C_6 
      \;=\;
    G_7 
    -
    \tfrac{1}{2} C_3 \, G_4
    \,.
  \end{array}
  \right.
\end{equation}
\end{itemize}

\noindent {\bf (ii)} This surjection respects equivalence classes, where
\begin{itemize}[leftmargin=.9cm]
\item equivalence of $\mathfrak{l}S^4$-coboundaries is by coboundaries of coboundaries (Def. \ref{CoboundariesOfCoboundaries}),
\item (gauge) equivalence of the pairs \eqref{PairsOfPotentialForms} is defined as follows: \footnote{Notice that \eqref{EquivalenceOfTraditionalCFieldGaugePotentials} is slightly more restrictive than what has been called ``gauge transformations'' of the C-field in \cite[(2.4)]{CJLP98}\cite[(3.3)]{LLPS99}\cite[(14)]{KS03}\cite[(4.9)]{Sati10}: 
Indeed, the transformations considered there are more general {\it symmetries} of the C-field, analogous to general shifts of 1-form connections by closed but possibly non-exact forms. Among these, the actual gauge transformations must satisfy an exactness condition, which for the C-field had previously been left unspecified. Our Proposition \ref{CoboundariesOfClosedLS4ValuedForms} shows that the correct C-field gauge transformations are as in \eqref{EquivalenceOfTraditionalCFieldGaugePotentials}. Notice that this coincides with \cite[(21)]{BNS04} up to an exact term.}
\begin{equation}
  \label{EquivalenceOfTraditionalCFieldGaugePotentials}
  (C_3,\, C_6)
  \;\sim\;
  (C'_3,\, C'_6)
  \hspace{.5cm}
  \Leftrightarrow
  \hspace{.5cm}
  \exists
  \,
  \left.
  \def\arraystretch{1.2}
  \begin{array}{l}
  B_2 \,\in\,\Omega^2_{\mathrm{dR}}(X)
  \\
  B_5 \,\in\,\Omega^5_{\mathrm{dR}}(X)
  \end{array}
  \!\!\! \right\}
  \;\;
  \mbox{\rm \small such that}
  \;\;
  \left\{\!\!\!
  \def\arraystretch{1.1}
  \begin{array}{l}
    \mathrm{d}\, B_2 \;=\;
    C'_3 - C_3
    \\
    \mathrm{d}\, B_5 
      \;=\; 
    C'_6 - C_6
    -
    \tfrac{1}{2}
    C'_3 \, C_3
    \,.
  \end{array}
  \right.
\end{equation}
\end{itemize}
{\bf (iii)} A section of the induced surjection on equivalence classes is given by 
\vspace{-2mm} 
\begin{equation}
  \label{LiftOfGaugeTransformations}
  \begin{tikzcd}[
    row sep=17pt
  ]
    \big(
      C_3, C_6
    \big)
    \ar[
      dd,
      "{\color{darkgreen} 
        (B_2, B_5)
      }"{description}
    ]
    &&
    \Big(
      \widehat{G}_4
      \,:=\,
      t \, G_4 + \mathrm{d}t\, C_3
      ,\;\;\;\;
      \widehat{G}_7
      \,:=\,
      t^2 \, G_7 \,+\,
      2\, t \mathrm{d}t \, C_6
    \Big)
    \ar[
      dd,
      "{\color{darkgreen} 
        \left(
          \def\arraystretch{1.1}
          \def\arraycolsep{2pt}
          \begin{array}{l}
          \doublehat{G}_4
          \,:=\,
          t \, G_4 + \mathrm{d}t\, C_3
          +
          s\, \mathrm{d}t
          (
            C'_3 - C_3
          )
          -
          \mathrm{d}s\, \mathrm{d}t
          \, 
          B_2
          \\
          \doublehat{G}_7
          \,:=\,
          t^2 \, G_7 
          + 2 \, t\mathrm{d}t
          \, C_6
          \,+\,
          2 \, s \, t \mathrm{d}t
          (
            C'_6 - C_6
          )
          \,-\,
          2
          \, \mathrm{d}s
          \, t \mathrm{d}t
          (
            B_5 + \frac{1}{2}
            B_2 \, C_3
          )
          \end{array}
       \!\! \right)
      }"{description}
    ]
    \\
    & \longmapsto
    \\
    \big(
      C'_3, C'_6
    \big)
    &&
    \Big(
      \widehat{G}'_4
      \,:=\,
      t \, G_4 + \mathrm{d}t\, C'_3
      ,\;\;\;\;
      \widehat{G}'_7
      \,:=\,
      t^2 \, G_7 \,+\,
      2\, t \mathrm{d}t \, C'_6
    \Big)
    \mathrlap{\,,}
  \end{tikzcd}
\end{equation}

\vspace{-2mm} 
\noindent
where all un-hatted differential forms denote their pullback along 
$
  \begin{tikzcd}[column sep=20pt]
  X \times [0,1]_t \times [0,1]_s 
  \ar[rr, ->>, "p_{X \times [0,1]_t}\;\;"] 
  &&  
  X \times [0,1]_t 
  \ar[r, ->>, "p_X"] 
  & 
  X\,,
  \end{tikzcd}
$ 
and where $(t,s) \colon [0,1]_t \times [0,1]_s \to \mathbb{R}^2$ denote the canonical coordinate functions.
\end{proposition}

\begin{proof}
  {\bf (i)}
  First, to describe the map itself we use the fiberwise Stokes Theorem (e.g. \cite[Lem. 6.1]{FSS23Char}) for differential forms $\widehat{F}$ on the product manifold 
  $$
    \begin{tikzcd}
      X
      \ar[
        r, 
        hook,
        shift left=3pt,
        "{\iota_0}"
      ]
      \ar[
        r, 
        hook,
        shift right=3pt,
        "{\iota_1}"{swap}
      ]
      &
      X \times [0,1]
      \,,
    \end{tikzcd}
  $$ 
  where it says that:
  \begin{equation}
    \label{FiberwiseStokes}
    \mathrm{d}
    \int_{[0,1]}
    \widehat F
    \;\;=\;\;
    \iota_1^\ast \widehat F
    \,-\,
    \iota_0^\ast \widehat F
    \,-\,
    \int_{[0,1]}
    \mathrm{d}
    \widehat{F}
    \,.
  \end{equation}
Now given a concordance 
\begin{equation}
\label{ANullConcordanceForClosedlS4ValuedForms}
\big(\widehat{G}_4\,, \widehat{G}_7\big) \,\in\, \Omega^1_{\mathrm{dR}}\big(X \times [0,1]; \, \mathfrak{l}S^4\big)_{\mathrm{clsd}}
\hspace{.4cm}
\mbox{\footnotesize with}
\hspace{.4cm}
\left\{\!\!\!
\def\arraystretch{1.5}
\begin{array}{l}
  \iota_1^\ast
  \big(
    \widehat{G}_4
    ,\,
    \widehat{G}_7
  \big)
  \;=\;
  (G_4,\, G_7)
  \,,
  \\
  \iota_0^\ast
  \big(
    \widehat{G}_4
    ,\,
    \widehat{G}_7
  \big)
  \;=\;
  0
\end{array}
\right.
\end{equation}
take its image to be
\begin{equation}
  \label{CFieldPotentialsFromlConcordances}
  \left.
  \def\arraystretch{1.9}
  \begin{array}{l}
    C_3 \,:=\,
     {\displaystyle \int}_{\!\!\![0,1]}
    \widehat{G}_4
    \\
    C_6 
    \,:=\,
    {\displaystyle \int}_{\!\!\![0,1]}
    \Bigg(
      \widehat{G}_7
      -
      \tfrac{1}{2}
      \,
      \underbrace{
      \bigg(
        \int_{[0,-]
      }
      \widehat{G}_4
      \bigg)
      }_{\color{gray}
        \widehat{C}_3
      }
      \,
      \widehat{G}_4
    \Bigg)
  \end{array}
 \!\!\!\! \right\}
  \hspace{.4cm}
  \mbox{\footnotesize which indeed satisfies}
  \hspace{.4cm}
  \left\{\!\!\!
  \def\arraystretch{1.5}
  \begin{array}{l}
    \mathrm{d}
    \,
    C_3
    \;=\;
    G_4
    \\
    \mathrm{d}
    \,
    C_6
    \;=\;
    G_7 -
    \tfrac{1}{2}
    C_3 \, G_4
    \,.
  \end{array}
  \right.
\end{equation}
Here over the brace on the left we have
$$
  \widehat{C}_3
  \;\in\;
  \Omega^3_{\mathrm{dR}}
  \big(
    X \times [0,1]
  \big)
  \,,
  \hspace{.6cm}
  \widehat{C}_3(-,t)
  \;:=\;
    \int_{[0,t]}
  \widehat{G}_4
  \,,
$$
which satisfies (by Lem. \ref{TrivializationOfClosedFormOnCylinder})
\begin{equation}
  \label{TrivializationOfCOncordance4Form}
  \mathrm{d} \widehat{C}_3
  \;=\;
  \widehat{G}_4
  \;\;\;\;\;\;
  \mbox{and}
  \;\;\;\;\;\;
  \iota_0^\ast \widehat{C}_3
  \;=\;
    \int_{[0,0]}
  \widehat{G}_4
  \;=\;
  0
  \,,
  \;\;\;\;\;
  \iota_1^\ast \widehat{C}_3
  \;=\;
    \int_{[0,1]}
  \widehat{G}_4
  \;=\;
  C_3
  \,,
\end{equation}
and on the right of \eqref{CFieldPotentialsFromlConcordances}
we computed as follows:
$$
  \def\arraystretch{1.7}
  \begin{array}{lll}
    \mathrm{d}
 {\displaystyle \int}_{\!\!\![0,1]}    
    \Bigg(
      \widehat{G}_7
      -
      \tfrac{1}{2}
      \bigg(
         {\displaystyle \int}_{\!\!\![0,-]}
        \widehat{G}_4
      \bigg)
      \widehat{G}_4
    \Bigg)
    &
    \!\!=\;
    \iota_1^\ast
    \Bigg(
      \widehat{G}_7
      -
      \tfrac{1}{2}
      \bigg(
         {\displaystyle \int}_{\!\!\![0,-]}
        \widehat{G}_4
      \bigg)
      \widehat{G}_4
    \Bigg)
    -
 {\displaystyle \int}_{\!\!\![0,1]}
    \underbrace{
    \mathrm{d}
    \Bigg(
      \widehat{G}_7
      -
      \tfrac{1}{2}
      \bigg(
          \int_{[0,-]}
        \widehat{G}_4
      \bigg)
      \widehat{G}_4
    \Bigg)
    }_{ \color{gray} = 0 }
    &
    \proofstep{
      by 
      \eqref{FiberwiseStokes}
    }
    \\[-6pt]
  & \!\! =\;
    G_7 - \tfrac{1}{2} C_3 \, G_4
    &
    \proofstep{
      by \eqref{TrivializationOfCOncordance4Form}.
    }
  \end{array}
$$
To see that this construction \eqref{CFieldPotentialsFromlConcordances} is a surjection as claimed, we demonstrate the explicit pre-images \eqref{LiftOfGaugeTransformations}:
For $\big(C_3, C_6\big)$ as in \eqref{PairsOfPotentialForms}, consider the concordance 
\begin{equation}
  \label{ConcordanceOflS4ValuedFormsFromTraditionalFormula}
  \left.
  \def\arraystretch{1.5}
  \begin{array}{l}
    \widehat{G}_4
    \;:=\;
    t\, G_4 + \mathrm{dt}\, C_3
    \\
    \widehat{G}_7
    \;:=\;
    t^2 \, G_7
    +
    2 t \mathrm{d}t
    \,
    C_6
  \end{array}
  \!\! 
  \right\}
  \hspace{0cm}
  \mbox{\footnotesize
    \def\arraystretch{.9}
    \begin{tabular}{c}
    which,  using \eqref{PairsOfPotentialForms},
    \\
indeed    satisfies
   :
    \end{tabular}
  }
  \hspace{0cm}
  \left\{\!\!
  \def\arraystretch{1.5}
  \begin{array}{l}
    \mathrm{d}
    \big(
      t \, G_4
      +
      \mathrm{d}t
      \, 
      C_3 
    \big)
    \;=\;
    0
    \\
    \mathrm{d}
    \big(
    t^2 \, 
    G_7
    +
    2 t
    \mathrm{d}t
    \,
    C_6
    \big)
    \;=\;
    \tfrac{1}{2}
    \big(
      t \, G_4
      +
      \mathrm{d}t
      \,
      C_3
    \big)
    \big(
      t \,  G_4
      +
      \mathrm{d}t
      \,
      C_3
    \big) \,,
  \end{array}
  \right. 
\end{equation}
This is indeed a preimage:
$$
  \def\arraystretch{2}
  \begin{array}{l}
    {\displaystyle \int}_{\!\!\![0,t']}
    \big(
      t \, G_4
      +
      \mathrm{d}t
      \,
      C_3
    \big)
    \;=\; 
    t' \, C_3
    \\
     {\displaystyle \int}_{\!\!\![0,1]}
    \Big(
      \underbrace{
      t^2 \, G_7
      +
      2 t \mathrm{d}t \, C_6
      }_{\color{gray}  \widehat{G}_7 }
      -
      \tfrac{1}{2}
      \underbrace{
        \mathclap{\phantom{\big(}}
        t 
        \,
        C_3
      }_{ \color{gray} 
        \widehat{C}_3 
      }
      \underbrace{
      \big(
        t\, G_4
        +
        \mathrm{d}t\, C_3
      \big)
      }_{\color{gray}  \widehat{G}_4}
    \Big)
    \;=\;
    2
    \,
    C_6
     {\displaystyle \int}_{\!\!\![0,1]}
    t\mathrm{d}t
    \;=\;
    C_6
    \,.
  \end{array}
$$

\noindent
{\bf (ii)}
To see that the construction \eqref{CFieldPotentialsFromlConcordances} respects equivalences, consider a pair of concordances
$\big(\widehat{G}_4, \widehat{G}_7\big), \big(\widehat{G}'_4, \widehat{G}'_7\big)$ as 
in \eqref{ANullConcordanceForClosedlS4ValuedForms}, with a 
concordance-of-concordances between them:
$$
  \big(\doublehat{G}_4
    ,\,
    \doublehat{G}_7
  \big)
  \;\;
  \in
  \;\;
  \Omega^1_{\mathrm{dR}}\big(
    X \times [0,1]_t \times [0,1]_s
    ;\,
    \mathfrak{l}S^4
  \big)_{\mathrm{clsd}}
  \,,
  \;\;\;\;\;\;
  \mbox{\footnotesize such that:}
  \;\;\;
  \left\{\!\!
  \def\arraystretch{1.5}
  \begin{array}{l}
    \iota_{s=1}^\ast
    \big(
      \doublehat{G}_4
      ,\,
      \doublehat{G}_7
    \big)
    \;=\;
    \big(
      \widehat{G}'_4
      ,\,
      \widehat{G}'_7
    \big)
    \,,
    \\
    \iota_{s=0}^\ast
    \big(
      \doublehat{G}_4
      ,\,
      \doublehat{G}_7
    \big)
    \;=\;
    \big(
      \widehat{G}_4
      ,\,
      \widehat{G}_7
    \big)
    \,.
  \end{array}
  \right.
$$
Then we obtain an equivalence \eqref{EquivalenceOfTraditionalCFieldGaugePotentials} between the corresponding images \eqref{CFieldPotentialsFromlConcordances} by setting
\begin{equation}
  \label{2ndCoboundariesFrom2ndConcordances}
  \def\arraystretch{1.8}
  \begin{array}{l}
  B_2 
    \;:=\;
  {\displaystyle \int}_{\!\!\!\! s \in [0,1]} \;
  {\displaystyle \int}_{\!\!\!\! t \in [0,1]}
  \doublehat{G}_4
  \\
  B_5 
    \;:=\;
 {\displaystyle \int}_{\!\!\!\! s \in [0,1]}\;
 {\displaystyle \int}_{\!\!\!\! t \in [0,1]}
  \Bigg(
    \doublehat{G}_7
    -
    \tfrac{1}{2}
    \bigg(
    {\displaystyle \int}_{\!\!\!\! t' \in [0,-]}
    \doublehat{G}_4
    \bigg)
    \doublehat{G}_4
  \Bigg)
  -
  \tfrac{1}{2}
  B_2 \, C_3
  \,.
  \end{array}
\end{equation}
To see that this pair satisfies the condition \eqref{EquivalenceOfTraditionalCFieldGaugePotentials}
we repeatedly use the fiberwise Stokes theorem \eqref{FiberwiseStokes} to find, first:
\begin{equation}
  \label{DifferentialOfDoublyIntegrated2nd4FormConcordance}
  \def\arraystretch{1.9}
  \begin{array}{ll}
  &
  \mathrm{d}
  \,
{\displaystyle \int}_{\!\!\!\! s \in [0,1]}\;
 {\displaystyle \int}_{\!\!\!\! t \in [0,1]}
  \,
  \doublehat{G}_4
  \\
  &
  \!\!=\;
  \iota_{s=1}^\ast
  {\displaystyle \int}_{\!\!\!\! t \in [0,1]}
  \,
  \doublehat{G}_4
  \;-\;
  \iota_{s=0}^\ast
{\displaystyle \int}_{\!\!\!\! t \in [0,1]}
  \,
  \doublehat{G}_4
  \;-\;
 {\displaystyle \int}_{\!\!\!\! s \in [0,1]}
  \mathrm{d}
  {\displaystyle \int}_{\!\!\!\! t \in [0,1]}
  \,
  \doublehat{G}_4
  \\[6pt]
  & \!\!=\;
 {\displaystyle \int}_{\!\!\!\! t \in [0,1]}
  \iota_{s=1}^\ast
  \doublehat{G}_4
  \;-\;
  {\displaystyle \int}_{\!\!\!\! t \in [0,1]}
  \iota_{s=0}^\ast
  \doublehat{G}_4
  \;-\;
  {\displaystyle \int}_{s \in [0,1]}
  \underbrace{
  \Big(
  \iota_{t=1}^\ast
  \doublehat{G}_4
  -
  \iota_{t=0}^\ast
  \doublehat{G}_4
  \Big)
  }_{ \color{gray} 
    = \, 0
    \mathrlap{
    \;\; 
    \scalebox{.7}{\eqref{2ndConcordanceConstantAtBoundary}}
    }
  }
  +
 {\displaystyle \int}_{\!\!\!\! s \in [0,1]}\;
 {\displaystyle \int}_{\!\!\!\! t \in [0,1]}
  \underbrace{
    \mathrm{d} 
    \,
    \doublehat{G}_4
  }_{=\, 0}
  \\
  & \!\! =\;
 {\displaystyle \int}_{\!\!\!\! t \in [0,1]}
  \widehat{G}'_4
  -
  {\displaystyle \int}_{\!\!\!\! t \in [0,1]}
  \widehat{G}_4
  \\
 & \!\!=\;
  C'_3 \,-\, C_3
  \,.
  \end{array}
\end{equation}
The computation for $B_5$ is similar, only that here the term $\int_{s \in [0,1]} \int_{t' \in [0,-]} \doublehat{G}_4$ survives the evaluation at $t = 1$:
\begin{equation}
  \label{DifferentialOfDoublyIntegrated2nd7FormConcordance}
  \def\arraystretch{2}
  \begin{array}{ll}
  &
  \mathrm{d}
  {\displaystyle \int}_{\!\!\!\! s \in [0,1]}\;
  {\displaystyle \int}_{\!\!\!\! t \in [0,1]}
  \Bigg(
    \doublehat{G}_7
    -
    \tfrac{1}{2}
    \bigg(
    {\displaystyle \int}_{\!\!\!\! t' \in [0,-]}
    \doublehat{G}_4
    \bigg)
    \doublehat{G}_4
  \Bigg)
  \\
  &
  \!\!=\;
  {\displaystyle \int}_{\!\!\!\! t \in [0,1]}
  \Big(
    \widehat{G}'_7
    -
    \tfrac{1}{2}
    \widehat{C}'_3
    \, 
    \widehat{G}'_4
  \Big)
  \;-\;
 {\displaystyle \int}_{\!\!\!\! t \in [0,1]}
  \Big(
    \widehat{G}_7
    -
    \tfrac{1}{2}
    \widehat{C}_3
    \, 
    \widehat{G}_4
  \Big)
  -
  \displaystyle{\int}_{
    \!\!\!\! s \in [0,1]}
  \mathrm{d}
  \displaystyle{\int}_{
    \!\!\!\! t \in [0,1]}
  \Bigg(
    \doublehat{G}_7
    -
    \tfrac{1}{2}
    \bigg(
    {\displaystyle \int}_{\!\!\!\! t' \in [0,-]}
    \doublehat{G}_4
    \bigg)
    \doublehat{G}_4
  \Bigg)
  \\
 & \!\!=\;
  C'_6 
    \,-\, 
  C_6
    \,+\,
  \tfrac{1}{2}
  \left(
  \displaystyle{\int}_{
    \!\!\!\! s \in [0,1] 
  }
  \displaystyle{\int}_{
   \!\!\!\! t \in [0,1]} 
  \doublehat{G}_4
  \right)
  \,
  G_4
  \\
  &
  \;=\;
  C'_6 
    \,-\, 
  C_6
    \,+\,
  \tfrac{1}{2} B_2 \, G_4
  \,.
  \end{array}
\end{equation}
Hence, in total, we have
$$
  \def\arraystretch{1.8}
  \begin{array}{lll}
    \mathrm{d}B_5
    &=\;
    \mathrm{d}
 {\displaystyle \int}_{\!\!\!\! s \in [0,1]}\;
 {\displaystyle \int}_{\!\!\!\! t \in [0,1]}
  \Bigg(
    \doublehat{G}_7
    -
    \tfrac{1}{2}
    \bigg(
    {\displaystyle \int}_{\!\!\!\! t' \in [0,-]}
    \doublehat{G}_4
    \bigg)
    \doublehat{G}_4
  \Bigg)
  -
  \mathrm{d}
  \tfrac{1}{2}
  B_2 \, C_3
    &
    \proofstep{
      by
      \eqref{2ndCoboundariesFrom2ndConcordances}
    }
    \\
 &   \;=\;
    C'_6 
      \,-\, 
    C_6 
      \,+\, 
    \tfrac{1}{2} B_2 \, G_4
    \,
    \underbrace{
    -
    \tfrac{1}{2}
    \big(
      C'_3 - C_3
    \big)
    C_3
    -
    \tfrac{1}{2}
    B_2 \, G_4
    }_{ \color{gray} 
      \mathrm{d}(
        - \frac{1}{2}
        B_2 \, C_3
      )
    }
    &
    \proofstep{
      by 
      \eqref{DifferentialOfDoublyIntegrated2nd7FormConcordance}
      \& 
      \eqref{DifferentialOfDoublyIntegrated2nd4FormConcordance}
    }
    \\[-5pt]
   & \;=\;
    C'_6 - C_6
    -
    \tfrac{1}{2}
    C'_3 \, C_3
    \,,
  \end{array}
$$
as required \eqref{EquivalenceOfTraditionalCFieldGaugePotentials}.

\smallskip

\noindent
{\bf (iii)}
Finally, for $(B_2, B_5)$ a gauge transformation \eqref{EquivalenceOfTraditionalCFieldGaugePotentials} we show that the pair $\big(\doublehat{G}_4,\, \doublehat{G}_7\big)$ \eqref{LiftOfGaugeTransformations} is a concordance-of-concordances between the concordances \eqref{ConcordanceOflS4ValuedFormsFromTraditionalFormula}: It is clear that the pullbacks to $s,t \in \{0,1\}$ are as required, and checking the Bianchi identities is straightforward: First we have

$
\hspace{1cm} 
  \def\arraystretch{1.6}
  \begin{array}{lclll}
    \mathrm{d}
    \big(
      t\, G_4
      +
      \mathrm{d}t\, 
      C_3
    \big)
    &=&
    0
    &
    \proofstep{
      by
      \eqref{ConcordanceOflS4ValuedFormsFromTraditionalFormula}
    }
    \\
    \mathrm{d}\left(
      s
      \,
      \mathrm{d}t
      (C'_3 - C_3)
    \right)
    &=&
    \phantom{+}
    \mathrm{d}s\, \mathrm{d}t
    \big(
      C'_3 - C_3
    \big)
    &
    \proofstep{
      by \eqref{PairsOfPotentialForms}.
    }
    \\
    \mathrm{d}\big( 
      - \mathrm{d}s\, \mathrm{d}t
      \, B_2
    \big)
    &=&
    - \mathrm{d}s\, \mathrm{d}t
    \big(
     C'_3 - C_3
    \big)
    &
    \proofstep{
      by
      \eqref{EquivalenceOfTraditionalCFieldGaugePotentials}
    }
    \\
    \cline{1-3}
    {\color{darkgreen}
    \mathrm{d}\big(
      \doublehat{G}_4
    \big)}
    & 
    {\color{darkgreen} =}
    &
    {\color{darkgreen}  0}
    &
    \proofstep{
      by 
      \eqref{LiftOfGaugeTransformations},
    }
  \end{array}
$

\vspace{.2cm}
\noindent
and, using from \eqref{PairsOfPotentialForms}
that
\begin{equation}
  \label{DifferentialOfC6Difference}
  \mathrm{d}(C'_6 - C_6)
  \;=\;
  -
  \tfrac{1}{2}\big(
     C'_3
    -
    C_3
  \big)
  G_4
  \,,
\end{equation}
we have:
$$
\hspace{1cm} 
  \def\arraystretch{1.8}
  \begin{array}{lcll}
    \mathrm{d}
    \big(
      t^2
      G_7
      +
      2t \mathrm{d}t\, C_6
    \big)
    &=&
    \phantom{+}
    \tfrac{1}{2}
    \big(
     t\, G_4
     +
     \mathrm{d}t\,
     C_3
    \big)^2
    &
    \proofstep{
      by 
      \eqref{ConcordanceOflS4ValuedFormsFromTraditionalFormula}
    }
    \\
    \mathrm{d}
    \big(
      2 s \, t \mathrm{d}t
      (C'_6 - C_6)
    \big)
    &=&
    \phantom{+}
    s\, t \mathrm{d}t
    \big(
      C'_3 - C_3
    \big)
    G_4
    \,+\,
    2 
    \, \mathrm{d}s
    \, t\mathrm{d}t
    (C'_6 - C_6)
    &
    \proofstep{
      by
      \eqref{DifferentialOfC6Difference}
    }
    \\
    \mathrm{d}
    \big(
      - 2 
      \, \mathrm{d}s
      \, t\mathrm{d}t
        B_5
    \big)
    &=&
    \hspace{3.2cm}
    -
    \,
    2 
    \, \mathrm{d}s 
    \, t \mathrm{d}t
    \big(
      C'_6 - C_6
    \big)
    \,+\, 
    \mathrm{d}s 
    \, t\mathrm{d}t
    \,
    C'_3 \, C_3
    &
    \proofstep{
      by 
      \eqref{EquivalenceOfTraditionalCFieldGaugePotentials}
    }
    \\
    \mathrm{d}\big(
      -
      \mathrm{d}s 
      \, t\mathrm{d}t
      \,
      B_2 \, C_3
    \big)
    &=&
    -
    \mathrm{d}s 
    \, t\mathrm{d}t
    \,
    B_2 \, G_4
    \hspace{4.2cm}
    -\,
    \mathrm{d}s 
    \, t\mathrm{d}t
    \,
    C'_3 \, C_3
    &
    \proofstep{
      by
      \eqref{EquivalenceOfTraditionalCFieldGaugePotentials}
      \&
      \eqref{PairsOfPotentialForms}
    }
    \\
    \cline{1-3}
    {\color{darkgreen}   
    \mathrm{d}
    \big(
    \doublehat{G}_7
    \big)
    }
    &
    {\color{darkgreen} =}
    &
    {\color{darkgreen}   
    \tfrac{1}{2}
    \doublehat{G}_4
    \,
    \doublehat{G}_4}
    &
    \proofstep{
      by
      \eqref{LiftOfGaugeTransformations} 
      \hspace{3.3cm}
      \qedhere
    }
  \end{array}
$$
\end{proof}

\medskip 
 
\subsubsection{Super Field Spaces}
\label{SuperSmoothFieldSpaces}

We discuss a notorious subtle issue of supergeometry, which is key both to the conceptual foundations of the subject as well as to its relation with observable physical reality --- such as to the question of what it actually means to {\it observe} a gravitino (or any classical fermion, for that matter). Nevertheless, since the point is somewhat tangential to our main results in \S\ref{SuperFluxQuantization} and \S\ref{Supergravity}, the reader who does not feel like bothering with the following slightly more topos-theoretic discussion can safely skip it, while we refer
the reader looking for more elaboration to \cite{GSS24}.

\medskip

The archetypical practical example of differential forms with coefficients (Def. \ref{SuperDifferentialFormsWithCoefficients}) in an odd vector space (Ex. \ref{PurelyOddVectorSpaces}) are classical fermionic fields with values in a Spin-representation $\mathbf{N}_{\mathrm{odd}}$ (regarded in odd degree). Specifically the {\it gravitino} field \eqref{GravitonAndGravitino}
in supergravity (considered in \S\ref{SuperSpaceTime} and brought to life in \S\ref{Supergravity}) is a differential 1-form with coefficients in some $\mathbf{N}_{\mathrm{odd}}$.
While a key move in \S\ref{SuperFluxQuantization} and \S\ref{Supergravity} is to discuss supergravity not on ordinary spacetime manifolds $\bos X$, but on super-manifold enhancements $\bos{X} \hookrightarrow X$  thereof, where such odd-valued 1-forms may exist (Ex. \ref{Super1FormsWithCoefficients}) as ordinary elements 
$$
  \psi 
    \,\in\, 
  \Omega^1_{\mathrm{dR}}
  \big(
    X;\, 
    \mathbf{N}_{\mathrm{odd}}
  \big)
  \,,
$$
one also does want to speak of (fermions in general and particularly) gravitinos  on an ordinary spacetime $\bos{X}$, cf. \eqref{OddFluxOnBosonicSupermanifold}. However, by Ex. \ref{OddFormsOnEvenManifold}, these do {\it not exist} as ordinary such elements, since the only element of $\Omega^1_{\mathrm{dR}}\big(\!\bos{X};
\, \mathbf{N}_{\mathrm{odd}}\big)$ is the zero 1-form.

\smallskip 
This notorious issue, which (in some guise or other) has occupied authors of texts on classical fermionic field theory in general and of supersymmetric field theory in particular, is naturally solved by our passage from plain sets to super sets (Def. \ref{SuperSmoothSets}). Namely the issue with Def. \ref{DifferentialFormsOnSuperCartesianSpaces} is simply that it defines only the {\it ordinary set} of (specifically) odd 1-forms, while these should clearly form a whole {\it super-set} instead: The odd forms on a bosonic manifold should not be the ordinary but the ``odd elements'' of the super-set that they form. 

\smallskip 
The mathematics that makes this idea a reality is known in category-theory (Rem. \ref{CategoryTheoryInTheBackground}) as the {\it internal hom}-construction (for exposition and pointers see \cite{Schreiber24}, for more on the bosonic analog see \cite[\S 2.2]{GS23}). We briefly spell this out in simple terms:
 
\begin{definition}[\bf{Smooth super mapping set}]\label{InternalHom}Let $X,Y \, \in\,  \mathrm{sSmthSet}$ be super smooth sets. The \textit{smooth super mapping set} $[X,Y]\, \in  \, \mathrm{sSmthSet}$ is defined by the assignment of plots
$$
[X,Y](\FR^{n|q}) \, \, := \, \, \mathrm{Hom}_{\mathrm{sSmthSet}}(X\times \FR^{n|q}, \, F )  \, ,
$$
where $\FR^{n|q}$ is viewed as a super smooth set via Ex. \ref{SmoothManifoldsAsSmoothSuperSets}.
\end{definition}

That is, the object $[X,Y]$ encodes not only the bare set of morphisms from $X$ to $Y$ via $[X,Y](*)$, but further \textit{defines} the smooth and super structure of the corresponding \textit{space} of morphisms. In particular, for two super-manifolds $X,Y \, \in \, \sManifolds$, this construction yields
$$
[X,Y](\FR^{n|q}) \cong  \mathrm{Hom}_{\sManifolds}(X\times \FR^{n|q}, \, Y )  \, ,
$$

\smallskip 
\noindent by the fully faithfulness of the embedding $\sManifolds\hookrightarrow \mathrm{sSmthSet}$ (Ex. \ref{SmoothManifoldsAsSmoothSuperSets}). This smooth super set may be interpreted, for instance, as the correct model for the smooth super field space of $\sigma$-models on a super-manifold $X$ with target a supermanifold $Y$. The $\FR^{n|q}$-plots given by $ \mathrm{Hom}_{\sManifolds}(X\times \FR^{n|q}, \, Y ) $ have the natural intepretation of (smoothly) $\FR^{n|q}$-parametrized maps from $X$ to $Y$.

\begin{example}[\bf{Fermionic scalar field space}]\label{FermionicScalarFields}
Consider the case of a bosonic manifold $\bos{X}\in \sManifolds$. Similar to Ex. \ref{OddFormsOnEvenManifold}, the set of maps from $\bos{X}$ to an $\textit{odd}$ vector space $V_\mathrm{odd}$ is trivial
$$
\mathrm{Hom}_{\sManifolds}\big(\bos{X}, \, V_{\mathrm{odd}}\big) \cong \mathrm{Hom}_{\mathrm{sCAlg}}\big(C^\infty(V_\mathrm{odd}), \, C^\infty(\bos{X})\big)\cong 0 \, .
$$
It follows that this bare 
\textit{set} is not quite the correct model for odd scalar (vector) fields on a bosonic spacetime, which however do appear non-trivially in the theoretical physics literature. The resolution is that the smooth super mapping set $[\bos{X}, V_\mathrm{odd}]$ is non-trivial. In particular, by computing the morphisms dually in the algebra picture, the $\FR^{0|1}$-plots of the field space are
\smallskip 
\begin{align}\label{1AuxFermCoord}
[\bos{X},V_\mathrm{odd}] \big(\FR^{0|1}\big) \cong \mathrm{Hom}_{\sManifolds}\big(\bos{X} \times \FR^{0|1}, \, V_\mathrm{odd}\big) \cong \big(C^\infty(\bos{X})\otimes \FR[\theta] \otimes V_\mathrm{odd}\big)_0 \, ,
\end{align}
which may be further identified with a copy of the usual bosonic $V$-valued smooth maps 
$$
C^\infty(\!\bos{X})\otimes V \cong \mathrm{Hom}_{\sManifolds}\big(\bos{X} , \, V \big)\, .
$$ 
General $\FR^{n|q}$-plots may be computed similarly to give
$$
\big[
  \bos{X}
  ,\,
  V_\mathrm{odd}
\big] 
(\FR^{n|q}) 
\;\cong\; 
\Big(
  C^\infty (\!\bos{X})
  \,\hat{\otimes}\, 
  C^\infty(\FR^{n|0}) 
  \,\otimes\, 
  C^\infty(\FR^{0|q}) 
  \,\otimes\, 
  V_\mathrm{odd}\Big)_0 
  \, , 
$$
where $\hat{\otimes}$ denotes the (completed) projective tensor product, so in particular $C^\infty\big(\!\bos{X}\times \FR^{n|0}\big)\cong C^\infty\big(\!\bos{X}\big) \,\hat{\otimes}\, C^\infty(\FR^{n|0})$.
\end{example}
\begin{remark}[\bf{Auxilliary fermionic coordinates}]\label{AuxilliaryFermionicCoordinates}
The extra odd `$\theta$-coordinates' appearing in the $\FR^{0|q}$-plots of super-field spaces, as in Eq. \eqref{1AuxFermCoord}, are often referred to as ``auxilliary fermionic coordinates" (see \cite{DeWitt}\cite{Rogers} \cite{Sachse08}). These are traditionally invoked in an ad-hoc manner to make certain polynomial formulas in a fermionic field $\psi$ non-trivial, that would otherwise vanish from the point-set perspective, as it happens for instance in $\big[\bos{X}, V_\odd\big](*)$ from Ex. \ref{FermionicScalarFields}. But our sheaf-topos of smooth super sets provides a natural interpretation for their appearance: they are nothing but the content of plots of the corresponding smooth super field space. It follows that the symbol $\psi$ used for fermionic fields
implicitly refers to an arbitrary $\FR^{0|q}$-plot of the smooth super field space (at the ``$q^{\mathrm{th}}$ Grassmann-stage''), or more generally, an arbitrary $\FR^{p|q}$-plot. Formulas made out of these are implicitly functorial under maps in the probe site, that is, they may naturally be interpreted as natural transformations -- maps within the category $\sSmoothSets$.
\end{remark}

In the present article, the relevant smooth super field spaces are those corresponding to forms on a super-manifold $X$, valued in super (graded) vector spaces (Def. \ref{SuperDifferentialFormsWithCoefficients}), and more generally to (closed) forms valued in super $L_\infty$-algebras (Def. \ref{ClosedLInfinityValuedDifferentialForms}). To employ the internal hom construction, one must first identify the bare set of form fields with an appropriate hom-set. One such option is as a hom-set into a classifying space from Eq. \eqref{ClassifyingSuperSetOfDifferentialForms}. However, it is easy to see that the internal hom-object
\begin{align}\label{WrongInternalHom}
\big[
  X
  ,\, 
  \Omega^1_{\mathrm{dR}}
    (
      -;
      \, 
      V
    ) 
  \big] 
  \quad \in \quad \sSmoothSets
\end{align}
does not yield the correct notion $\FR^{n|q}$-parametrized $V$-valued forms on $X$. Indeed, by the Yoneda 
Lemma the $\FR^{n|q}$-plots of this smooth super set are 
$$ \mathrm{Hom}_{\sSmoothSets}\big(
    X\times \FR^{n|q}
    ,\, 
    \Omega^1_{\mathrm{dR}}
    (
      -;
      \, 
      V
    )
  \big) \cong \Omega^1_{\mathrm{dR}}( X\times \FR^{n|q};\, V) \, ,
$$
elements of which have `form-legs along the probe space' $\FR^{n|q}$. The same issue arises, verbatim, even in the purely smooth setting (\cite{dcct}\cite{GS23}) with $V$ an even (ungraded) vector space. 

\medskip 
For the case of $V$ being a super vector space concentrated in degree 0, there are at least two (equivalent) ways to obtain the correct super smooth set structure on such form-field spaces: 
\smallskip
\begin{itemize}
\item[\bf{(i)}] By applying a certain ``concretification'' functor \cite{dcct} on the internal hom set \eqref{WrongInternalHom}, which essentially removes the form-legs along the probe space $\FR^{n|q}$. 

\item[\bf{(ii)}] By identifying the bare set of  forms with a different hom-set, that is, as fiber-wise maps out of the tangent bundle (Eq. \eqref{VectorValuedFormsAsBundleMaps}), and then considering the corresponding (fiber-wise linear) internal hom subobject (\cite{GSS24}) 
$$ [TX,V]^\mathrm{fib.lin.}\longhookrightarrow [TX,V] $$
defined by, under the Yoneda Lemma,
$$ [TX,V]^\mathrm{fib.lin.}(\FR^{n|q}):=  \mathrm{Hom}_{\sManifolds}^{\mathrm{fib.lin.}}\big(T X\times \FR^{n|q}, \, V \big) \, .$$
Spelling this out explicitly, dually in terms of function super-algebras, yields 
$$
[TX,V]^\mathrm{fib.lin.}(\FR^{n|q})   \cong
    \Big(
      \Omega^1_{\mathrm{dR}}
      (X
      )
      \,\hat{\otimes}\,
      C^\infty(  \FR^{n \vert q})\otimes 
      (V^\vee)^\ast
    \Big)_{(1,0)}\, ,
$$
\end{itemize}
which we take as the definition. Arguing along similar lines \cite{GSS24}, this naturally extends to a definition of the field space corresponding to $\mathbb{Z}$-graded super vector space valued 1-forms. 
\begin{definition}[\bf{Vector valued form field space}]
Given $X\in \sManifolds$ and $V\in \mathrm{sgMod}$, the smooth super field space  
 $\bfOmega^1_{\mathrm{dR}}\big(
    X
    ;\,
    V
  \big)\, \in \, \sSmoothSets$ of forms on $X$ valued in $V$ is defined by 
  $$
    \bfOmega^1_{\mathrm{dR}}
    (
    X
    ;\,
    V 
  )(\FR^{n|q}) 
  \, \, := \, \, 
  \Big(
      \Omega^\bullet_{\mathrm{dR}}
      (X
      )
      \,\hat{\otimes}\,
      C^\infty(  \FR^{n \vert q})\otimes 
      (V)^\ast
    \Big)_{(1,0)}\, . $$
\end{definition}

\begin{example}[\bf{Fermionic gravitino field space}]\label{FermionicGravitinoFieldSpace}
Part of the (off-shell) field space of supergravity consists of the  fermionic gravitino $\psi$, being a $1$-form valued in the odd (ungraded) vector space $\mathbf{32}_\odd$. When supergravity is formulated on a \textit{bosonic} spacetime $\bos{X}$, the bare \textit{set} such 1-forms vanishes (Ex.  \ref{OddFormsOnEvenManifold}). Nevertheless, the corresponding smooth super field space of the gravitino is non-trivial with 
\smallskip 
 $$
 \bfOmega^1_{\mathrm{dR}}\big(\!
    \bos{X}
    ;\,
    \mathbf{32}_\odd  
  \big)(\FR^{n|q}) \, \, = \, \, 
      \Big( \Omega^1_{\mathrm{dR}}
      (\!\bos{X}
      )
      \,\hat{\otimes}\,
      C^\infty(  \FR^{n \vert q})\otimes 
      \mathbf{32}_\odd 
    \Big)_{0}\, . 
$$
Thus, in the canonical (odd) basis for $\mathbf{32}_\odd$, we may write
$$
\psi^\alpha = \psi^\alpha{}_r \cdot \dd x^r
$$ for the odd $1$-form component of an arbitrary $\FR^{n|q}$-plot $\psi$ of the gravitino field space, and hence with $\psi^\alpha{}_r \in (C^\infty(\bos{X})\,\hat{\otimes}\, C^{\infty}(\FR^{n|q})\big)_1$ being implicitly an odd $\FR^{n|q}$-parametrized function on $\bos{X}$.

Along the same lines, it follows that the pullback of (super) gravitino 1-forms  along 
$\eta_X :  \bos{X}\hookrightarrow X$ is necessarily the trivial 0-map at the point-set level, but is nevertheless non-trivial 
as a map of super smooth field spaces 
\smallskip 
\begin{align*}
\eta_X^* \; : \; \bfOmega^1_{\mathrm{dR}}\big(
    X
    ;\,
    \mathbf{32}_\odd  
  \big) &
  \; \longrightarrow \; \bfOmega^1_{\mathrm{dR}}\big(\!
    \bos{X}
    ;\,
    \mathbf{32}_\odd  
  \big)\,  
  \\[3pt] 
   \psi^\alpha{}_\evencoordinateindex 
     \cdot 
   \dd x^\evencoordinateindex 
     + 
   \psi^\alpha{}_{\oddcoordinateindex} 
     \cdot 
   \dd \theta^\oddcoordinateindex 
   & 
   \; \longmapsto \;
   \psi^\alpha{}_\evencoordinateindex|_{\theta^\oddcoordinateindex=0}
     \cdot 
   \dd x^\evencoordinateindex
   \,,
\end{align*}
with the understanding that the coefficient functions are implicitly $\FR^{n|q}$-parametrized. Despite the above delicate details,
we shall conform with the standard theoretical physics literature and denote the super smooth field spaces by the corresponding
set-theoretic symbols.
\end{example}

This argument further and naturally generalizes \cite{GSS25} to a definition of the field space corresponding to (closed) $L_\infty$-algebra valued forms (Def. \ref{ClosedLInfinityValuedDifferentialForms}).
\begin{definition}[\bf{$L_\infty$-algebra valued form field space}]
Given $X\in \sManifolds$ and $\mathfrak{a}\,\in\, \mathrm{shLAlg}^{\mathrm{ft}}$, the smooth super field space  
 $\bfOmega^1_{\mathrm{dR}}(
    X
    ;\,
    \mathfrak{a}
  )_{\mathrm{clsd}}\, \in \, \sSmoothSets$
  of (closed) forms on $X$ valued in $\mathfrak{a}$ is defined by 
  \smallskip 
  $$
    \bfOmega^1_{\mathrm{dR}}
    \big(
    X
    ;\,
    \mathfrak{a}  
  \big)_{\mathrm{clsd}}(\FR^{n|q}) 
  \, \, := \, \, 
  \mathrm{MC} \Big(
      \Omega^\bullet_{\mathrm{dR}}
      (X
      )
      \,\hat{\otimes}\,
      C^\infty(  \FR^{n \vert q})\otimes 
      (\mathfrak{a}^\vee)^\ast
    \Big)_{0}\, . $$
\end{definition}
For the particular choice of a bosonic 11-dimensional spacetime $\!\bos{X}$ and $\mathfrak{a}= \mathfrak{l}S^4$,
this field space accommodates the spacetime bosonic fluxes from \eqref{OddFluxOnBosonicSupermanifold}, 
including the possible spacetime gravitino contributions at higher Grassmann-stage (Ex. \ref{FermionicScalarFields}).

\medskip

\subsubsection{Super Moduli Stacks}
\label{SuperModuliStacks}

We briefly explain  higher structures in supergeometry that allow for the discussion of flux quantization on super-spacetimes (\S\ref{SuperFluxQuantization}).
With the above ``functorial'' formulation (in the terminology of \cite{Grothendieck73}) of super-geometry in hand, it is now fairly immediate to generalize further to {\it higher} supergeometry. The basic facts from homotopy theory that we need here are all surveyed in \cite[\S 1]{FSS23Char}.
Not to overburden this paper with abstract machinery, we here briefly {\it introduce} and motivate
the key concepts of $\infty$-groupoids and $\infty$-stacks directly by way of the relevant example of moduli of super-flux densities.

\begin{example}[\bf Moduli of super-flux deformations]
\label{InfinityGroupoidOfFluxDeformations}
Given a super-manifold $X$ and a super $L_\infty$-algebra $\mathfrak{a}$,
in \eqref{CoboundaryInNonabelianDeRhamCohomology} we considered deformation paths of closed $\mathfrak{a}$-valued differential forms on $X$ looking as follows:
\begin{equation}
\label{DeformationPathInSimplicialForm}
\adjustbox{}{
$
  \begin{tikzcd}[row sep=13pt]
    \overset{
      \mathclap{
        \raisebox{10pt}{
          \scalebox{.7}{
            \color{darkblue}
            \bf
            \def\arraystretch{.9}
            \begin{tabular}{c}
              Deformation paths
              \\
              of flux densities
            \end{tabular}
          }
        }
      }
    }{
    \Omega^1_{\mathrm{dR}}
    \big(
      X
      \times
      [0,1]
      ;\, 
      \mathfrak{a}
    \big)_{\mathrm{clsd}}
    }
    \ar[
      dd,
      shift left=22pt,
      "{
        (-)_0
      }"{description},
      "{
        \scalebox{.6}{
          \color{darkgreen}
          \bf
          \def\arraystretch{.9}
          \begin{tabular}{c}
            take starting point
            \\
            of deformation path
          \end{tabular}
        }
      }"
    ]
    \ar[
      dd,
      shift right=22pt,
      "{
        (-)_1
      }"{description},
      "{
        \scalebox{.6}{
          \color{darkgreen}
          \bf
          \def\arraystretch{.9}
          \begin{tabular}{c}
            take endpoint of
            \\
            deformation path
          \end{tabular}
        }
      }"{swap, xshift=-4pt}
    ]
    \ar[
      from=dd,
      "{
        \mathrm{pr}_X^\ast
      }"{description}
    ]
    \\
    \\
    \underset{
      \mathclap{
        \raisebox{-10pt}{
          \scalebox{.7}{
            \color{darkblue}
            \bf
            \def\arraystretch{.9}
            \begin{tabular}{c}
              Flux densities satisfying
              \\
              their Bianchi identities
            \end{tabular}
          }
        }
      }
    }{
    \Omega^1_{\mathrm{dR}}
      (
        X
        ;\, 
        \mathfrak{a}
      )_\mathrm{clsd}
    }
  \end{tikzcd}
$
}
\end{equation}
But in higher gauge theory it is clearly relevant to consider not just deformations of fluxes, but also  deformations-of-defomations, and so on. An immediate idea may be to model these as closed forms parametrized over higher cubes $[0,1]^n$, but for subtle technical reasons it turns out to be equivalent but more tractable to parametrize over the higher-dimensional analogs of triangles and tetrahedra, instead, called higher ``simplices'', denoted:
\begin{equation}
  \label{GeometricSimplices}
  \Delta^n_{\mathrm{geo}}
  \;:=\;
  \Big\{
    (x^0, x^1, \cdots, x^n)
    \in 
    \big(
      \mathbb{R}_{\geq 0}
    \big)^n
    \;\Big\vert\;
    \textstyle{\sum_{i = 0}^n} 
    \,
    x^i \,=\, 1
  \Big\}
  \hspace{.6cm}
\raisebox{-45pt}{
\begin{tikzpicture}[scale=0.7]
\begin{scope}[shift={(-7,-.1)}]

\draw[
  |-Latex
] 
(0,0)-- (2.6,0);

\draw[fill=black] 
  (1.8,0) circle (0.055);

\node
  at (2.3,+.2)
  {
    \scalebox{.7}{
      $x^0$
    }
  };

\node
  at (1.8,-.35)
  {
    \scalebox{.7}{
      \color{darkblue}
      $\Delta^0_{\mathrlap{\mathrm{geo}}}$
    }
  };
\end{scope}

\begin{scope}[shift={(-3.4,-.8)}]

\draw[
  -Latex
]
  (0,0) -- (2.6,0);

\draw[
  -Latex
]
  (0,0) -- (0,2.6);

\draw[
  line width=1.3
]
  (1.7,0) -- (0,1.7);

\draw[fill=black]
  (1.7,0) circle (.055);
\draw[fill=black]
  (0,1.7) circle (.055);

\node
  at (2.3,+.2)
  {
    \scalebox{.7}{$x^0$}
  };

\node
  at (+.24, 2.3)
  {
    \scalebox{.7}{$x^1$}
  };

\node
 at (1.1,1.1)
 {
   \scalebox{.7}{
     \color{darkblue}
     $\Delta^1_{\mathrlap{\mathrm{geo}}}$
   }
 };

\end{scope}

\begin{scope}

 \draw[
   -Latex
  ]
   (0,0)
   --
   (90:2.6);

\draw[
  -Latex
]
   (0,0)
    --
   (20-10:3);

\draw[
  -Latex
]
   (0,0)
    --
   (-40-10:2.5);

\draw[
  draw=black,
  line width=1.3,
  fill=gray,
  fill opacity=.4
]
  (90:1.8)
  --
  (20-10:2)
  --
  (-40-10:1.7)
  --
  cycle;

\draw[fill=black]
  (90:1.8) circle (.055);
\draw[fill=black]
  (10:2) circle (.055);
\draw[fill=black]
  (-50:1.7) circle (.055);

\node
  at (-50+6:2.2)
  {
    \scalebox{.7}{$x^0$}
  };
\node
  at (10+4:2.7)
  {
    \scalebox{.7}{$x^2$}
  };
\node
  at (90-5:2.35)
  {
    \scalebox{.7}{$x^1$}
  };

\node
  at (1.85,-.5)
  {
    \scalebox{.7}{
      \color{darkblue}
      $\Delta^2_{\mathrlap{\mathrm{geo}}}$
    }
  };

\end{scope}
\end{tikzpicture}
}
\end{equation}

In terms of these higher simplices, the system of deformations-of-deformations of closed $\mathfrak{a}$-valued differential forms, extending \eqref{DeformationPathInSimplicialForm} to higher order, looks as follows, where we now leave the parameter super-space unspecified, denoted just by a blank:

\begin{equation}
  \label{SystemOfHigherDeformations}
  \shape
  \,
  \Omega^1_{\mathrm{dR}}\big(
    -;
    \mathfrak{a}
  \big)_{\mathrm{clsd}}
  \;
  =
  \;
  \left(
  \hspace{-8pt}
  \adjustbox{}{
  \begin{tikzcd}[row sep=12pt, column sep=large]
    &[50pt]
    {}
    \ar[
      d,
      -,
      dotted,
      shift right=24
    ]
    \ar[
      d,
      -,
      dotted,
      shift right=16
    ]
    \ar[
      d,
      -,
      dotted,
      shift right=8
    ]
    \ar[
      d,
      -,
      dotted,
      shift right=0
    ]
    \ar[
      d,
      -,
      dotted,
      shift left=8
    ]
    \ar[
      d,
      -,
      dotted,
      shift left=16
    ]
    \ar[
      d,
      -,
      dotted,
      shift left=24
    ]
    \\
    \adjustbox{}{
       \begin{tikzpicture}
         \draw
           (.02, .03) node {
             \scalebox{.5}{$3$}
           };
         \draw[gray] 
           (90:.7) to (.02, .2);
         \draw[gray] 
           (120+90:.7) to (-.06, .03);
         \draw[gray] 
           (240+90:.7) to (+.1, .04);
         \draw[
          fill=gray, 
          draw opacity=.6,
          fill opacity=.6
         ]
           (90:.7) -- (120+90:.7) -- (240+90:.7) -- cycle;
          \draw (90:.81) node {\scalebox{.7}{$1$}};
          \draw (240+90-4:.9) node {\scalebox{.7}{$2$}};
          \draw (120+90+4:.9) node {\scalebox{.7}{$0$}};
       \end{tikzpicture}
    }
    &[70pt]
    \mathllap{
          \scalebox{.7}{
            \color{darkblue}
            \bf
            \def\arraystretch{.9}
            \begin{tabular}{c}
              Deformation paths
              \\
              of deformation paths
              \\
              of deformation paths
              \\
              of flux densities
            \end{tabular}
      }    
    }
    \Omega^1_{\mathrm{dR}}\big(
      -
      \!\times \Delta^3_{\mathrm{geo}}
      ;\,
      \mathfrak{a}
    \big)_\mathrm{clsd}
    \ar[
      dd,
      shift right=51pt,
      "{
        (-)_{[1,2,3]}
      }"{description}
    ]
    \ar[
      dd,
      shift right=17.5pt,
      "{
        (-)_{[0,2,3]}
      }"{description}
    ]
    \ar[
      dd,
      shift left=17.5pt,
      "{
        (-)_{[0,1,3]}
      }"{description}
    ]
    \ar[
      dd,
      shift left=51pt,
      "{
        (-)_{[0,1,2]}
      }"{description}
    ]
    \\
    \\
    \adjustbox{}{
       \begin{tikzpicture}
         \draw[
          fill=gray,
          draw opacity=.6,
          fill opacity=.6
         ]
           (90:.7) -- (120+90:.7) -- (240+90:.7) -- cycle;
          \draw (90:.81) node {\scalebox{.7}{$1$}};
          \draw (240+90-4:.9) node {\scalebox{.7}{$2$}};
          \draw (120+90+4:.9) node {\scalebox{.7}{$0$}};
       \end{tikzpicture}
    }
    &[70pt]
    \mathllap{
          \scalebox{.7}{
            \color{darkblue}
            \bf
            \def\arraystretch{.9}
            \begin{tabular}{c}
              Deformation paths
              \\
              of deformation paths
              \\
              of flux densities
            \end{tabular}
      }    
    }
    \Omega^1_{\mathrm{dR}}\big(
      -
      \!\times \Delta^2_{\mathrm{geo}}
      ;\,
      \mathfrak{a}
    \big)_\mathrm{clsd}
    \ar[
      dd,
      shift left=15,
      "{
        (-)_{[1,2]}
      }"{description}
    ]
    \ar[
      from=dd,
      shift left=7.5
    ]
    \ar[
      dd,
      shift left=0,
      "{
        (-)_{[0,2]}
      }"{description}
    ]
    \ar[
      from=dd,
      shift right=7.5
    ]
    \ar[
      dd,
      shift right=15,
      "{
        (-)_{[0,1]}
      }"{description}
    ]
    \\
    \\
    \adjustbox{}{
       \begin{tikzpicture}
         \draw[line width=1.2, gray]
           (-.6, .07) -- (+.6,.07); 
         \draw 
           (-.75,0) node {
            \scalebox{.7}{$0$}
           };
         \draw 
           (+.75,0) node {
            \scalebox{.7}{$1$}
           };
       \end{tikzpicture}
    }
    &
    \mathllap{
          \scalebox{.7}{
            \color{darkblue}
            \bf
            \def\arraystretch{.9}
            \begin{tabular}{c}
              Deformation paths
              \\
              of flux densities
            \end{tabular}
      }    
    }
    \Omega^1_{\mathrm{dR}}
    \big(
      - 
      \!\times
      \Delta^1_{\mathrm{geo}}
      ;\, 
      \mathfrak{a}
    \big)_{\mathrm{clsd}}
    \ar[
      dd,
      shift left=18pt,
      "{
        (-)_0
      }"{description},
      "{
        \scalebox{.6}{
          \color{darkgreen}
          \bf
          \def\arraystretch{.9}
          \begin{tabular}{c}
            take starting point
            \\
            of deformation path
          \end{tabular}
        }
      }"
    ]
    \ar[
      dd,
      shift right=18pt,
      "{
        (-)_1
      }"{description},
      "{
        \scalebox{.6}{
          \color{darkgreen}
          \bf
          \def\arraystretch{.9}
          \begin{tabular}{c}
            take endpoint of
            \\
            deformation path
          \end{tabular}
        }
      }"{swap, xshift=-4pt}
    ]
    \ar[
      from=dd
    ]
    \\
    \\
    &
    \mathllap{
          \scalebox{.7}{
            \color{darkblue}
            \bf
            \def\arraystretch{.9}
            \begin{tabular}{c}
              Flux densities satisfying
              \\
              their Bianchi identities
            \end{tabular}
          }
    }
    \Omega^1_{\mathrm{dR}}
      (
        -
        ;\, 
        \mathfrak{a}
      )_\mathrm{clsd}
  \end{tikzcd}
  }
  \hspace{-8pt}
  \right)
\end{equation}
For example, the coboundaries-of-coboundaries from Def. \ref{CoboundariesOfCoboundaries} are captured this way as deformations over triangles one of whose sides is degenerate (cf. \eqref{2ndConcordanceConstantAtBoundary}):
$$
  \begin{tikzcd}[row sep=small]
    F^{(0)}
    \ar[
      rr,
      bend left=30,
      "{ \widehat{F} }",
      "{\ }"{swap, name=s}
    ]
    \ar[
      rr,
      bend right=30,
      "{ \widehat{F}' }"{swap},
      "{\ }"{name=t}
    ]
    \ar[
      from=s,
      to=t,
      Rightarrow,
      "{
        {\color{orangeii} \doublehat{F}}
      }"
    ]
    &&
    F^{(1)}
  \end{tikzcd}
  \hspace{.5cm}
  \Leftrightarrow
  \hspace{.5cm}
  \begin{tikzcd}
    & 
    F^{(0)}
    \ar[
      dr,
      "{
        \widehat{F}
      }"{pos=.4},
      "{\ }"{swap, pos=.2, name=s}
    ]
    \\
    F^{(0)}
    \ar[
      ur, 
      equals
    ]
    \ar[
      rr,
      "{
        \widehat{F}'
      }"{swap},
      "{\ }"{name=t}
    ]
    &&
    F^{(1)}
    \ar[
      from=s,
      to=t,
      Rightarrow,
      "{
       {\color{orangeii} \doublehat{F}}
      }"{swap}
    ]
  \end{tikzcd}
$$

Now for any Cartesian super-space $\mathbb{R}^{n\vert q}$ (in fact for any super-manifold $X$) the above diagram \eqref{SystemOfHigherDeformations} is a system of sets (of higher-order deformations) indexed by simplices \eqref{GeometricSimplices}, hence called a {\it simplicial set} (exposition in \cite{Friedman12} \cite[\S 2]{Jardine15}).
$$
  \begin{tikzcd}[
    column sep=50pt,
    row sep=15pt
  ]
    &
    \hspace{-2.8cm}
    \mathrlap{
      \mbox{
        \includegraphics[width=12cm]{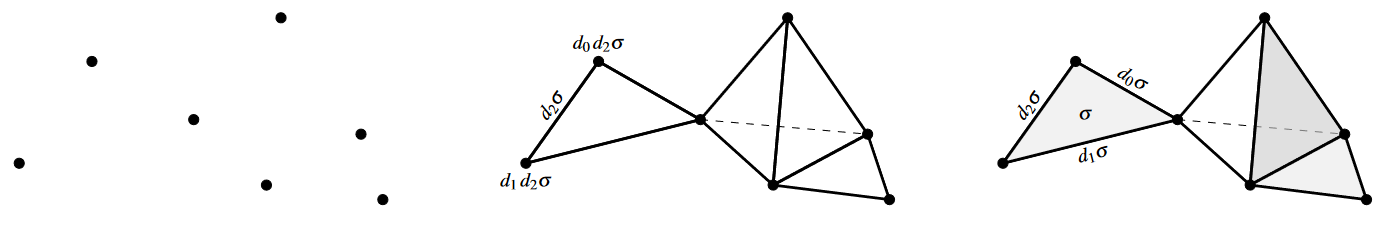}
      }
    }
    &&
    &&
    &[-25pt]
    \\[-14pt]
    \mathcal{
      \scalebox{.7}{
        \bf
        simplicial set
      }
    }
    &
    \mathclap{
      \scalebox{.7}{
        \bf
        \def\arraystretch{.9}
        \begin{tabular}{c}
          set of 
          \\
          points
        \end{tabular}
      }
    }
    \ar[
      from=rr,
      shorten=20pt,
      shift right=4pt,
      "{
        \scalebox{.6}{
          endpoints
        }
      }"{description}
    ]
    \ar[
      rr,
      shorten=20pt,
      shift right=4pt,
      "{
        \scalebox{.6}{
          constant edges
        }
      }"{description}
    ]
    &&
    \mathclap{
      \scalebox{.7}{
        \bf
        \def\arraystretch{.9}
        \begin{tabular}{c}
          set of 
          \\
          edges
        \end{tabular}
      }
    }
    \ar[
      from=rr,
      shorten=20pt,
      shift right=4pt,
      "{
        \scalebox{.6}{
          boundary edges
        }
      }"{description}
    ]
    \ar[
      rr,
      shorten=20pt,
      shift right=4pt,
      "{
        \scalebox{.6}{
          thin surfaces
        }
      }"{description}
    ]
    &&
    \mathclap{
      \scalebox{.7}{
        \bf
        \def\arraystretch{.9}
        \begin{tabular}{c}
          set of 
          \\
          triangles
        \end{tabular}
      }
    }
    \ar[
      from=r,
      shorten <=7pt,
      shorten >=22pt,
      shift right=4pt,
    ]
    \ar[
      r,
      shorten >=7pt,
      shorten <=22pt,
      shift right=4pt,
    ]
    &
    \mathclap{
      \scalebox{1.4}{$\cdots$}
    }
    \\
  \mathcal{X}
  \ar[
    r,
    phantom,
    "{ = }"
  ]
  &
  \Bigg(
    X_0
    \ar[from=rr, shift left=9pt, "{ d_0 }"{description}]
    \ar[rr, "{s_0}"{description}]
    \ar[from=rr, shift right=9pt, "{ d_1 }"{description}]
    &&
    X_1
    \ar[from=rr, shift left=18pt, "{ d_0 }"{description}]
    \ar[rr, shift left=9pt, "{ s_0 }"{description}]
    \ar[from=rr, "{d_1}"{description}]
    \ar[rr, shift right=9pt, "{ s_1 }"{description}]
    \ar[from=rr, shift right=18pt, "{ d_2 }"{description}]
    &&
    X_2
    \ar[r, shift left=21pt, -, dotted]
    \ar[r, shift left=14pt, -, dotted]
    \ar[r, shift left=7pt, -, dotted]
    \ar[r, -, dotted]
    \ar[r, shift right=7pt, -, dotted]
    \ar[r, shift right=14pt, -, dotted]
    \ar[r, shift right=21pt, -, dotted]
    &
    \scalebox{1.4}{$\cdots$}
  \Bigg)
  \end{tikzcd}
$$

\smallskip 
The particular simplicial set \eqref{SystemOfHigherDeformations} is \cite[Prop. 5.10 with Def. 9.1]{FSS23Char} {\it Kan fibrant}, which means that if one finds deformations along all faces of an $n$-simplex but the $i$th  -- called the $i$th {\it $n$-horn} $\Lambda^n_i$ -- then there exists a deformation along the full $n$-simplex.
One readily sees in dimension $n = 2$ that this condition means that deformation paths have composites if they coincide, and have inverses:
\vspace{-1mm} 
\begin{equation}
\adjustbox{raise=-1.3cm}{
\begin{tikzpicture}[scale=0.68]

\begin{scope}  
  \draw[
    draw=lightgray,
    fill=lightgray
  ]
    (0-30:1.5)
    --
    (120-30:1.5)
    --
    (240-30:1.5)
    --
    cycle;

\node
 at (-150:1.5) 
 {
   \scalebox{.7}{0}
 };
\node
 at (-150-120:1.5) 
 {
   \scalebox{.7}{1}
 };
\node
 at (-150+120:1.5) 
 {
   \scalebox{.7}{2}
 };

\begin{scope}[shift={(-150:1.5)}]
  \draw[
    line width=.8,
    -Latex 
  ]
    (60:.15) -- (60:2.45);
  \draw[
    line width=.8,
    -Latex 
  ]
    (0:.15) -- (0:2.45);
\end{scope}

\begin{scope}[shift={(+90:1.5)}]
  \draw[
    line width=.75,
    -Latex,
    gray
  ]
    (-60:.15) -- (-60:2.45);
\end{scope}

\node
 at (0,-.05)
 {
  \scalebox{.65}{$
    \Lambda^2_{\color{purple}0} 
    \subset
    \Delta^2
  $}
 };
\end{scope}

\begin{scope}[shift={(3.5,0)}]
  \draw[
    draw=lightgray,
    fill=lightgray
  ]
    (0-30:1.5)
    --
    (120-30:1.5)
    --
    (240-30:1.5)
    --
    cycle;

\node
 at (-150:1.5) 
 {
   \scalebox{.7}{0}
 };
\node
 at (-150-120:1.5) 
 {
   \scalebox{.7}{1}
 };
\node
 at (-150+120:1.5) 
 {
   \scalebox{.7}{2}
 };

\begin{scope}[shift={(-150:1.5)}]
  \draw[
    line width=.8,
    -Latex 
  ]
    (60:.15) -- (60:2.45);
  \draw[
    line width=.75,
    -Latex,
    gray
  ]
    (0:.15) -- (0:2.45);
\end{scope}

\begin{scope}[shift={(+90:1.5)}]
  \draw[
    line width=.8,
    -Latex,
  ]
    (-60:.15) -- (-60:2.45);
\end{scope}

\node
 at (0,-.05)
 {
  \scalebox{.65}{$
    \Lambda^2_{\color{purple}1} 
    \subset
    \Delta^2
  $}
 };
\end{scope}

\begin{scope}[shift={(7,0)}]
  \draw[
    draw=lightgray,
    fill=lightgray
  ]
    (0-30:1.5)
    --
    (120-30:1.5)
    --
    (240-30:1.5)
    --
    cycle;

\node
 at (-150:1.5) 
 {
   \scalebox{.7}{0}
 };
\node
 at (-150-120:1.5) 
 {
   \scalebox{.7}{1}
 };
\node
 at (-150+120:1.5) 
 {
   \scalebox{.7}{2}
 };

\begin{scope}[shift={(-150:1.5)}]
  \draw[
    line width=.75,
    -Latex,
    gray
  ]
    (60:.15) -- (60:2.45);
  \draw[
    line width=.8,
    -Latex
  ]
    (0:.15) -- (0:2.45);
\end{scope}

\begin{scope}[shift={(+90:1.5)}]
  \draw[
    line width=.8,
    -Latex,
  ]
    (-60:.15) -- (-60:2.45);
\end{scope}

\node
 at (0,-.05)
 {
  \scalebox{.65}{$
    \Lambda^2_{\color{purple}2} 
    \subset
    \Delta^2
  $}
 };
\end{scope}

\end{tikzpicture}
}
\end{equation}
A similar inspection of the 3-horns shows that these composites are associative and unital up to yet higher deformations, hence that deformations of flux densities behave much like a symmetry group, only that they may act between distinct configurations -- which makes this group a group{\it oid} -- and that the usual group laws hold only up to higher deformations -- which makes such {\it Kan simplicial sets} be models of {\it higher} groupoids or {\it $\infty$-groupoids}, for short.

\medskip 
Finally, all this structure is manifestly compatible with the pullback of differential forms along maps of Cartesian super-spaces. This means that the construction \eqref{SystemOfHigherDeformations} of higher deformations of flux densities constitutes a contravariant functor from Cartesian super-spaces to Kan-simplicial sets:
\begin{equation}
  \label{DeformationGroupoidAsFunctor}
  \shape
  \,
  \Omega^1_{\mathrm{dR}}\big(
    -;\,
    \mathfrak{a}
  \big)_{\mathrm{clsd}}
  \;:\;
  \mathrm{sCartSp}^{\mathrm{op}}
  \xlongrightarrow{\;\;}
  \mathrm{SimpSet}_{\mathrm{Kan}}
  \,.
\end{equation}
Since every ordinary set may naturally be regarded as a simplicial set all whose simplices are degenerate, this is a simplicial extension of $\Omega^1_{\mathrm{dR}}(-; \mathfrak{a})_{\mathrm{clsd}}$ \eqref{ClosedVValuedDifferentialFormsOnSmoothSuperSet}; in fact we have a natural inclusion (a natural transformation of plot-assigning functors)
\begin{equation}
  \label{ShapeUnitForDeformationModuli}
  \begin{tikzcd}[sep=0pt]
    \Omega^1_{\mathrm{dR}}(
      -
      ;\,
      \mathfrak{a}
    )_{\mathrm{clsd}}
    \ar[
      rr,
      hook,
      "{
        \eta^{\scalebox{.7}{\shape}}
      }"
    ]
    &&
    \shape
    \,
    \Omega^1_{\mathrm{dR}}(
      -
      ;\,
      \mathfrak{a}
    )_{\mathrm{clsd}}
    \\[-2pt]
    \vec F
    &\longmapsto&
    \big(p^\bullet\big)^\ast 
    \vec F
  \end{tikzcd}
\end{equation}
which for each probe super-space $\mathbb{R}^{n\vert q}$ and in each simplicial degree $n$ pulls back closed $\mathfrak{a}$-valued differential forms on $\mathbb{R}^{n \vert q}$ along the projection map $p^n : \mathbb{R}^{n \vert q} \times \Delta^n_{\mathrm{geo}} \xrightarrow{\;} \mathbb{R}^{n \vert q}$.
\end{example}

\smallskip

\noindent
{\bf Higher smooth super sets.}
Therefore we want to think of functors such as \eqref{DeformationGroupoidAsFunctor} as being just like the plot-assigning functors of 
smooth super-sets from Def. \ref{SuperSmoothSets}, but now such that these plots may have gauge transformations (homotopies) 
between them with higher-gauge-of-gauge transformations between these, etc.;  hence forming not just sets but Kan-simplicial sets, 
hence $\infty$-groupoids. With \eqref{DeformationGroupoidAsFunctor} regarded as a {\it smooth super $\infty$-groupoid} this way, 
it may serve as a ``classifying space'', or {\it moduli stack}, for deformations of super-flux densities.

\smallskip
From this perspective, a natural transformation between higher-plot assigning functors $\mathcal{X}, \mathcal{Y}
: \mathrm{sCartSp}^{\mathrm{op}} \to \mathrm{SimpSet}_{\mathrm{Kan}}$ should count as {\it identifying} $\mathcal{X}$ 
with $\mathcal{Y}$ if it does so 

{\bf (i)} {\it locally}, namely for arbitrarily ``small'' plots (germs of plots) and 

{\bf (ii)} {\it up to homotopy}, namely up to gauge transformations. 

\smallskip 
\noindent This is made precise by the following notion of {\it local homotopy equivalences} (lhe), which are the higher generalization
of the local isomorphism \eqref{LocalIsomorphisms} and constitute in more generality the basis of {\it local homotopy theory} \cite{Jardine15}.

 \newpage 
\begin{definition}[{\bf Local homotopy equivalences of higher super-plots}]
\noindent
\begin{itemize}[leftmargin=.75cm]
\item[{\bf (i)}] 
The {\it simplicial 1-simplex} is
$\Delta^1 := \big\{ 0 \to 1 \big\}
\in \mathrm{SimpSet}$.

\item[{\bf (ii)}] 
Given $\mathcal{X}, \mathcal{Y} \,\in\, \mathrm{SimpSet}_{\mathrm{Kan}}$ 
then 
\vspace{0mm} 
\begin{itemize}[
  leftmargin=.4cm
]

\item 
A {\it homotopy}
$
  \begin{tikzcd}
    \mathcal{X}
    \ar[
      rr,
      bend left=25pt,
      "{ f }",
      "{\ }"{swap, name=s}
    ]
    \ar[
      rr,
      bend right=25pt,
      "{ g }"{swap},
      "{\ }"{name=t}
    ]
    \ar[
      from=s,
      to=t,
      Rightarrow,
      "{ \eta }"
    ]
    &&
    \mathcal{Y}
  \end{tikzcd}
$
is a diagram of the form
$
  \begin{tikzcd}[
    row sep=7pt
  ]
    \mathcal{X}
    \ar[
      d,
      shorten=-3pt, 
      "{ (\mathrm{id}, 0) }"{swap}
    ]{pos=.4}
    \ar[
      drr,
      end anchor={[yshift=+3pt]},
      "{ f }"{pos=.4}
    ]
    \\
    \mathcal{X} \times \Delta^1
    \ar[
      rr,
      dashed,
      "{ 
        \eta 
      }"{pos=.3, description}
    ]
    &&
    \mathcal{Y}\;.
    \\
    \mathcal{X}
    \ar[
      u,
      shorten=-3pt,
      "{ (\mathrm{id}, 1) }"{pos=.3}
    ]
    \ar[
      urr,
      end anchor={[yshift=-3pt]},
      "{ g }"{swap, pos=.4}
    ]
  \end{tikzcd}
$
\vspace{0mm} 
\item 
A {\it homotopy equivalence} is maps
$
  \begin{tikzcd}
    \mathcal{X}
    \ar[
      r,
      bend left=20,
      "f"
    ]
    \ar[
      from=r,
      bend left=20,
      "\overline{f}"
    ]
    &
    \mathcal{Y}
  \end{tikzcd}
$
with homotopies
$
  \begin{tikzcd}
    \mathcal{X}
    \ar[
      r,
      bend left=30,
      "{ 
        \overline{f} \circ f 
      }",
      "{\ }"{swap, name=s}
    ]
    \ar[
      from=r,
      bend left=30,
      "{ 
        \mathrm{id}
      }",
      "{\ }"{swap, name=t}
    ]
    \ar[
      from=s,
      to=t,
      Rightarrow
    ]
    &
    \mathcal{X}
  \end{tikzcd}
$
and
$
  \begin{tikzcd}
    \mathcal{Y}
    \ar[
      r,
      bend left=30,
      "{ 
        f \circ \overline{f} 
      }",
      "{\ }"{swap, name=s}
    ]
    \ar[
      from=r,
      bend left=30,
      "{ 
        \mathrm{id}
      }",
      "{\ }"{swap, name=t}
    ]
    \ar[
      from=s,
      to=t,
      Rightarrow
    ]
    &
    \mathcal{Y} 
  \end{tikzcd}
$\!.
\end{itemize}
\item[{\bf (iii)}]  
Given $\mathcal{X} \in \mathrm{Func}\big(\mathrm{sCartSp}^{\mathrm{op}}, \mathrm{SimpSet}_{\mathrm{Kan}}\big)$, 
\begin{itemize}[leftmargin=.4cm]
\item
its {\it $(n\vert q)$-stalks} is the simplicial set of equivalence classes $\mathrm{Plt}\big(\mathbb{R}^{n\vert q},\, \mathcal{X}\big)\big/\sim$, 

where plots $\phi \sim \phi'$ iff they agree on an open neighborhood of the origin,
\end{itemize}

  \smallskip  
\item[{\bf (iv)}]  
Given $\mathcal{X}, \mathcal{Y} \in \mathrm{Func}\big(\mathrm{sCartSp}^{\mathrm{op}}, \mathrm{SimpSet}_{\mathrm{Kan}}\big)$ then:
\vspace{1mm} 
\begin{itemize}[leftmargin=.4cm]
\item
The {\it $(n\vert q)$-germ} of a map
$f : \mathcal{X} \xrightarrow{} \mathcal{Y}$ is 
its (co)restriction to the $(n \vert q)$-stalks of $\mathcal{X}$ and $\mathcal{Y}$,
\item
A {\it local homotopy equivalence} is a map 
$\!
  \begin{tikzcd}[column sep=12pt]
    \mathcal{X}\!
    \ar[
      r,
      "f"
    ]
    &\!
    \mathcal{Y}
  \end{tikzcd}
\!$
which on $(n\vert q)$-stalks is part of a homotopy equivalence.
\end{itemize}
\end{itemize}
\end{definition}

With this we may introduce higher super-spaces in the guise of {\it smooth super $\infty$-groupoids} in direct analogy with the definition of smooth super-sets in Def. \ref{SuperSmoothSets}, just with sets of plots replaced by Kan-simplicial sets, and with local isomorphism replaced by local homotopy equivalences (see \cite{Jardine15} and specifically \cite[pp. 41]{FSS23Char} for details):

\begin{definition}[{\bf Smooth super $\infty$-groupoids} {\cite[\S 3.1.3]{SS20Orb}}]
\label{SmoothSuperInfinityGroupoids}
The simplicial category of {\it smooth super $\infty$-groupoids} is the simplicial localization of the higher super-plot assigning functors at the local homotopy equivalences (Def. \ref{SmoothSuperInfinityGroupoids}) between them:
\begin{equation}
  \label{SimplicialCategoryOfSmoothSuperInfinityGroupoids}
  \mathrm{sSmthGrpd}_\infty
  \;:=\;
  L^{\mathrm{lhe}}
  \,
  \mathrm{Func}\big(
    \mathrm{sCartSp}^{\mathrm{op}}
    \,,\,
    \mathrm{SimpSet}_{\mathrm{Kan}}
  \big)
  \,,
\end{equation} 
which means, in particular, that:
\begin{itemize}[leftmargin=.65cm]
\item[\bf (i)] Smooth super $\infty$-groupoids $\mathcal{X}$ are (represented by) functors
\begin{equation}
  \label{PlotsOfSmoothSuperInfinityGroupoids}
  \begin{tikzcd}[row sep=-2pt, column sep=small]
    \mathrm{sCartSp}^{\mathrm{op}}
    \ar[
      rr
    ]
    &&
    \mathrm{SimpSet}_{\mathrm{Kan}}
    \\
    \mathbb{R}^{n \vert q}
    &\longmapsto&
    \mathrm{Plt}\big(
      \mathbb{R}^{n \vert q}
      ,\,
      \mathcal{X}
    \big)
  \end{tikzcd}
\end{equation}

\vspace{-2mm} 
\noindent 
which we think of as assigning to a Cartesian super-space $\mathbb{R}^{n \vert q}$ the Kan-simplicial set ($\infty$-groupoid) of ways-and-their-higher-equivalences of mapping it into the would-be smooth super $\infty$-groupoid $\mathcal{X}$.
\item[\bf (ii)] If $\mathcal{Y}$ is {\it projectively fibrant} (which we do not further explain here, 
see \cite[Ex. 1.20]{FSS23Char}, but which is the case for all examples considered here) then maps $\mathcal{X} \xrightarrow{\;} \mathcal{Y}$ of smooth super $\infty$-groupoids are natural transformations between these plot-assigning functors of the form
\vspace{-1mm} 
\begin{equation}
  \label{ZigZagOfSimplicialPresheaves}
  \begin{tikzcd}
    \mathcal{X}
    \ar[
      from=r,
      "{
        \mathrm{lheq}
      }"{},
      "{
        p
      }"{swap}
    ]
    &
    \widehat{\mathcal{X}}
    \ar[
      r,
      "{ f }"
    ]
    &
    \mathcal{Y}
  \end{tikzcd}
\end{equation}
where the left one is a local homotopy equivalence. 
\end{itemize}
\end{definition}
The situation \eqref{ZigZagOfSimplicialPresheaves} means that for representing all maps between smooth super $\infty$-groupoids, the domain $\mathcal{X}$ may first need to be ``puffed up'' by a locally homotopy equivalent ``resolution'' $\widehat{\mathcal{X}}$ which supports more ``homotopical freedom'' for the map to act. Among all resolutions, there are universal ones, called {\it projectively cofibrant}, which are guaranteed to support all maps. Without going into their theory here (for details see \cite[pp. 43]{FSS23Char}) we state the one example needed here (cf. \cite[Ex. 3.3.44]{SS21EBund}), being the higher generalization of Ex. \ref{OperCoverOfSupermanifoldAsLocalResolution}:

\begin{example}[\bf {\v C}ech resolution of supermanifold]
  \label{CechResolutionOfSupermanifold}
  Let $X$ be a supermanifold equipped with an open cover $\big\{ U_i \xhookrightarrow{ \iota_i } X \big\}_{i \in I}$.
  We obtain a smooth super $\infty$-groupoid $\widehat X$ \eqref{PlotsOfSmoothSuperInfinityGroupoids} whose $\mathbb{R}^{n\vert q}$-plots form the simplicial set of ways of mapping $\mathbb{R}^{n \vert q}$ into any one of the charts $U_i$, with gauge transformations being the transitions to overlapping charts:
  
  \smallskip 
  \begin{equation}
    \label{PlotsOfCechResolution}
    \mathrm{Plt}\big(
      \mathbb{R}^{n \vert q}
      ,\,
      \widehat X
    \big)
    \;\;
    :=
    \;\;
    \left(\!\!\!
    \begin{tikzcd}[
      row sep=10pt
    ]
      {}
      \ar[
        d,
        dotted,
        shift left=60pt
      ]
      \ar[
        from=d,
        dotted,
        shift right=40pt
      ]
      \ar[
        d,
        dotted,
        shift left=20pt
      ]
      \ar[
        from=d,
        dotted,
        shift right=0pt
      ]
      \ar[
        d,
        dotted,
        shift left=-20pt
      ]
      \\
      \mathrm{Hom}_{\mathrm{sSmthMfd}}\Big(
        \mathbb{R}^{n \vert q}
        ,\,
        \underset{i_1, i_2 \in I}{\coprod}
        U_{i_1}
        \cap 
        U_{i_2}
      \Big)
      \ar[
        d, 
        shift left=40pt
      ]
      \ar[
        d, 
        shift left=0pt
      ]
      \ar[
        from=d,
        shift right=20pt
      ]
      \\
      \mathrm{Hom}_{\mathrm{sSmthMfd}}\Big(
        \mathbb{R}^{n \vert q}
        ,\,
        \underset{i \in I}{\coprod}
        U_i
      \Big)
    \end{tikzcd}
   \!\!\! \right)
    \,
  \end{equation}
  If the open cover is chosen to be {\it super-differentiably good} \cite[\S A]{FSS12}\cite[Def. 5.3.1]{GuilleminHaine73} in that all finite intersections of charts $U_{i_1} \cap \cdots \cap U_{i_n}$ are either empty or super-diffeomorphic to a Cartesian super space, then this constitutes a cofibrant resolution of $X$, in that, particularly, all maps out of $X$ of the form appearing in \eqref{FluxQuantizedSugraFields} are represented by natural transformations (of plot-assigning functors) out of $\widehat{X}$.
\end{example}

Where the previous example imports domain spacetimes into higher super-geometry; the following example does the same for coefficients of ``ordinary'' cohomology:

\begin{example}[{\bf Dold-Kan construction}, e.g. {\cite[Ex. 1.30]{FSS23Char}}]
  \label{DoldKanConstruction}
  For $n \in \mathbb{N}$, write 
   $$
     N_\bullet(\Delta^n)
     \;\in\;
     \mathrm{SimpAb}
   $$
   for the {\it normalized chain complex} of the $n$-simplex, which in degree $k$ is the free abelian group on the non-degenerate $k$-simplices in $\Delta^n$ with differential given by the alternating sum of the face maps.

  Then for 
  $$
    A_\bullet \,\in\, \mathrm{Ch}_{\geq 0}\Big(\mathrm{Ab}\big(\mathrm{Sh}(\mathrm{sCartSp})\big)\Big)
  $$
  a chain complex (in non-negative degrees, with differential of degree -1) of sheaves of abelian groups, we obtain a smooth super $\infty$-groupoid
  $$
    H A_\bullet
    \;\in\;
    \mathrm{sSmthGrpd}_\infty
  $$
  whose $k$-simplices of plots are the images of $N_\bullet(\Delta^k)$ in $A_\bullet$:
  $$
    \mathrm{Plt}\big(
      \mathbb{R}^{n\vert q}
      ,\,
      H A_\bullet
    \big)_k
    \;\;
      :=
    \;\;
    \mathrm{Hom}_{\mathrm{sAb}}
    \big(
      N_\bullet(\Delta^k)
      ,\,
      A_\bullet(\mathbb{R}^{n \vert q})
    \big)
    \,.
  $$
\end{example}

The final class of examples of smooth super $\infty$-groupoids relevant to our purpose is the following:
\begin{example}[{\bf Path $\infty$-groupoids} {(cf. \cite[p. 144]{SS21EBund})}]
  \label{PathInfinityGroupoids}
  For $A$ a topological space, its {\it path $\infty$-groupoid} is the smooth $\infty$-groupoid --- here to be denoted $\mathcal{A}$ --- whose plots, independently of the probe space, form the traditional {\it singular simplicial complex} of $A$, hence the simplicial set of continuous images of geometric $n$-simplicies \eqref{GeometricSimplices} (i.e. order-$n$ paths of continuous paths) in $A$:
  \begin{equation}
    \mathrm{Plt}\big(
      \mathbb{R}^{n \vert q}
      ;\,
      \mathcal{A}
    \big)
    \;\;
    :=
    \;\;
    \left(
    \begin{tikzcd}[row sep=14pt]
      {}
      \ar[
        d,
        dotted,
        shift left=+30pt
      ]
      \ar[
        d,
        dotted,
        <-,
        shift left=+20pt
      ]
      \ar[
        d,
        dotted,
        shift left=+10pt
      ]
      \ar[
        d,
        dotted,
        <-,
        shift left=0
      ]
      \ar[
        d,
        dotted,
        shift left=-10pt
      ]
      \ar[
        d,
        dotted,
        <-,
        shift left=-20pt
      ]
      \ar[
        d,
        dotted,
        shift left=-30pt
      ]
      \\
      \mathrm{Hom}_{\mathrm{TopSp}}
      \big(
        \Delta^2_{\mathrm{geo}}
        ,\,
        \mathcal{A}
      \big)
      \ar[
        d,
        shift left=30pt
      ]
      \ar[
        from=d,
        shift right=15pt
      ]
      \ar[
        d,
        shift left=0pt
      ]
      \ar[
        from=d,
        shift right=-15pt
      ]
      \ar[
        d,
        shift left=-30pt
      ]
      \\
      \mathrm{Hom}_{\mathrm{TopSp}}\big(
        \Delta^1_{\mathrm{geo}}
        ,\,
        A
      \big)
      \ar[
        d,
        shift left=20pt
      ]
      \ar[
        from=d,
        shift right=0pt
      ]
      \ar[
        d,
        shift left=-20pt
      ]
      \\
      \mathrm{Hom}_{\mathrm{TopSp}}\big(
        \Delta^0_{\mathrm{geo}}
        ,\,
        A
      \big)
    \end{tikzcd}
    \right).
  \end{equation}
\end{example}

\medskip

\subsubsection{Super Flux Quantization}
\label{BackgroundOnFluxQuantization}

With the super-moduli constructions of \S\ref{SuperModuliStacks} in hand, homotopical flux quantization on super-manifolds --- as shown in \eqref{TheGaugePotentials} ---  follows verbatim by the same rules \cite{FSS23Char} as on 
ordinary manifolds. Here we briefly review a couple of illustrating examples and then the 
specialization to the half-integral flux quantization of the C-field (\cite{FSS20-H}\cite{FSSHopf}\cite{FSS22Twistorial}, surveyed in \cite[\S 12]{FSS23Char});
for more expository survey see \cite{SS24Flux}.

\medskip

\noindent
{\bf Dirac charge quantization in homotopical language.}
To begin with, it helps to understand ordinary Dirac charge quantization in this language. For this purpose, first recall the formulation of abelian gauge fields via {\v C}ech cohomology (first highlighted in \cite{Alvarez85a}\cite{Alvarez85b}):

The ordinary integral cohomology of $X$ in degree 2 (classifying usual Dirac monopole charges) is computed by the {\it {\v C}ech cohomology} with respect to a differentiably good open cover $\big\{ U_i \xhookrightarrow{\iota_i} X \big\}_{i \in I}$ (cf. Ex. \ref{CechResolutionOfSupermanifold}), namely by assignments of integers $c_{i j k} \in \mathbb{Z}$ to the non-empty triple intersections $U_i \cap U_j \cap U_k$ such that on all non-trivial quadruple intersections $U_i \cap U_j \cap U_k \cap U_l$ the cocycle conditon $c_{i j l} + c_{j k l} = c_{i j k} +  c_{i k l}$ holds, and subject to the equivalence relation $\{c_{i j k}\}_{i, j, k} \sim \{c'_{i j k}\}_{i, j, k}$ iff there exist integers $h_{i j}$ for all non-empty double intersections $U_i \cap U_k$ such that $c'_{i j k} = c_{i j k} + h_{i j} + h_{j k} - h_{i k}$.

Direct inspection shows that this data may neatly be re-packaged by 
writing $\widehat X$ for the {\v C}ech groupoid on the open cover (Ex. \ref{CechResolutionOfSupermanifold}) and $H \mathbb{Z}[2]$ for the Dold-Kan construction (Ex. \ref{DoldKanConstruction}) of the chain complex concentrated on $\mathbb{Z}$ in degree 2
$$
  \mathbb{Z}[2]
  \;:=\;
  \big[
    \begin{tikzcd}
      \ar[r, dotted]
      &
      0
      \ar[r]
      &
      0
      \ar[r]
      &
      \mathbb{Z}
      \ar[r]
      &
      0
      \ar[r]
      &
      0
    \end{tikzcd}
  \big]
  \,,
$$
in terms of which the above {\v C}ech cocycles are just maps of smooth $\infty$-groupoids modeled as simplicial sheaves (Def. \ref{NonabelianCohomology}) from $\widehat X$ to $H \mathbb{Z}[2]$, and coboundaries are simplicial homotopies between these (see \cite[\S 2]{FSS13CupCS}\cite[Rem. 3.3.45]{SS21EBund} and \cite{Schreiber22} for more exposition of this translation):

\vspace{-5pt}
$$
  \begin{tikzcd}[
    column sep=40pt
  ]
  \\
    \widehat{X}
    \ar[
      rr,
      bend left=30,
      "{ c }"{description,pos=.63},
      "{\ }"{swap, name=s, pos=.67},
      "{
        \mathclap{
        \scalebox{.7}{
          \color{darkblue}
          \bf {\v C}ech cocycle
        }
        }
      }"{pos=.6}
    ]
    \ar[
      rr,
      bend right=30,
      "{
        c'
      }"{description, pos=.65},
      "{\ }"{name=t, pos=.67},
      "{
        \mathclap{
        \scalebox{.7}{
          \color{darkblue}
          \bf {\v C}ech cocycle
        }
        }
      }"{swap,pos=.6}
    ]
    &&
    H \mathbb{Z}[2]
    &[-43pt]
    \;\simeq\;
    B^2 \mathbb{Z}
    \ar[
      from=s, to=t,
      shift left=2pt,
      Rightarrow,
      "{ h }",
      "{
        \scalebox{.65}{
          \color{orangeii}
          \bf
          \def\arraystretch{.85}
          \def\tabcolsep{-2pt}
          \begin{tabular}{c}
            {\v C}ech
            \\
            coboundary
          \end{tabular}
        }
      }"{swap, pos=.45}
    ]
  \end{tikzcd}
$$
But this means that the {\it Eilenberg-MacLane space} which is alternatively denoted as
$$
  B^2 \mathbb{Z}
  \;:=\;
  K(\mathbb{Z}, 2)
  \;=\;
  H \mathbb{Z}[2]
$$ 
serves as the classifying space for ordinary integral cohomology in degree 2 (cf. \cite[Ex. 2.1]{FSS23Char}) 
$$
  H^2\big(
    X
    ;\,
    \mathbb{Z}
  \big)
  \;:=\;
  H^1\big(
    X
    ;
    B\mathbb{Z}
  \big)
  \;\simeq\;
  \pi_0
  \,
  \mathrm{Map}
  \big(
    \widehat{X}
    ,\,
    H \mathbb{Z}[2]
  \big)
  \,.
$$

In the same manner, there is the classifying space $H \mathbb{R}[2]$ for ordinary real cohomology. However, in this case, its defining chain complex is actually quasi-isomorphic to the 2-shifted de Rham complex (due to the Poincar{\'e} lemma) whose Dold-Kan construction is, in turn, equivalent to the moduli of 2-flux deformations from Ex. \ref{InfinityGroupoidOfFluxDeformations}, as indicated in the following diagram (cf.\cite[Lem. 9.2]{FSS23Char}):
$$
\hspace{-2mm}
  \begin{tikzcd}[
    column sep=20pt
  ]
    &[-21pt]
    \mathrm{Plt}
    \big(
      \mathbb{R}^{n\vert q}
      ,\,
      H \mathbb{R}[2]
    \big)
    \ar[
      d,
      "{
        \sim
      }"{swap, sloped}
    ]
    &[-22pt] 
      =
    &[-22pt]
    H 
    \big[
    &[-30pt]
      \mathbb{R}
    \ar[r]
    \ar[
      d,
      hook
    ]
    &[-4pt]
      0
    \ar[r]
    \ar[d]
    &[-4pt]
      0
    \ar[d]
    &[-38pt]
    \big]
    \\
    \mathrm{Plt}
    \big(
      \mathbb{R}^{n\vert q}
      ,\,
      B^2 \mathbb{R}
    \big)
    \ar[
      ur,
      "{ \sim }"{sloped}
    ]
    \ar[
      from=r,
      shorten=-2pt,
      "{ \sim }"{sloped}
    ]
    \ar[
      dr,
      "{ \sim }"{sloped}
    ]
    &
    \mathrm{Plt}
    \big(
      \mathbb{R}^{n\vert q}
      ,\,
      H \Omega^\bullet_{\mathrm{dR}}[2]
    \big)
    &=&
    H
    \big[
    &
      \Omega^0_{\mathrm{dR}}
      \big(
        \mathbb{R}^{n \vert q}
      \big)
      \ar[
        r,
        "{
          \mathrm{d}
        }"
      ]
      &
      \Omega^1_{\mathrm{dR}}
      \big(
        \mathbb{R}^{n \vert q}
      \big)
      \ar[
        r,
        "{ \mathrm{d} }"
      ]
      &
      \Omega^2_{\mathrm{dR}}
      \big(
        \mathbb{R}^{n \vert q}
      \big)_{\mathrm{clsd}}
      &
    \big]
    \\
    &
    \hspace{-24pt}
    \mathrm{Plt}
    \big(
      \mathbb{R}^{n\vert q}
      ,\,
      \shape
      \,
      \Omega^1_{\mathrm{dR}}
      (-;b \mathbb{R})_{\mathrm{clsd}}
    \big)  
    \ar[
      u,
      "{ \sim }"{sloped}
    ]
    &=&
    &
    \hspace{-11pt}
    \Omega^2_{\mathrm{dR}}
    \big(
      \mathbb{R}^{n\vert q}
      \!\times\! 
      \Delta^2_{\mathrm{geo}}
    \big)_{\mathrm{clsd}}
    \ar[
      u,
      "{
        \int_{\Delta^2_{\mathrm{geo}}}
      }"{swap}
    ]
    \ar[
      r,
      shift left=7pt
    ]
    \ar[
      r,
      shift left=0pt
    ]
    \ar[
      r,
      shift left=-7pt
    ]
    &
    \Omega^1_{\mathrm{dR}}
    \big(
      \mathbb{R}^{n\vert q}
      \!\times\! 
      \Delta^1_{\mathrm{geo}}
    \big)_{\mathrm{clsd}}
    \ar[
      u,
      "{
        \int_{\Delta^1_{\mathrm{geo}}}
      }"{swap}
    ]
    \ar[
      r,
      shift left=3.5pt
    ]
    \ar[
      r,
      shift left=-3.5pt
    ]
    &
    \Omega^0_{\mathrm{dR}}
    \big(
      \mathbb{R}^{n\vert q}
      \!\times\! 
      \Delta^0_{\mathrm{geo}}
    \big)_{\mathrm{clsd}}
    \ar[
      u,
      equals
    ]
  \end{tikzcd}
$$
A compatible model for $B^2 \mathbb{Z}$ is obtained via
$$
  \begin{tikzcd}[
    row sep=1pt,
    column sep=45pt
  ]
    \mathrm{Plt}\big(
      \mathbb{R}^{n \vert q}
      ,\,
      H \mathbb{Z}[2]
    \big)
    \ar[
      dd,
      "{ \sim }"{sloped}
    ]
   \ar[
     dddd,
     rounded corners,
     to path={
           ([xshift=-00pt]\tikztostart.west)  
        -- ([xshift=-12pt]\tikztostart.west)  
        -- node[xshift=-9pt] {
          \scalebox{.7}{\colorbox{white}{
            $
              \mathrm{Plt}\big(
                \mathbb{R}^{n \vert q}
                ,\,
                \mathbf{ch}
              \big)
            $
            }
          }
        }
           ([xshift=-07pt]\tikztotarget.west)  
        -- ([xshift=-00pt]\tikztotarget.west)  
     }
   ]
   &[-30pt]
   =
   &[-30pt]
   H\big[
   &[-33pt]
   \mathbb{Z}
   \ar[
     d,
     "{
       \scalebox{.7}{$
       \def\arraystretch{.9}
       \begin{array}{c}
         n
         \\
         \rotatebox{-90}{$
           \mapsto
        $}
         \\
         (n,n)
       \end{array}
       $}
     }"
   ]
   \ar[
     r
   ]
   &
   0
   \ar[d]
   \ar[r]
   &
   0
   \ar[d]
   &[-35pt]
   \big]
    \\[20pt]
    &
    &&
    \mathbb{Z}
    \ar[
      r,
      hook
    ]
    &
    \Omega^0_{\mathrm{dR}}\big(
      \mathbb{R}^{n \vert q}
    \big)
    \ar[
      r,
      "{
        \mathrm{d}
      }"
    ]
    &
    \Omega^1_{\mathrm{dR}}\big(
      \mathbb{R}^{n \vert q}
    \big)
    \ar[
      ddd,
      "{
        \mathrm{d}
      }"
    ]
    \\[-15pt]
    \mathrm{Plt}\big(
      \mathbb{R}^{n \vert q}
      ,\,
      H \widehat{\mathbb{Z}[2]}
    \big)
    \ar[
      dd,
      "{
        \mathrm{fib}
      }"{description}
    ]
    &=&
    H
    \Bigg[
    &
    \oplus
    &
    \oplus
    &
    &
    \Bigg]
    \\[-15pt]
    &
    &&
    \Omega^0_{\mathrm{dR}}\big(
      \mathbb{R}^{n \vert q}
    \big)
    \ar[
      r,
      "{
        \mathrm{d}
      }"
    ]
    \ar[
      uur,
      "{
        \mathllap{-}\mathrm{id}
      }"{description}
    ]
    \ar[
      d,
      "{
        \mathrm{pr}_2
      }"
    ]
    &
    \Omega^1_{\mathrm{dR}}\big(
      \mathbb{R}^{n \vert q}
    \big)
    \ar[
      uur,
      "{
        \mathrm{id}
      }"{description}
    ]
    \ar[
      d,
      "{
        \mathrm{pr}_2
      }"
    ]
    \\[20pt]
    \mathrm{Plt}\big(
      \mathbb{R}^{n \vert q}
      ,\,
      H \Omega^\bullet_{\mathrm{dR}}[2]
    \big)
    &=&
    H\big[
    &
    \Omega^0_{\mathrm{dR}}\big(
      \mathbb{R}^{n \vert q}
    \big)
    \ar[
      r,
      "{ \mathrm{d} }"
    ]
    &
    \Omega^1_{\mathrm{dR}}\big(
      \mathbb{R}^{n \vert q}
    \big)
    \ar[
      r,
      "{ \mathrm{d} }"
    ]
    &
    \Omega^2_{\mathrm{dR}}\big(
      \mathbb{R}^{n \vert q}
    \big)
    &
    \big]
    \,,
  \end{tikzcd}
$$
where $H \widehat{\mathbb{Z}[2]}$ serves as a {\it fibrant resolution} (a degreewise surjective chain map) of the character map from (coefficient spaces for) integral to de Rham cohomology. This has the effect that the homotopy pullback along the character map may be computed as an ordinary fiber product with this resolution (by the model-category theoretic arguments reviewed in \cite[\S 1]{FSS23Char}).

\medskip
Therefore, with these equivalent models, the defining diagram \eqref{TheGaugePotentials} for gauge potentials induced by the 2-flux quantization given by the space $\mathcal{A} \defneq B^2 \mathbb{Z}$ 
$$
  \begin{tikzcd}[column sep=large]
    &
    \big(B^2 \mathbb{Z}\big)_{\mathrm{diff}}
    \ar[r]
    \ar[
      d,
      "{
        \mathbf{ch}
      }"
    ]
    \ar[
      dr,
      phantom,
      "{
        \scalebox{.6}{
          (pb)
        }
      }"
    ]
    &[-10pt]
    B^2 \mathbb{Z}
    \ar[d]
    \\
    X
    \ar[
      r,
      "{
        F_2
      }"{swap}
    ]
    \ar[
      ur,
      dashed,
      "{
        \widehat{A}
      }"
    ]
    &
    \Omega^1_{\mathrm{dR}}
    \big(
      -;
      b\mathbb{R}
    \big)_{\mathrm{clsd}}
    \, \ar[
      r,
      shorten <=-5pt,
      "{
        \eta^{\scalebox{.6}{$\shape$}}
      }"{swap}
    ]
    &
    \shape
    \,
    \Omega^1_{\mathrm{dR}}
    \big(
      -;
      b\mathbb{R}
    \big)_{\mathrm{clsd}}
  \end{tikzcd}
$$
is now modeled by the Dold-Kan construction applied to the fiber product of the above chain complexes. But this is manifestly the Deligne complex in degree 2 (\cite[Ex. 9.4, Prop. 9.5]{FSS23Char} ):
\begin{equation}
  \label{ObtainingTheDeligneComplexInDegree2}
  (B^2 \mathbb{Z})_{\mathrm{diff}}
  \;\;\;\simeq\;\;\;
  \Omega^2_{\mathrm{dR}}(-)_{\mathrm{clsd}}
  \quad 
  \underset{
    \;
    \mathclap{
      H \Omega^\bullet_{\mathrm{dR}}[2]
    }
    \;
  }{\times}
  \quad 
  H \widehat{\mathbb{Z}[2]}  
  \;\;\;=\;\;\;
    H\big[
      \underbrace{
      \mathbb{Z}
      \xhookrightarrow{\;}
      \Omega^0_{\mathrm{dR}}(-)
      \xrightarrow{\mathrm{d}}
      \Omega^1_{\mathrm{dR}}(-)
      }_{
        \scalebox{.7}{
          \color{darkblue}
          \bf
          Deligne complex
        }
      }
    \big]
    \,.
\end{equation}
This means  (cf. \cite[\S 2]{FSS13CupCS}\cite{FSS15-Stacky}) that the gauge potentials $\widehat{A}$ are locally 1-forms $A$ which globally glue to constitute connections on $\mathrm{U}(1)$-principal bundles, as it should be for the electromagnetic field subject to Dirac charge quantization.

\medskip

\noindent
{\bf RR-Flux quantization in homotopical language.}
In variation of this situation, consider now the duality-symmetric RR-field flux densities in type IIB supergravity, first for vanishing B-field. These are closed differential forms in every odd degree, hence characterized by the direct sum $L_\infty$-algebra $\underset{k \in \mathbb{N}}{\oplus} b^{2k} \mathbb{R}$ of higher line $L_\infty$-algebras, in that
$$
  \underset{k \in \mathbb{N}}{\oplus}
  \;
  \Omega^{2k+1}_{\mathrm{dR}}\big(
    -
  \big)_{\mathrm{clsd}}
  \;\;=\;\;
  \Omega^1_{\mathrm{dR}}\big(
    -
    ;\,
    \underset{k \in \mathbb{N}}{\oplus}
    b^{2k} \mathbb{R}
  \big)_{\mathrm{clsd}}
  \,.
$$
Evident choices of  topological spaces with this algebra as their Whitehead $L_\infty$-algebra are 
\begin{itemize}[
  leftmargin=.6cm,
  topsep=1pt,
  itemsep=2pt
]
\item[\bf (i)] the product $\mathcal{A} \,\defneq\, \underset{k \in \mathbb{N}}{\prod} B^{2k+1} \mathbb{Z}$ of integral Eilenberg-MacLane spaces in every positive odd degree;

\item[\bf (ii)] the classifying space $\mathcal{A} \,\defneq\, \mathrm{ku}_1 \,\simeq\, \mathrm{U}$ of complex topological K-theory in degree=1.
\end{itemize}

\smallskip

\noindent
The first choice leads, in direct generalization of the previous example \eqref{ObtainingTheDeligneComplexInDegree2} to flux quantization in even-periodic ordinary differential cohomology (cf. \cite[Prop. 9.5]{FSS23Char}), while the second choice leads to flux quantization in differential K-theory (cf. \cite[Ex. 9.2]{FSS23Char}).

However, if the B-field flux $H_3$ is not assumed to vanish, then the duality-symmetric RR-flux densities instead satisfy the Bianchi identity
$$
  \mathrm{d}H_3
  \;=\;
  0
  \,,
  \;\;\;\;\;
  \mathrm{d}
  \, 
  F_{2k+1}
  \;=\;
  H_3 \, F_{2 k-1}
$$
whose characteristic $L_\infty$-algebra is the Whitehead $L_\infty$-algebra of the classifying space $\mathcal{A} \,\defneq\, \mathrm{k u}_1 \!\sslash\! \mathrm{PU}$ for  {\it 3-twisted} K-theory \cite[Lem. 2.31]{BMSS19}\cite[Ex. 6.6]{FSS23Char}\cite[Ex. 4.5.4]{SS21EBund} (with no direct analog for ordinary integral cohomology):
\begin{equation}
  \label{TwistedKTheory}
  \mathllap{
    \scalebox{.7}{
      \color{darkblue}
      \bf
      \def\arraystretch{.9}
      \begin{tabular}{c}
        twisted
        \\
        K-theory
      \end{tabular}
    }
  }
  \mathrm{KU}^{1 + b_2}(X)
  \;=\;
  \left\{\!\!\!\!
  \adjustbox{raise=8pt}{
  \begin{tikzcd}[column sep=huge]
    &&
    \mathrm{ku}_1 \!\sslash\! \mathrm{PU}
    \ar[d]
    \\
    X 
    \ar[
      rr,
      "{
        b_2
      }",
      "{
        \scalebox{.7}{
          \color{darkgreen}
          \bf
          \def\arraystretch{.9}
          \begin{tabular}{c}
            background 
            \\
            B-field charge
          \end{tabular}
        }
      }"{swap}
    ]
    \ar[
      urr,
      dashed,
      "{
        \scalebox{.7}{
          \color{orangeii}
          \bf
          RR-field charge
        }
      }"{sloped}
    ]
    &&
    B \mathrm{PU}
  \end{tikzcd}
  }
  \!\!\!\!\right\}_{\!\!\!\Big/\mathrm{rel.hmtp.}}
\end{equation}
For this choice the induced gauge potentials according to diagram \eqref{TheGaugePotentials} are cocycles in twisted differential K-theory \cite[Ex. 11.2-3]{FSS23Char}, as assumed by the widely discussed {\it Hypothesis K} that D-brane charge is quantized in K-theory (\cite{GS-RR}, for further pointers and references see \cite[\S 4.1]{SS24Flux}).

\smallskip
On the backdrop of these examples, we turn to the case of interest here:
\smallskip

\noindent
{\bf Shifted C-field flux quantization.}
The plain Bianchi identities for the duality-symmetric C-field flux densities happen to be characterized by the Whitehead $L_\infty$-algebra of the 4-sphere \cite[\S 2.5]{Sati13}\cite[Ex. 5.3]{FSS23Char}
\smallskip 
$$
  \Omega^1_{\mathrm{dR}}
  \big(
    -;\,
    \mathfrak{l}S^4
  \big)
  \;\;
  \simeq
  \;\;
  \left\{
    \def\arraystretch{1}
    \begin{array}{l}
      G_4 \,\in\, \Omega^4_{\mathrm{dR}}(-)
      \\
      G_7 \,\in\, \Omega^7_{\mathrm{dR}}(-)
    \end{array}
    \middle\vert
    \def\arraystretch{1}
    \begin{array}{l}
    \mathrm{d}
    \,
    G_4 \;=\; 0
    \,
    \\
    \mathrm{d}
    \, G_7
    \;=\;
    \frac{1}{2} G_4 \, G_4
    \end{array}
  \right\}
  \,,
$$

\smallskip 
\noindent whence a compatible choice of flux quantization law for the C-field is in {\it 4-CoHomotopy} (''Hypothesis H'' \cite{FSS20-H} \cite{FSSHopf}), whose classifying space $\mathcal{A} \,\defneq\, S^4$ is the homotopy type of the 4-sphere.

Similar to the above case of twisted K-theory \eqref{TwistedKTheory}, there is an evident twisting of Cohomotopy cohomology theory via group actions on the classifying space $S^4$. Evident group actions on $S^4$ are induced by its various coset space realizations, such as $S^4 \,\simeq\, \mathrm{O}(5)/\mathrm{O}(4) \,\simeq\, \mathrm{Spin}(5)/\mathrm{Spin}(4)$. Since the corresponding twists are classified by $B \mathrm{Spin}(4)$ they ought to be regarded as tangential twists expressing a coupling of gravity to the C-field, quantized in 4-Cohomotopy twisted by a $\mathrm{Spin}(5)$-structure $\tau$ on spacetime (\cite[\S 2]{FSS20-H}):

\smallskip 
\begin{equation}
  \mathllap{
    \scalebox{.7}{
      \color{darkblue}
      \bf
      \def\arraystretch{.9}
      \begin{tabular}{c}
        twisted 
        \\
        Cohomotopy
      \end{tabular}
    }
  }
  \pi^\tau(X)
  \;\;:=\;\;
  \left\{\!\!
  \begin{tikzcd}[
    column sep=20pt
  ]
    &&
    { S^4 }
    \!\sslash\!
    { \mathrm{Spin}(5) }
    \ar[
      d
    ]
    \\
    X
    \ar[
      rr,
      "{ \tau }",
      "{
        \scalebox{.7}{
          \color{darkgreen}
          \bf
          \def\arraystretch{.9}
          \begin{tabular}{c}
            background 
            \\
            grav. charge
          \end{tabular}
        }
      }"{swap}
    ]
    \ar[
      dr,
      "{
        \vdash \mathrm{Fr}_X
      }"{swap, sloped}
    ]
    \ar[
      urr,
      dashed,
      "{
        \scalebox{.7}{
          \color{orangeii}
          \bf
          C-field charge
        }
      }"{sloped}
    ]
    &&
    B \mathrm{Spin}(5)
    \ar[
      dl
    ]
    \\
    &
    B \mathrm{Spin}(1,2)
    \times B\mathrm{Spin}(8)
    \,.
  \end{tikzcd}
  \!\!\right\}_{\!\!\!\Big/{\mathrm{rel.hmtp.}}}
\end{equation}
Indeed, one finds \cite[\S 3.4, Prop. 3.13]{FSS20-H} that the character map on $\mathcal{A} \,\defneq\,  S^4 \!\sslash\! \mathrm{Spin}(5)$ lands in differential 4-forms which satisfy the notorious half-shifted flux quantization expected \cite[(1.2)]{Witten97a}\cite[(1.2)]{Witten97b} for the M-theory C-field, shifted by a quarter of the Pontrjagin form $p_1$ of the tangent bundle of spacetime:
\begin{equation}
  \label{HalfIntegrallyShifted4Form}
  {[G_4, G_7]}
  \;\in\;
  \mathrm{im}
  \Big(
    \pi^\tau(X)
    \xrightarrow{\mathrm{ch}}
    H^\tau_{\mathrm{dR}}
    \big(
      X;
      \,
      \mathfrak{l}S^4
    \big)
  \Big)
  \;\;\;\;
    \Rightarrow
  \;\;\;\;
  [G_4 + \tfrac{1}{4}p_1]
  \;\in\;
  H^4(X;\mathbb{Z})
  \xrightarrow{\;}
  H^4_{\mathrm{dR}}(X)
  \,.
\end{equation}
This is one strong indication among several others (cf. review in \cite[\S 12]{FSS23Char}\cite[\S 4.2]{SS24Flux}) that the assumption of C-field flux quantization in (twisted) Cohomotopy (``Hypothesis H'') captures the non-perturbative aspects of C-field flux in M-theory.

\smallskip

\noindent
{\bf Shifted flux quantization as higher curvature correction.}
However, for the present purpose of comparing strictly to 11d supergravity, it must be noted that the half-integral shift \eqref{HalfIntegrallyShifted4Form}
is a first {\it higher curvature correction} from the point of view of supergravity \cite{Tsimpis04b}. Indeed, one finds that the assumption on the left of \eqref{HalfIntegrallyShifted4Form} also implies that the Bianchi identity for $G_7$ receives a correction by an 8-form proportional to the ``1-loop term'' $I_8$ \cite[Prop. 3.8]{FSS20-H}\cite[\S 5.3]{FSS23Char}, which is expected to be the next higher curvature correction in 11d SuGra \cite[(56)]{HT03}\cite[(4.11)]{SoueresTsimpis17}, cf. \cite[Rem. 7]{SS21M5Anomaly}.
However, the supersymmetric form of these higher curvature corrections to 11d SuGra remains incompletely understood to date.
It is expected \cite{CGNT05} that to realize them requires relaxing the torsion constraint which otherwise drives the theory (cf. Rem. \ref{FormOfTheSuperTorsionConstraint} below).

\medskip 
Therefore, in \S\ref{Supergravity} below we consider 11d SuGra in the absence of higher curvature corrections. It would be interesting to generalize this discussion to higher curvature-corrected superspace supergravity, but that is beyond the scope of the present article. Key examples, beyond flat spacetime, of supergravity solutions whose Pontrjagin forms vanish, so that the difference becomes insubstantial, are (cf. \cite[Prop. 22]{SS21M5Anomaly})  the Freund-Rubin compactications $\mathrm{AdS}_{p+2} \times S^{D-p+2}$ \cite{FreundRubin80}, here for $p = 2$ or $p = 5$. 
 
\medskip 
In conclusion:

\begin{example}[\bf Nonabelian cohomology of smooth super $\infty$-groupoids]
\label{NonabelianCohomology}
 For $X$ a supermanifold, the homotopy classes of maps of smooth super $\infty$-groupoids 
 \eqref{SimplicialCategoryOfSmoothSuperInfinityGroupoids}
 into classifying objects reflect  the generalized nonabelian cohomology of $X$ \cite[\S 2]{FSS23Char}\cite[p. 6]{SS20Orb}:
 \begin{center}
 \small 
 \def\arraystretch{1.5}
 \begin{tabular}{ccc}
   \hline
   $\mathcal{A}$ 
   &
   $\pi_0 \, \mathrm{sSmthGrpd_\infty}(X,\mathcal{A})$
   & {\bf Cohomology theory}
   \\
   \hline
   \hline
   \rowcolor{lightgray}
   $B G$ 
   & 
   $H^1(X;G)$
   &
   Ordinary nonabelian cohomology
   \\
   $B^n \mathbb{Z}$ 
   &
   $H^n(X; \mathbb{Z})$
   & 
   Ordinary integral cohomology
   \\
   \rowcolor{lightgray}
   $B \mathrm{U}
   \times \mathbb{Z}$
   &
   $\mathrm{K}(X)$
   &
   Complex K-theory
   \\
   $\mathrm{M}\mathrm{U}_n$
   &
   $\mathrm{MU}^n(X)$
   &
   Complex Cobordism cohomology
   \\
   \rowcolor{lightgray}
   $S^n$ & $\pi^n(X)$ & Cohomotopy
   \\
   $\mathcal{A}_0 \cong 
   B \big(\Omega \mathcal{A}\big)$
   &
   $H^1\big(X; \Omega \mathcal{A}\big)$
   &
   Generalized nonabelian cohomology
   \\
   \rowcolor{lightgray}
   $\shape\,\Omega^1_{\mathrm{dR}}(\mbox{-};\mathfrak{a})_{\mathrm{clsd}}$
   &
   $H^1_{\mathrm{dR}}(X;\mathfrak{a})$
   &
   Nonabelian de Rham cohomology
   \\
   \hline
 \end{tabular}
 \end{center}
\end{example}

\medskip 
This is the basis for flux-quantization on superspacetime as discussed in \S\ref{SuperFluxQuantization}.

\medskip

Here we do not further dwell on the question of which {\it choice} of flux quantization to make for the C-field in 11d SuGra and what the consequences of these choices on the global field content are
(this is surveyed in \cite[\S 4.3]{SS24Flux}). 

\smallskip

\noindent
Instead, the upshot here is that {\it every} choice of flux quantization $\mathcal{A}$ (subject to $\mathfrak{l}\mathcal{A} \simeq \mathfrak{l}S^4$) lifts to super-space, since 

\smallskip 
\begin{itemize}[leftmargin=.7cm]
\item[\bf (i)] the super-flux densities $(G_4^s, G_7^s)$ \eqref{SuperFluxDensitiesInIntroduction} satisfy the same kind of duality-symmetric Bianchi identity \eqref{SuperCFieldBianchiInIntro} as their ordinary bosonic components -- iff the super-spacetime is a solution of 11d SuGra (Thm. \ref{11dSugraEoMFromSuperFluxBianchiIdentity}),
hence they still constitute flat differential forms with coefficients in the ``M-theory gauge $L_\infty$-algebra'' $\mathfrak{l}S^4$ (Ex. \ref{Rational4Sphere}, \ref{ClosedlS4ValuedDifferentialForms}), now on super-spacetime.

\item[\bf (ii)] The differential homotopy-theory of quantization (\cite{FSS23Char}) of such $L_\infty$-algebra valued flux densities lifts to higher supergeometry as just indicated (more details in \cite{GSS25}).

\item[\bf (iii)] At the same time, irrespective of the choice of compatible flux quantization law, homotopy-theoretically defined globally-defined C-field gauge potentials \eqref{FluxQuantizedSugraFields} {\it locally} still look as expected (Prop. \ref{RecoveringTraditionalSuperCFieldGaugePotentials}, \ref{CoboundariesOfClosedLS4ValuedForms}).
\end{itemize}

\medskip
\noindent
In short, this means that the higher super Cartan geometry discussed here is a proper context for discussing the (UV-)completion of 11d supergravity.

\subsection{Super Spacetime Geometry}
\label{SuperSpaceTime}

Here we specialize the general super geometry from \S\ref{HigherSuperGeometry} to the super-Poincar{\'e} Cartan geometry of super-spacetimes. Much of the discussion applies to all dimensions $D$ and spinor representations (``number of supersymmetries'') $\mathbf{N}$, but for definiteness we specialize to the case of present interest, where $D = 11$ and $\mathbf{N} = \mathbf{32}$ (from which most other supergravity theories are obtained by dimensional reduction, anyway).
We will not shy away from recalling basics; our aim is to record all the details that make the delicate proof of Thm. \ref{11dSugraEoMFromSuperFluxBianchiIdentity} below self-contained and thus readily verifiable.

\medskip

\S\ref{TheSpinRep} -- Majorana Spinors in $D = 11$.

\S\ref{SuperFrameAndSupergravityFields} -- Super-Frame and Supergravity Fields.

\subsubsection{Majorana Spinors in $D=11$}
\label{TheSpinRep}

For reference, we spell out basic definitions and relations concerning the irreducible real (``Majorana'') spinor representation $\mathbf{32}$ of $\mathrm{Spin}(1,10)$. 
Everything here is standard, but in totality not easily referenced; we spell out some of the arguments for completeness. Similar reviews may in parts be found in \cite[\S 2.5]{MiemiecSchnakenburg06} \cite[\S A.1]{HSS19}, whose Clifford algebra conventions agree with the one used here \eqref{TheCliffordAlgebra}. Beware that a different (but easily related) convention is used in \cite{CDF91} and related literature (Rem. \ref{CliffordConventionsInCDF}).

\begin{remark}[\bf Commuting spinors]
  \label{CommutingSpinors}
  Throughout this section, the symbol ``$\psi$'' denotes a generic element in the ordinary vector space (in even super-degree, cf. Def. \ref{SuperVectorSpaces}) underlying the $\mathrm{Spin}(1,10)$-representation $\mathbf{32}$ (which we recall below). 

  \smallskip 
  \noindent {\bf (i)}  This is in contrast to the corresponding elements in the super-vector space $\FR^{1,10\vert \mathbf{32}}$, where the copy of $\mathbf{32}$ is in odd super-degree, $\mathbf{32}_{\mathrm{odd}}$.

\smallskip
\noindent {\bf (ii)} 
On the other hand, the (component) gravitino 1-forms \eqref{GravitonAndGravitino} in $\Omega^1_{\mathrm{dR}}\big(X^D; \mathbf{32}_{\mathrm{odd}}\big)$ (see Def. \ref{SuperDifferentialFormsWithCoefficients}, Def. \ref{FermionicGravitinoFieldSpace})
  again commute {\it among each other}, because their commutator picks up one sign from $\mathbf{32}_{\mathrm{odd}}$ being in odd degree and {\it another} sign from the form degree 1 being odd:
  \smallskip 
  \begin{equation}
    \label{GravitinoFormsCommutingAmongEachOther}
    \psi^\alpha, \psi^\beta
    \;\in\;
    \Omega^1_{\mathrm{dR}}\big(
      -
      ;\,
      \mathbf{32}_{\mathrm{odd}}
    \big)
    \;\;\;\;\;\;
      \Rightarrow
    \;\;\;\;\;\;
    \psi^\alpha
    \,
    \psi^\beta
    \;=\;
    \psi^\beta
    \,
    \psi^\alpha
    \;\;\;\;
    \in
    \;
    \Omega^2_{\mathrm{dR}}\big(
      -
      ;\,
      \mathbf{32}_{\mathrm{odd}}
      \otimes
      \mathbf{32}_{\mathrm{odd}}
    \big)
    \,.
  \end{equation}
  (This is the case independently of the super-homological sign rule being used, cf. Rem. \ref{SignsInHomotopicalSuperAlgebra}.)

\smallskip     
\noindent {\bf (iii)}  Therefore, all the statements about multilinear expressions on $\mathbf{32}$ in the following hold verbatim whether the symbol ``$\psi$'' that enters them is regarded as an element of $\mathbf{32}$ or as an element of $\Omega^1_{\mathrm{dR}}\big(-; \mathbf{32}_{\mathrm{odd}}\big)$, and for this reason it is not only harmless but in fact suggestive to use the same symbol ``$\psi$'' in both cases --- as is usual in the supergravity literature.
\end{remark}

\noindent
{\bf Octonionic spinors.}
In analogy to how spinors in $D=4$ are controlled by $4 \times 4$ Dirac matrices with coefficients in the  complex numbers $\mathbb{C}$, so spinors in $D = 11$ are controlled by Dirac-like $4 \times 4$ matrices with coefficients in the algebra of octonions $\mathbb{O}$ 
(this is due to \cite{KT83}, expanded on in \cite{BH11}, we follow \cite[Ex. A.12]{HSS19}\cite[\S 3.2]{FSS21-SU2}).

We do not need octonion algebra anywhere else in the article, but here it serves to neatly establish the all-important existence of the real $\mathrm{Spin}(1,10)$-representation $\mathbf{32}$ \eqref{TheSpinRepresentation} with its bilinear form \eqref{SpinorPairing}.

\smallskip

\noindent
\hspace{-.1cm}
\def\tabcolsep{0pt}
\begin{tabular}{p{11cm}l}
The $\FR$-algebra $\mathbb{O} \cong_{{}_\FR} \FR^8$ of {\it octonions} is generated by seven elements $\mathrm{e}_1, \cdots, \mathrm{e}_7$ subject to the relations  $\mathrm{e}_i \cdot \mathrm{e}_i = -1$, for all $i \in \{1, \cdots, 7\}$,
and 
$$
  \begin{array}{l}
  a \cdot b = c,
  \;\;
  c \cdot a = b
  ,\;\;
  b \cdot c = a
  \,,
  \;\;\;\;
  b \cdot a = - c
  \end{array}
$$
for every consecutive pair of arrows $a \to b \to c$ in the diagram on the right.
This becomes a real star-algebra under $1^\ast = 1$, $\mathrm{e}_i^\ast = - \mathrm{e}_i$, and it becomes a real inner product space with $\langle v,w\rangle = \mathrm{Re}\big( v^\ast \cdot w \big)$.

For any imaginary octonion $v\in \mathrm{Im}(\mathbb{O}) := \FR\langle \mathrm{e}_1, \cdots \mathrm{e}_7 \rangle \subset \mathbb{O}$ (i.e., excluding a scalar summand), 
we write
$$
  L_{v}
  \;:\;
  \mathbb{O} \to \mathbb{O}
$$
for its left multiplication action. From the above relations, one finds that these operators represent the Clifford algebra $\mathrm{C}\ell\big(\mathrm{Im}(\mathbb{O}), -\vert\mbox{--}\vert^2\big)$:
$$
  L_{v} \circ L_v
  \;=\;
  -\vert v\vert^2 \,
  \mathrm{id}_{\mathbb{O}}
$$
and that 
\begin{equation}
  \label{LeftProductOfAllUnitQuaternions}
  L_{\mathrm{e}_7}
  L_{\mathrm{e}_6}
  L_{\mathrm{e}_5}
  L_{\mathrm{e}_4}
  L_{\mathrm{e}_3}
  L_{\mathrm{e}_2}
  L_{\mathrm{e}_1}
  \;=\;
  \mathrm{id}_{\mathbb{O}}
  \,.
\end{equation}
&
\hspace{25pt}
\adjustbox{
  scale=.65,
  raise=-4.7cm
}{
\begin{tikzpicture}

 \begin{scope}[rotate={(-20)}]
 \draw[->, -{>[scale=1.24]}, >=Latex, line width=1.2pt]
   (0+4:2) arc (0+4:-120:2);
 \draw[->, -{>[scale=1.24]}, >=Latex, line width=1.2pt]
   (-120+4:2) arc (-120+4:-240:2);
 \draw[->, -{>[scale=1.24]}, >=Latex, line width=1.2pt]
   (-240+4:2) arc (-240+4:-360:2);
 \end{scope}

 \begin{scope}[rotate={(+0)}]
 \draw[white, line width=2.8pt]
   (90+60:-4) to (90+60:-1);
 \draw[->, -{>[scale=1.24]}, >=Latex, line width=1.2pt]
   (90+60:2) to (90+60:.7);
 \draw[->, -{>[scale=1.24]}, >=Latex, line width=1.2pt]
   (90+60:1.24) to (90+60:-2.8);
 \draw[->, -{>[scale=1.24]}, >=Latex, line width=1.2pt]
   (90+60:-2.4) to (90+60:-4);
 \end{scope}

 \begin{scope}[rotate={(+120)}]
 \draw[white, line width=2.8pt]
   (90+60:-4) to (90+60:-1);
 \draw[->, -{>[scale=1.24]}, >=Latex, line width=1.2pt]
   (90+60:2) to (90+60:.7);
 \draw[->, -{>[scale=1.24]}, >=Latex, line width=1.2pt]
   (90+60:1.24) to (90+60:-2.8);
 \draw[->, -{>[scale=1.24]}, >=Latex, line width=1.2pt]
   (90+60:-2.4) to (90+60:-4);
 \end{scope}

 \begin{scope}[rotate={(-120)}]
 \draw[white, line width=2.8pt]
   (90+60:-4) to (90+60:-1);
 \draw[->, -{>[scale=1.24]}, >=Latex, line width=1.2pt]
   (90+60:2) to (90+60:.7);
 \draw[->, -{>[scale=1.24]}, >=Latex, line width=1.2pt]
   (90+60:1.24) to (90+60:-2.8);
 \draw[->, -{>[scale=1.24]}, >=Latex, line width=1.2pt]
   (90+60:-2.4) to (90+60:-4);
 \end{scope}

 \draw[->, -{>[scale=1.24]}, >=Latex, line width=1.1pt]
   (3.45,-2) to ({3.45/2},-2);
 \draw[->, -{>[scale=1.24]}, >=Latex, line width=1.1pt]
   ({3.45/2+.4},-2) to ({-3.45/2-.2},-2);
 \draw[line width=1.1pt]
   ({-3.45/2+.4},-2) to (-3.45,-2);

 \begin{scope}[rotate=(+120)]
 \draw[->, -{>[scale=1.24]}, >=Latex, line width=1.1pt]
   (3.45,-2) to ({3.45/2},-2);
 \draw[->, -{>[scale=1.24]}, >=Latex, line width=1.1pt]
   ({3.45/2+.4},-2) to ({-3.45/2-.2},-2);
 \draw[line width=1.1pt]
   ({-3.45/2+.4},-2) to (-3.45,-2);
 \end{scope}

 \begin{scope}[rotate=(-120)]
 \draw[->, -{>[scale=1.24]}, >=Latex, line width=1.1pt]
   (3.45,-2) to ({3.45/2},-2);
 \draw[->, -{>[scale=1.24]}, >=Latex, line width=1.1pt]
   ({3.45/2+.4},-2) to ({-3.45/2-.2},-2);
 \draw[line width=1.1pt]
   ({-3.45/2+.4},-2) to (-3.45,-2);
 \end{scope}

 \draw[fill=white] (90+60:2) circle (.6);
 \draw[fill=white] (90-60:2) circle (.6);
 \draw[fill=white] (90-60-120:2) circle (.6);

 \draw[fill=white] (90-60-120:0) circle (.6);

 \draw[fill=white] (90+60:-4) circle (.6);
 \draw[fill=white] (90-60:-4) circle (.6);
 \draw[fill=white] (90-60-120:-4) circle (.6);

 \draw (90+180:2) node
   {
     \scalebox{1.2}{
     $
       {\rm e}_1
     $
     }
   };
 \begin{scope}[shift={(0,-.36)}]
 \draw (90+180:2) node
   {
     \scalebox{.75}{
     $
       = {\rm i}
     $
     }
   };
 \end{scope}

 \draw (90+60:2) node
   {
     \scalebox{1.2}{
     $
       {\rm e}_2
     $
     }
   };
 \begin{scope}[shift={(0,-.36)}]
 \draw (90+60:2) node
   {
     \scalebox{.75}{
     $
       = {\rm j}
     $
     }
   };
 \end{scope}

 \draw (90-60:2) node
   {
     \scalebox{1.2}{
     $
       {\rm e}_3
     $
     }
   };
 \begin{scope}[shift={(0,-.36)}]
 \draw (90-60:2) node
   {
     \scalebox{.75}{
     $
       = {\rm k}
     $
     }
   };
 \end{scope}

 \draw (90-60:0) node
   {
     \scalebox{1.2}{
     $
       {\rm e}_4
     $
     }
   };
 \begin{scope}[shift={(0,-.36)}]
 \draw (90-60:0) node
   {
     \scalebox{.75}{
     $
       = \ell
     $
     }
   };
 \end{scope}

 \draw (90+60:-4) node
   {
     \scalebox{1.2}{
     $
       {\rm e}_6
     $
     }
   };
 \begin{scope}[shift={(0,-.36)}]
 \draw (90+60:-4) node
   {
     \scalebox{.75}{
     $
       = {\rm j} \ell
     $
     }
   };
 \end{scope}

 \draw (90-60:-4) node
   {
     \scalebox{1.2}{
     $
       {\rm e}_7
     $
     }
   };
 \begin{scope}[shift={(0,-.36)}]
 \draw (90-60:-4) node
   {
     \scalebox{.75}{
     $
       = {\rm k} \ell
     $
     }
   };
 \end{scope}

 \draw (90+180:-4) node
   {
     \scalebox{1.2}{
     $
       \mathrm{e}_5
     $
     }
   };
 \begin{scope}[shift={(0,-.36)}]
 \draw (90+180:-4) node
   {
     \scalebox{.75}{
     $
       = {\rm i} \ell
     $
     }
   };
 \end{scope}

 \draw[fill=green, draw opacity=0, fill opacity=.15]
   (90+60:2) circle (.6);
 \draw[fill=green, draw opacity=0, fill opacity=.15]
   (90-60:2) circle (.6);
 \draw[fill=lightgray, draw opacity=0, fill opacity=.2]
   (90-60-120:2) circle (.6);
 \draw[fill=green, draw opacity=0, fill opacity=.15]
   (90-60-120:2) circle (.6);

 \draw[fill=cyan, draw opacity=0, fill opacity=.15]
   (90-60-120:0) circle (.6);

 \draw[fill=cyan, draw opacity=0, fill opacity=.15]
   (90+60:-4) circle (.6);
 \draw[fill=cyan, draw opacity=0, fill opacity=.15]
   (90-60:-4) circle (.6);
 \draw[fill=cyan, draw opacity=0, fill opacity=.15]
   (90-60-120:-4) circle (.6);

\end{tikzpicture}
}
\end{tabular}

In this form, the $\mathrm{Pin}^+(1,10)$-Clifford algebra 
\eqref{TheCliffordAlgebra}
is naturally realized by ``octonionic Dirac matrices'', namely by the following
$4 \times 4$ octonionic matrices:
\begin{equation}
  \label{OctonionicCliffordReprsentation}
  \def\arraystretch{1.1}
  \def\arraycolsep{00pt}
  \begin{array}{ccccccc}
    \Gamma_a
    &\in&
    \mathrm{End}(\FR^2)
    &\otimes&
    \mathrm{End}(\FR^2)
    &\otimes&
    \mathrm{End}(\mathbb{O})
    \\
    \Gamma_0
    &=&
    J 
    &\otimes&
    1
    &\otimes&
    1
    \\
    \Gamma_1
    &=&
    \epsilon
    &\otimes&
    \tau
    &\otimes&
    1
    \\
    \Gamma_2
    &=&
    \epsilon 
    &\otimes&
    \epsilon
    &\otimes&
    1
    \\
    \Gamma_{2+i}
    &=&
    \epsilon
    &\otimes&
    J
    &\otimes&
    L_{\mathrm{e}_{i}}
    \\
    \Gamma_{10}
    &=&
    \tau
    &\otimes&
    1
    &\otimes&
    1
    \mathrlap{\,,}
  \end{array}
\end{equation}
where
$$
 \tau
 \;:=\;
 \bigg(
 \def\arraycolsep{3pt}
 \def\arraystretch{1.2}
 \begin{array}{cc}
   1 & 0 
   \\
   0 & -1
 \end{array}
 \bigg)
 \,,
 \hspace{.4cm}
 \epsilon
 \;:=\;
 \bigg(
 \def\arraycolsep{3pt}
 \def\arraystretch{1.2}
 \begin{array}{cc}
   0 & 1 
   \\
   1 & 0
 \end{array}
 \bigg)
 \,,
 \hspace{.4cm}
 J
 \;:=\;
 \tau
 \cdot
 \epsilon
 \;:=\;
 \bigg(
 \def\arraycolsep{3pt}
 \def\arraystretch{1.2}
 \begin{array}{cc}
   0 & 1 
   \\
   -1 & 0
 \end{array}
 \bigg)
 \,.
$$
and thus canonically represented on 
\begin{equation}
  \label{TheSpinRepresentation}
  \mathbf{32}
  \;:=\;
  \mathbb{O}^4
  \cong_{{}_{\mathbb{R}}}
  \FR^{32}
  .
\end{equation}

\medskip

\begin{lemma}[{\bf Hodge duality for Clifford basis elements}, {e.g. \cite[Prop. 6]{MiemiecSchnakenburg06}}]
\label{HodgeDualityForCliffordBasisElements}
For $p \in \{1, 2, \cdots, 11\}$, we have
\begin{equation}
  \label{HodgeDualityOnCliffordAlgebra}
  \Gamma^{a_1 \cdots a_p}
  \;=\;
  \tfrac{
    (-1)^{
      (p+1)(p-2)/2
    }
  }{
    (11-p)!
  }
  \,
  \epsilon^{ 
    a_1 \cdots a_p
    \,
    b_1 \cdots a_{11-p}
  }
  \,
  \Gamma_{b_1 \cdots b_{11-p}}
  \,.
\end{equation}
\end{lemma}
\noindent
For instance:
\begin{equation}
  \label{ExamplesOfHodgeDualCliffordElements}
  \def\arraystretch{1.6}
  \def\arraycolsep{10pt}
  \begin{array}{l}
    \Gamma^{
      a_1 \cdots a_{11}
    }
    =
    \epsilon^{a_1 \cdots a_{11}}
    \mathrm{Id}_{\mathbf{32}}
    \,,
    \;\;\;\;\;\;\;
    \Gamma^{a_1 \cdots a_6}
    \;=\;
    +
    \tfrac{
      1
    }{
      5!
    }
    \,
    \epsilon^{
      a_1 \cdots a_6
      \,
      \color{darkblue}
      b_1 \cdots b_5
    }
    \,
    \Gamma_{
      \color{darkblue}
      b_1 \cdots b_5
    }
    \,,
    \\
    \Gamma^{a_1 \cdots a_{10}}
    =
    \epsilon^{
      a_1 \cdots a_{10} 
      \color{darkblue}
      b
    }
    \,
    \Gamma_{
      \color{darkblue}
      b
    }
    \,,
    \;\;\;\;\;\;\;\;
    \Gamma^{a_1 \cdots a_5}
    \;=\;
    -
    \tfrac{
      1
    }{
      6!
    }
    \,
    \epsilon^{
      a_1 \cdots a_5
      \,
      \color{darkblue}
      b_1 \cdots b_6
    }
    \,
    \Gamma_{
      \color{darkblue}
      b_1 \cdots b_6
    }
    \,.
  \end{array}
\end{equation}
\begin{proof}
Using \eqref{OctonionicCliffordReprsentation} with \eqref{LeftProductOfAllUnitQuaternions}, 
we find
  \begin{align*}
  \Gamma_{10}
  \cdot
  \Gamma_{9}
  \cdot
  \Gamma_{8}
  \cdot
  \Gamma_{7}
  \cdot
  \Gamma_{6}
  \cdot
  \Gamma_{5}
  \cdot
  \Gamma_{4}
  \cdot
  \Gamma_{3}
  \cdot
  \Gamma_{2}
  \cdot
  \Gamma_{1}
  &
  \;=\;
  \tau \epsilon^3 J
  \,\otimes\,
  J \epsilon \tau
  \,\otimes\,
  1
 \;=\;
  -
  \mathrm{Id}_{\mathbf{32}}
  \,.
  \end{align*}
Switching the order of the factors produces another sign $(-1)^{10\cdot 11/2} = -1$, so that
\begin{equation}
  \label{OrderedVolumeGamma}
  \Gamma_0 
  \cdot 
  \Gamma_1 
  \cdot
  \Gamma_2
  \cdots
  \Gamma_{10}
  \;=\;
  +1
  \,,
\end{equation}
and hence
\begin{equation}
  \label{VolumeGamma}
  \Gamma_{
    a_1 \cdots a_{11}
  }
  \;=\;
  \epsilon_{a_1 \cdots a_{11}}
  \,
  \mathrm{Id}_{\mathbf{32}}
  \,.
\end{equation}
With this we compute as follows:
$$
  \def\arraystretch{1.7}
  \begin{array}{lll}
    \Gamma^{a_1 \cdots a_p}
    &    
    \;=\;
    \tfrac
      {-1}
      { (11-p)! }
    \epsilon^{
      {\color{darkblue}
      b_1 \cdots b_{11-p}
      }
      \,
      {
      \color{darkgreen}
      a_p a_{p-1} \cdots a_{1}
      }
    }
    \Gamma_{
      \color{darkblue}
      b_1 \cdots b_{11-p}
    }
    \underbrace{
    \Gamma_{
      \color{darkgreen}
      a_p a_{p-1} \cdots a_1
    }
    }_{
      \;\;\;
      \scalebox{.6}{
        no sum
      }
    }
    \Gamma^{a_1 a_{2} \cdots a_{p}}
    &
    \proofstep{
      by 
      \eqref{VolumeGamma}
      \& \eqref{NormalizationOfLeviCivitaSymbol}
    }
    \\[-10pt]
    &    
    \;=\;
    \,
    \tfrac{
      -1
    }{
      (11-p)!
    }
    \,
    \epsilon^{
      {
        \color{darkblue}
        b_1 
          \cdots 
        b_{11-p}
      }
      \,
      a_{ p } 
        \cdots 
      a_{ 1 }
    }
    \,
    \Gamma_{
      \color{darkblue}
      b_1 \cdots b_{11-p}
    }
    &
    \proofstep{
      by
      $\Gamma_{a_{\sigma(i)}} \Gamma^{a_{\sigma(i)}} = 1$
    }
    \\
    &
    \;=\;
    \tfrac{
      -(-1)^{p(p-1)/2}
    }{(11-p)!}
    \,
    \epsilon^{
      {
        \color{darkblue}
        b_1 \cdots b_{11-p} 
      }
      a_1 \cdots a_p
    }
    \,\Gamma_{
      \color{darkblue}
      b_1 \cdots b_{11-p}
    }
    \\
    &
    \;=\;
    \tfrac{
      -(-1)^{
        p(p-1)/2
        +
        p(11-p)        
      }
    }{(11-p)!}
    \,
    \epsilon^{
      a_1 \cdots a_p
      \,
      {
        \color{darkblue}
        b_1 \cdots b_{11-p} 
      }
    }
    \,
    \Gamma_{
      \color{darkblue}
      b_1 \cdots b_{11-p}
    }
    \\
    &
    \;=\;
    \tfrac{
      (-1)^{
        (p+1)(p-2)/2
      }
    }{(11-p)!}
    \,
    \epsilon^{
      a_1 \cdots a_p
      \,
      {
        \color{darkblue}
        b_1 \cdots b_{11-p} 
      }
    }
    \,
    \Gamma_{
      \color{darkblue}
      b_1 \cdots b_{11-p}
    }
    \,.
  \end{array}
$$

\vspace{-4mm}
\end{proof}

\smallskip 
\begin{lemma}[{\bf Product of linear Clifford basis elements}, e.g. {\cite[Prop. 2]{MiemiecSchnakenburg06}}]
\label{ProductOfLinearCliffordBasisElements}
\begin{equation}
  \label{GeneralCliffordProduct}
  \Gamma^{a_j \cdots a_1}
  \,
  \Gamma_{b_1 \cdots b_k}
  \;=\;
  \sum_{l = 0}^{
    \mathrm{min}(j,k)
  }
  \pm
  l!
\binom{j}{l}
 \binom{k}{l}
  \,
  \delta
   ^{[a_1 \cdots a_l}
   _{[b_1 \cdots b_l}
  \Gamma^{a_j \cdots a_{l+1}]}
  {}_{b_{l+1} \cdots b_k]}
\end{equation}
\end{lemma}
\begin{proof}
First, observe that if the $a$-indices are not pairwise distinct or the $b$-indices are not pairwise distinct then both 
sides of the equation are zero. 
Hence assume next that the indices are separately pairwise distinct, and consider their sets $A := \{a_1, \cdots, a_j\}$, $B := \{b_1, \cdots, b_k\}$ 
and the intersection $C := A \cap B$, with cardinality $\mathrm{card}(C) = l$. 
The idea is to recursively contract one pair $\Gamma^c \Gamma_{c} = 1$
with $c \in C$ at a time. 
We claim that, in the first step, this can be written as 
$$
  \Gamma^{a_j \cdots a_1}
  \Gamma_{b_1 \cdots b_k}
  \;=\;
  \frac{j k}{l}
  \,
  \Gamma^{[a_j \cdots a_2}
  \delta^{a_1]}_{[b_1}
  \Gamma_{b_2 \cdots b_k]}  
  \,.
$$
Namely, notice that for any tensor $X^{a_1 \cdots a_k}$, the expression $k X^{[a_k \cdots a_1]}$ is the signed sum over all ways of moving any one index to
the far right, and similarly $l Y^{[b_1 \cdots b_l]}$ is the signed sum over all ways of moving any one index to the far left. In contracting all the 
indices that thus become coincident ``in the middle'' of our expression, we are contracting the one index that we set out to contract, 
but since we are doing this for all $c \in C$ we are overcounting by a factor of $l$.

In order to conveniently recurse on this expression, we just move the Kronecker-delta to the left to obtain
$$
  \Gamma^{a_j \cdots a_1}
  \Gamma_{b_1 \cdots b_k}
  \;=\;
  (-1)^{j-1}
  \frac{j k}{l}
  \,
  \delta^{[a_1}_{[b_1}
  \Gamma^{a_j \cdots a_2]}
  \Gamma_{b_2 \cdots b_k]}  
  \,.
$$
Now working recursively, we arrive at
$$
  \Gamma^{a_j \cdots a_1}
  \Gamma_{b_1 \cdots b_k}
  \;=\;
  (-1)^{(j-1) \cdots (j-l)}
  \underbrace{
  \frac{
    j \cdots (j-l)
    \,
    k \cdots (k-l)
  }{l!}
  }_{\color{gray} 
    l! 
    \binom{k}{l}
    \binom{j}{l}
  }
  \,
  \delta
    ^{[a_1 \cdots a_l}
    _{[b_1 \cdots b_l}
  \underbrace{
    \Gamma^{a_j \cdots a_{l+1}]}
    \Gamma_{b_{l+1} \cdots b_{k}]}  
  }_{\color{gray} 
    \Gamma
      ^{ a_j \cdots a_{l+1}] }
      {}_{ b_{j + l} \cdots b_k] }
  }
$$
Under the brace on the far right we use that by assumption no further contraction is possible.
With the substitution under the brace made, the right-hand side can just as well be summed over $l$, since it gives zero whenever $l \,\neq\, \mathrm{card}(C)$. This yields the claimed formula \eqref{GeneralCliffordProduct}.
\end{proof}

\begin{lemma}[\bf Vanishing trace of Clifford elements]
  For $1 \leq p \leq 10$, we have
  \begin{equation}
    \label{VanishingTraceOfCliffordElements}
    \mathrm{Tr}(
      \Gamma_{a_1 \cdots a_p}
    )
    \;=\;
    0
    \,.
  \end{equation}
\end{lemma}
\begin{proof}
  By combining the plain cyclic invariance of the trace with the signed cyclic invariance of $\Gamma_{a_1 \cdots a_p}$.
\end{proof}

By the general classification of Clifford algebras, we know that every linear map $\mathbf{32} \to \mathbf{32}$ is represented by some element in the Clifford algebra. Moreover, by the identity \eqref{HodgeDualityOnCliffordAlgebra} it follows that it is sufficient to expand in $\Gamma_{a_1 \cdots a_p}$ for $p \leq 5$ and by \eqref{VanishingTraceOfCliffordElements} it follows that the coefficients are given by tracing the composite of the linear map with the given Clifford element:

\begin{proposition}[{\bf Clifford expansion of any matrix}, e.g. {\cite[(2.61)]{MiemiecSchnakenburg06}}]
Every $\FR$-linear endomorphism $  \phi
    \;\in\;
  \mathrm{End}(\mathbf{32})
$ on $\mathbf{32}$ 
may be expanded as:
\begin{equation}
  \label{CliffordExpansionOnAnyMatrix}
  \phi
    \;=\;
  \tfrac{1}{32}
  \sum_{p = 0}^5
  \;
  \frac{
    (-1)^{p(p-1)/2}
  }{ p! }
  \mathrm{Tr}\big(
    \phi \circ 
    \Gamma_{a_1 \cdots a_p}
  \big)
  \Gamma^{a_1 \cdots a_p}.
\end{equation}
\end{proposition}

\begin{lemma}[{\bf Spinor pairing}, e.g. {\cite[Prop. 10]{BH11}}]
\label{TheSpinorPairing}
In terms of the octonionic $\mathrm{Spin}(1,10)$-representation $\mathbf{32}$
from
\eqref{TheSpinRepresentation}, 
the {\it spinor pairing}
\begin{equation}
  \label{SpinorPairing}
  \begin{tikzcd}[
    row sep=-2pt, column sep=small 
  ]
    \mathbf{32}
    \times
    \mathbf{32}
    \ar[
      rr,
    ]
    &&
    \FR
    \\
    (\psi, \phi)
    &\longmapsto&
    (
      \,
      \overline{\psi}
      \,
      \phi
      \,
    )
    \mathrlap{
      \;:=\;
      \mathrm{Re}\big(
        \psi^\dagger
        \cdot
        \Gamma_0
        \cdot
        \phi
    \big)
    }
  \end{tikzcd}
\end{equation}
is bi-linear, 
$\mathrm{Spin}(1,10)$-equivariant.
and
skew-symmetric
\begin{equation}
  \label{SkewSymmetryOfSpinorPairing}
  \big(\, \overline{\psi} \, \phi \, \big)
  \;=\;
  -
  \big(\, \overline{\phi} \, \psi\, \big)\,.
\end{equation}
\end{lemma}
\begin{remark}[\bf Adjointness of Clifford generators]
Noticing from \eqref{OctonionicCliffordReprsentation} that
\begin{equation}
  \label{DirectAdjointnessOfCliffordGenerators}
  (\Gamma_a)^\dagger
  \;=\;
  \left\{\!\!\!
  \def\arraystretch{1.3}
  \begin{array}{lcl}
    - \Gamma_a &\vert& a = 0
    \\
    + \Gamma_a &\vert& a\neq 0
  \end{array}
  \!\! \right\}
  \;=\;
  \Gamma_0 \, \Gamma_a\, \Gamma_0
  \,,
\end{equation}
the Clifford generators are skew self-adjoint with respect to the spinor pairing \eqref{SpinorPairing}
\begin{equation}
  \label{SkewSelfAdjointnessOfCliffordGenerators}
  \def\arraystretch{1.4}
  \begin{array}{ll}
    \mathllap{
    \big(\,
      \overline{
        \Gamma_a
      \psi}
      \, \phi
    \big)
    \;
    }
    =\;
    \mathrm{Re}\big(
      (\Gamma_a \, \psi)^\dagger
      \,
      \Gamma_0
      \,
      \phi
    \big)
    &
    \proofstep{
      by \eqref{SpinorPairing}
    }
    \\
    \;=\;
    \mathrm{Re}\big(
      \psi^\dagger
      \,
      \underbrace{
      \Gamma_0 
        \Gamma_a
      \Gamma_0
      }_{\color{gray} 
        (\Gamma_a)^\dagger
      }
      \,
      \Gamma_0
      \,
      \phi
    \big)
    &
    \proofstep{
      by \eqref{DirectAdjointnessOfCliffordGenerators}
    }
    \\
    \;=\;
    -
    \mathrm{Re}\big(
      \psi^\dagger
      \,
      \Gamma_0
      \,
      \Gamma_a
      \,
      \phi
    \big)
    &
    \proofstep{
      by \eqref{TheCliffordAlgebra}
    }
    \\
    \;=\;
    -
    \big(\,
      \overline{\psi}
      \,
      \Gamma_a
      \phi
    \big)
    &
    \proofstep{
      by \eqref{SpinorPairing}.
    }
  \end{array}
\end{equation}
In general:
\begin{equation}
  \label{BarAdjointnessOfGammaMatrices}
  \overline{
    \Gamma_{a_1 \cdots a_p}
  }
  \;=\;
  (-1)^{p + p(p-1)/2}
  \,
  \Gamma_{a_1 \cdots a_p}\;.
\end{equation}
\end{remark}

\begin{proposition}[{\bf Basic Fierz expansion}, {e.g. \cite[p. 113]{DF82}\cite[Prop. 5]{MiemiecSchnakenburg06}}] The following identity holds: 
\begin{equation}
  \label{FierzDecomposition}
  \hspace{-3mm} 
  \big(\,
  \overline{\phi}_1
  \,
  \psi
  \big)
  \big(\,
  \overline{\psi}
  \,
  \phi_2
  \big)
  \;
  =
  \;
  \tfrac{1}{32}\Big(
    \big(\,
      \overline{\psi}
      \,\Gamma^a\,
      \psi
    \big)
    \big(\,
      \overline{\phi}_1
      \,\Gamma_a\,
      \phi_2
    \big)
    -
    \tfrac{1}{2}
    \big(\,
      \overline{\psi}
      \,\Gamma^{a_1 a_2}\,
      \psi
    \big)
    \big(\,
      \overline{\phi}_1
      \,\Gamma_{a_1 a_2}\,
      \phi_2
    \big)
    +
    \tfrac{1}{5!}
    \big(\,
      \overline{\psi}
      \,\Gamma^{a_1 \cdots a_5}\,
      \psi
    \big)
    \big(\,
      \overline{\phi}_1
      \,\Gamma_{a_1 \cdots a_5}\,
      \phi_2
    \big)
 \Big).
\end{equation}
\end{proposition}

Due to this relation, it is often suggestive to denote the scalar multiple of a given $\psi \,\in\, \mathbf{32}$ with an expression $\big(\, \overline{\psi} \,\Gamma_{a_1 \cdots a_p}\,\psi\big) \,\in\, \mathbb{R}$ by multiplication from the right
$$
  \psi
  \big( \,
    \overline{\psi} 
    \,\Gamma_{a_1 \cdots a_p}\,
    \psi
  \big)
  \;\in\;
  \mathbf{32}
  \,.
$$
However, since the scalars $  \big( \,
    \overline{\psi} 
    \,\Gamma_{a_1 \cdots a_p}\,
    \psi
  \big)
$ 
themselves span (as the values of their indices vary) a tensor-representation of $\mathrm{Spin}(1,10)$, we may regard the span of the above expressions as a higher-spin representation
$$
  \Big\langle
    \psi^\alpha
    \big( \,
      \overline{\psi} 
      \,\Gamma_{a_1 \cdots a_p}\,
      \psi
    \big)
  \Big\rangle_{
    a_i \in \{0, \cdots, 10\}
    ,
    \alpha \in \{1, \cdots, 32\}
  }
  \;\in\;
  \mathrm{Rep}_{\mathbb{R}}\big(
    \mathrm{Spin}(1,10)
  \big)
  \,.
$$
Moreover,  since the $\psi$ are commuting variables (Rem. \ref{CommutingSpinors}) this representation must be the polarization of a sub-representation of the third symmetric tensor power $\big(\mathbf{32} \otimes \mathbf{32} \otimes \mathbf{32}\big)_{\mathrm{sym}}$.
This perspective  allows to use basic but powerful tools from representation theory to bear on the analysis of these and similar compound spinorial expressions:
 
\begin{proposition}[{\bf The general Fierz identities} {\cite[(3.1-3) \& Table 2]{DF82}\cite[(II.8.69) \& Table II.8.XI]{CDF91}}]
 
  \noindent {\bf (i)}  The $\mathrm{Spin}(1,10)$-irrep decomposition of the first few symmetric tensor powers of $\mathbf{32}$ is:
  \begin{equation}
    \label{IrrepsInSymmetricPowersOf32}
    \def\arraystretch{1.3}
    \begin{array}{rcl}
       \big(
        \mathbf{32} \otimes \mathbf{32}
      \big)_{\mathrm{sym}}
      &\cong&
      \mathbf{11}
      \,\oplus\,
      \mathbf{55}
      \,\oplus\,
      \mathbf{462}
      \\
       \big(
        \mathbf{32}
          \otimes
        \mathbf{32} 
          \otimes 
        \mathbf{32}
      \big)_{\mathrm{sym}}
      &\cong&
      \mathbf{32}
      \,\oplus\,
      \mathbf{320}
      \,\oplus\,
      \mathbf{1408}
      \,\oplus\,
      \mathbf{4424}
      \\
       \big(
        \mathbf{32}
          \otimes
        \mathbf{32} 
          \otimes 
        \mathbf{32}
          \otimes 
        \mathbf{32}
      \big)_{\mathrm{sym}}
      &\cong&
      \mathbf{1}
      \,\oplus\,
      \mathbf{165}
      \,\oplus\,
      \mathbf{330}
      \,\oplus\,
      \mathbf{462}
      \,\oplus\,
      \mathbf{65}
      \,\oplus\,
      \mathbf{429}
      \,\oplus\,
      \mathbf{1144}
      \,\oplus\,
      \mathbf{17160}
      \,\oplus\,
      \mathbf{32604}\,.
    \end{array}
  \end{equation}
  \noindent 
  {\bf (ii)} In more detail, the irreps appearing on the right are tensor-spinors spanned by basis elements
  \begin{equation}
    \label{TheHigherTensorSpinors}
    \def\arraystretch{1.6}
    \begin{array}{l}
    \big\langle
      \Xi^\alpha_{a_1 \cdots a_p}
      \;=\;
      \Xi^\alpha_{[a_1 \cdots a_p]}
    \big\rangle_{
      a_i \in \{0,\cdots, 10\}, 
      \alpha \in \{1, \cdots 32\}
    }
    \;\;\;
    \in
    \;\;
    \mathrm{Rep}_{\mathbb{R}}\big(
      \mathrm{Spin}(1,10)
    \big)
    \\
    \mbox{\rm with}
    \;\;\;
    \Gamma^{a_1} \Xi_{a_1 a_2 \cdots a_p}
    \;=\;
    0
    \end{array}
  \end{equation}
  (jointly to be denoted $\Xi^{(N)}$ for the case of the irrep $\mathbf{N}$)
  such that:
  \begin{equation}
    \label{GeneralCubicFierzIdentities}
    \def\arraycolsep{3pt}
    \def\arraystretch{1.5}
    \begin{array}{rcrrrr}
      \psi
      \big(\,
        \overline{\psi}
        \,\Gamma_a\,
        \psi
      \big)
      &=&
      \tfrac{1}{11}
      \,\Gamma_a\,
      \Xi^{(32)}
      &+\;
      \Xi^{(320)}_a
      \mathrlap{\,,}
      \\
      \psi 
      \big(\,
        \overline{\psi}
        \,\Gamma_{a_1 a_2}\,
        \psi
      \big)
      &=&
      \tfrac{1}{11}
      \,\Gamma_{a_1 a_2}\,
      \Xi^{(32)}
      &-\;
      \tfrac{2}{9}
      \,\Gamma_{[a_1}\,
      \Xi^{(320)}_{a_2]}
      &+\;
      \Xi^{(1408)}_{a_1 a_2}
      \mathrlap{\,,}
      \\
      \psi
      \big(\,
        \overline{\psi}
        \,\Gamma_{a_1 \cdots a_5}\,
        \psi
      \big)
      &=&
      -\tfrac{1}{77}
      \Gamma_{a_1 \cdots a_5}
      \Xi^{(32)}
      &+\;
      \tfrac{5}{9}
      \Gamma_{[a_1 \cdots a_4}
      \Xi^{(320)}_{a_5]}
      &+\;
      2
      \,\Gamma_{[a_1 a_2 a_3}\,
      \Xi^{(1408)}_{a_4 a_5]}
      &+\;
      \Xi^{(4224)}_{a_1 \cdots a_5}
      \mathrlap{\,.}
    \end{array}
  \end{equation}
\end{proposition}

\begin{lemma}[\bf Quadratic forms on spinors]
\label{SpinorQuadraticForms}
$\,$

\noindent {\bf (i)} The following quadratic forms on $\psi \in \mathbf{32}$ vanish:
\begin{equation}
\label{VanishingQuadraticForms}
\hspace{-5mm} 
\def\arraystretch{1.4}
\begin{array}{r}
  \overline{\psi}
  \psi 
  \;=0
  \,, \;\; 
  \overline{\psi}
  \Gamma_{[a_1 a_2 a_3]}
  \psi
  \;= 
  0
    \,, \quad 
  \overline{\psi}
  \Gamma_{[a_1 \cdots a_4]}
  \psi
  \;=
  0
    \,, \;\;
  \overline{\psi}
  \Gamma_{[a_1 \cdots a_7]}
  \psi
  \;=
  0
    \,, \;\;
  \overline{\psi}
  \Gamma_{[a_1 \cdots a_8]}
  \psi
  \;=
  0
    \,, \;\;
  \overline{\psi}
  \Gamma_{[a_1 \cdots a_{11}]}
  \psi
  \;=
  0
  \mathrlap{\,,}
\end{array}
\end{equation}
and so on. 

\noindent {\bf (ii)}  Conversely, all non-trivial quadratic forms on $\mathbf{32}$ are unique linear combinations of the following ones: 
\begin{equation}
  \label{TheNontrivialQuadraticFormsOn32}
  \def\arraystretch{1.5}
  \begin{array}{ll}
    \big(\,
    \overline{\psi}
    \Gamma_a
    \psi
    \big)
    &
    \left(
      11 \atop 1
    \right)
    \;=\;
    11
    \\
    \big(\,
    \overline{\psi}
    \Gamma_{a b}
    \psi
    \big)
    &
    \left(
      11 \atop 2
    \right)
    \;=\;
    55
    \\
    \big(\,
    \overline{\psi}
    \Gamma_{a_1 \cdots a_5}
    \psi
    \big)
    &
    \left(
      11 \atop 5
    \right)
    \;=\;
    462
    \,.
  \end{array}
\end{equation}
\end{lemma}
\begin{proof}
With the skew-symmetry of the spinor pairing \eqref{SpinorPairing} we compute as follows: 
$$
  \def\arraystretch{1.6}
  \begin{array}{ll}
    \big(\,
    \overline{\psi}
    \Gamma_{[a_1 \cdots a_p]}
    \phi
    \big)
       & \;\defneq\;
    \mathrm{Re}\big(
      \psi^\dagger
      \Gamma_0
      \Gamma_{[a_1 \cdots a_p]}
      \phi
    \big) 
    \\
   & \;=\;
    -
    \mathrm{Re}\Big(
      \phi^\dagger
      \big(\Gamma_{[a_1 \cdots a_p]}\big)^\dagger
      \Gamma_0
      \psi
    \Big) 
    \\
  &  \;=\;
    -
    \mathrm{Re}\Big(
      \phi^\dagger
      \Gamma_0
      \Gamma^{-1}_0
      \big(\Gamma_{[a_1 \cdots a_p]}\big)^\dagger
      \Gamma_0
      \psi
    \Big) 
    \\
   & \;=\;
    -
    (-1)^{p + p(p-1)/2}
    \mathrm{Re}\big(
      \phi^\dagger
      \Gamma_0
      \Gamma_{[a_1 \cdots a_p]}
      \psi
    \big) 
    \\
  &  \;=\;
    -(-1)^{p(p+1)/2}
    \big(\,
      \overline{\phi}
      \,\Gamma_{[a_1 \cdots a_p]}\,
      \psi
    \big)
    \,.
  \end{array}
$$
Moreover, due to \eqref{HodgeDualityOnCliffordAlgebra}
only the first three of the relations \eqref{VanishingQuadraticForms} are independent statements. This implies that all non-vanishing quadratic forms are linear combinations of those in \eqref{TheNontrivialQuadraticFormsOn32}, and by \eqref{IrrepsInSymmetricPowersOf32} all of these are nontrivial and independent.
\end{proof}
Below we need, among others, the following corollary of the above: 
\begin{lemma}[{\bf Mixed nondegeneracy}]
  \label{XiVanishesIfAllSymmetricPairingsOntoPsiVanish}
  Given $\psi, \, \xi \,\in\, \mathbf{32}$ with   $\psi \,\neq\, 0$ such that 
  $$
    \big(\,
      \overline{\psi}
      \,\Gamma_{a_1}\,
      \xi
    \big)
    \;=\;
    0
    \,,
    \hspace{.4cm}
    \mbox{and}
    \hspace{.4cm}
    \big(\,
      \overline{\psi}
      \,\Gamma_{a_1 a_2}\,
      \xi
    \big)
    \;=\;
    0
    \hspace{.4cm}
    \mbox{and}
    \hspace{.4cm}
    \big(\,
      \overline{\psi}
      \,\Gamma_{a_1 \cdots a_5}\,
      \xi
    \big)
    \;=\;
    0
    \;\;\;\;\;\;\;
    \mbox{for all $a_1, a_2 \cdots a_5$},
  $$
  then $\xi = 0$.
\end{lemma}
\begin{proof}
  By Lem. \ref{SpinorQuadraticForms}, the statement reduces to observing the following: 
  given two vectors in $\mathbb{R}^n$ with one of them non-vanishing but having vanishing pairing onto the other vector
  with respect to {\it all} symmetric bilinear forms on $\mathbb{R}^n$, then the second vector must be zero. 

      This is the case: For instance, we may assume without restriction that the first vector has components $(1,0,0, \cdots, 0)$, and consider as a linear basis for the space of bilinear forms $B$ those whose representing matrices have all entries vanishing except for $B_{1 i}  = B_{i 1} = 1$ for any fixed index $i$. Then the vanishing of the pairing of the two vectors with respect to all bilinear forms is equivalent to their vanishing in all these basis elements, which is equivalently the vanishing of the components $\xi_i$ for all $i$.
\end{proof}

\begin{proposition}[\bf The Fierz identities controlling $D=11$ supergravity]
\label{TheFierzIdentitiesOf11dSupergravity}
The following relations hold between quartic symmetric forms on $\mathbf{32}$:
\begin{equation}
  \label{TheQuarticFierzIdentities}
    \def\arraystretch{1.5}
    \begin{array}{rcl}
    \big(\,
      \overline{\psi}
      \Gamma_{a b}
      \psi
    \big)
    \big(\,
      \overline{\psi}
      \,\Gamma^{a}\,
      \psi
    \big)
    &=&
    \phantom{+}
    0
    \,,
    \\
    \big(\,
      \overline{\psi}
      \Gamma_{a b_1\cdots b_4}
      \psi
    \big)
    \big(\,
      \overline{\psi}
      \,\Gamma^{a}\,
      \psi
    \big)
    &=&
    \phantom{+}
    3
    \,
    \big(\,
      \overline{\psi}
      \,\Gamma_{[b_1 b_2}\,
      \psi
    \big)
    \big(\,
      \overline{\psi}
      \,\Gamma_{b_3 b_4]}\,
      \psi
    \big)
    \\
    &=&
    -
    \,
    \tfrac{1}{6}
    \,
    \big(\,
     \overline{\psi}
     \,\Gamma^{a_1 a_2}\,
     \psi
    \big)
    \big(\,
      \overline{\psi}
      \,\Gamma_{a_1 a_2 b_1 \cdots b_4}\,
      \psi
    \big)
    \,.
    \end{array}
\end{equation}
\end{proposition}
\noindent
(The first two are equivalent to the fundamental $\mathfrak{l}S^4$-valued super-cocycle relation \eqref{TheSupercocycles} and as such control the super-flux Bianchi identities in Lem. \ref{SuperBianchiIdentityForG4InComponents} \& \ref{SuperBianchiIdentityForG7InComponents}, while the last line appears in the gravitino Bianchi identity in Lem. \ref{SuperFluxAndGravitinoBianchiEquivalentToRaritaSchwinger}.)
\begin{proof}
  On the first expression:
  This is the quartic diagonal of a $\mathrm{Spin}(1,10)$-equivariant map
  $$
    \big(
      \mathbf{32}
      \,\otimes\,
      \mathbf{32}
      \,\otimes\,
      \mathbf{32}
      \,\otimes\,
      \mathbf{32}
    \big)_{\mathrm{sym}}
    \longrightarrow
    \mathbf{11}
    \,.
  $$
  But by \eqref{IrrepsInSymmetricPowersOf32} the irrep summand $\mathbf{11}$ does not appear on the left, hence 
  this map has to vanish by Schur's Lemma (\cite[(3.13)]{DF82}).
  
  For the second expression one needs a closer analysis \cite[(3.28a), Table 2]{DF82}\cite[(2.28)]{NOF86}.

  For the third line we dualize a further such relation proven in \cite{NOF86}:
  $$
    \def\arraystretch{1.6}
    \begin{array}{lll}
    & \phantom{aaa}
    \big(\,
        \overline{\psi}
        \,\Gamma^{a_1 a_2}\,
        \psi
      \big)
      \big(\,
        \overline{\psi}
        \,\Gamma_{a_1 a_2 b_1 \cdots b_4}\,
        \psi
      \big)
      \\
      & \! =
      \tfrac{1}{5!}
      \big(\,
        \overline{\psi}
        \,\Gamma^{a_1 a_2}\,
        \psi
      \big)
      \big(\,
        \overline{\psi}
        \,\Gamma^{
          \color{darkblue}
          c_1 \cdots c_5
        }\,
        \psi
      \big)
      \epsilon_{
        a_1 a_2
        \,
        b_1 \cdots b_4
        \,
        {\color{darkblue}
          c_1 \cdots c_5
        }
      }
      &
      \proofstep{
        by \eqref{ExamplesOfHodgeDualCliffordElements}
      }
      \\
      &  
      \!=
      \tfrac{1}{5!}
      \big(\,
        \overline{\psi}
        \,\Gamma^{a_1}\,
        \psi
      \big)
      \big(\,
        \overline{\psi}
        \,\Gamma^{
          a_2
          \,
          {
            \color{darkblue}
            c_1 \cdots c_5
          }
        }\,
        \psi
      \big)
      \,
      \epsilon_{
        a_1 a_2 
        \,
        b_1 \cdots b_4 
        \,
        {\color{darkblue}
          c_1 \cdots c_5
        }
        }      
      &
      \proofstep{
        by \cite[(2.29)]{NOF86}
      }
      \\
      &   
      \!=
      \tfrac{1}{5! \cdot 5!}
      \big(\,
        \overline{\psi}
        \,\Gamma^{a_1}\,
        \psi
      \big)
      \big(\,
        \overline{\psi}
        \,
        \Gamma_{
          \color{darkorange}
          d_1 \cdots d_5
        }
        \,
        \psi
      \big)
      \,
      \epsilon^{
        a_2 
        \,
        {\color{darkblue}
          c_1 \cdots c_5
        }
        \,
        {\color{darkorange}
          d_1 \cdots d_5
        }
      }
      \,
      \epsilon_{
        a_1 a_2
        \,
        b_1 \cdots b_4
        \,
        {\color{darkblue}
          c_1 \cdots c_5
        }
      }      
      &
      \proofstep{
        by \eqref{ExamplesOfHodgeDualCliffordElements}
      }
      \\
      &    
      \!=
      -
      \tfrac{5! \cdot 6!}{5! \cdot 5!}      
      \big(\,
        \overline{\psi}
        \,\Gamma^{a_1}\,
        \psi
      \big)
      \big(\,
        \overline{\psi}
        \,
        \Gamma_{
          \color{darkorange}
          d_1 \cdots d_5
        }
        \,
        \psi
      \big)
      \delta
        ^{
          {
            \color{darkorange}
            d_1 \cdots d_5
          }
        }
        _{a_1 b_1 \cdots b_4}      
      &
      \proofstep{
        by \eqref{ContractingKroneckerWithSkewSymmetricTensor}      
      }
      \\
      & 
      \!=
      -
      6
      \,
      \big(\,
        \overline{\psi}
        \,\Gamma^{a_1}\,
        \psi
      \big)
      \big(\,
        \overline{\psi}
        \,\Gamma_{a_1 b_1 \cdots b_4}\,
        \psi
      \big)
      &
      \proofstep{
        by \eqref{ContractingKroneckerWithSkewSymmetricTensor}
      }
      \\
      &    
      \!=
      -
      \,
      3 \cdot 6
      \,
      \big(\,
        \overline{\psi}
        \,\Gamma_{[b_1 b_2}\,
        \psi
      \big)
      \big(\,
        \overline{\psi}
        \,\Gamma_{b_1 b_2]}\,
        \psi
      \big)
      &
      \proofstep{
        by previous claim in \eqref{TheQuarticFierzIdentities}.
      }
    \end{array}
  $$

  \vspace{-4mm} 
\end{proof}

\subsubsection{Super-frame and Super-gravity fields}
\label{SuperFrameAndSupergravityFields}

A {\it super-spacetime} should be a super-manifold equipped with a field configuration of super-gravity (not necessarily satisfying any equations of motion, at this point).
Mathematically this means, for our purposes, that a $(D\vert\mathbf{N})$-dimensional {\it super-spacetime} of super-dimension $(D\vert\mathbf{N})$ with $D \in \mathbb{N}_{\geq 1}$ and $\mathbf{N} \,\in\, \mathrm{Rep}\big( \mathrm{Spin}(1,D-1) \big)$ is:

\vspace{1mm} 
\begin{itemize}[leftmargin=.4cm]
\setlength\itemsep{2pt}
\item
 a supermanifold equipped with a super-frame filed $(e,\psi)$ and a (super-)torsion-free spin-connection $\omega$, locally ``soldering'' the supermanifold to the super-Minkowski-spacetime $\FR^{1,D-1\vert \mathbf{N}}$.

\item
More abstractly, this is a torsion-free super-Cartan geometry modeled on the super-Poincar{\'e} group $\mathrm{Iso}(\FR^{1,D\vert \mathbf{N}})$.

\item
Yet more abstractly, in the language of \cite{SS20Orb}, this is a $\FR^{1,D-1\vert \mathbf{N}}$-fold equipped with a smooth $\mathrm{Spin}(1,D-1)$-structure which coincides with the left-invariant one on $\FR^{1,D-1\vert \mathbf{N}}$ on the bosonically first-order infinitesimal neighborhood of every point.
\end{itemize}

\medskip

Concretely:

\begin{definition}[\bf Super-spacetime and Super-gravity fields]
\label{SuperSpacetime}
A $D\vert \mathbf{N}$-dimensional {\it super-spacetime} is:

\begin{itemize}[leftmargin=1cm]

 \item[{\bf (i)}]
 A {\it supermanifold} $X$, admitting a cover by local diffeomorphisms (``{\'e}tale maps'') from the supermanifold underlying the super-Minkowski Lie algebra \eqref{TheSuperMinkowskiLieAlgebra}:
 $$
   \big\{
     \hspace{-4pt}
     \begin{tikzcd}[column sep=8pt]
       \FR^{1,D-1\vert \mathbf{N}} 
       \ar[
         from=r, 
         shorten=-3pt,
         "{ \sim }"{pos=.4,swap}
       ]
        &
       \;
       U_i 
       \; 
       \ar[
        r, 
        shorten=-5pt, 
        "{ \scalebox{.7}{\'et} }"{yshift=1pt}
      ] 
        & 
       \,  X 
     \end{tikzcd}
     \hspace{-4pt}
   \big\}_{i \in I}
   \,;
 $$

 \item[{\bf (ii)}] equipped with a {\it super Cartan connection} with respect to the canonical subgroup inclusion $\mathrm{Spin}(1,D-1) \hookrightarrow \mathrm{Iso}\big(\FR^{1,D-1\vert \mathbf{N}}\big)$ into the super-Poincar{\'e} group (e.g. \cite[\S 6.5]{Varadarajan04}), namely equipped with:

 \begin{itemize}[leftmargin=.7cm]
 \item[\bf (a)]
  A {\it super-coframe field}, hence on each $U_i$\footnote{Crucially, this does not necessarily mean that $\phi_i$ \eqref{CartanProperty} 
  is $T(U_i\xrightarrow{\sim} \FR^{1,D-1\vert \mathbf{N}})$, i.e., not necessarily the pushfoward along the corresponding local chart trivialization.
  } 
  \begin{equation}
  \label{SuperFrameField}
  \Big(
    (e^a_i)_{a=0}^{D-1}
    ,\,
    (\psi^\alpha_i)_{\alpha=1}^N
  \Big)
  \;\in\;
  \Omega^1_{\mathrm{dR}}\big(
    U_i
    ;\,
    \FR^{1,D-1\vert \mathbf{N}}
  \big)
  \end{equation}
such that 
\begin{itemize}[leftmargin=.4cm]
\item[$\bullet$] on every $U_i$ these differential forms constitute
an isomorphism from the super-tangent bundle over $U_i$ to that of super-Minkowski space (this extra property makes $(e, \psi,\omega)$ a {\it Cartan connection}):
\begin{equation}
  \label{CartanProperty}
  \phi_i
  \,:=\,
  (e_i,\psi_i)
  \;:\;
  \begin{tikzcd}
    T U_i
    \ar[
      r,
      "{ \sim }"{swap}
    ]
    &
    T \FR^{1,D-1\vert \mathbf{N}} 
    \,;
  \end{tikzcd}
\end{equation}
\item[$\bullet$] on double overlaps $ U_i \cap U_j$ the transition
$$
  \gamma_{i j}
  \;:\;
  \begin{tikzcd}
    T \big(
      U_i 
        \cap 
      U_j
    \big)
    \ar[
      rr,
      "{
        \phi_i
      }"
    ]
    &&
    T \FR^{1,D-1\vert\mathbf{N}}
    \ar[
      rr,
      "{
        \phi_j^{-1}
      }"
    ]
    &&
    T \big(
      U_i 
        \cap 
      U_j
    \big)
  \end{tikzcd}
$$
is by the action of an element of 
$\mathrm{Spin}(1,D-1)$, hence 
$$
  \Big(
    \big(\gamma_{i j}\big)^a{}_b
  \Big)_{a,b=0}^{D-1}
  \;\in\;
  \Omega^0_{\mathrm{dR}}\big(
    U_i \cap U_j
    ;\,
    \mathrm{Spin}(1,D-1) 
  \big).
$$
\end{itemize}
\item[\bf (b)]
A {\it spin-connection} hence on each $U_i$ 
\begin{equation}
  \label{TheSpinConnectionForm}
  \big(
    \omega^a{}_b
  \big)_{a,b = 0}^d
  \;\in\;
  \Omega^1_{\mathrm{dR}}\big(
    U_i
    ;\,
    \mathfrak{so}(1,D-1)
  \big)
\end{equation}
such that on double overlaps $U_i \cap U_j$ we have
$$
  (\omega_i)^a{}_b
  =
  \big(\gamma_{i j}\big)^{a}{}_{a'}
  \,
  (\omega_j)^{a'}{}_b'
  \,
  \big(\gamma^{-1}_{i j}\big)^{b'}{}_{b}
  +
  \big(\gamma_{i j}\big)^{a}{}_{c}
  \,
  \differential
  \,
  \big(\gamma^{-1}_{i j}\big)^{c}{}_{b}\;,
$$
\end{itemize}
which represents a supergravity field configuration on $X$ (not necessarily on-shell):
\begin{equation}
  \label{GravitonAndGravitino}
  \def\arraystretch{1.5}
  \begin{array}{rl}
    \scalebox{.8}{
      \color{darkblue} 
      \bf 
      Graviton
    }
    &
    \big(
      (e^a)_{a=0}^{D-1},
      \,
      (\omega^{ab})_{a,b=0}^{D-1}
    \big)
    \;\; \in\,
    \Omega^1_{\mathrm{dR}}\big(
      U
      ;\,
      \mathfrak{iso}(\FR^{1,D-1})
    \big)
    \\
    \scalebox{.8}{
      \color{darkblue} 
      \bf 
      Gravitino}    
    &
    \big(\psi^\alpha\big)_{\alpha=1}^{N}
    \qquad \qquad \qquad \, \in\,
    \Omega^1_{\mathrm{dR}}\big(
      U
      ;\,
      \mathbf{N}_{\mathrm{odd}}
    \big) , 
  \end{array}
\end{equation}
\item[{\bf (iii)}]
such that the (bosonic coframe field-component of the super-)torsion \eqref{GravitationalFieldStrengths} vanishes, on each super-chart:
  \begin{equation}
    \label{TorsionConstraint}
    T^a_i
    \;:=\;
    \differential
    \,
    e^a_i 
    +
    (\omega_i)^a{}_b
    \,
    e_i^{b}
    \;-\;
    \big(\,
    \overline{\psi}_i
    \,\Gamma^a\,
    \psi_i
    \big)
    \;=\;
    0
    \,.
  \end{equation}
\end{itemize}
\end{definition}

\smallskip 
\begin{remark}[\bf Frame- and Coordinate-indices]
\label{FrameAndCoordinateIndices}
On a given super-coordinate chart $U$, a coframe field \eqref{SuperFrameField} is expressed in the coordinate differentials as
\begin{equation}
  \label{FrameAndCoordinateDifferentials}
  e^a
  \;=\;
  e^a{}_\evencoordinateindex
  \,
  \mathrm{d}x^\mu
  \,+\,
  e^a{}_\oddcoordinateindex
  \,
  \mathrm{d}\theta^{\oddcoordinateindex}
  \,,\;\;\;\;\;\;\;
  \psi^\alpha
  \;=\;
  \psi^\alpha{}_\evencoordinateindex
  \,
  \mathrm{d}x^\evencoordinateindex
  \,+\,
  \psi^\alpha{}_\oddcoordinateindex
  \,
  \mathrm{d}\theta^\oddcoordinateindex
  \,.
\end{equation}
As usual, one uses these coefficients to translate between frame- and coordinate-indices other tensors, such as:
$$
  \partial_a
  \;:=\;
  e_a{}^{\evencoordinateindex}
  \frac{\partial}{\partial x^{\evencoordinateindex}}
  \,+\,
  e_a{}^{\oddcoordinateindex}
  \frac{\partial}{\partial \theta^{\oddcoordinateindex}}
  \,,
  \;\;\;\;\;\;\;\;\;\;
  \Gamma_{\evencoordinateindex_1 \cdots  \evencoordinateindex_p} 
    \;:=\; 
  \Gamma_{a_1 \cdots a_p} e^{a_1}{}_{[\evencoordinateindex_1} \cdots e^{a_p}{}_{\evencoordinateindex_p]}
  \,.
$$
See also \eqref{SuperBianchiIdentityForG4InComponents} and Rem. \ref{ExteriorAndCovariantDerivatives} below.
\end{remark}

\smallskip 
\begin{example}[\bf Super Minkowski Spacetime]
  \label{SuperMinkowskiSpacetime}
  The supermanifold $\mathbb{R}^{1,10\vert \mathbf{32}}$ with its canonical coordinate functions denoted $\big( (x^\evencoordinateindex)_{\evencoordinateindex=0}^{10}, (\theta^\oddcoordinateindex)_{\oddcoordinateindex=1}^{32}\big)$ becomes a super-spacetime (Def. \ref{SuperSpacetime})  
  by choosing a single chart $U := \mathbb{R}^{1,10\vert \mathbf{32}}$ and equipping it with coframe fields defined by
  \begin{equation}
    \label{SuperMinkowskiFrameField}
    \scalebox{.7}{
      \color{darkblue}
      \bf
      \def\arraystretch{.9}
      \begin{tabular}{c}
        Supergravity fields on
        \\
        super-Minkowski spacetime
      \end{tabular}
    }
    \def\arraystretch{1.2}
    \begin{array}{ccl}
      e^a \
      &:=&
      \delta^a_\evencoordinateindex
      \,\mathrm{d}x^{\evencoordinateindex}
      +
      \big(\,
        \overline{\theta}
        \,\Gamma^a\,
        \mathrm{d}
        \theta
      \big)
      \\
      \psi^\alpha
      &:=&
      \delta^\alpha_\oddcoordinateindex
      \,
      \theta^\oddcoordinateindex
      \\
      \omega^a{}_b
      &:=&
      0
      \,.
    \end{array}
  \end{equation}

  Notice that the SDG-algebra generated (over $\mathbb{R}$) by these fields is just the Chevalley-Eilenberg algebra of the super-Minkowski Lie algebra (Ex. \ref{SuperMinkowskiLieAlgebra}), thus identifying the linear span of these fields with the {\it left-invariant} 1-forms on the super-Minkowski group. It is this identification which requires the $\big(\overline{\theta} \,\Gamma^a\, \mathrm{d}\theta\big)$-term in \eqref{SuperMinkowskiFrameField} and thus the super-torsion constraint in \eqref{TorsionConstraint}, cf. Rem. \ref{FormOfTheSuperTorsionConstraint}.

\smallskip 
  Said more conceptually: A crucial difference between the bosonic translation group $\mathbb{R}^{11}$ (under addition) and the super-Minkowski translation group $\mathbb{R}^{1,10 \vert \mathbf{32}}$ is that the latter is (mildly but crucially) {\it nonabelian}, with non-trivial super Lie bracket being the super-symmetry bracket \eqref{TheSupersymmetryBracketInComponents}. It is the condition that a super-spacetime locally (on each first-order infinitesimal neighborhood) looks like super-Minkowski spacetime {\it with} its super-translation symmetry structure which demands the super-torsion constraint \eqref{TorsionConstraint} 
  that is so crucial for the theory of supergravity; cf. again Rem. \ref{FormOfTheSuperTorsionConstraint} below.
\end{example}

\begin{example}[\bf Super-Spacetimes extending ordinary pseudo-Riemannian Spin-manifolds]
\label{SuperSpacetimesExtendingOrdinarySpinManifolds}
Consider an ordinary $D=11$ manifold $\bosonic{X}$ equipped with geometric $\mathrm{Spin}(1,10)$-structure represented by a 
$\mathrm{Spin}(1,10)$-bundle 
$P \xrightarrow{\;} \bosonic{X}$. 
Via the action of $\mathrm{Spin}(1,10)$ on the representation space $\mathbf{32}$ \eqref{TheSpinRepresentation} this induces the associated spinor bundle 
$
  \mathbf{32} 
  \underset{
    \;
    \mathclap{
      \scalebox{.5}{$
      \mathrm{Spin}(1,10)
      $}
    }
    \;
  }{\times} 
  P 
  \xrightarrow{\;} \bosonic{X}$,
which in turn induces the  supermanifold extension of $\bosonic{X}$ (cf. \cite[\S 2]{Rogers84}) that in the notation 
\eqref{PlotsOfOddVctorBundle} of Ex. \ref{SmoothManifoldsAsSmoothSuperSets} reads 
\vspace{-.3cm}
\begin{equation}
  \label{ExtendingSpinManifoldToSupermanifold}
  X
  \;:=\;
  \bosonic{X}
  \,
  \big\vert
  \,
  \big(
  \mathbf{32} 
  \underset{
    \;
    \mathclap{
      \scalebox{.5}{$
      \mathrm{Spin}(1,10)
      $}
    }
    \;
  }{\times} 
  P 
  \big)
  \,.
\end{equation}
Via Batchelor's theorem \cite{Batchelor79} every $11\vert\mathbf{32}$-dimensional super-spacetime in the sense of Ex. \ref{SuperSpacetime} has underlying super-manifold of this form $X$ \eqref{ExtendingSpinManifoldToSupermanifold}, and Thm. \ref{11dSugraEoMFromSuperFluxBianchiIdentity} below implies that solutions of 11d SuGra on $\bosonic{X}$ extend to the extending super-spacetime $X$ in a unique rheonomic way (Cor. \ref{11dSugraOnBosonicSpacetimeFromQuantizableSuperFlux} below).
\end{example} 

The following Def. \ref{SuperGravitationalFieldStrengths} is the evident generalization (\cite[p. 362]{WZ77}\cite[\S 2]{GWZ79}, cf. \cite[\S III.3.2]{CDF91}) to super-geometry of the classical {\it Cartan structural equations} and their {\it Bianchi identities} (\cite[p. 368]{Cartan23}\cite[\S 2]{Scholz19}\cite[p. 748]{Chern44}\cite[\S I.2]{CDF91}\cite[\S 22.2]{Tu17}):

\begin{definition}[\bf Super-Gravitational field strengths]
\label{SuperGravitationalFieldStrengths}
Given a super-spacetime (Def. \ref{SuperSpacetime}), we define
\footnote{
  \label{ConventionsForTorsionAndCurvature}
  Our notation and conventions in \eqref{GravitationalFieldStrengths} is that of \cite[(3.5,18)]{DF82}\cite[(III.8.5,14)]{CDF91} up to an overall sign on the connection. 
  The gravitino field strength is equivalently the odd frame component 
  $\rho^\alpha \,=\, T^\alpha$ of the torsion tensor of the full super-coframe field $E := (e,\psi)$. This latter perspective is conceptually more homogeneous (used elsewhere in the literature, e.g. \cite{WZ77}\cite[\S 2]{GWZ79}) but notationally less transparent in component computations.
  
  Our choice of sign in the covariant derivative $\dd \mathcolor{purple}{+} 
 \omega$ means that the 
 connection and curvature are related by 
 $\omega_{\mathrm{CDF}} \leftrightarrow \mathcolor{purple}{-}\omega_\mathrm{GSS}$ and  $R_\mathrm{CDF} \leftrightarrow \mathcolor{purple}{-}R_\mathrm{GSS}$, respectively.
  This choice ultimately governs the sign in front of the energy-momentum tensor in the Einstein equation \eqref{TheEinsteinEquation}. With our sign convention, 
the scalar curvature of a compact Riemannian manifold contributes with a {\it positive} sign (as is most quickly verified for the round $S^3 \simeq \mathrm{SU}(2)$), which is clearly desirable. But the opposite convention is used, besides by \cite{CDF91}, also in other classical texts, e.g. \cite[below (4b)]{FreundRubin80}.

Regarding infinitesimal gauge transformations, these are given in our convention by $\delta_{\lambda}^\mathrm{GSS}(e^a,\psi) = ( + \lambda^{a}{}_b \cdot e^b\, , \, + \tfrac{1}{4}\lambda^{ab}\Gamma_{ab} \psi )$ and $\delta_{\lambda}^\mathrm{GSS} \omega^{ab} = - \dd \lambda^{ab} - [\omega, \lambda]^{ab} $, while in \cite{CDF91} the former is the same but that of the connection is given by $\delta_{\lambda}^\mathrm{CDF} \omega^{ab} =  \mathcolor{purple}{+}\dd \lambda^{ab} - [\omega, \lambda]^{ab}$.
} super-chartwise the {\it structural equations}:
\vspace{-2mm} 
\begin{equation}
  \label{GravitationalFieldStrengths}
  \def\arraystretch{1.7}
  \begin{array}{ccccll}
    \scalebox{.7}{
      \color{darkblue} 
      \bf
      \def\arraystretch{.9}
      \begin{tabular}{c}
        Super-
        \\
        Torsion
      \end{tabular}
    }
    &
    \big(
    \,
    T^a
    &:=&
    \differential
    \,
        e^a
    &
    + \,
    \omega^{a}{}_b \, e^b
    -
    (\,
      \overline{\psi}
      \,\Gamma^a\,
      \psi
    )
    \,
    \big)_{a=0}^{D-1}
    &
    \in
    \Omega^2_{\mathrm{dR}}\big(
      U;\,
      \FR^{1,D-1}
    \big),
    \\
    \scalebox{.7}{
      \color{darkblue} 
      \bf 
      Curvature
    }
    &
    \big(
    \,
    R^{ab}
    &:=&
    \differential
    \,
      \omega^{a b}
    &
    +\, 
    \omega^a{}_c 
    \,
    \omega^{c b}
    \,
    \big)_{a,b = 0}^{D-1}
    &
    \in
    \Omega^2_{\mathrm{dR}}\big(
      U;\,
      \mathfrak{so}(1,D-1)
    \big),
    \\
    \scalebox{.7}{ \color{darkblue} \bf 
      \def\arraystretch{.9}
      \def\tabcolsep{-.2cm}
      \begin{tabular}{c}
        Gravitino
        \\
        field strength
      \end{tabular}
    }
    &
    \big(
    \,
    \rho
    &:=&
    \differential 
    \,
        \psi
    &
    +\, 
    \tfrac{1}{4}
    \omega^{a b} 
     \,
    \Gamma_{a b}
    \psi
    \,
    \big)_{\alpha=1}^{N}
    &
    \in
    \Omega^2_{\mathrm{dR}}\big(
      U;\,
      \mathbf{N}_{\mathrm{odd}}
    \big).
  \end{array}
\end{equation}
By exterior calculus, these satisfy the following {\it Bianchi identities}:
\begin{equation}
  \label{SuperGravitationalBianchiIdentities}
  \def\arraystretch{1.5}
  \begin{array}{ccl}
    \overbrace{
      \differential
      \,
      T^a
      +
      \omega^a{}_b 
      \, 
      T^b
    }^{{\color{gray} 0}}
    &=&
    +
    R^{ab}
    \,
    e_b
    +
    2
    \,
    \big(\,
      \overline{\psi}
      \,\Gamma^a\,
      \rho
    \big),
    \\
    \differential
    \,
    \rho
    +
    \tfrac{1}{4}
    \,
    \omega^{a b}
    \,\Gamma_{ab}\,
    \rho
    &=&
    +
    \tfrac{1}{4}
    \,
    R^{a b}
    \,\Gamma_{ab}\,
    \psi,
    \\
    \differential
    \,
    R^{ab}
    +
    \omega^a{}_{a'}
    \,R^{a'b}
    -
    R^{a b'}
    \omega^{b}{}_{b'}
    &=&
    0
    \,,
  \end{array}
\end{equation}
  where over the brace we used the torsion constraint \eqref{TorsionConstraint}.
\end{definition}

\begin{remark}[Two meanings of ``Torsion'']
  \label{OnTheTermTorsion}
  Unfortunately, the term {\it torsion} is used for two completely unrelated notions in different fields of mathematics. This happens and is usually of little concern since these fields are rarely discussed in common -- not so for us, though:
  \begin{itemize}[leftmargin=.7cm]
    \item[\bf (i)] 
      In group theory and cohomology theory, a {\it torsion element} of an abelian group $A$ (such as a cohomology group), is an element $a \in A$ such that some multiple of it vanishes: $n \cdot a = 0$ for some $n \in \mathbb{N}$.

      One says a cohomology group {\it has torsion} if it contains such torsion elements.
    \item[\bf (ii)] 
    In differential geometry, the {\it torsion tensor} of a $G$-structure (such as a pseudo-Riemannian metric structure) on a smooth manifold is a measure for this $G$-structure being infinitesimally non-trivial. If compatible coframe fields and connections are introduced, then the torsion tensor is the covariant derivative of the coframe field \eqref{GravitationalFieldStrengths}. One says that a $G$-structure {\it has torsion} if the torsion tensor of the vielbein is non-vanishing. 
    \end{itemize}
\end{remark}

\begin{notation}[\bf coframe field expansion of super-gravitational field strengths]
\label{FrameFieldExpansionOfGravitationalFieldStrength}
Due to the Cartan property \eqref{CartanProperty} of the coframe fields, every differential form on super-spacetime has a unique expansion in the super-coframe fields $(e,\psi)$.
The corresponding expansion of the supergravity field strengths \eqref{GravitationalFieldStrengths} we denote as follows:
\begin{equation}
    \label{RhoFrameFieldExpansion}
    \mathllap{
      \scalebox{.7}{
        \color{darkblue}
        \bf
        \def\arraystretch{.9}
        \begin{tabular}{c}
          gravitino field  
          \\
          strength expansion
        \end{tabular}
      }
      \;\;\;
    }
    \rho
    \;\;=:\;\;
    \tfrac{1}{2}
    \rho_{a b}
    \,
    e^a\, e^b
    \;+\;
    H_a \psi \, e^a
    \;+\;
    \big(\,
      \overline{\psi}
        \,\kappa\,
      \psi
    \big)
  \end{equation}
  
  \begin{equation}
    \label{CurvatureTensorInFrameField}
    \mathllap{
      \scalebox{.7}{
        \color{darkblue}
        \bf
        \def\arraystretch{.9}
        \begin{tabular}{c}
          curvature 
          \\
          expansion
        \end{tabular}
      }
      \;\;\;
    }
    R^{a_1 a_2}
    \;=:\;
    \tfrac{1}{2}
    R^{a_1 a_2}{}_{b_1 b_2}
    \,
    e^{b_1}
    \,
    e^{b_2 }
    \;+\;
    \big(\,
     \overline{\CurvatureAtPsiOne}{}^{a_1 a_2}{}_b
     \,
     \psi
    \big)
    e^b
    \;+\;
    \big(\,
      \overline{\psi}
      K^{a_1 a_2}
      \psi
    \big)
    \,.
  \end{equation}
  Here the choice of symbols for the components follows \cite[\S III.8]{CDF91}, except that:
  \begin{itemize}[leftmargin=.6cm]
  \item[(i)]
  we do not set to zero the term denoted $\kappa$ in \eqref{RhoFrameFieldExpansion}; instead, we will show that this term vanishes as a consequence of the duality-symmetric super-flux Bianchi identity, see Lem. \ref{SuperFluxAndGravitinoBianchiEquivalentToRaritaSchwinger} and Rem. \ref{RoleOfDualitySymmetricBianchiIdentitiesInEnforcingTheSuperTorsionConstraint} below;
  \item[(ii)] we use ``$\CurvatureAtPsiOne$'' instead of ``$\theta$'' in \eqref{CurvatureTensorInFrameField} so as not to clash with the standard symbol for odd coordinate functions 
 \eqref{SpaceOfDifferentialFormsOnSuperpoint}\eqref{FrameAndCoordinateDifferentials}.
  \end{itemize}
\end{notation}

\medskip

Some remarks are in order:

\begin{remark}[{\bf Role of the super-torsion constraint}]
\label{FormOfTheSuperTorsionConstraint}
$\,$

\noindent {\bf (i)} The extra summand
$
  \big(\,\overline{\psi} \,\Gamma_a\, \psi\big)
$ 
which the super-torsion tensor \eqref{TorsionConstraint} \eqref{GravitationalFieldStrengths}, has 
(reflecting the {\it intrinsic torsion} of super-Minkowski spacetime, cf. Ex. \ref{SuperMinkowskiSpacetime} and Rem. \ref{TheSupersymmetryBracket}) on top of the ordinary torsion $\mathrm{d}\, e^a - \omega^a{}_b\, e^b$
is the all-important term that drives essentially everything that is non-trivial about 11d supergravity (cf. \cite{Howe97}).

\smallskip 
\noindent {\bf (ii)} For instance, without this term all quadratic spinorial expressions of the form 
$\tfrac{1}{p!}\big(\,\overline{\psi} \, \Gamma_{a_1 \cdots a_p} \, \psi\big) e^{a_1} \cdots e^{a_p}$ would be 
closed for vanishing gravitino field strength, while with this term it takes delicate Fierz identities 
(Prop. \ref{TheFierzIdentitiesOf11dSupergravity}) to make 
$\tfrac{1}{2}\big(\,\overline{\psi} \,\Gamma_{a_1 a_2}\,\psi\big) e^{a_1}\, e^{a_2}$ 
closed and the differential of 
$\tfrac{1}{5!}\big(\,\overline{\psi} \,\Gamma_{a_1 \cdots a_5}\,\psi\big) e^{a_1} \cdots e^{a_5}$ 
to be proportional to the square of 
$\tfrac{1}{2}\big(\,\overline{\psi} \,\Gamma_{a_1 a_2}\,\psi\big) e^{a_1}\, e^{a_2}$. This is how the 
structure of C-field flux Bianchi identity \eqref{SuperCFieldBianchiInIntro} is preconfigured in the
fermion structure of $11\vert\mathbf{32}$-dimensional super-spacetimes.

\smallskip 
\noindent {\bf (iii)} Even more: Next requiring that these relations remain intact even for non-vanishing 
gravitino field strength is what implies nothing less than the equations of motion of 11d supergravity 
(Thm. \ref{11dSugraEoMFromSuperFluxBianchiIdentity}).
\end{remark}

\begin{remark}[\bf Role of the super-gravitational Bianchi identities]
\label{RoleOfTheSuperGravitationalBianchiIdentities}
Equations \eqref{SuperGravitationalBianchiIdentities} are not conditions but identities satisfied by any super-spacetime.  Conversely, this means that when {\it constructing} a super-spacetime (say subject to further constraints, such as Bianchi identities for flux densities), these equations \eqref{SuperGravitationalBianchiIdentities} are a necessary condition to be satisfied by any {\it candidate} super-vielbein field, and as such they may play the role of equations of motion for the super-gravitational field, as we will see next section.
\end{remark}

\begin{remark}[{\bf Exterior and covariant derivatives on super-spacetime}]
\label{ExteriorAndCovariantDerivatives}
On a given coordinate chart, the exterior derivative on super-spacetime is given (cf. Rem. \ref{FrameAndCoordinateIndices}) by
$$
  \mathrm{d}
  \;\;=\;\;
  \mathrm{d}x^\evencoordinateindex
  \frac{\partial}{\partial x^\evencoordinateindex} 
  + 
  \dd \theta^\oddcoordinateindex 
  \frac{\partial}{\partial \theta^\oddcoordinateindex}
  \;\;=\;\;
  e^a
  \partial_a
  + 
  \psi^\alpha
  \partial_\alpha
  \,.
$$
Besides the explicit covariantizations of $\mathrm{d}$ that appear on the left of \eqref{SuperGravitationalBianchiIdentities}, the differentials of the super-frame forms \eqref{GravitationalFieldStrengths} induce covariantization of contracted indices, for example:

\noindent {\bf (i)} 
The differential of
$$
  \omega 
  \;:=\;
  \tfrac{1}{p!}
  \,
  \omega_{a_1 \cdots a_p}
  \,
  e^{a_1} \cdots e^{a_p}
$$
may be expanded as 
\begin{equation}
  \def\arraystretch{1.8}
  \def\arraycolsep{10pt}
  \begin{array}{rcl}
  \mathrm{d}
  \,
  \underset{
    \mathclap{
      \raisebox{-3pt}{
        \scalebox{.7}{
          \color{darkblue}
          \bf
          \def\arraystretch{.9}
          \begin{tabular}{c}
          exterior derivative
          on super-spacetime
          \end{tabular}
        }
      }
    }
  }{
   \big(
   \tfrac{1}{p_1}
   \, 
   \omega_{a_1 \cdots a_p}
   \, 
   e^{a_1} \cdots e^{a_p}
   \big)
  }
  &=&
  \tfrac{1}{p!}
  \underset{
    \mathclap{
      \raisebox{-3pt}{
        \scalebox{.7}{
          \color{darkblue}
          \bf
          \def\arraystretch{.9}
          \begin{tabular}{c}
          ordinary 
          \\
          covariant derivative
          \end{tabular}
        }
      }
    }
  }{
  \big(
    \covariantderivative_{\!a_0}
    \,
    \omega_{a_1 \cdots a_p}
  \big)
  }
  e^{a_0} \cdots e^{a_p}
  \;+\;
  \tfrac{1}{(p-1)!}
  \omega_{a_1 a_2 \cdots a_p}
  \underset{
    \mathclap{
      \raisebox{-3pt}{
        \scalebox{.7}{
          \color{orangeii}
          \bf
          \def\arraystretch{.9}
          \begin{tabular}{c}
          intrinsic
          \\
          super-torsion
          \end{tabular}
        }
      }
    }
  }{
  \big(\,
    \overline{\psi}
    \,\Gamma^{a_1}\,
    \psi
  \big)
  }
  e^{a_2} \cdots e^{a_p}
  \\
  &&
  \;+\;
  \tfrac{1}{p!}
  \underset{
    \mathclap{
      \raisebox{-3pt}{
        \scalebox{.7}{
          \color{darkgreen}
          \bf
          \def\arraystretch{.9}
          \begin{tabular}{c}
          spinorial
          \\
          covariant derivative
          \end{tabular}
        }
      }
    }
  }{
  \big(
    \covariantderivative_{\!\alpha}
    \,
    \omega_{a_1 \cdots a_p}
  \big)
  }
  \psi^\alpha
  e^{a_1} \cdots e^{a_p}
  \,.
  \end{array}
\end{equation}

Here the expressions $\covariantderivative_a, \covariantderivative_\alpha$ denote the components of the (super) {\it covariant derivative} $\covariantderivative(\omega_{a_1 \cdots a_p})$ of the component functions along the super-frame. The second term arises from the contribution of the intrinsic super-torsion $(\,\overline{\psi}\Gamma^a\psi)$ 
\eqref{TorsionConstraint}.

\noindent {\bf (ii)} Similarly, for the exterior derivative of differential forms with only gravitino components, 
where the gravitino field strength $\rho$ is in general a non-vanishing torsion term, we decompose:
\begin{equation}
  \def\arraystretch{1.8}
  \def\arraycolsep{10pt}
  \begin{array}{rcl}
  \mathrm{d}
  \,
  \underset{
    \mathclap{
      \raisebox{-3pt}{
        \scalebox{.7}{
          \color{darkblue}
          \bf
          \def\arraystretch{.9}
          \begin{tabular}{c}
            exterior derivative
            \\
            on super-spacetime
          \end{tabular}
        }
      }
    }
  }{
   \big(
     \kappa_\alpha
     \,
     \psi^\alpha
   \big)
  }
  &=&
  \underset{
    \mathclap{
      \raisebox{-3pt}{
        \scalebox{.7}{
          \color{darkblue}
          \bf
          \def\arraystretch{.9}
          \begin{tabular}{c}
          ordinary 
          \\
          covariant derivative
          \end{tabular}
        }
      }
    }
  }{
  \big(
    \covariantderivative_{\!a}
    \,
    \kappa_\alpha
  \big)
  }
  e^a
  \,
  \psi^\alpha
  \;+\;
  \kappa_\alpha
  \,
  \underset{
    \mathclap{
      \raisebox{-3pt}{
        \scalebox{.7}{
          \color{orangeii}
          \bf
          \def\arraystretch{.9}
          \begin{tabular}{c}
          intrinsic
          \\
          super-torsion
          \end{tabular}
        }
      }
    }
  }{
   \rho^\alpha
  }
  \\
  &&
  \;+\;
  \underset{
    \mathclap{
      \raisebox{-3pt}{
        \scalebox{.7}{
          \color{darkgreen}
          \bf
          \def\arraystretch{.9}
          \begin{tabular}{c}
          spinorial
          \\
          covariant derivative
          \end{tabular}
        }
      }
    }
  }{
  \big(
    \covariantderivative_{\!\beta}
    \,
    \kappa_\alpha
  \big)
  }
  \psi^\beta
  \,
  \psi^\alpha
  \,.
  \end{array}
\end{equation}

\noindent {\bf (iii)} In particular, for $\mathrm{Spin}(1,10)$-invariant pairings of super-frame forms with constant coefficients, the covariant derivative vanishes and just the torsion component remains, e.g.:
$$
  \mathrm{d}
  \,
  \Big(
    \tfrac{1}{p!}
    \big(\,
      \overline{\psi}
      \,\Gamma_{a_1 \cdots a_p}\,
      \psi
    \big)
    e^{a_1} \cdots e^{a_p}
  \Big)
  \;=\;
  \tfrac{1}{(p-1)!}
  \big(\,
    \overline{\psi}
    \,\Gamma_{a_1 \cdots a_p}\,
    \psi
  \big)
  \underset{
    \mathclap{
      \raisebox{-4pt}{
        \scalebox{.7}{
          \color{orangeii}
          \bf
          intrinsic super-torsion
        }
      }
    }
  }{
    {
    \color{orangeii}
    \big(\,
      \overline{\psi}
      \Gamma^{a_1}
      \psi
    \big)
    }
  }
  \,
  e^{a_2} \cdots e^{a_p}
  \;-\;
    \tfrac{2}{p!}
  \underset{
    \mathclap{
      \raisebox{-4pt}{
        \scalebox{.7}{
          \color{darkblue}
          \bf
          gravitino field strength
        }
      }
    }
  }{
    \big(\,
      \overline{\psi}
      \,\Gamma_{a_1 \cdots a_p}\,
      {\color{darkblue} \rho}
    \big)
    }
    e^{a_1} \cdots e^{a_p}
  \,.
$$

\smallskip 
\noindent
These kinds of manipulations govern the computations in \S\ref{Supergravity}.
\end{remark}

\begin{example}[{\bf Gamma-matrices are covariantly constant}, e.g. {\cite[(10)]{Pollock10}}]
With the Clifford generators, regarded as sections of the tensor product of the tangent bundle with the
endomorphism bundle of the spinor bundle, hence with all three indices being free, they are covariantly 
constant (shown here on any chart):
\begin{equation}
  \label{GammaMatricesAreCovariantlyConstant}
  \hspace{-3.5mm} 
  \def\arraystretch{2}
  \begin{array}{lll}
    \nabla_\evencoordinateindex
    \,
    \Gamma_a{}^{\alpha}{}_{\beta}
    &
    \;=\;
    \underbrace{
      \partial_{\evencoordinateindex}
      \Gamma_{a}{}^\alpha{}_\beta
    }_{ \color{gray} = 0}
    \,+\,
    \omega_{\evencoordinateindex}{}_{a}{}^{a'}
    \,
    \Gamma_{a'}{}^\alpha{}_\beta
    \,+\,
    \tfrac{1}{4}
    \omega_\evencoordinateindex{}^{b_1 b_2}
    \Big((\Gamma_{b_1 b_2})^{\alpha}{}_{\alpha'}
    \Gamma_a{}^{\alpha'}{}_\beta -
    \Gamma_a{}^{\alpha}{}_{\beta'}
    (\Gamma_{b_1 b_2})^{\beta'}{}_{\beta}
    \Big)
    &
    \proofstep{
      by 
      Rem. \ref{ExteriorAndCovariantDerivatives}
      \&
      \eqref{GravitationalFieldStrengths}
    }
    \\
 &   \;=\;
    \omega_{\evencoordinateindex}{}_{a}{}^{a'}
    \,
    \Gamma_{a'}{}^\alpha{}_\beta
    \,+\,
    \tfrac{1}{4}
    \underbrace{
    \omega_\evencoordinateindex{}^{b_1 b_2}
    \big[
      \Gamma_{b_1 b_2}
      ,\,
      \Gamma_{a}
    \big]}_{\color{gray} 
      =\;- 
      4 
      \,
      \omega
        _{\evencoordinateindex a}
        {}^{a'}
        \Gamma{a'}
        
    }{}^\alpha{}_\beta
    \\[-2pt]
 &   \;=\;
    \omega_{\evencoordinateindex}{}_{a}{}^{a'}
    \,
    \Gamma_{a'}{}^\alpha{}_\beta
    \,-\,
    \omega_{\evencoordinateindex a}{}^{a'}{}
    \,
    \Gamma_{a'}{}^\alpha{}_\beta
    \;=\;
    0
    &
    \proofstep{
      by \eqref{TheSpinConnectionForm}.
    }
  \end{array}
\end{equation}
Alternatively, this is just the component-incarnation of the fact that the Lorentz-invariant expression
$\big(\, \overline{\psi}\,\Gamma^a\,\psi\big) e_a$ is closed up to torsion terms.
\end{example}

\medskip

\noindent
{\bf Metric and Spinor metric.} We mention also the following very basic fact, since it is important in carefully checking Lem. \ref{ImplicationsOfFluxBianchiIdentityOnGravitinoFieldStrength} below.  
The spin connection \eqref{TheSpinConnectionForm} is of course compatible with the Minkowski metric (even if torsionful),
witnessed by the skew-symmetry of its indices, in that the covariant derivative of the metric vanishes (e.g. \cite[(3.16)]{Krasnov20}):
$$
  \nabla_{\!A} 
  \,
  \eta_{a_1 a_2}
  \;=\;
  \underbrace{
  \partial_A 
  \,
  \eta_{a_1 a_2}
  }_{\color{gray} 0}
  +
  \underbrace{
  \omega_A{}_{a_1}{}^{a'_1}
  \eta_{a'_1 a_2}
  }_{ \color{gray}
    \omega_{A a_1 a_2}
  }
  +
  \underbrace{
  \omega_A{}_{a_2}{}^{a'_2}
  \eta_{a_1 a'_2}
  }_{\color{gray}
    \omega_{A a_2 a_1}
  }
  \;=\;
  0
  \,,
$$
 reflecting the fact that the vector pairing $(v,w) \mapsto v^{a} \eta_{a b} w^b$
is $\mathrm{Spin}(1,10)$-equivariant. With \eqref{ContractingKroneckerWithSkewSymmetricTensor} this 
implies for instance that also the Levi-Civita tensor \eqref{NormalizationOfLeviCivitaSymbol} is covariantly constant:
\begin{equation}
  \label{LCTensorIsCovariantyConstant}
  0
  \,=\,
  \nabla_{\!A}
  \big(
  \epsilon_{a_1 \cdots a_{11}}
  \epsilon^{a_1 \cdots a_{11}}
  \big)
  \,=\,
  2
  \,
  \epsilon_{a_1 \cdots a_{11}}
  \nabla_{\! A}
  \epsilon^{a_1 \cdots a_{11}}
  \hspace{.5cm}
  \Rightarrow
  \hspace{.5cm}
  \nabla_{\!A}
  \,
  \epsilon_{a_1 \cdots a_{11}}
  \;=\;
  0
  \,.
\end{equation}

\smallskip 
We recall this because also the spinor pairing 
\eqref{SpinorPairing} is $\mathrm{Spin}(1,10)$-equivariant (Lem. \ref{TheSpinorPairing}), thus immediately serving as the {\it spinor metric} (the odd-odd component of a super-metric), with analogous statements holding for spinors: If we denote by $\eta_{\alpha \beta}$ the components of the spinor pairing
$$
  \big(\,
    \overline{\psi}
    \,
    \phi
  \big)
  \;\;
  =
  \;\;
  \psi^\alpha 
  \, \eta_{\alpha \beta} \,
  \phi^\beta
$$
and use it to shift spinor indices, then also this shifting passes through the covariant derivative:
\begin{equation}
  \label{ShiftingOfSpinorialIndicesUnderCovariantDerivative}
  \nabla_{\!A}
  \psi_\beta
  \;=\;
  \nabla_{\!A}
  \big(
    \eta_{\beta \beta'}
    \psi^{\beta'}
  \big)
  \;=\;
  \eta_{\beta \beta'}
  \nabla_{\!A}
  \big(
    \psi^{\beta'}
  \big)
  \,.
\end{equation}

\newpage 

\section{11d SuGra EoM from Super-Flux Bianchi}
\label{Supergravity}

Here we spell out in detail the proof of the following theorem, which enters our main Claim \ref{FluxQuantizedSuperFieldsOf11dSugra}:

\begin{theorem}[\bf 11d SuGra EoM from super-flux Bianchi identity]
\label{11dSugraEoMFromSuperFluxBianchiIdentity}
$\,$

\hspace{-6mm} 
\colorbox{lightblue}{
\begin{tabular}{l}
\begin{minipage}{17cm}
An $(11\vert\mathbf{32})$-dimensional super-spacetime $\big(X, (e,\psi,\omega)  \big)$  (according to Def. \ref{SuperSpacetime})
carries super C-field flux 
\begin{center}
  $(G_4^s, \, G_7^s) 
  \,\in\, 
  \Omega^1_{\mathrm{dR}}\big(X;\, \mathfrak{l}S^4\big)_{\mathrm{clsd}}
$
\end{center}
of the form of expressions \eqref{SuperFluxDensitiesInIntroduction} and diagram 
\eqref{FluxQuantizedSugraFields}
iff
\begin{itemize}[leftmargin=.6cm]
\item[\bf (i)]
it solves the equations of motion of 11d supergravity \footnotemark
with the given $G_4$-flux source:
\begin{itemize}
\item[\bf (a)] the Maxwell equation for the C-field flux 
\eqref{MaxwellEquationDerived},
\item[\bf (b)] the Rarita-Schwinger equation for the gravitino \eqref{GravitinoEquation},
\item[\bf (c)] the Einstein equation for the field of gravity \eqref{TheEinsteinEquation}.
\end{itemize}
\item[\bf (ii)]
the super-fields form a unique (``rheonomic'') extension of their restriction (in the sense of \S\ref{SuperSmoothFieldSpaces}) to $\bosonic{X}$.
\end{itemize}
\end{minipage}
\end{tabular}
}
\end{theorem}

\footnotetext{
  The 11d SuGra EoMs in their superspace form that we are deriving are neatly summarized in \cite[Table 3]{DF82}\cite[Table III.8.1]{CDF91},
  from which their original formulation on ordinary spacetime (e.g. \cite[\S 3.1]{MiemiecSchnakenburg06}) follows by expanding as in Ex. \ref{FermionicGravitinoFieldSpace}; cf. \cite[(3.33)]{DAuria19}.
}

Thm. \ref{11dSugraEoMFromSuperFluxBianchiIdentity} is a mild but consequential reformulation (as explained in \S\ref{DualitySymmetryOnSuperSpacetime}) of the claim of \cite[\S III.8.5]{CDF91} where some easy parts of the proof are indicated (and we do not assume constraints on the gravitino field strength but show that these are implied, cf. Rem. \ref{RoleOfDualitySymmetricBianchiIdentitiesInEnforcingTheSuperTorsionConstraint}), which in turn is a manifestly duality-symmetric reformulation of the original claim in \cite{CF80}\cite{BrinkHowe80} (see also \cite[\S 6]{CL94}\cite{Howe97}\cite[\S 2.5]{CGNT05}) where less details were given.

\smallskip

We spell out the detailed proof broken up into the following Lemmas \ref{SuperBianchiIdentityForG4InComponents}, \ref{SuperBianchiIdentityForG7InComponents}, \ref{SuperFluxAndGravitinoBianchiEquivalentToRaritaSchwinger},
\ref{DerivingTheEinsteinEquation},
and \ref{RheonomyForSuperCFieldFLux},
where we invoke mechanized computer algebra \cite{AncillaryFiles} to verify the steps that are heavy on Clifford algebra. \footnote{
The run-time of our computer code \cite{AncillaryFiles} suggests that a complete hand-checked proof of Thm. \ref{11dSugraEoMFromSuperFluxBianchiIdentity} would be remarkable. Comparatively easy is the derivation of the equations of motion from the Bianchi identities, but a full proof requires verifying also the converse implication that no further contraints are implied by the Bianchi identities, which may previously have received less attention.}

The following computations make intensive use of the super-coframe field components declared in Ntn. \ref{FrameFieldExpansionOfGravitationalFieldStrength} and their (covariant) differentials computed according to  Rem. \ref{ExteriorAndCovariantDerivatives}.

\begin{lemma}[\bf Bianchi identity for $G_4^s$ in components]
\label{SuperBianchiIdentityForG4InComponents}
$\,$

\noindent   The Bianchi identity $\differential \, G_4^s \,=\, 0$
  is equivalent to the following set of conditions:
  \begin{enumerate}
  \item[{\bf (i)}]
  The $G_4$-Bianchi identity holds, in that:
  \begin{equation}
    \label{DerivingClosureOfBosonic4FluxDensity}
    \hspace{.6cm}
   \colorbox{lightblue}{$
      \covariantderivative_{[a}(G_4)_{a_1 \cdots a_4]}
    \;=\;
    0
    \,.
    $}
  \end{equation}
  \item[{\bf (ii)}]
  The $(\psi^1)$-component of the gravitino field strength 
  \eqref{RhoFrameFieldExpansion}
  is a linear function of $G_4$:
  \begin{equation}
    \label{FormOfGravitinoFieldStrengthImpliedByG4sBianchi}
     \colorbox{lightblue}{$
     H_a
     \;\;
     =
     \;\;
      \tfrac{1}{6}
      \tfrac{1}{3!}
      (G_4)_{a \, b_1 b_2 b_3}
      \Gamma^{b_1 b_2 b_3}
      \,-\,
      \tfrac{1}{12}
      \tfrac{1}{4!}
      (G_4)^{b_1 \cdots b_4}
      \Gamma_{a \, b_1 \cdots b_4}
      \,.
      $}
  \end{equation}
  \item[{\bf (iii)}]
  Rheonomy for $G_4$: the odd covariant derivatives of $G_4$
  are fixed by the components $\rho_{a_1 a_2}$ of the gravitino field strength:
  \begin{equation}
    \label{OddCovariantDerivativeOfFluxDensity}
     \colorbox{lightblue}{$
      \psi^\alpha
      \,
      \covariantderivative_\alpha
      (G_4)_{a_1 \cdots a_4}
      \;=\;
    12
    \,
    \big(\,
      \overline{\psi}
      \,
      \Gamma_{[a_1 a_2}
      \,
      \rho_{a_3 a_4]}
    \big)
    \,.
    $}
  \end{equation}
  \item[{\bf (iv)}]
  The $(\psi^2)$-component of the gravitino field strength \eqref{FormOfGravitinoFieldStrengthImpliedByG4sBianchi}
  satisfies
  \begin{equation}
    \label{FirstConditionOnSpuriousRhoComponent}
     \colorbox{lightblue}{$
    \Big(
      \overline{\psi}
      \,\Gamma_{a_1 a_2}\,
      \big(\,
        \overline{\psi}
        \,\kappa\,
        \psi
      \big)
    \Big)
    \;=\;
    0
    \,.
    $}
  \end{equation}
  \end{enumerate}
\end{lemma}
\begin{proof}
  In terms of the coframe field expansion \eqref{RhoFrameFieldExpansion}
  of $\rho$, the $G_4$-Bianchi identity has the following components:   
  \begin{equation}
    \label{ComponentsOfBianchiForGs4}
    \def\arraystretch{1.6}
    \begin{array}{l}
      \differential
      \Big(
      \,
      \tfrac{1}{4!}
      (G_4)_{a_1 \cdots a_4}
      \,
      e^{a_1} \cdots e^{a_4}
      \,+\,
      \tfrac{1}{2}
        \big(\,
          \overline{\psi}
          \Gamma_{a_1 a_2}
        \psi
        \big)
        \,
        e^{a_1}\, e^{a_2}
      \Big)
      \;=\;
      0
      \\[5pt]
      \;\Leftrightarrow\;
      \left\{
      \def\arraycolsep{0pt}
      \begin{array}{l}
        \scalebox{.7}{
          \color{gray}
          $\big(\psi^0\big)$
          \;
        }
        \big(
          \covariantderivative_{[a}
          (G_4)_{a_1 \cdots a_4]}
        \big)
        e^{a}\, e^{a_1} \cdots e^{a_4}
        \;=\;
        0
        \,,
        \\
        \scalebox{.7}{
          \color{gray}
          $\big(\psi^1\big)$
          \;
        }
        \Big(
          \tfrac{1}{4!}
          \psi^\alpha
          \covariantderivative_\alpha
          (G_4)_{a_1 \cdots a_4}
          \;-\;
          \tfrac{1}{2}
          \big(\,
            \overline{\psi}
            \,\Gamma_{[a_1 a_2}\,
            \rho_{a_3 a_4]}
          \big)
        \Big)
        \,
        e^{a_1} \cdots e^{a_4}
        \;=\;
        0
        \,,
        \\
        \scalebox{.7}{
          \color{gray}
          $\big(\psi^2\big)$
          \;
        }
        \tfrac{1}{3!}
        (G_4)_{a b_1 b_2 b_3}
        \big(\,
          \overline{\psi}
          \,\Gamma^a\,
          \psi
        \big)
        \,
        e^{b_1 b_2 b_3}
        \,-\,
        \big(\,
        \overline{\psi}
        \,\Gamma_{[a_1 a_2}\,
        H_{b]}
        \psi
        \big)
        e^{a_1} \, e^{a_2} \, e^b
        \;=\;
        0
        \,,
        \\
        \scalebox{.7}{
          \color{gray}
          $(\psi^3)$
          \;\;
        }
        \Big(\,
          \overline{\psi}
          \,\Gamma_{a_1 a_2}\,
          \big(
            \overline{\psi}
            \,\kappa\,
            \psi
          \big)
        \Big)
        e^{a_1}\, e^{a_2}
        \;=\;
        0
        \,.
      \end{array}
      \right.
    \end{array}
  \end{equation}
  \noindent
  Here:
  \begin{itemize}[leftmargin=.4cm]
  \setlength\itemsep{2pt}
  \item The $(\psi^0)$-component is the claimed relation \eqref{DerivingClosureOfBosonic4FluxDensity}.

  \item The $(\psi^1)$-component is the claimed relation \eqref{OddCovariantDerivativeOfFluxDensity}.
  
  \item The $(\psi^2)$-component is solved for $H_a$ by (e.g. \cite[(III.8.43-49)]{CDF91}) 
  expanding $H_a$ in the Clifford algebra basis according to \eqref{CliffordExpansionOnAnyMatrix}, 
  observing that for $\Gamma_{a_1 a_2} H_{a_3}$ to be a linear combination of the $\Gamma_a$ the matrix $H_a$ needs to have a $\Gamma_{a_1}$-summand or a $\Gamma_{a_1 a_2 a_3}$-summand. The former does not admit a Spin-equivariant linear combination with coefficients $(G_4)_{a_1 \cdots a_4}$, hence it must be the latter. But then we may also need a component $\Gamma_{a_1 \cdots a_5}$ in order to absorb the skew-symmetric product in $\Gamma_{a_1 a_2} H_a$. Hence $H_a$ must be of this form:
  \smallskip 
  \begin{equation}
    \label{AnsatzForHa}
    H_a 
    \;=\;
    \mathrm{const}_1
    \,
    \tfrac{1}{3!}
    (G_4)_{a b_1 b_2 b_3}
    \Gamma^{b_1 b_2 b_3}
    +
    \mathrm{const}_2
    \,
    \tfrac{1}{4!}
    (G_4)^{b_1 \cdots b_4}
    \Gamma_{a b_1 \cdots b_4}
    \,.
  \end{equation}
  With this, we compute:
  \vspace{-2mm} 
  $$
    \def\arraystretch{1.3}
    \begin{array}{lll}
      \big(\,
      \overline{\psi}
      \Gamma_{a_1 a_2} H_{a_3}
      \psi
      \big)
      e^{a_1} \, e^{a_2} \, e^{a_3}
      &
      =\;
      \mathrm{const}_1
      \,
      \tfrac{1}{3!}
      (G_4)_{a_3 b_1 b_2 b_3}
      \,
     \big(\,
     \overline{\psi}
     \Gamma_{a_1 a_2}
     \Gamma^{b_1 b_2 b_3}
     \psi
     \big)
     e^{a_1} \, e^{a_2} \, e^{a_3}
     &
     \proofstep{
       by \eqref{AnsatzForHa}
     }
     \\
     &
     \;\;\;+\,
     \mathrm{const}_2
     \,
     \tfrac{1}{4!}
     \,
     (G_4)^{b_1 \cdots b_4}
     \,
     \big(\,
     \overline{\psi}
     \Gamma_{a_1 a_2}
     \Gamma_{a_3 b_1 \cdots b_4}
     \psi
     \big)
     e^{a_1} \, e^{a_2} \, e^{a_3}
     \\
     &
     \;=\;
      1
      \,
      \mathrm{const}_1
      \,
      \tfrac{1}{3!}
      \,
      (G_4)_{a_3 b_1 b_2 b_3}
      \big(\,
       \overline{\psi}
       \,\Gamma_{a_1 a_2}{}^{b_1 b_ 2 b_3}\,
       \psi
      \big)
      e^{a_1} \, e^{a_2} \, e^{a_3}
      &
      \proofstep{
        by
        \eqref{GeneralCliffordProduct} 
        \&
        \eqref{VanishingQuadraticForms}
      }
      \\
      &
      \;\;\;+\,
      6
      \,
      \mathrm{const}_1
      \,
      \tfrac{1}{3!}
      \,
      (G_4)_{b_3 a_1 a_2 a_3}
       \big(\,
       \overline{\psi}
       \,\Gamma^{b_3}\,
       \psi
       \big)
     e^{a_1} \, e^{a_2} \, e^{a_3}
     \\
     &
     \;\;\;+\,
     8
     \,
     \mathrm{const}_2
     \,
     \tfrac{1}{4!}
     \,
     (G_4)^{b_1 \cdots b_3}{}_{a_3}
     \,
     \big(\,
     \overline{\psi}
     \Gamma^{a_1 a_2}{}_{b_1 \cdots b_3}
     \psi
     \big)
     e^{a_1} \, e^{a_2} \, e^{a_3}
     \,,
    \end{array}
  $$
  where we used the following multiplicities \eqref{GeneralCliffordProduct} of the contractions that have non-vanishing spinor pairing:
  \vspace{-2mm} 
  $$
    \def\arraystretch{1.7}
    \begin{array}{l}
      1 \;=\;
      1! \binom{2}{0}
      \binom{3}{0} 
      \,, \qquad 
      6 \;=\;
      2! \binom{2}{2}
      \binom{3}{2}
      \,, \qquad 
      8 \;=\;
      1! \binom{2}{1}
      \binom{4}{1}
      \,.
    \end{array}
  $$

  \vspace{-2mm} 
\noindent  Inserting this in \eqref{ComponentsOfBianchiForGs4} yields:
$
    \mathrm{const}_1 = \tfrac{1}{6}
$
and
$
    \mathrm{const}_2 
      =     
    \,
    -
    \frac{
      4!
    }{
      3! \, 8
    }
    \,
    \mathrm{const}_1
    =
    -
    \tfrac{1}{12}
    \,,
  $
  as claimed in \eqref{FormOfGravitinoFieldStrengthImpliedByG4sBianchi}.
  \item The $(\psi^3)$-component is the claimed condition \eqref{FirstConditionOnSpuriousRhoComponent}.

  \item The would-be $(\psi^4)$-component holds due to the Fierz identity \eqref{TheQuarticFierzIdentities}:
  $
    -
    \tfrac{1}{2}
    \big(\,
      \overline{\psi}
      \Gamma_{a_1 a_2}
      \psi
    \big)
    \big(\,
      \overline{\psi}
      \Gamma^{a_1}
      \psi
    \big)
    e^{a_2}
    \;=\;
    0
    $.
  \end{itemize}

  \vspace{-4mm}
\end{proof}

\begin{lemma}[\bf Bianchi identity for $G_7^s$ in components]
\label{SuperBianchiIdentityForG7InComponents}
  Given the Bianchi identity for $G_4^s$ (cf. \ref{SuperBianchiIdentityForG4InComponents}),
  the Bianchi identity $\mathrm{d}\, G_7^s = \tfrac{1}{2} G_4^s G_4^s$ is equivalent to the following set of conditions:

  \begin{itemize} 
  \setlength\itemsep{3pt}
    \item[\bf (i)]
      The $G_7$-Bianchi identity:
      \vspace{-1mm} 
      \begin{equation}
        \label{BosonicComponentOfG7BianchiIdentity}
         \colorbox{lightblue}{$
        \big(
        \covariantderivative
        _{a_1}
        \tfrac{1}{7!}
        (G_7)_{a_2 \cdots a_8}
        \big)
        e^{a_1}
        \cdots 
        e^{a_8}
        \;=\;
        \tfrac{1}{2}
        \Big(
        \tfrac{1}{4!}
        (G_4)_{a_1 \cdots a_4}
        \, 
        \tfrac{1}{4!}
        (G_4)_{a_5 \cdots a_8}
        \Big)
        e^{a_1} \cdots e^{a_8}
        \,.
        $}
      \end{equation}
    \item[\bf (ii)]
     Rheonomy for $G_7$: the odd covariant derivatives of $(G_7)$ are fixed by the bosonic frame component of the gravitino field strength \eqref{FormOfGravitinoFieldStrengthImpliedByG4sBianchi}:
     \vspace{-2mm} 
     \begin{equation}
       \label{RheonomyConditionForG7}
        \colorbox{lightblue}{$
       \psi^\alpha 
       \covariantderivative_{\alpha}
       (G_7)_{a_1 \cdots a_7}
       \;=\;
       \tfrac{7!}{5!}
       \big(\,
         \overline{\psi}
         \,
         \Gamma_{[a_1 \cdots a_5}
         \,
         \rho_{a_6 a_7]}
       \big)
       \,.
       $}
     \end{equation}

    \item[\bf (iii)]  
      Hodge duality between $G_7$ and $G_4$:
      \begin{equation}
        \label{G7HodgeDualToG4InComponents}
        \colorbox{lightblue}{$
        (G_7)_{a_1 \cdots a_7}
        \;=\;
        \epsilon_{
          a_1 \cdots a_7
          \,
          b_1 \cdots b_4
        }
        \tfrac{1}{4!}
        (G_4)^{b_1 \cdots b_4}
        \,,
        $}
        \;\;\;\;\;\;
        \colorbox{lightblue}{$
        (G_4)_{a_1 \cdots a_4}
        \;=\;
        -
        \epsilon_{
          a_1 \cdots a_4
          \,
          b_1 \cdots b_7
        }
        \tfrac{1}{7!}
        (G_7)^{b_1 \cdots b_7}
        \,.
        $}
      \end{equation}

    \item[\bf (iv)] The $(\psi^2)$-component of the gravitino field strength \eqref{FormOfGravitinoFieldStrengthImpliedByG4sBianchi} satisfies
    \begin{equation}
      \label{SecondConditionOnSpuriousRhoComponent}
       \colorbox{lightblue}{$
      \Big(
        \overline{\psi}
        \,\Gamma_{a_1 \cdots a_5}\,
        \big(\,
          \overline{\psi}
          \,\kappa\,
          \psi
        \big)
      \Big)
      \;=\;
      0
      \,.
      $}
    \end{equation}

  \end{itemize}

\end{lemma}

\vspace{-.2cm} 
\begin{proof}
The coframe field components of the $G^s_7$-Bianchi identity are:
\vspace{-3mm} 
\begin{equation}
  \label{ComponentsOfBianchiOfGs7}
  \hspace{-2.7cm}
  \def\arraystretch{1.7}
  \begin{array}{l}
  \mathrm{d}
  \,
  G_4^s \;=\; 0
  \\
 \;\;\;\; \Rightarrow
  \left\{
  \def\arraystretch{1.8}
  \begin{array}{l}
  \mathrm{d}
  \Big(
    \tfrac{1}{7!}
    (G_7)_{a_1 \cdots a_7}
    \, 
    e^{a_1} \cdots e^{a_7}
    \,+\,
    \tfrac{1}{5!}
    \big(\,
    \overline{\psi}
    \Gamma_{a_1 \cdots a_5}
    \psi
    \big)
    e^{a_1} \cdots e^{a_5}
  \Big)
  \\
  \;=\;
  \tfrac{1}{2}
  \Big(
    \tfrac{1}{4!}
    (G_4)_{a_1 \cdots a_4}
    e^{a_1}
    \cdots
    e^{a_4}
    +
    \tfrac{1}{2}
    \big(\,
      \overline{\psi}
      \Gamma_{a_1 a_2}
      \psi
    \big)
  \Big)
  \Big(
    \tfrac{1}{4!}
    (G_4)_{a_1 \cdots a_4}
    e^{a_1}
    \cdots
    e^{a_4}
    +
    \tfrac{1}{2}
    \big(\,
      \overline{\psi}
      \Gamma_{a_1 a_2}
      \psi
    \big)
  \Big)
  \\
  \;\Leftrightarrow\;
  \left\{
  \def\arraycolsep{0pt}
  \def\arraystretch{2}
  \begin{array}{l}
  \scalebox{.7}{
    \color{gray}
    $(\psi^0)$
    \;
  }
  \Big(
    \covariantderivative_{a_1}
    \tfrac{1}{7!}
    (G_7)_{a_2 \cdots a_8}  
    \;=\;
    \;\tfrac{1}{2}\;
    \tfrac{1}{4!}
    (G_4)_{a_1 \cdots a_4}
    \,
    \tfrac{1}{4!}
    (G_4)_{a_5 \cdots a_8}
  \Big)
  e^{a_1} \cdots e^{a_8}
  \\
  \scalebox{.7}{
    \color{gray}
    $(\psi^1)$
    \;
  }
    \Big(
      \psi^\alpha
      \covariantderivative_\alpha
      \tfrac{1}{7!}
      (G_7)_{a_1 \cdots a_7}
      \;-\;
      \tfrac{1}{5!}
      \big(\,
      \overline{\psi}
      \,\Gamma_{a_1 \cdots a_5}\,
      \rho_{a_6 a_7}
      \big)
    \Big)
    \,
    e^{a_1} \cdots e^{a_7}
    \;=\;
    0
  \\
  \left.
  \scalebox{.7}{
    \color{gray}
    $(\psi^2)$
    \;
  }
  \def\arraystretch{1.6}
  \begin{array}{l}
    \mbox{\vspaceabove}
    \tfrac{1}{6!}
    (G_7)_{a_1 \cdots a_6 b}
    \big(\,
      \overline{\psi}
      \,\Gamma^b\,
      \psi
    \big)
    e^{a_1}
    \cdots
    e^{a_6}
    \\
    \;\;-\;
    \tfrac{2}{6}
    \tfrac{1}{5!}
    \tfrac{1}{3!}
    (G_4)_{a b_1 b_2 b_3}
    \big(\,
      \overline{\psi}
      \,\Gamma_{a_1 \cdots a_5}\,
      \Gamma^{b_1 b_2 b_3}
      \psi
    \big)
    e^{a}
    \,
    e^{a_1} \cdots e^{a_5}
    \\
    \;\;+\,
    \tfrac{2}{12}
    \,
    \tfrac{1}{5!}
    \,
    \tfrac{1}{4!}
    \,
      (G_4)^{b_1 \cdots b_4}
    \big(\,
      \overline{\psi}
      \,
      \Gamma_{a_1 \cdots a_5}
      \,
      \Gamma_{a b_1 \cdots b_4}\,
      \psi
    \big)
    e^a
    \,
    e^{a_1} \cdots e^{a_5}
    \\
    \;\;+\,
    \Big(
    \tfrac{1}{2}
    \big(\,
      \overline{\psi}
       \Gamma_{a_1 a_2}
      \psi
    \big)
    e^{a_1} \, e^{a_2}
    \Big)
    \tfrac{1}{4!}
    (G_4)_{b_1 \cdots b_4}
    \,
    e^{b_1} \cdots e^{b_4}
    \;\;=\;\; 0
  \end{array}
  \right\}
  \Leftrightarrow
  \;\;\;\;
  \mathrlap{
    \def\arraystretch{1.4}
    \begin{array}{l}
    (G_7)_{a_1 \cdots a_6 b}
    \\
    \;=\;
    \tfrac{1}{4!}
    \epsilon_{a_1 \cdots a_6 b
    b_1 \cdots b_4}
    (G_4)^{b_1 \cdots b_4}
    .
    \end{array}
  }
  \\
  \scalebox{.7}{
    \color{gray}
    $(\psi^3)$
    \;\;
  }
  \Big(
    \overline{\psi}
    \,\Gamma_{a_1 \cdots a_5}\,
    \big(
      \overline{\psi}
        \,\kappa\,
      \psi
    \big)
  \Big)
  e^{a_1} \cdots e^{a_5}
  \;=\;
  0
  \end{array}
  \right.
  \end{array}
  \right.
  \end{array}
\end{equation}

\noindent
Here:
\begin{itemize}[leftmargin=.5cm]
\item
The $(\psi^0)$-component is manifestly the ordinary Bianchi identity \eqref{BosonicComponentOfG7BianchiIdentity}.
\item
The $(\psi^1)$-component is manifestly the rheonomy condition \eqref{RheonomyConditionForG7}.
\item
In the $(\psi^2)$-component we inserted the expression for $\rho$ from \eqref{FormOfGravitinoFieldStrengthImpliedByG4sBianchi}, then contracted $\Gamma$-factors using \eqref{GeneralCliffordProduct}.  Observe, with \eqref{GeneralCliffordProduct}, that of the three spinorial quadratic forms \eqref{TheNontrivialQuadraticFormsOn32} 
the coefficients of $\big(\,\overline{\psi}\Gamma_{a_1 a_2} \psi\big)$ and of $\big(\,\overline{\psi}\Gamma_{a_1 \cdots a_6} \psi\big)$ vanish identically, by a moderately remarkable cancellation of combinatorial prefactors:
\begin{equation}  \label{CancellationOfCOmbinatorialPrefactorsInG4Bianchi}
  \def\arraystretch{2.1}
  \begin{array}{l}
    \overset{
    {\color{gray}   = 0}
    }{
    \overbrace{
    \mbox{$
    \Big(
    -
    \frac{2}{6}
    \frac{1}{5!}
    \frac{1}{3!}
    3!
    \binom{5}{3}
    \binom{ 3}{ 3 }
    +
    \frac{2}{12}
    \frac{1}{5!}
    \frac{1}{4!}
    4!
    \binom{ 5}{4 }
    \binom{ 4}{ 4 }
    \;+\;
    \frac{1}{2}
    \frac{1}{4!}
    \Big)
    $}
    }
    }
    \;
    (G_4)_{a_2 \cdots a_5}
    \big(\,
      \overline{\psi}
      \,\Gamma_{a a_1}\,
      \psi
    \big)
    e^{a}
    \,
    e^{a_1} \cdots e^{a_6},
    \\
    \underbrace{
    \Big(
    \mbox{$
    \;-\;
    \frac{2}{6}
    \frac{1}{5!}
    \frac{1}{3!}
    1
    \binom{ 5}{ 1 }
    \binom{ 3}{ 1 }
    \;+\;
    \frac{2}{12}
    \frac{1}{5!}
    \frac{1}{4!}
    2
   \binom{ 5}{2 }
    \binom{ 4}{ 2 }
    \Big)
    $}
    }_{ {\color{gray} = 0} }
    \;
    (G_4)_{a_1 a_2 b_1 b_2}
    \big(\,
      \overline{\psi}
      \,\Gamma_{a_3 \cdots a_6}{}^{b_1 b_2}\,
      \psi
    \big)
    e^{a_1} \cdots e^{a_6}
    \,.
  \end{array}
\end{equation}
What remains is the coefficient of 
$
\big(\,
  \overline{\psi}\,\Gamma_{a_1\cdots a_5 a b_1 \cdots b_4}\psi
\big) 
  =
  +
\epsilon_{
  a_1\cdots a_5 
  a 
  b_1 \cdots b_4 
  b
}
\big(\,
  \overline{\psi}\,\Gamma^b\,\psi
\big)$ (see \eqref{HodgeDualityOnCliffordAlgebra}):
\begin{equation}
  \label{ObtainingHodgeDualityOfFluxDensities}
  \def\arraystretch{1.9}
  \begin{array}{l}
  \Big(
  \tfrac{1}{6!}
  (G_7)_{a_1 \cdots a_6 b}
  \;-\;
  \tfrac{2}{12}
  \tfrac{1}{5!}
  \tfrac{1}{4!}
  (G_4)^{b_1 \cdots b_4}
  \epsilon_{a_1 \cdots a_6 b b_1 \cdots b_4}
  \Big)
  \big(\,
   \overline{\psi}
   \,\Gamma^b\,
   \psi
  \big)
  e^{a_1} \cdots e^{a_6}
  \;=\;
  0
  \,,
  \end{array}
\end{equation}
which is manifestly the claimed Hodge duality relation \eqref{G7HodgeDualToG4InComponents}
(cf. \cite[p. 878]{CDF91}).

Dually:
$$
  \def\arraystretch{1.5}
  \begin{array}{ll}
    \mathllap{
    (G_4)_{a_1 \cdots a_4}
    \;
    }
    \;=\;
    \delta
      ^{a_1 \cdots a_4}
      _{b_1 \cdots b_4}
    (G_4)_{a_1 \cdots a_4}
    \;=\;
    -
    \frac{1}{4! \cdot 7!}
    \epsilon
      ^{c_1 \cdots c_7 a_1 \cdots a_4}
    \epsilon
      _{c_1 \cdots c_7 b_1 \cdots b_4 }
    (G_4)_{a_1 \cdots a_4}
    &
    \proofstep{
      by
      \eqref{ContractingKroneckerWithSkewSymmetricTensor}
    }
    \\
    \;\;=\;
    -
    \tfrac{1}{7!}
    \epsilon_{a_1 \cdots a_4 c_1 \cdots c_7}
    (G_7)^{c_1 \cdots c_7}
    &
    \proofstep{
      by
      \eqref{ObtainingHodgeDualityOfFluxDensities}.
    }
  \end{array}
$$
\item The $(\psi^3)$-component is 
manifestly the condition \eqref{SecondConditionOnSpuriousRhoComponent}.
\item
The would-be $(\psi^4)$-component holds identically, due to the Fierz identity \eqref{TheQuarticFierzIdentities}:$$
  \tfrac{5}{5!}
  \big(\,
    \overline{\psi}
    \,\Gamma_{a_1 \cdots a_5}\,
    \psi
  \big)
  \big(\,
    \overline{\psi}
    \Gamma^{a_1}
  \big)
  e^{a_2} \cdots e^{a_5}
  \;=\;
  \tfrac{1}{2^3}
  \Big(
  \big(\,
    \overline{\psi}
    \,\Gamma_{a_1 a_2}\,
    \psi
  \big)
  e^{a_1} e^{a_2}
  \Big)
  \Big(
  \big(\,
    \overline{\psi}
    \,\Gamma_{a_1 a_2}\,
    \psi
  \big)
  e^{a_1} e^{a_2}
  \Big)
  \,.
  \qedhere
$$
\end{itemize}
\end{proof}

\smallskip 
\begin{remark}[\bf Maxwell equation for C-field]
  To be explicit, Lem. \ref{SuperBianchiIdentityForG7InComponents} implies that the divergence of $(G_4)$ is:
  \begin{equation}
    \label{MaxwellEquationDerived}
    \def\arraystretch{1.6}
    \begin{array}{ll}
      \mathllap{
      \covariantderivative_b
      (G_4)^{b a_1 a_2 a_3}
      }
      \;=\;
      -
      \tfrac{1}{7!}
      \epsilon^{
        b 
        a_1 a_2 a_3
        c_1 \cdots c_7
      }
      \covariantderivative_b
      (G_7)_{c_1 \cdots c_7}
      &
      \proofstep{
        by
        \eqref{G7HodgeDualToG4InComponents}
      }
      \\
      \;\;=\;
      \tfrac{1}{2}
      \tfrac{1}{7!}
      \epsilon^{
        a_1 a_2 a_3
        \,
        b 
        \,
        c_1 \cdots c_7
      }
      \big(
        \tfrac{1}{4!}
        (G_4)_{b c_1 \cdots c_3}
        \tfrac{1}{4!}
        (G_4)_{c_4 \cdots c_7}
      \big)
      &
      \proofstep{
        by
        \eqref{BosonicComponentOfG7BianchiIdentity}.
      }
    \end{array}
  \end{equation}
\end{remark}
\begin{remark}[\bf Duality-symmetric gravitino super-field strength]
With Lem. \ref{SuperBianchiIdentityForG7InComponents} we may rewrite the $(\psi^1)$-component of the gravitino field strength \eqref{FormOfGravitinoFieldStrengthImpliedByG4sBianchi} in a form where $G_4$ and $G_7$ enter on the same footing:
\begin{equation}
  \label{DualitySymmetricGravitinoSuperFieldStrength}
  \def\arraystretch{1.6}
  \begin{array}{lcll}
     H_a
     &=&
      \tfrac{1}{6}
      \tfrac{1}{3!}
      (G_4)_{a \, b_1 b_2 b_3}
      \Gamma^{b_1 b_2 b_3}
      \,-\,
      \tfrac{1}{12}
      \tfrac{1}{4!}
      (G_4)^{b_1 \cdots b_4}
      \Gamma_{a \, b_1 \cdots b_4}
    &
    \proofstep{
      by
      \eqref{FormOfGravitinoFieldStrengthImpliedByG4sBianchi}
    }
    \\
    &=&
      \tfrac{1}{6}
      \tfrac{1}{3!}
      (G_4)_{a \, b_1 b_2 b_3}
      \Gamma^{b_1 b_2 b_3}
      \,+\,
      \tfrac{1}{12}
      \tfrac{1}{4!}
      \tfrac{1}{6!}
      (G_4)^{b_1 \cdots b_4}
      \epsilon_{
        a \, b_1 \cdots b_4
        \,
        c_1 \cdots c_6
      }
      \Gamma^{c_1 \cdots c_6}
      &
    \proofstep{
      by
      \eqref{HodgeDualityOnCliffordAlgebra}
    }
    \\
    &=&
      \tfrac{1}{6}
      \tfrac{1}{3!}
      (G_4)_{a \, b_1 b_2 b_3}
      \Gamma^{b_1 b_2 b_3}
      \,+\,
      \tfrac{1}{12}
      \tfrac{1}{ 6! }
      (G_7)_{a\, c_1 \cdots c_6}
      \Gamma^{c_1 \cdots c_6}
    & 
    \proofstep{
      by \eqref{G7HodgeDualToG4InComponents}
      .
    }
  \end{array}
\end{equation}
\end{remark}

Before proceeding to analyze the gravitational Bianchi identities (Lem. \ref{SuperFluxAndGravitinoBianchiEquivalentToRaritaSchwinger} below), we record the following implications of the gravitino equation of motion:

\begin{lemma}[\bf Algebraic implications of the Rarita-Schwinger equation]
\label{AlgebraicImplicationsOfGravitinoequation}
$\,$

\noindent {\bf (i)} The Rarita-Schwinger equation  \eqref{GravitinoEquation} for the gravitino has the following algebraic implications:
\begin{equation}
  \label{TheAlgebraicImplicationsOfGravitinoEquation}
  \colorbox{lightblue}{$
  \scalebox{.7}{
    \color{darkblue}
    \bf
    \def\arraystretch{.9}
    \begin{tabular}{c}
      Rarita-Schwinger
      \\
      gravitino equation
    \end{tabular}
  }
 \;\; \Gamma^{a \, b_1 b_2}
  \,\rho_{b_1 b_2}
  \;=\;
  0
  \qquad 
    \Rightarrow
  \qquad
  \left\{\!\!
  \def\arraystretch{1.3}
  \begin{array}{rccl}
    \Gamma^{b_1 b_2}
    \,
    \rho_{b_1 b_2}
    &=&
    0
    \,,
    \\
    \Gamma^{b_2} \, \rho_{b_1 b_2} 
    &=& 
    0
    \mathrlap{
    \;\;\;
    \scalebox{.7}{
      \color{darkblue}
      \bf
      irreducibility
    }    
    }
    \,,
    \\
    \Gamma^{a b_1}
    \,
    \rho_{b_1 b_2}
    &=&
    -
    \rho^a{}_{b_2}
    \,,
    \\
    \Gamma^{a_1 a_2 \, b_1 b_2}
    \,
    \rho_{b_1 b_2}
    &=&
    -2\, \rho^{a_1 a_2}
    \,,
    \\
    \Gamma_{[a_1 \cdots a_5}
  \,
  \rho_{a_6 a_7]}
  &=&
  \tfrac{1}{84}
  \epsilon_{
    a_1 \cdots a_7
    b_1 \cdots b_4
  }
  \Gamma^{b_1 b_2}
  \rho^{b_3 b_4}
  \,
  \,.
  \end{array}
  \right.
  $}
\end{equation}
\noindent {\bf (ii)}  Moreover, together with the $(G^s_4)$-Bianchi identity {\rm (Lem. \ref{SuperBianchiIdentityForG4InComponents})} it implies 
that $\rho_{a_1 a_2}$ is fixed as a linear function of the flux density {\rm (cf. \cite[(12)]{CF80}\cite[(19)]{BrinkHowe80}\cite[(12)]{Howe97})}:
\begin{equation}
  \label{eeComponentOfRhoAsFunctionOfG4}
  \colorbox{lightblue}{$
  \cdots
  \hspace{.8cm}
  \Rightarrow
  \hspace{.8cm}
  \overline{\rho_{a_1 a_2}}_\alpha
  \;=\;
  + 6
  \,
  \Gamma
    _{b_1 b_2}
    {}^\beta{}_\alpha
  \,
  \covariantderivative_\beta
  (G_4)^{a_1 a_2 b_1 b_2}
  \,.
  $}
\end{equation}
\end{lemma}

\smallskip 
\begin{proof}
The first two equations follow immediately from the following evident Clifford contractions:
$$
  \underbrace{
  \Gamma_a
  \,
  \Gamma^{a \, b_1 b_2}
  \,
  \rho_{b_1 b_2}
  }_{\color{gray} = 0}
  \;=\;
  9\, \Gamma^{b_1 b_2}
  \,
  \rho_{b_1 b_2}
  \,,
  \qquad \quad 
  \underbrace{
  \Gamma_{c a}
  \,
  \Gamma^{a \, b_1 b_2}
  \,
  \rho_{b_1 b_2}
  }_{\color{gray} = 0}
  \;=\;
  \underbrace{
  8\,
  \Gamma^{c\, b_1 b_2}
  \,
  \rho_{b_1 b_2}
  }_{\color{gray} = 0}
  \;+\;
  18
  \,
  \Gamma^b \, \rho_{c b}
  \,,
$$
where the summands over the braces vanish by the assumption \eqref{TheAlgebraicImplicationsOfGravitinoEquation} that the gravitino equation holds. Now the third equation follows as
\vspace{-2mm} 
$$
  \def\arraystretch{1.5}
  \begin{array}{lcl}
    \Gamma^{a c}
    \rho_{c b}
    &=&
    \overbrace{
    \tfrac{1}{2}
    \Gamma^a \, \Gamma^c
    \,
    \rho_{c b}
    }^{ {\color{gray} = 0 }}
    -
    \tfrac{1}{2}
    \Gamma^c 
    \,
    \Gamma^a
    \,
    \rho_{c b}
    \\
    &=&
    \hspace{1.9cm}
    \underbrace{
    \tfrac{1}{2}
    \Gamma^a
    \,
    \Gamma^c
    \,
    \rho_{c b}    
    }_{\color{gray} = 0 }
    \,-\, 
    \eta^{a c}
    \rho_{c b}    
    \;=\;
    - \rho^a{}_v
    \,,
  \end{array}
$$
where over the braces we used the previous equation (``$\Gamma$-extraction''). Next follows the fourth equation by
$$
  \def\arraystretch{1.7}
  \begin{array}{ll}
    \mathllap{
      0
    }\!\!
    \;=\;
    \Gamma_{
      c_1 c_2 a
    }
    \,
    \Gamma^{a b_1 b_2}
    \rho_{b_1 b_2}
    &
    \proofstep{
      by assumption
    }
    \\
    \;=\;
    7
    \,
    \Gamma_{c_1 c_2 b_1 b_2}
    \,
    \rho^{b_1 b_2}
    +
    32
    \,
    \Gamma_{[c_1}{}^b
    \rho_{c_2] b}
    -
    18
    \,
    \rho_{c_1 c_2}
    &
    \proofstep{
      by contraction
    }
    \\
    \;=\;
    7
    \,
    \Gamma_{c_1 c_2 b_1 b_2}
    \,
    \rho^{b_1 b_2}
    +
    14\, \rho_{c_1 c_2}
    &
    \proofstep{
      by previous statement.
    }
  \end{array}
$$
From this follows the claim in \eqref{eeComponentOfRhoAsFunctionOfG4} as:
\begin{equation}
  \def\arraystretch{1.5}
  \begin{array}{ll}
    \mathllap{
      \overline{
        \rho^{a_1 a_2}
      }_\alpha    
      \;\;
    }
    =\;
    -\tfrac{1}{2}
    \,
    \overline{
      \Gamma^{a_1 a_2 b_1 b_2}
      \,
      \rho_{b^1 b^2}
    }_{\;\alpha}
    &
    \proofstep{
      by \eqref{TheAlgebraicImplicationsOfGravitinoEquation}
    }
    \\
    \;=\;
    -\tfrac{1}{2}
    \,
    \overline{
      \Gamma_{b_1 b_2}
      \Gamma^{[a_1 a_2}
      \rho^{b_1 b_2]}
    }
    _{\;\alpha}
    \\
    \;=\;
    +\tfrac{1}{2}
    \Gamma_{b_1 b_2}
      {}^\beta
      {}_\alpha
    \,
    \overline{
      \Gamma^{[a_1 a_2}
      \rho^{b_1 b_2]}
    }
    _\beta
    &
    \proofstep{
      by
      \eqref{SkewSelfAdjointnessOfCliffordGenerators}
    }
    \\
    \;=\;
    +6
    \,
    \Gamma
      _{b_1 b_2}{}
    ^\beta{}_\alpha
    \covariantderivative_{\beta}
    (G_4)^{a_1 a_2 b_1 b_2}
    &
    \proofstep{
      by
      \eqref{ComponentsOfBianchiForGs4}.
    }
  \end{array}
\end{equation}
The last claim in \eqref{TheAlgebraicImplicationsOfGravitinoEquation} is checked mechanically in \cite{AncillaryFiles}.
\end{proof}

\begin{lemma}[\bf Implications of $\big(G_4^s, G_7^s\big)$-Bianchi identity on Gravitino field strength]
\label{ImplicationsOfFluxBianchiIdentityOnGravitinoFieldStrength}
If a super-spacetime is equipped with super-flux $(G_4^s, G^s_7)$ {\rm (as in Lem. \ref{SuperBianchiIdentityForG4InComponents}, \ref{SuperBianchiIdentityForG7InComponents})}, and the component $\rho_{ab}$ \eqref{RhoFrameFieldExpansion} of its gravitino field strength is irreducible \eqref{TheAlgebraicImplicationsOfGravitinoEquation}, then the exterior derivative of the latter
is given by
\begin{equation}
  \label{ExteriorDerivativeOfGravitinoFieldStrength}
  \def\arraystretch{1.8}
  \begin{array}{ll}
 \colorbox{lightblue}{$    \nabla_{[a_1}
    \rho_{a_2 a_3]}
    \;\;=\;\;
    \tfrac{1}{3}
    \underbrace{
      \Gamma_{b[a_1}
      \nabla^b\rho_{a_2 a_3]}
    }_{
      \mathclap{
      \scalebox{.7}{ \eqref{DiracLikeDerivativeOfGravitinoFieldStrength}
      }}
    }
    \;-\;
    \tfrac{1}{15}
    \Gamma_{[a_1 a_2}
    \underbrace{
    \nabla_b
    \,
    \rho^b{}_{a_3]}
    }_{
      \mathclap{
        \scalebox{.7}{
          \eqref{DivergenceOfGravitinoFieldStrength}
        }
      }
    }
    \;-\;
    \tfrac{1}{3}
    \underbrace{
    \Gamma^{b_1 b_2}
    \Gamma_{[b_1 b_2}
    \nabla_{a_1}
    \rho_{a_2 a_3]} 
    }_{
      \mathclap{
        \scalebox{.7}{
          \eqref{5IndexDerivativeOfRho}        
        }
      }
    }.
    $}
    
  \end{array}
\end{equation}
Here each of the three terms on the right, and hence the exterior derivative itself, is an algebraic expression in $\rho$ and the flux density:

\def\arraystretch{1.7}
\begin{align}
  \label{5IndexDerivativeOfRho}
    \Gamma_{[a_1 a_2}
    \covariantderivative_{a_3}
    \rho_{a_4 a_5]}
    \;
    &=\;
      \overline{H}_{[a_1}
      \Gamma_{a_2 a_3}
      \rho_{a_4 a_5]}
    -
    \tfrac{1}{3}
    \,
    (G_4)_{
      b
      \,
      [a_1 a_2 a_3
    }
    \Gamma^b
    \rho_{a_4 a_5]}
 \\[15pt]\label{DivergenceOfGravitinoFieldStrength}
  \covariantderivative_b
  \,
  \rho^b{}_a
  \;
  &=\;
  \tfrac{5}{84}
  \,
  \Gamma^{b_1 \cdots b_4}
  \underbrace{
    \Gamma_{[b_1 b_2}
    \covariantderivative_{b_3}
    \rho_{b_4 \, a]}
  }_{
    \mathclap{
      \scalebox{.7}{
       \eqref{5IndexDerivativeOfRho} 
      }
    }
  }
 \\\label{DiracLikeDerivativeOfGravitinoFieldStrength}
  \Gamma^{a [c_1}
  \nabla_{a}
  \rho^{c_2 c_3]}
  \;
  &=\;
  \raisebox{-3pt}{$
  \begin{array}{l}
  -
  \Gamma^{[c_1 c_2}
  \overbrace{
    \nabla_b 
    \,
    \rho^{|b| c_3]}
  }^{
    \mathclap{
      \scalebox{.7}{
        \eqref{DivergenceOfGravitinoFieldStrength}
      }
    }
  }
  \,+\,
  2
  \,
  \overline{H}_{b}
  \,
  \Gamma^{[b \, c_1}
  \rho^{c_2 c_3]}
  \\
    +
    \;
    2
    \tfrac
      {5! \cdot 84}
      {7! \cdot 4!}
    \,
    \epsilon^{
      c_1 c_2 c_3 
      \,
      a_1 \cdots a_8
    }
    \big(
      \tfrac
        {12}
        {4! \cdot 4!}
      (G_4)_{a_1 \cdots a_4}
      \Gamma_{a_5 a_6}\rho_{a_7 a_8}
      -
      \tfrac{1}{6!}
      (G_7)_{b a_1 \cdots a_6}
      \Gamma^b
      \rho_{a_7 a_8}
    \big)
    \,.
    \end{array}
  $}
\end{align}
\end{lemma}
\begin{proof}
The strategy is to combine the $(\psi^1)$-components of the fact that $\mathrm{d}^2 = 0$ on the flux densities $G_4$ and $G_7$, while using the Bianchi identities for $G_4^s$ and $G_7^s$ and observing the following two Clifford-contraction identities (checked in \cite{AncillaryFiles}, using the assumption that $\rho_{ab}$ is irreducible and the fact that the Gamma-matrices are covariantly constant \eqref{GammaMatricesAreCovariantlyConstant}):
\vspace{0cm}
\def\arraystretch{1.9}
\begin{align}
\label{DivergenceOfRhoAsContractionOf2ndDiracExteriorDerivative}
  \Gamma^{b_1 \cdots b_4}
  \,
  \Gamma_{[b_1 b_2}
  \covariantderivative_{b_3}
  \rho_{b_4 a]}
  &\;
  =
  \;\;
  \tfrac{84}{5}
  \,
  \covariantderivative_{b}
  \,
  \rho^{b}{}_{a}
  \,,
  \\
  \label{A2IndexContraction}
  \Gamma^{b_1 b_2}
  \,
  \Gamma_{[b_1 b_2}
  \covariantderivative_{a_1}
  \rho_{a_2 a_3]}
  &
  \;=\;\;
  \,
  \Gamma_{b[a_1}
  \covariantderivative^{b}
  \rho_{a_2 a_3]}
  -
  \tfrac{1}{5}
  \Gamma_{[a_1 a_2}
  \covariantderivative^b
  \rho_{|b|\, a_3]}
  -
  3
  \covariantderivative_{[a_1}
  \rho_{a_2 a_3]}
  \,.
\end{align}

\smallskip
To this end, we will write $\mathcal{O}(\psi^{\neq 1})$ for all summands of an expression whose order in $\psi$ 
is different from 1, hence to be disregarded for the present purpose.
Moreover, notice in the following the use of the Dirac adjoint \eqref{BarAdjointnessOfGammaMatrices}
$
  \overline{H_a \psi}
  \;=\;
  \overline{\psi} \; \overline{H}_a
$
of $H_a$ \eqref{FormOfGravitinoFieldStrengthImpliedByG4sBianchi}, given by:
\begin{equation}
 \label{DiracAdjointOfH}
     \overline{H}_a
     \;\;
     =
     \;\;
      \tfrac{1}{6}
      \tfrac{1}{3!}
      (G_4)_{a \, b_1 b_2 b_3}
      \Gamma^{b_1 b_2 b_3}
      \,{\color{purple}+}\,
      \tfrac{1}{12}
      \tfrac{1}{4!}
      (G_4)^{b_1 \cdots b_4}
      \Gamma_{a \, b_1 \cdots b_4}
      \,.
\end{equation}

Now first consider the condition obtained from $G_4$:
$$
  \def\arraystretch{1.7}
  \begin{array}{ll}
    \mathllap{0\;}\!\!
    \;=\;
    \mathrm{d}
    \,
    \mathrm{d}
    \,
    \tfrac{1}{4!}
    \,
    (G_4)_{a_1 \cdots a_4}
    \,
    e^{a_1}\cdots e^{a_4}
    \\
    \;=\;
    \mathrm{d}
    \Big(
      \tfrac{1}{4!}
      \big(
      \covariantderivative_{[a_1}
      (G_4)_{a_2 \cdots a_5]}
      \big)
      e^{a_1} \cdots e^{a_5}
      +
      \tfrac{1}{4!}
      \psi^\beta 
      \big(
      \covariantderivative_\beta
      (G_4)_{a_1 \cdots a_4}
      \big)
      e^{a_1} \cdots e^{a_4}
      \\
      \hspace{1.1cm}
      +
      \tfrac{1}{3!}
      (G_4)_{b a_1 a_2 a_3}
      \big(
        \overline{\psi}
        \,
        \Gamma^b
        \,
        \psi
      \big)
      e^{a_1} e^{a_2} e^{a_3}
    \Big)
    \\
    \;=\;
    \tfrac{1}{4!}
    (H_{[a_1}\psi)^\beta
    \big(
      \nabla_{|\beta|} 
      (G_4)_{a_2 \cdots a_5]}
    \big)
    e^{a_1}
    \cdots
    e^{a_5}  -
    \tfrac{1}{4!}
    \psi^\beta
    \big(
    \covariantderivative_{[a_1}
    \covariantderivative_{|\beta|}
    (G_4)_{a_2 \cdots a_5]}
    \big)
    e^{a_1} \cdots e^{a_5}
    & \proofstep{
      by
      \hspace{-6pt}
      \def\arraystretch{.9}
      \begin{tabular}{l}

 \eqref{RhoFrameFieldExpansion}
 \\ \eqref{DerivingClosureOfBosonic4FluxDensity}
    \end{tabular}
    }
    \\
    \hspace{.7cm}
   
    -
    \tfrac{1}{3!}
    (G_4)_{b a_1 a_2 a_3}
    \big(
      \overline{\psi}
      \,
      \Gamma^b
      \,
      \rho_{a_4 a_5}
    \big)
    e^{a_2} \cdots e^{a_5}
    \\
    \hspace{.6cm}
    +
    \mathcal{O}(\psi^{\neq 1})
    \\
    \;=\;
    \overline{\psi}
    \Big(
    \tfrac{1}{2}
      \overline{H}_{[a_1}
      \Gamma_{a_2 a_3}
      \rho_{a_4 a_5]}
    -\tfrac{1}{2}
    \covariantderivative_{[a_1}
    \Gamma_{a_2 a_3}
    \rho_{a_4 a_5]}
    -
    \tfrac{1}{3!}
    (G_4)_{
      b
      \,
      [a_1 a_2 a_3
    }
    \Gamma^b
    \rho_{a_4 a_5]}
    \Big)
    e^{a_1} \cdots e^{a_5}
    &
    \proofstep{
      by
      \hspace{-6pt}
      \def\arraystretch{.9}
      \begin{tabular}{l}
\eqref{ShiftingOfSpinorialIndicesUnderCovariantDerivative}
   \\ 
   \eqref{OddCovariantDerivativeOfFluxDensity}
    \end{tabular}
    }
    \\
    \hspace{.9cm}
    +
    \mathcal{O}\big(
      \psi^{\neq 1}
    \big)
    \,,
  \end{array}
$$
which proves \eqref{5IndexDerivativeOfRho}.
Inserting this into \eqref{DivergenceOfRhoAsContractionOf2ndDiracExteriorDerivative} proves \eqref{DivergenceOfGravitinoFieldStrength}.

\smallskip

Notice for the following that the divergence
of $\rho_{ba}$, on the right hand side of \eqref{DivergenceOfRhoAsContractionOf2ndDiracExteriorDerivative}, appears in half of the summands of the divergence of $\Gamma_{[b\, c_1} \rho_{c_2 c_3]}$ as follows, just by the combinatorics of skew-symmetrization:
\begin{equation}
  \label{DistributingDifferentialOverGammaabrhocd}
  \nabla^{b}
  \,
  \Gamma_{
    [b 
    \,
    c_1
  }
  \rho_{c_2 c_3]}
  \;\;=\;\;
  \tfrac{1}{2}
  \underbrace{
  \Gamma_{
    b 
    \,
    [c_1
  }
  \nabla^{b}
  \rho_{c_2 c_3]}
  }_{
    \mathclap{
      \scalebox{.7}{
        \eqref{DiracLikeDerivativeOfGravitinoFieldStrength}
      }
    }
  }
  +
  \tfrac{1}{2}
  \Gamma_{
    [c_1
    c_2
  }
  \underbrace{
    \nabla^{b}
    \rho_{|b| c_3]}
  }_{
    \mathclap{
      \scalebox{.7}{
        \eqref{DivergenceOfGravitinoFieldStrength}
      }
    }
  }
\end{equation}
(where, just for emphasis, we also moved the covariant derivative, using again that the $\Gamma$-matrices are covariantly constant \eqref{GammaMatricesAreCovariantlyConstant}).
Then consider the corresponding condition obtained from $G_7$:
$$
  \def\arraystretch{2.1}
  \begin{array}{ll}
    \mathllap{0\;}\!\!
    \;=\;
    \mathrm{d}
    \,
    \mathrm{d}
    \,
    \tfrac{1}{7!}
    (G_7)_{a_1 \cdots a_7}
    \,
    e^{a_1} \cdots e^{a_7}
    \\
    \;=\;
    \mathrm{d}
    \Big(
      \tfrac{1}{7!}
      \big(
        \nabla_{a_1}
        (G_7)_{a_2 \cdots a_8}
      \big)
      e^{a_1} \cdots e^{a_8}
      +
      \tfrac{1}{7!}
      \psi^\beta
      \nabla_\beta
      (G_7)_{a_1 \cdots a_7}
      e^{a_1}
      \cdots 
      e^{a_7}
      \\
      \hspace{1.1cm}
      +
      \tfrac{1}{6!}
      (G_7)_{b a_1 \cdots a_6}
      \big(\,
        \overline{\psi}
        \,
        \Gamma^b
        \,
        \psi
      \big)
      \,
      e^{a_1} \cdots e^{a_6}
    \Big)
    \\
    \;=\;
    \bigg(
    \tfrac{1}{7!}
    \psi^\beta
    \nabla_\beta
    \nabla_{[a_1}
    (G_7)_{a_2 \cdots a_8]}
    \\
    \hspace{1cm}
    +
    \tfrac{1}{7!}
    \,
    \psi^\beta
    \big(\,
      \overline{H}_{[a_1}
      -
      \nabla_{[a_1}
    \big)
      \nabla_{|\beta|}
      (G_7)_{a_2 \cdots a_8]}
    -
    \tfrac{1}{6!}
    (G_7)_{b a_1 \cdots a_6}
    \big(\,
      \overline{\psi}
      \,\Gamma^b\,
      \rho_{a_7 a_8}
    \big)
\!    \bigg)
    e^{a_1}
    \cdots
    e^{a_8}
    &
    \proofstep{
      by 
      \eqref{RhoFrameFieldExpansion}
    }    
    \\
    \hspace{.6cm}
    +
    \mathcal{O}(\psi^{\neq 1})
    \\
    \;=\;
    \overline{\psi}
    \bigg(
    \tfrac
      {12}
      {4! \cdot 4!}
    \,
    (G_4)_{
      [a_1 \cdots a_4}
    \,
    \Gamma_{a_5 a_6}
    \rho_{a_7 a_8]}
    &
    \proofstep{
      by
      \hspace{-6pt}
      \def\arraystretch{.9}
      \begin{tabular}{l}
      \eqref{BosonicComponentOfG7BianchiIdentity}
      \\
      \eqref{OddCovariantDerivativeOfFluxDensity}
      \end{tabular}
    }
    \\
    \hspace{1cm}
    +
    \tfrac
      {1}
      {5! \cdot 84}
    \big(\,
      \overline{H}_{[a_1}
    -
    \nabla_{[a_1}
    \big)
    \epsilon_{
      a_2 \cdots a_8]
      \,
      b_1 \cdots b_4
    }
    \Gamma^{[b_1 b_2}
    \rho^{b_3 b_4]}
    -
    \tfrac{1}{6!}
    (G_7)_{b a_1 \cdots a_6}
    \big(\,
      \overline{\psi}
      \,\Gamma^b\,
      \rho_{a_7 a_8}
    \big)
    \! \bigg)
    e^{a_1}
    \cdots
    e^{a_8}
    \\
    \hspace{.6cm}
    +
    \mathcal{O}(\psi^{\neq 1})    \;.
  \end{array}
$$
Contracting this with 
$\epsilon^{a_1 \cdots a_8 c_1 c_2 c_3}$, using \eqref{ContractingKroneckerWithSkewSymmetricTensor} and \eqref{LCTensorIsCovariantyConstant}, yields (where under the brace we recall \eqref{G7HodgeDualToG4InComponents}):
$$
  \def\arraystretch{1.8}
  \begin{array}{l}
    \tfrac{7! \cdot 4!}{5! \cdot 84}
    \big(\,
      \overline{H}_{a_1}
      -
      \nabla_{a_1}
    \big)
    \Gamma^{[a_1 c_1}
    \rho^{c_2 c_3]}
    \\
    \hspace{.6cm}
    +
    \;
    \epsilon^{
      c_1 c_2 c_3 
      \,
      a_1 \cdots a_8
    }
    \Big(
      \tfrac
        {12}
        {4! \cdot 4!}
      (G_4)_{a_1 \cdots a_4}
      \Gamma_{a_5 a_6}\rho_{a_7 a_8}
      -
      \tfrac{1}{6!}
      \underbrace{
        (G_7)_{b a_1 \cdots a_6}
      }_{\color{gray} 
        \mathclap{
          \tfrac{1}{4!}
          \epsilon_{
            b 
            \,
            a_1 \cdots a_6
            d_1 \cdots d_4
          }
          (G_4)^{d_1 \cdots d_4}
        }
      }
        \Gamma^b
        \rho_{a_7 a_8}
    \Big)
    \;=\;
    0
    \,.
  \end{array}
$$
Inserting \eqref{DistributingDifferentialOverGammaabrhocd} for the differential term in this last expression yields \eqref{DiracLikeDerivativeOfGravitinoFieldStrength}.
Finally, inserting these three equations into \eqref{A2IndexContraction} manifestly gives the final claim
\eqref{ExteriorDerivativeOfGravitinoFieldStrength}.
\end{proof}

With these preliminaries in hand, we dive into the analysis of the torsion and gravitino Bianchi identities:

\newpage 
\begin{lemma}[\bf Gravitational Bianchi identities in components]
\label{SuperFluxAndGravitinoBianchiEquivalentToRaritaSchwinger}
 Assuming the Bianchi identities for $G_4^s$ (Lem. \ref{SuperBianchiIdentityForG4InComponents}) and $G_7^s$ (Lem. \ref{SuperBianchiIdentityForG7InComponents}), the  gravitational Bianchi identities \eqref{SuperGravitationalBianchiIdentities} (Rem. \ref{RoleOfTheSuperGravitationalBianchiIdentities}) are equivalent to the combination of the following conditions:
 \begin{itemize}
   \item[\bf (i)]
   the bosonic coframe component of the gravitino field strength \eqref{RhoFrameFieldExpansion}
   satisfies the Rarita-Schwinger equation:
   \begin{equation}
     \label{GravitinoEquation}
  \colorbox{lightblue}{$      \Gamma^{a \, b_1 b_2}
      \rho_{b_1 b_2}
      \;=\;
      0
      \,,
      $}
   \end{equation}
   \item[{\bf (ii)}]
   the odd coframe components of the super-curvature \eqref{CurvatureTensorInFrameField} are fixed by:
  \begin{equation}
  \label{SolutionForTheta}
  \colorbox{lightblue}{$
  \CurvatureAtPsiOne_{a b c}
  \;=\;
    \,+\, \Gamma_a \, \rho_{b c}
    \,-\, \Gamma_c \, \rho_{a b}
    \,+\, \Gamma_b \, \rho_{c a}
    \,,
  $}
  \end{equation}

  \vspace{-3pt}
  
  \begin{equation}
    \label{SolutionForK}
    \colorbox{lightblue}{$
    \def\arraystretch{1.6}
    \begin{array}{ll}
    K^{a_1 a_2}
    & = \;
    +
    \tfrac{1}{6}
    \Big(
      (G_4)^{a_1 a_2 \, b_1 b_2}
      \Gamma_{b_1 b_2}
      \;+\;
      \tfrac{1}{4!}
      (G_4)_{b_1 \cdots b_4}
      \Gamma
        ^{a_1 a_2 \, b_1 \cdots b_4}
    \Big)
    \\
    & = \;
    +
    \tfrac{1}{6}
    \Big(
      (G_4)^{a_1 a_2 \, b_1 b_2}
      \Gamma_{b_1 b_2}
      \;+\;
      \tfrac{1}{5!}
      (G_7)
        ^{a_1 a_2 \, b_1 \cdots b_5}
      \Gamma
        _{b_1 \cdots b_5}
    \Big)
    \,,
    \end{array}
    $}
  \end{equation}
   \item[{\bf (iii)}] 
   the $(\psi^2)$-component of the gravitino field strength \eqref{RhoFrameFieldExpansion} vanishes:   
   \begin{equation}
     \label{VanishingOfpsi2ComponentOfRho}
  \colorbox{lightblue}{$     
    \big(\, 
      \overline{\psi}
      \,\kappa\, 
      \psi
    \big)
     \;=\;
     0
     \,.
     $}
   \end{equation}
\end{itemize}
\end{lemma}
\begin{proof}
  First, the \colorbox{lightgray}{\bf torsion Bianchi identity}
  \eqref{SuperGravitationalBianchiIdentities}
  has the following coframe field components, in terms of those of the curvature tensor  \eqref{CurvatureTensorInFrameField}:
  \begin{equation}
  \label{ComponentsOfTorsionBianchiIdentity}
  \def\arraystretch{1.6}
  \begin{array}{l}
    R^{a b}\, e_b
    \;=\;
    - \, 2 
    \big(\,
      \overline{\psi}
      \,\Gamma^{a}\,
      \rho
    \big)
    \\
    \;\Leftrightarrow\;
    \left\{
    \def\arraycolsep{0pt}
    \def\arraystretch{1.6}
    \begin{array}{l}
      \scalebox{.7}{
        \color{gray}
        $(\psi^0)$
        \;
      }
      R^{a}{}_{[b_1 b_2 b_3]}
      \,
      e^{b_1}\, e^{b_2}\, e^{b_3}
      \;=\;
      0
      \,,
      \\
      \scalebox{.7}{
        \color{gray}
        $(\psi^1)$
        \;
      }
      \big(
        \,
        \overline{\psi}
        \,
        \CurvatureAtPsiOne^{a}{}_{b_1 b_2}
        \,
      \big)
      \, e^{b_1}\, e^{b_2}
      \;=\;
      +
      \big(\,
        \overline{\psi}
        \,\Gamma^a\,
        \rho_{b_1 b_2}
      \big)
      e^{b_1}\, e^{b_2}
      \,,
      \\
      \scalebox{.7}{
        \color{gray}
        $(\psi^2)$
        \;
      }
      \big(\,
        \overline{\psi}
        \,K^{a b}\,
        \psi
      \big)
      e_b
      \;=\;
      - 2 
      \big(\,
        \overline{\psi}
        \,\Gamma^a\,        
        H_b
        \psi
      \big)
        \,
        e^b
      \,,
      \\
      \scalebox{.7}{
        \color{gray}
        $(\psi^3)$
        \;\;
      }
      2
      \Big(
        \overline{\psi}
        \,\Gamma^a\,
        \big(\,
          \overline{\psi}
          \,\kappa\,
          \psi
        \big)
      \Big)
      \;=\;
      0
      \,.
    \end{array}
    \right.
  \end{array}
\end{equation}
Here:

\smallskip 
\noindent \fbox{{\bf Torsion Bianchi at} $\psi^0$} The $(\psi^0)$-component in \eqref{ComponentsOfTorsionBianchiIdentity} holds identically (via Rem. \ref{RoleOfTheSuperGravitationalBianchiIdentities}) as it does not involve the prescribed field $G_4$.

\smallskip 
\noindent \fbox{{\bf Torsion Bianchi at} $\psi^1$}
The $(\psi^1)$-component  says that
\begin{equation}
  \label{EquationForThetaComponentOfR}
  \tfrac{1}{2}
  \big(
    \CurvatureAtPsiOne_{a b_1 b_2}
    -
    \CurvatureAtPsiOne_{a b_2 b_1}
  \big)
  \;=\;
  + \Gamma_a \rho_{b_1 b_2}
  \,.
\end{equation}
Hence adding up three copies of this equation with cyclically permuted indices, and using the skew symmetries $\CurvatureAtPsiOne^{a b}{}_c = \CurvatureAtPsiOne^{[a b]}{}_c$ and $\rho_{a b} = \rho_{[a b]}$ 
$$
  \def\arraystretch{1.5}
  \begin{array}{r}
  \tfrac{1}{2}
  \big(
    \CurvatureAtPsiOne_{a b_1 b_2}
    -
    \CurvatureAtPsiOne_{a b_2 b_1}
  \big)
  \\
  \;-\;
  \tfrac{1}{2}
  \big(
    \CurvatureAtPsiOne_{b_2 a b_1 }
    -
    \CurvatureAtPsiOne_{b_2 b_1 a}
  \big)
  \\
  \;+\;
  \tfrac{1}{2}
  \big(
    \CurvatureAtPsiOne_{b_1 b_2 a }
    -
    \CurvatureAtPsiOne_{b_1 a b_2}
  \big)
  \\
  \hline
  \CurvatureAtPsiOne_{a b_1 b_2}
  \mathrlap{\,,}
  \phantom{\big)}
  \end{array}
  \;=\;
  \def\arraystretch{1.5}
  \begin{array}{r}
    + \Gamma_a \rho_{b_1 b_2}
    \\
    - \Gamma_{b_2} \rho_{a b_1}
    \\
    + \Gamma_{b_1} \rho_{b_2 a}
    \\
    \phantom{A}
  \end{array}
$$
this implies (cf. \cite[(III.3.218)]{CDF91}) that $\CurvatureAtPsiOne$ is as claimed \eqref{SolutionForTheta}, and conversely this solution for $\CurvatureAtPsiOne$ already solves the original equation \eqref{EquationForThetaComponentOfR} for all $\rho$, since
$$
  \def\arraystretch{1.5}
  \begin{array}{l}
    + \Gamma_{a} \rho_{b_1 b_2}
    - \Gamma_{b_2} \rho_{a b_1}
    + \Gamma_{b_1} \rho_{b_2 a}
    \\
    - \Gamma_{a} \rho_{b_2 b_1}
    + \Gamma_{b_1} \rho_{a b_2}
    - \Gamma_{b_2} \rho_{b_1 a}
    \\
    =\;
    +2\,
    \Gamma_a \,
    \rho_{b_1 b_2}
    \,.
  \end{array}
$$

\noindent \fbox{{\bf Torsion Bianchi at} $\psi^2$}
The $(\psi^2)$-component in \eqref{ComponentsOfTorsionBianchiIdentity} has a solution for $K^{a b}$ iff 
$\big(\, \overline{\psi} \,\Gamma^a H^b\, \psi\big)$ is skew-symmetric in $a,b$, in which case the solution is unique.
And indeed, by
\eqref{FormOfGravitinoFieldStrengthImpliedByG4sBianchi},
\eqref{GeneralCliffordProduct} and 
\eqref{VanishingQuadraticForms} we have
$$
  \def\arraystretch{1.7}
  \begin{array}{rcl}
 \big(\,
   \overline{\psi}
     \,
     \Gamma^a 
     H^b
     \,
   \psi
 \big)
 &=&
 \tfrac{1}{6}
 \tfrac{1}{3!}
 (G_4)^b{}_{b_1 b_2 b_3}
 \big(\,
    \overline{\psi}
      \,
      \Gamma^a\Gamma^{b_1 b_2 b_3}
      \,
    \psi
  \big)
  \;-\;
  \tfrac{1}{12}
  \tfrac{1}{4!}
  (G_4)_{b_1 \cdots b_4}
  \big(\,
  \overline{\psi}
    \,
    \Gamma^a
    \Gamma^{
      b
      \,
      b_1 \cdots b_4
    }
    \,
  \psi
  \big)
  \\
  &=&
  -
 \tfrac{1}{6}
 \tfrac{1}{2!}
 (G_4)^{a b \, b_2 b_3}
 \big(\,
    \overline{\psi}
      \,
      \Gamma_{b_2 b_3}
      \,
    \psi
  \big)
  \;-\;
  \tfrac{1}{12}
  \tfrac{1}{4!}
  (G_4)_{b_1 \cdots b_4}
  \big(\,
  \overline{\psi}
    \,
    \Gamma^{
      a b
      \,
      b_1 \cdots b_4}
    \,
  \psi
  \big)
  \end{array}
$$
and hence (cf. \cite[(III.8.58)]{CDF91}) $K^{a b}$ is as claimed in \eqref{SolutionForK}:
\begin{equation}
  \def\arraystretch{1.8}
  \begin{array}{lll}
  K^{a b}
  & = \;
  +
 \tfrac{1}{6}
 \Big(
   (G_4)^{a b \, b_1 b_2}
   \Gamma_{b_1 b_2}
   \;+\;
   \tfrac{1}{4!}
   (G_4)_{b_1 \cdots b_4}
   \Gamma
     ^{a b \, b_1 \cdots b_4}
  \Big)
  \\
  & =
  +
 \tfrac{1}{6}
 \Big(
   (G_4)^{a b \, b_1 b_2}
   \Gamma_{b_1 b_2}
   \;+\;
   \tfrac{1}{4! \cdot 5!}
   (G_4)_{b_1 \cdots b_4}
   \epsilon
     ^{
       a b \, b_1 \cdots b_4
       \,
       c_1 \cdots c_5
     }
    \Gamma_{c_1 \cdots c_5}
  \Big)
    &
    \proofstep{
      by
      \eqref{ExamplesOfHodgeDualCliffordElements}
    }
  \\
  & =
  +
 \tfrac{1}{6}
 \Big(
   (G_4)^{a b \, b_1 b_2}
   \Gamma_{b_1 b_2}
   \;+\;
   \tfrac{1}{5!}
   (G_7)^{
     a b \, c_1 \cdots c_5
    }
    \Gamma_{c_1 \cdots c_5}
  \Big)
  &
  \proofstep{
    by
    \eqref{G7HodgeDualToG4InComponents}.
  }
  \end{array}
\end{equation}

\noindent \fbox{{\bf Torsion Bianchi at} $\psi^3$}
The $(\psi^3)$-component of the torsion Bianchi \eqref{SuperGravitationalBianchiIdentities}, combined with that of the $G^s_7$-Bianchi \eqref{ComponentsOfBianchiOfGs7} and that of the $G^s_4$-Bianchi \eqref{ComponentsOfBianchiForGs4}, says that {\it all} the following expressions vanish:
\begin{equation}
  \label{AllPairingsWithThePsi2ComponentOfRhoVanish}
  \Big(
    \overline{\psi}
    \,\Gamma_a\,
    \big(\,
      \overline{\psi}
      \,\kappa\,
      \psi
    \big)
  \Big)
  \;=\;
  0
  \,,
  \;\;\;\;\;\;
  \Big(
    \overline{\psi}
    \,\Gamma_{a_1 a_2}\,
    \big(\,
      \overline{\psi}
      \,\kappa\,
      \psi
    \big)
  \Big)
  \;=\;
  0
  \,,
  \;\;\;\;\;\;
  \Big(
    \overline{\psi}
    \,\Gamma_{a_1 \cdots a_5}\,
    \big(\,
      \overline{\psi}
      \,\kappa\,
      \psi
    \big)
  \Big)
  \;=\;
  0
  \,.
\end{equation}
By 
Lem. \ref{XiVanishesIfAllSymmetricPairingsOntoPsiVanish} this finally implies the vanishing of $\big(\overline{\psi}\,\kappa\, \psi\big)$, as claimed \eqref{VanishingOfpsi2ComponentOfRho}.

\medskip

\noindent
Next, the \colorbox{lightgray}{\bf gravitino Bianchi identity} \eqref{SuperGravitationalBianchiIdentities}
  has the following coframe field components:
  \begin{equation}
    \label{ComponentsOfGravitinoBianchiIdentity}
    \begin{array}{l}
      \mathrm{d}
      \,
      \rho
      \;+\;
      \tfrac{1}{4}
      \omega^{a b}
      \,
      \Gamma_{a b}
      \rho
      \;\;
      =
      \;\;
       +
      \tfrac{1}{4}
      R^{a b}
      \Gamma_{a b}
      \psi
      \\[5pt]
      \;\Leftrightarrow\;
      \left\{
      \def\arraycolsep{0pt}
      \def\arraystretch{1.6}
      \begin{array}{l}
        \scalebox{.7}{
          \color{gray}
          $(\psi^0)$
          \;
        }
        \Big(
          \covariantderivative_{[a_1}
          \rho_{a_2 a_3]}
          \;+\;
          H_{[a_1}
          \rho_{a_2 a_3]}
        \Big)
        e^{a_1} \, e^{a_2} \, e^{a_3}
        \;=\;
        0
        \,,
        \\
        \scalebox{.7}{
          \color{gray}
          $(\psi^1)$
          \;
        }
        \Big(
        \psi^\alpha
        \tfrac{1}{2}
        \big(
          \covariantderivative_\alpha
          \,
          \rho_{a_1 a_2}
        \big)
        \;+\;
        \big(
        \covariantderivative_{[a_1}
          H_{a_2]}
        \big)
        \psi
        \;-\;
        H_{a_1}
        H_{a_2}
        \psi
        \;-\;
        \tfrac{1}{4}
        \tfrac{1}{2}
        R^{a b}{}_{a_1 a_2}
        \Gamma_{ab}\psi
        \Big)
        e^{a_1} \, e^{a_2}
        \;=\;
        0
        \,,
        \\
        \scalebox{.7}{
          \color{gray}
          $(\psi^2)$
          \;
        }
        \rho_{a b}
        \big(\,
          \overline{\psi}
          \Gamma^a
          \psi
        \big)
        \,
        e^b
        \;+\;
        \Big(
          \tfrac{1}{6}
          \tfrac{1}{3!}
          \,
          \psi^\alpha
          \big(
            \covariantderivative_\alpha
            (G_4)_{a b_1 b_2 b_3}
          \big)
          \,
          \Gamma^{b_1 b_2 b_3}
          \,-\,
          \tfrac{1}{12}
          \tfrac{1}{4!}
          \,
          \psi^\alpha
          \big(
            \covariantderivative_\alpha
            (G_4)^{b_1 \cdots b_4}
          \big)
          \,
          \Gamma_{a b_1 \cdots b_4}
        \Big)
        \psi
        \, e^a
        \\
        \hspace{.6cm}
        \;\;\;+\;
        \big(\,
          \overline{\psi}
          \,
          \CurvatureAtPsiOne^{b_1 b_2}{}_{a}
        \big)
        \tfrac{1}{4}
        \Gamma_{b_1 b_2}\psi
        \,
        e^a
        \;\;
        =
        \;
        0
        \,,
        \\
        \scalebox{.7}{
          \color{gray}
          $(\psi^3)$
          \;
        }
        H_a \psi
        \big(\,
          \overline{\psi}
          \,\Gamma^a\,
          \psi
        \big)
        -
        \tfrac{1}{2}
        \Gamma_{ab}
        \,
        \psi
        \,
        \big(\,
          \overline{\psi}
          \,
          \Gamma^{[a}H^{b]}
          \,
          \psi
        \big)
        \;\;=\;\;
        0
        \,.
      \end{array}
      \right.
    \end{array}
  \end{equation}

\noindent
Here we used that the $(\psi^2)$-component of $\rho$ vanishes by \eqref{VanishingOfpsi2ComponentOfRho}.
\begin{itemize}[leftmargin=.4cm]

\item 
The $(\psi^1)$-component in
\eqref{ComponentsOfGravitinoBianchiIdentity} gives the rheonomic propagation of $\rho_{a b}$ along the odd super-spacetime directions.
\item
In the  $(\psi^2)$- and $(\psi^3)$-component
we have inserted the particular form of $\rho$ from \eqref{FormOfGravitinoFieldStrengthImpliedByG4sBianchi} and the form of $R^{a b}$ from \eqref{CurvatureTensorInFrameField} and \eqref{ComponentsOfTorsionBianchiIdentity}; 
\end{itemize}

\smallskip

\noindent
Next we discuss the $(\psi^2)$- and then the $(\psi^3)$- and $(\psi^0)$-components in detail.

\medskip

\noindent
\fbox{\bf Gravitino Bianchi at $\psi^2$} 
Observe that the $(\psi^2)$-component in
\eqref{ComponentsOfGravitinoBianchiIdentity} is equivalent to the vanishing of this expression:
$$
  \def\arraystretch{1.6}
  \begin{array}{ll}
        \rho_{c a}
        \big(\,
          \overline{\psi}
          \Gamma^{c}
          \psi
        \big)
        e^{a}
        \;+\;
        \Big(
          \tfrac{1}{6}
          \tfrac{1}{3!}
          \underset{\color{gray}
            \tfrac{4!}{2}
            \big(
              \overline{\psi}
              \,\Gamma_{[a b_1}\,
              \rho_{b_2 b_3]}
            \big)
          }{
            \underbrace{
            \psi^\alpha
            \big(
              \covariantderivative_\alpha
              (G_4)_{a b_1 b_2 b_3}
            \big)
            }
          }
          \Gamma^{b_1 b_2 b_3}
          +
          \tfrac{1}{12}
          \tfrac{1}{4!}
          \underset{\color{gray}
            \tfrac{4!}{2}
            \big(
              \overline{\psi}
              \,\Gamma^{[b_1 b_2}\,
              \rho^{b_3 b_4]}
            \big)
          }{
            \underbrace{
          \psi^\alpha
          (
            \covariantderivative_\alpha
            (G_4)^{b_1 \cdots b_4}
          )
          }
        }
          \Gamma_{a b_1 \cdots b_4}
        \Big)
        \psi
        \, e^a
        \;
        \underset{\color{gray} 
          \mathclap{
           +
           \big(\,
           \overline{\psi}
           (
              \Gamma_{b_1}\rho_{b_2 a}
              -
              \Gamma_a\rho_{b_1 b_2}
              +
              \Gamma_{b_2}\rho_{a b_1}
            )
            \big)
          }
        }{
        \underbrace{
        +\;
        \big(\,
            \overline{\psi}
            \,
            \CurvatureAtPsiOne_{b_1 b_2 a}
          \big)
        }
        }
        \tfrac{1}{4}
        \Gamma^{b_1 b_2}\psi
        \,
        e^a
    \\
    \;=\;
    \rho_{c a}
    \big(\,
      \overline{\psi}
      \,\Gamma^c\,
      \psi
    \big)
    \,
    e^a
    \;-\;
    \frac{1}{3}
    \Gamma^{b_1 b_2 b_3}
    \psi
    \,
    \big(\,
      \overline{\psi}
      \,
      \Gamma_{[a b_1}
      \,
      \rho_{b_2 b_3]}
    \big)
    e^a
    +
    \tfrac{1}{24}
    \Gamma_{a b_1 \cdots b_4}
    \psi
    \,
    \big(\,
      \overline{\psi}
      \,
      \Gamma^{[b_1 b_2}
      \,
      \rho^{b_3 b_4]}
      \,
      \psi
    \big)
    e^a
    \\
    \hspace{8cm}
    \;-\;
    \tfrac{1}{4}
    \Gamma^{b_1 b_2}
    \psi
    \Big(\!
    \big(\,
      \overline{\psi}
      \,
      \Gamma_{b_1}
      \,
      \rho_{b_2 a}
    \big)
    -
    \big(\,
      \overline{\psi}
      \,
      \Gamma_{a}
      \,
      \rho_{b_1 b_2}
    \big)
    +
    \big(\,
      \overline{\psi}
      \,
      \Gamma_{b_2}
      \,
      \rho_{a b_1}
    \big)
   \! \Big)
    e^a
    \\
    \;=:\;
    Q_{c a}
    \big(\,
     \overline{\psi}
     \,\Gamma^c\,
     \psi
    \big)
    \;+\;
    \tfrac{1}{2}
    Q_{c_1 c_2 a}
    \big(\,
     \overline{\psi}
     \,\Gamma^{c_1 c_2}\,
     \psi
    \big)
    \;+\;
    \tfrac{1}{5!}
    Q_{c_1 \cdots c_5 a}
    \big(\,
     \overline{\psi}
     \,\Gamma^{c_1 \cdots c_5}\,
     \psi
    \big)
    \,,
  \end{array}
$$
where under the braces we used 
\eqref{OddCovariantDerivativeOfFluxDensity}
and \eqref{SolutionForTheta};
then we moved the bispinorial coefficients -- observing that we pick up a sign when passing the odd components of $\rho$ past $\psi$ -- in order to bring out the product $\psi \overline{\psi}$
on which we finally apply Fierz decomposition \eqref{FierzDecomposition} to obtain the following three independent quadratic forms:
\smallskip 
\begin{align*}
    32
    Q_{c a}
  &  =
    32
    \cdot
    \rho_{ca}
    \;-\;
    \tfrac{1}{3}
    \Gamma^{b_1 b_2 b_3}
    \Gamma_c
    \Gamma_{[a b_1}
    \rho_{b_2 b_3]}
    +
    \tfrac{1}{24}
    \Gamma_{ab_1 \cdots b_4}
    \Gamma_c
    \Gamma^{[b_1 b_2}\rho^{b_3 b_4]}
    -
    \tfrac{1}{4}
    \Gamma^{b_1 b_2}
    \Gamma_c
    \big(
      \Gamma_{b_1}
      \rho_{b_2 a}
      -
      \Gamma_a \rho_{b_1 b_2}
      +
      \Gamma_{b_2}\rho_{a b_1}
    \big),
    \\[5pt]
    32 
    Q_{c_1 c_2 a}
    &
    =
    +
    \tfrac{1}{3}
    \Gamma^{b_1 b_2 b_3}
    \Gamma_{c_1 c_2}
    \Gamma_{[a b_1}
    \rho_{b_2 b_3]}
    -
    \tfrac{1}{24}
    \Gamma_{ab_1 \cdots b_4}
    \Gamma_{c_1 c_2}
    \Gamma^{[b_1 b_2}\rho^{b_3 b_4]}
    +
    \tfrac{1}{4}
    \Gamma^{b_1 b_2}
    \Gamma_{c_1 c_2}
    \big(
      \Gamma_{b_1}
      \rho_{b_2 a}
      -
      \Gamma_a \rho_{b_1 b_2}
      +
      \Gamma_{b_2}\rho_{a b_1}
    \big),
    \\[5pt]
    32 
    Q_{c_1 \cdots c_5 a}
    &
    =
    -
    \tfrac{1}{3}
    \Gamma^{b_1 b_2 b_3}
    \Gamma_{c_1 \cdots c_5}
    \Gamma_{[a b_1}
    \rho_{b_2 b_3]}
    +
    \tfrac{1}{24}
    \Gamma_{ab_1 \cdots b_4}
    \Gamma_{c_1 \cdots c_5}
    \Gamma^{[b_1 b_2}\rho^{b_3 b_4]}
    -
    \tfrac{1}{4}
    \Gamma^{b_1 b_2}
    \Gamma_{c_1 \cdots c_5}
    \big(
      \Gamma_{b_1}
      \rho_{b_2 a}
      -
      \Gamma_a \rho_{b_1 b_2}
      +
      \Gamma_{b_2}\rho_{a b_1}
    \big)
    .
  \end{align*}

 \vspace{2mm}  
\noindent Hence the $(\psi^2)$-component of the gravitino Bianchi identity is equivalent to the joint vanishing of $Q_{c a}$, $Q_{c_1 c_2 a}$ and $Q_{c_1 \cdots c_5 a}$.
Now, direct computation shows \cite{AncillaryFiles} that the Clifford-contractions of $Q_{c a}$  are as follows:
\begin{equation}
  \def\arraystretch{1.6}
  \begin{array}{ccl}
    \Gamma^c \, Q_{c a}
    &=&
    -\frac{261}{2}
    \,
    \Gamma^{b} \, \rho_{a b}
    -
    \frac{31}{12}
    \,
    \Gamma^{a b_1 b_2}
    \rho_{b_1 b_2}
    \\
    \Gamma^a \, Q_{c a}
    &=&
    \frac{43}{2}
    \Gamma^b \, \rho_{c b}
    +
    \frac{53}{12}
    \Gamma^{c b_1 b_2}
    \rho_{b_1 b_2}
    \,.
  \end{array}
\end{equation}
Since the two lines are not multiples of each other, their joint vanishing implies both the  gravitino equation and the irreducibility of $\rho$ (which itself follows already from the gravitino equation, by Lem. \ref{AlgebraicImplicationsOfGravitinoequation}), so that we have found the implications:
\begin{equation}
  \label{RaritaSchwingerEquationFollows}
  Q_{ca}
  \;=\;0
  \hspace{.6cm}
  \Rightarrow
  \hspace{.6cm}
  \Gamma^{a b_1 b_2} \rho_{b_1 b_2} = 0
  \hspace{.6cm}
  \Rightarrow
  \hspace{.6cm}
  \Gamma^b \rho_{a b}
  \;=\;
  0
  \,.
\end{equation}

Conversely, direct but lengthy computation shows \cite{AncillaryFiles} that the irreducibility condition implies that all three terms vanish:
\begin{equation}
  \Gamma^{b'} \rho_{b b'}
  \;=\;
  0
  \hspace{.6cm}
  \Rightarrow
  \hspace{.6cm}
  \left\{
  \def\arraycolsep{0pt}
  \def\arraystretch{1.1}
  \begin{array}{l}
    Q_{ca} = 0
    \,,
    \\
    Q_{c_1 c_2 a} = 0
    \,,
    \\
    Q_{c_1 \cdots c_5 a} = 0
    \,.
  \end{array}
  \right.
\end{equation}
Together this shows that the $(\psi^2)$-component of the gravitino Bianchi identity is equivalent to the gravitino's Rarita-Schwinger equation \eqref{OddCovariantDerivativeOfFluxDensity}.

\medskip

\noindent
\fbox{\bf Gravitino Bianchi at $\psi^3$}
Using \eqref{SolutionForK}, 
the $(\psi^3)$-component in \eqref{ComponentsOfGravitinoBianchiIdentity}
is equivalent to
\begin{equation}
\label{GravitinoBianchiAtPsi3}
\hspace{-2mm} 
  \def\arraystretch{1.5}
  \begin{array}{l}
  \Big(\!\!
  -
  \tfrac{1}{6}
  \tfrac{1}{3!}
  \Gamma_{[a_1 a_2 a_3}
  \psi
  \big(\,
    \overline{\psi}
    \,\Gamma_{a_4]}\,
    \psi
  \big)
  -
  \tfrac{1}{12}
  \tfrac{1}{4!}
  \Gamma_{b a_1 \cdots a_4}
  \psi
  \big(\,
    \overline{\psi}
    \,\Gamma^b\,
    \psi
  \big)
  \\
\quad   +
  \tfrac{
    1
  }{
    4 \cdot 6
  }
  \Gamma_{[a_1 a_2}
  \,
  \psi
  \big(\,
    \overline{\psi}
    \,\Gamma_{a_3 a_4]}\,
    \psi
  \big)
  +
  \tfrac{1}{4\cdot 6}
  \tfrac{1}{24}
  \Gamma^{b_1 b_2}
  \,
  \psi
  \underbrace{
    \big(\,
      \overline{\psi}
      \,\Gamma_{
        b_1 b_2 
        a_1 \cdots a_4
      }\,
      \psi
    \big)
  }_{
    \mathclap{
    \epsilon_{
      b_1 b_2\, a_1 \cdots a_4
      \,
      {\color{darkblue}
        c_1 \cdots c_5
      }
    }
    \big(
      \overline{\psi}
      \,
      \Gamma^{
        \color{darkblue}
        c_1 \cdots c_5
      }
      \,
      \psi
    \big)
    }
  }
  \Big)
  (G_4)^{a_1 \cdots a_4}
 =\;
  0
  \,,
  \end{array}
\end{equation}
where under the brace we recalled \eqref{ExamplesOfHodgeDualCliffordElements}, for use in the following computations.
We claim that the coefficient of $(G_4)^{a_1 \cdots a_4}$ in  \eqref{GravitinoBianchiAtPsi3} vanishes identically, hence that the whole expression holds identically, independently of $G_4$.

To check this, it may be satisfactory to first consider a weaker consequence which may still reasonably be checked by hand, namely the vanishing of this term after its pairing with $\big(\overline{\psi}\, -\big)$: This makes the first summand vanish by \eqref{VanishingQuadraticForms} and the remaining summands become proportional to each other by the Fierz identities \eqref{TheQuarticFierzIdentities}, such as to cancel out:
$$
  \def\arraystretch{1.8}
  \begin{array}{l}
    -\tfrac{1}{12}
    \tfrac{1}{4!}\big(\,
      \overline{\psi}
      \,\Gamma_{b a_1 \cdots a_4}\,
      \psi
    \big)
    \big(\,
      \overline{\psi}
      \,\Gamma^b\,
      \psi
    \big)
    +
    \tfrac{1}{4 \cdot 6}
    \big(\,
      \overline{\psi}
      \,\Gamma_{[a_1 a_2}\,
      \psi
    \big)
    \big(\,
      \overline{\psi}
      \,\Gamma_{a_3 a_4]}\,
      \psi
    \big)
    +
    \tfrac{1}{4 \cdot 6}
    \tfrac{1}{24}
    \big(\,
      \overline{\psi}
      \,\Gamma^{b_1 b_2}\,
      \psi
    \big)
    \big(\,
      \overline{\psi}
      \,
      \Gamma_{
        b_1 b_2 
        \,
        a_1 \cdots a_4
      }\,
      \psi
    \big)
    \\
    \;=\;
    \Big(
      -
      \tfrac{1}{12}
      \tfrac{1}{4!}
      +
      \tfrac{1}{3 \cdot 4 \cdot 6}
      -
      \tfrac{1}{4}
      \tfrac{1}{24}
    \Big)
    \big(\,
      \overline{\psi}
      \,\Gamma_{b a_1 \cdots a_4}\,
      \psi
    \big)
    \big(\,
      \overline{\psi}
        \Gamma^b
      \psi
    \big)
    \\
    \;=\;
    \tfrac{1}{12}
    \Big(
      -
      \tfrac{2}{48}
      +
      \tfrac{8}{48}
      -
      \tfrac{6}{48}
    \Big)
    \big(\,
      \overline{\psi}
      \,\Gamma_{b a_1 \cdots a_4}\,
      \psi
    \big)
    \big(\,
      \overline{\psi}
        \Gamma^b
      \psi
    \big)
    \;=\;
    0
    \,.
  \end{array}
$$

Now to see the vanishing of the full term \eqref{GravitinoBianchiAtPsi3} using heavier Clifford algebra, we first expand its summands into the $\mathrm{Spin}(1,10)$-irreps from \eqref{GeneralCubicFierzIdentities}, which makes its vanishing equivalent to the following four conditions:
$$
  \def\arraystretch{1.8}
  \hspace{-1mm} 
  \begin{array}{r}
    \Big(
    -\tfrac{1}{6}\tfrac{1}{3!}
    \tfrac{1}{11}
    \Gamma_{[a_1 a_2 a_3}
    \Gamma_{a_4]}
    -
    \tfrac{1}{12}
    \tfrac{1}{4!}
    \tfrac{1}{11}
    \Gamma_{b a_1 \cdots a_4}
    \Gamma^b
    +
    \tfrac{1}{4 \cdot 6}
    \tfrac{1}{11}
    \Gamma_{[a_1 a_2}
    \Gamma_{a_3 a_4]}
    -
    \tfrac{1}{4 \cdot 6}
    \tfrac{1}{24}
    \tfrac{1}{77}
    \tfrac{1}{5!}
    \Gamma^{b_1 b_2}
    \epsilon_{
      b_1 b_2
      \, 
      a_1 \cdots a_4
      \,
      \color{darkblue}
      c_1 \cdots c_5
    }
    \Gamma^{
      \color{darkblue}
      c_1 \cdots c_5
    }
    \Big)
    \Xi^{(32)}
    = 0\,,
    \\
    -\tfrac{1}{6}
    \tfrac{1}{3!}
    \Gamma_{[a_1 a_2 a_3}
    \Xi^{(320)}_{a_4]}
    -
    \tfrac{1}{12}
    \tfrac{1}{4!}
    \Gamma^b{}_{a_1 \cdots a_4}
    \Xi^{(320)}_{b}
    -
    \tfrac{1}{4 \cdot 6}
    \tfrac{2}{9}
    \Gamma_{[a_1 a_2}
    \Gamma_{a_3}
    \Xi^{(320)}_{a_4]}
    +
    \tfrac{1}{4 \cdot 6}
    \tfrac{1}{24}
    \tfrac{5}{9}
    \tfrac{1}{5!}
    \Gamma_{b_1 b_2}
    \,
    \epsilon^{
      b_1 b_2
      \,
      a_1 \cdots a_4
      \,
      \color{darkblue}
      c_1 \cdots c_5
    }
    \Gamma_{
      [
      \color{darkblue}
      c_1 \cdots c_4
    }
    \Xi^{(320)}_{
      {\color{darkblue}
        c_5
      }
    ]}
    \;=\;
    0\,,
    \\
    \tfrac{1}{4 \cdot 6}
    \Gamma_{[a_1 a_2}
    \Xi^{(1408)}_{a_3 a_4]}
    \;+\;
    \tfrac{1}{4 \cdot 6}
    \tfrac{1}{24}
    2
    \tfrac{1}{5!}
    \Gamma_{b_1 b_2}
    \epsilon^{
      b_1 b_2
      \,
      a_1 \cdots a_4
      \,
      \color{darkblue}
      c_1 \cdots c_5
    }
    \Gamma_{
      [
      \color{darkblue}
      c_1 c_2 c_3
    }
    \Xi^{(1408)}_{
      {\color{darkblue}
        c_4 c_5
      }
    ]}
    \;=\;
    0\,,
    \\
    \tfrac{1}{4 \cdot 6}
    \tfrac{1}{24}
    \epsilon^{
      b_1 b_2
      \,
      a_1 \cdots a_4
      \,
      \color{darkblue}
      c_1 \cdots c_5
    }
    \Gamma_{b_1 b_2}
    \Xi^{(4224)}_{
      \color{darkblue}
      c_1 \cdots c_5
    }
    \;=\;
    0\,.
  \end{array}
$$
Direct but lengthy computation,
using the irreducibility ($\Gamma^{a_1} \Xi_{a_1 a_2 \cdots a_p} = 0$) of these representations \eqref{TheHigherTensorSpinors},
shows \cite{AncillaryFiles} that these four terms indeed vanish.
This means that the gravitino Bianchi identity at $(\psi^3)$ provides no further condition on the field components beyond the previous conclusion at $(\psi^2)$.

\smallskip

\noindent
\fbox{\bf Gravitino Bianchi at $\psi^0$} Similarly, also the $(\psi^0)$-component in \eqref{ComponentsOfGravitinoBianchiIdentity} is already  implied by the $(\psi^2)$-component: Namely by the irreducibility of $\rho_{a b}$, the exterior derivative $\covariantderivative_{a_1}\rho_{a_2 a_3}$ is already expressed algebraically via Lem. \ref{ImplicationsOfFluxBianchiIdentityOnGravitinoFieldStrength}, and extremely lengthy Clifford algebra manipulations show \cite{AncillaryFiles} that this expression solves the $(\psi^0)$-component in \eqref{ComponentsOfGravitinoBianchiIdentity}.

\smallskip

With the torsion- and gravitino-Bianchi identity thus solved, it follows (e.g. \cite[p. 63]{CF80}\cite[p. 884]{BBLPT88}) on general grounds (Dragon's
Theorem \cite{Dragon79}\cite{Smith84}\cite[Prop. 7]{Lott90}) that also the
\colorbox{lightgray}{\bf curvature Bianchi identity}

\vspace{-.5cm}
\begin{equation}
  \label{ComponentsOfCurvatureBianchi}
  \hspace{-3mm} 
  \def\arraystretch{1.8}
  \begin{array}{l}
    \mathrm{d}
    \,
    R^{a_1 a_2}
    \,+\,\omega^{a_1}{}_{a'_1}
    R^{a'_1 a_2 }
    \,-\,
    R^{a_1 a'_2 }    
    \,
    \omega^{a_2}{}_{a'_2}
    \;=\;
    0
    \\
    \;\;\;\Leftrightarrow\;
    \left\{\!\!\!
    \def\arraystretch{1.6}
    \begin{array}{l}
      \scalebox{.7}{
        \color{gray}
        $(\psi^0)$
        \;\;
      }
      \Big(
        \big(
          \covariantderivative_{b_1}
          R^{a_1 a_2}{}_{b_2 b_3}
        \big)        
      \;-\;
      \big(\,
        \overline{\CurvatureAtPsiOne}
          ^{a_1 a_2}{}_{b_1}
        \,
        \rho_{b_2 b_3}
      \big)
      \Big)
      e^{b_1}
      \,
      e^{b_2}
      \,
      e^{b_3}
      \;=\;
      0
      \,,
      \\
      \scalebox{.7}{
        \color{gray}
        $(\psi^1)$
        \;\;
      }
      \Big(
      \psi^\alpha
      \big(
        \covariantderivative_{\alpha}
        R^{a_1 a_2}{}_{b_1 b_2}
      \big)  
      \;-\;
      \big(\,
        \overline{\psi}
        \,
        \covariantderivative_{b_1}
        \CurvatureAtPsiOne^{a_1 a_2}{}_{b_2}
      \big)
      \;+\;
      \big(\,
        \overline{\CurvatureAtPsiOne}{}^{a_1 a_2}{}_{b_1}
        \,
        H_{b_2} \psi
      \big)
      \;-\;
      \big(\,
        \overline{\psi}
        \,K^{a_1 a_2}\,
        \rho_{b_1 b_2}
      \big)
      \Big)
      e^{b_1}
      \,
      e^{b_2}
      \;=\;
      0
      \,,
      \\
      \scalebox{.7}{
        \color{gray}
        $(\psi^2)$
        \;\;
      }
      \Big(
      2
      R^{a_1 a_2}{}_{ b c }
      \big(\,
        \overline{\psi}
        \,\Gamma^b\,
        \psi
      \big)
      \;+\;
      \psi^\alpha
      \big(\,
      \overline{
      (
      \covariantderivative_\alpha
      \CurvatureAtPsiOne^{a_1 a_2}{}_c
      )
      }
      \,
      \psi
      \,
      \big)
      \;+\;
      \big(\,
        \overline{\psi}
        \,
        \covariantderivative_{c}
        K^{a_1 a_2}
        \,
        \psi
      \big)
      \;-\;
      2
      \big(\,
        \overline{\psi}
        \,K^{a_1 a_2}\,
        H_c
        \psi
      \big)
      \Big)
      e^c
      \;=\;
      0
      \,,
      \\
      \scalebox{.7}{
        \color{gray}
        $(\psi^3)$
        \;\;
      }
      \big(\,
        \overline{\CurvatureAtPsiOne}{}^{a_1 a_2}{}_b
        \,
        \psi
      \big)
      \big(\,
        \overline{\psi}
        \,\Gamma^b\,
        \psi
      \big)
      \;-\;
      \psi^\alpha
      \big(\,
        \overline{\psi}
        \,
        \covariantderivative_\alpha
        K^{a_1 a_2}
        \psi
      \big)
      \;=\;
      0
      \,,
    \end{array}
    \right.
  \end{array}
\end{equation}

\smallskip 
\noindent is already solved in that it implies no further constraints on the fields.
\end{proof}

\smallskip 
This means in particular that the Einstein equation must already be implied from the gravitino Bianchi identity and hence from the gravitino equation of motion. Remarkably, this is the case:


\begin{lemma}[\bf Einstein equation is Susy partner of Rarita-Schwinger equation]
\label{DerivingTheEinsteinEquation}
$\,$

\noindent Given super-flux densities $(G_4^s, G_7^s) \in \Omega^1_{\mathrm{dR}}\big(X; \mathfrak{l}S^4\big)_{\mathrm{clsd}}$
 {\rm\eqref{SuperCFieldBianchiInIntro}} on a super-spacetime $\big(X, (e, \psi, \omega)\big)$, the latter satisfies 
the Einstein equation for the energy-momentum of the C-field flux: 
\footnote{
Due to our sign convention (cf. footnote \ref{ConventionsForTorsionAndCurvature}) the Einstein equation \eqref{TheEinsteinEquation} agrees with that of \cite[(24)]{Figueroa13}.
}
\begin{equation}
  \label{TheEinsteinEquation}
  \colorbox{lightblue}{$
  R_{a}{}^c{}_{b c}
  -
  \tfrac{1}{2}
  R^{c_1 c_2}{}_{\! c_1 c_2 }
  \,
  \eta_{a b}
  \;=\;
  +
  \tfrac
    { 1 }
    { 12 }
  \Big(
    (G_4)_{
      a
      \,
      c_1 c_2 c_3
    }
    (G_4)_b{}^{c_1 c_2 c_3}
    -
    \tfrac{1}{8}
    (G_4)_{c_1 \cdots c_4}
    (G_4)^{c_1 \cdots c_4}
    \,
    \eta_{a b}
  \Big)
  $}
\end{equation}
\begin{equation}
\label{EinsteinEquationInTermsOfRicciCurvature}
  \hspace{1cm}
  \Leftrightarrow
  \hspace{1cm}
  \colorbox{lightblue}{$
  R_{a}{}^c{}_{b c}
  \;=\;
  +
  \tfrac
    { 1 }
    { 12 }
  \Big(
    (G_4)_{
      a
      \,
      c_1 c_2 c_3
    }
    (G_4)_b{}^{c_1 c_2 c_3}
    -
    \tfrac{1}{12}
    (G_4)_{c_1 \cdots c_4}
    (G_4)^{c_1 \cdots c_4}
    \,
    \eta_{a b}
  \Big)
  $}
  \mathrlap{\,.}
\end{equation}
\end{lemma}
\begin{proof}
Consider the spinorial covariant derivative of the gravitino equation evaluated in $(\overline{\psi}\,-)$:
\begin{equation}
  \label{OddExteriorDerivativeOfGravitinoEquation}
  \def\arraystretch{1.8}
  \begin{array}{ll}
    \mathllap{
      0
      \;=\;\;
    }
    -
    \tfrac{1}{2}
    \Big(\,
    \overline{\psi}
    \,
      \Gamma_a{}^{b_1 b_2}
      \,
      \psi^\alpha
      \covariantderivative_\alpha
    \,
    \rho_{b_1 b_2}
    \Big)
    &
    \proofstep{
      by
      \eqref{GravitinoEquation}
      \&      \eqref{GammaMatricesAreCovariantlyConstant}
    }
    \\
    \;=\;
    \underbrace{
    \Big(\,
      \overline{\psi}
      \,
      \Gamma_a{}^{b_1 b_2}
      \covariantderivative_{[b_1}
      H_{b_2]}
      \psi
    \Big)
    }_{{\color{orangeii} \bf (C)}}
    \;-\;
    \underbrace{
    \big(\,
      \overline{\psi}
      \,
      \Gamma_a{}^{b_1 b_2}
      H_{b_1}
      H_{b_2}
      \,
      \psi
    \big)
    }_{{\color{darkblue} \bf (B)}}
    \;-\;
    \underbrace{
    \tfrac{1}{4}
    \tfrac{1}{2}
    R_{b_1 b_2 a_1 a_2}
    \big(\,
    \overline{\psi}
    \,
    \Gamma_a{}^{b_1 b_2}
    \Gamma^{a_1 a_2}
    \,
    \psi
    \big)
    }_{
      {\color{darkgreen} \bf (A)}
    }
    &
    \proofstep{
      by \eqref{ComponentsOfGravitinoBianchiIdentity}.
    }
  \end{array}
\end{equation}
\fbox{\bf {\color{darkgreen} \bf (A)}} Direct Gamma-expansion \eqref{GeneralCliffordProduct} shows that the rightmost
summand (A) is (minus) the superspace Einstein tensor contracted with 
$\big(\, \overline{\psi}\,\Gamma^c\, \psi\big)$:

\vspace{-5mm} 
$$
  \def\arraystretch{1.7}
  \def\arraycolsep{2pt}
  \begin{array}{lcl}
    \tfrac{1}{4}
    R_{b_1 b_2 a_1 a_2}
    \big(\,
    \overline{\psi}
    \Gamma_a{}^{b_1 b_2}
    \Gamma^{a_1 a_2}
    \psi
    \big)
   &=& 
    \tfrac{1}{4}
    \overbrace{
      R^{b_1 [b_2 a_1 a_2]}
    }^{{
      \color{gray}
      = \, 0
      \mathrlap{
        \scalebox{.6}{
          \eqref{ComponentsOfTorsionBianchiIdentity}
        }
      }
    }}
    \big(\,
      \overline{\psi}
      \,\Gamma_{a b_1 b_2 a_1 a_2}\,
      \psi
    \big)
    -
    \tfrac{1}{2}
    R^{b_1 b_2}{}_{a_1 a_2}
    \big(
      \delta
        ^{a_1 a_2}
        _{b_1 b_2}
        \eta_{a c}
        -
        \delta
          ^{a_1 a_2}
          _{a \; b_2}
        \eta
          _{b_1 c}
        +
        \delta
          ^{a_1 a_2}
          _{a \; b_1}
         \eta
           _{b_2 c}
    \big)
    \big(
      \overline{\psi}
      \Gamma^c
      \psi
    \big)
    \\
    &=& 
    -
    \tfrac{1}{2}
    \big(
      R^{b_1 b_2}{}_{b_1 b_2}
      \,
      \eta_{ac}
      -
      R_a{}^b{}_{c b}
      -
      R_a{}^b{}_{c b}
    \big)
    \big(\,
      \overline{\psi}
      \,\Gamma^c\,
      \psi
    \big)
   \\ 
   &=& 
    \big(
    R_a{}^b{}_{cb}
    -
    \tfrac{1}{2}
    R^{b_1 b_2}{}_{b_1 b_2}
    \,
    \eta_{ac}
    \big)
    \big(\,
      \overline{\psi}
      \Gamma^c
      \psi
    \big)
    \,.
  \end{array}
$$
Therefore we need to consider only the $\big(\,\overline{\psi}\Gamma^c\psi\big)$-components of the other two summands.

\smallskip 
\noindent
\fbox{\color{darkblue} \bf (B)} Direct but laborious Gamma-expansion inside the (B)-summand of \eqref{OddExteriorDerivativeOfGravitinoEquation} shows \cite{AncillaryFiles}
that its $\big(\,\overline{\psi}\Gamma^c\psi\big)$-component is the energy-momentum tensor of the C-field:
$$
\hspace{-1mm}
  \def\arraystretch{1.9}
  \begin{array}{l}
    \big(\,
      \overline{\psi}
      \,\Gamma_a{}^{b_1 b_2}\,
      H_{b_1}
      H_{b_2}
      \psi
    \big)
    \\
    \;= \!
    \bigg(
      \overline{\psi}
      \,
      \Gamma_a{}^{b_1 b_2}
      \Big(
        \tfrac{1}{6}\tfrac{1}{3!}
        (G_4)_{b_1 c_1 c_2 c_3}
        \Gamma^{c_1 c_2 c_3}
        -
        \tfrac{1}{12}
        \tfrac{1}{4!}
        (G_4)^{c_1 \cdots c_4}
        \Gamma_{b_1 c_1 \cdots c_4}
      \Big)
      \Big(
        \tfrac{1}{6}\tfrac{1}{3!}
        (G_4)_{b_2 c_1 c_2 c_3}
        \Gamma^{c_1 c_2 c_3}
        -
        \tfrac{1}{12}
        \tfrac{1}{4!}
        (G_4)^{c_1 \cdots c_4}
        \Gamma_{b_2 c_1 \cdots c_4}
      \Big)
      \psi
    \bigg)
    \\
    \;=
    -
    \tfrac
      {1}
      {24}
    \Big(
    (G_4)_{
      a
      \,
      b_1 b_2 b_3
    }
    (G_4)
     _c
     {}
     ^{b_1 b_2 b_3}
     -
     \tfrac{1}{8}
    (G_4)_{b_1 \cdots b_4}
    (G_4)^{b_1 \cdots b_4}
    \,
    \eta_{a c}
    \Big)
    \big(\,
      \overline{\psi}
      \,\Gamma^c\,
      \psi
    \big)
    \,+\,
    \scalebox{.9}{$
  Q_{a a_1 a_2}
  \big(\,
    \overline{\psi}
    \,
    \Gamma^{a_1 a_2}
    \,
    \psi
  \big)
  \;+\;
  Q_{a a_1 \cdots a_5}
  \big(\,
    \overline{\psi}
    \,
    \Gamma^{a_1 \cdots a_5}
    \,
    \psi
  \big)
  $}
  \,.
  \end{array}
$$

\smallskip 
\noindent
\fbox{\color{orangeii} \bf (C)} Finally, simple inspection shows that Gamma-expansion inside the (C)-summand in \eqref{OddExteriorDerivativeOfGravitinoEquation} produces vanishing $\big(\, \overline{\psi}\Gamma^c\psi\big)$-component:
$$
  \def\arraystretch{1.6}
  \begin{array}{lll}
  \big(\,
    \overline{\psi}
    \,\Gamma_a{}^{b_1 b_2}\,
    \covariantderivative_{[b_1}
    H_{b_2]}
    \,
    \psi
  \big)
  &
  =\;
  \Big(
    \overline{\psi}
    \,\Gamma_a{}^{b_1 b_2}\,
    \covariantderivative_{[b_1}
    \big(
      \tfrac{1}{6}
      \tfrac{1}{3!}
      (G_4)_{b2] c_1 c_2 c_3}
      \Gamma^{c_1 c_2 c_3}
      -
      \tfrac{1}{12}
      \tfrac{1}{4!}
      (G_4)^{c_1 \cdots c_4}
      \Gamma_{b_2] c_1 \cdots c_4}
    \big)
    \,
    \psi
  \Big)
  &
  \proofstep{
   by \eqref{FormOfGravitinoFieldStrengthImpliedByG4sBianchi}
  }
  \\
 & =\;
  Q'_{a a_1 a_2}
  \big(\,
    \overline{\psi}
    \,
    \Gamma^{a_1 a_2}
    \,
    \psi
  \big)
  \;+\;
  Q'_{a a_1 \cdots a_5}
  \big(\,
    \overline{\psi}
    \,
    \Gamma^{a_1 \cdots a_5}
    \,
    \psi
  \big)
  &
  \proofstep{
    by 
    \eqref{GeneralCliffordProduct}.
  }
  \end{array}
$$
Inserting these three expressions
for $\big(\,\overline{\psi}\,\Gamma^c\,\psi\big)$-components back into \eqref{OddExteriorDerivativeOfGravitinoEquation} evidently yields the claimed Einstein equation \eqref{TheEinsteinEquation}.

Finally, just to observe that both \eqref{TheEinsteinEquation} and \eqref{EinsteinEquationInTermsOfRicciCurvature} imply that the scalar curvature is given by
\begin{equation}
  \label{ImplicationOfEinsteinequationOnScalarCurvature}
  R^{c_1 c_2}{}_{\! c_1 c_2 }
  \;=\;
  +
    \tfrac
      { 1 }
      { 24 }
    \tfrac{1}{12}
    (G_4)_{c_1 \cdots c_4}
    (G_4)^{c_1 \cdots c_4}
\end{equation}
and that the difference between \eqref{TheEinsteinEquation} and \eqref{EinsteinEquationInTermsOfRicciCurvature} is just half this equation \eqref{ImplicationOfEinsteinequationOnScalarCurvature}.
\end{proof}

\begin{remark}[\bf Role of duality-symmetric Bianchi identities in enforcing the super-torsion constraints]
 \label{RoleOfDualitySymmetricBianchiIdentitiesInEnforcingTheSuperTorsionConstraint}
 The conclusion in \eqref{AllPairingsWithThePsi2ComponentOfRhoVanish} that the $(\psi^2)$-component of $\rho$ vanishes \eqref{VanishingOfpsi2ComponentOfRho} is ultimately enforced by our independent (duality-symmetric) imposition of the $G^s_4$ and $G_7^s$-flux Bianchi identities, since this is what gives the necessary 2-index and the 5-index constraints in \eqref{AllPairingsWithThePsi2ComponentOfRhoVanish}. In previous discussions the same constraint is obtained instead from an extra scaling condition \cite[(16)]{BrinkHowe80}\cite[(III.8.37)]{CDF91}\cite[(53)]{Howe97}.
\end{remark}

\smallskip

\noindent
{\bf Rheonomy.}
\label{Rheonomy}
It just remains to observe that the super-fields used in this super-space formulation of 11d supergravity 
carry -- despite their plethora of super-components -- no further data than expected. This is the property 
called {\it rheonomy} in \cite[\S III.3.3]{CDF91}, where the sketch of a general argument is given. 
A detailed recursive expression of the 11d on-shell superfields on the super-spacetime starting from their restriction to the bosonic body $\bosonic{X}$ is worked out in \cite{Tsimpis04}.

\smallskip 
For our purpose, we highlight rheonomy of the flux density forms:

\begin{lemma}[\bf Rheonomy for Super C-Field flux]
\label{RheonomyForSuperCFieldFLux}
Choosing super-flux densities
$$
  \big(
    G_4^s
    ,\,
    G_7^s
  \big)
  \;\;
    :
  \;\;
  \begin{tikzcd}
    X
    \ar[rr]
    &&
    \Omega^1_{\mathrm{dR}}(
      -
      ;\,
      \mathfrak{l}S^4
    )
  \end{tikzcd}
$$
on a super-spacetime $X$ is tantamount to choosing a solution of 11d SuGra with respect to ordinary flux densities
$$
  \big(
    \eta^{\rightsquigarrow}_X
  \big)^\ast
  \big(
    G_4^s
    ,\,
    G_7^s
  \big)
  \;\;
    :
  \;\;
  \begin{tikzcd}
    \bos{X}
    \ar[
      r,
      "{
        \eta^{\rightsquigarrow}_X
      }"
    ]
    &
    X
    \ar[
      rr,
      "{
        (
          G^s_4
          ,\,
          G^s_7
        )
      }"
    ]
    &&
    \Omega^1_{\mathrm{dR}}(
      -
      ;\,
      \mathfrak{l}S^4
    )
    \,.
  \end{tikzcd}
$$
\end{lemma}
\begin{proof}
  On a neighborhood of any point $ x_0 \in \bos{X}$, we may find {\it super-Riemann normal coordinates} 
  $\big\{(x^\evencoordinateindex)_{\evencoordinateindex = 0}^{10}, (\theta^\oddcoordinateindex)_{\oddcoordinateindex=1}^{32}\big\}$ 
  on $X$ such that (\cite[(43)-(44)]{Tsimpis04} following \cite{McArthur84}):
  \vspace{-2mm} 
  \begin{equation}
    \label{PropertiesOfSuperNormalCoordinates}
    \def\arraystretch{1.4}
    \def\arraycolsep{2pt}
    \begin{array}{rrr}
      \theta^{\oddcoordinateindex}
      \,
      e_{\oddcoordinateindex}^a 
      =
      0\,,
      \\
      \theta^{\oddcoordinateindex}
      \,
      \psi_{\oddcoordinateindex}^{\alpha}
      =
      0\,,
      \\ 
      \theta^{\oddcoordinateindex}
      \,
      \omega_{\oddcoordinateindex}{}^a{}_b
      =
      0
      \,.
   \end{array}
  \end{equation}
  Therefore the $(\psi^1)$-component 
  $\psi^\alpha \covariantderivative_{\alpha} (G_4)_{a_1 \cdots a_4} \,=\, 12 \, \big(\, \overline{\psi} \, \Gamma_{[a_1 a_2}\, \rho_{a_3 a_4]} \big)$ 
  of the $G_4$-Bianchi identity \eqref{OddCovariantDerivativeOfFluxDensity}  says that at any point $x_0$ 
  we may decompose the super-flux density as (cf. \cite[(53)]{Tsimpis04}):
  $$
    (G_4)_{a_1 \cdots a_4}
      \big(
        x_0, 
        \{\theta^{\oddcoordinateindex}\}_{\oddcoordinateindex = 1}^{32}
      \big)
    \;=\;
    \underbrace{
      \big(
        \eta^\rightsquigarrow_X
      \big)^\ast
      (G_4)_{a_1 \cdots a_4}
    (x_0)
    }_{
      \mathclap{
        \raisebox{-4pt}{
          \scalebox{.7}{
            \color{gray}
            \bf
            ordinary flux density
          }
        }
      }
    }
    \;+\;
    \underbrace{
    12
    \,
    \Big(
      \overline{\theta}
      \,\Gamma_{[a_1 a_2}
      \,
      \rho_{a_3 a_4]}(x_0, \{\theta^{\oddcoordinateindex}\}_{\oddcoordinateindex = 1}^{32})
    \Big)
    }_{
      \mathclap{
        \scalebox{.7}{
          \color{gray}
          \bf
          its higher superfield components
        }
      }
    } 
    \,.
  $$
  
  Conversely, given the ordinary 4-flux density, this equation defines its extension to a super-flux-density which is closed, since
  $$
    \def\arraystretch{1.5}
    \begin{array}{lll}
      \Big(
        \covariantderivative_{a_0}
        \big(\,
          \overline{\theta}
          \,
          \Gamma_{[a_1 a_2}
          \,
          \rho_{a_3 a_4]}
        \big)
      \Big)
      e^{a_0} \cdots e^{a_4}
         & \;=\;
      \Big(
        \overline{\theta}
        \,
        \Gamma_{[a_1 a_2}
        \,
        \covariantderivative_{a_0}
        \rho_{a_3 a_4]}
      \Big)
      e^{a_0} \cdots e^{a_4}
      \\
    &  \;=\;
      0
      &
      \proofstep{
        by $(\psi^0)$-component of 
        \eqref{ComponentsOfGravitinoBianchiIdentity}
      .}
    \end{array}
  $$
  
However, we need to show more, since the $(\psi^1)$-component of the $G_7^s$-Bianchi identity \eqref{ComponentsOfBianchiOfGs7} similarly prescribes the rheonomic extension of $G_7$ \eqref{RheonomyConditionForG7}, which however by the $(\psi^2)$-component of \eqref{ComponentsOfBianchiOfGs7} is linearly dependent on $G_4$ \eqref{G7HodgeDualToG4InComponents}. In order for this not to be a further constraint, we observe that the rheonomy \eqref{RheonomyConditionForG7} of $G_7$ is already implied by that for $G_4$ (using their Hodge duality and the gravitino equation of motion):
$$
  \def\arraystretch{1.6}
  \begin{array}{lll}
    \psi^\alpha
    \,
    \covariantderivative_\alpha
    \tfrac{1}{7!}
    (G_7)_{a_1 \cdots a_7}
    &
    \;=\;
    \tfrac{1}{7!\cdot 4!}
    \,
    \psi^\alpha
    \,
    \covariantderivative_\alpha
    \epsilon_{a_1 \cdots a_7 b_1 \cdot b_4} 
    (G_4)^{b_1 \cdots b_4}
    &
    \proofstep{
      by \eqref{G7HodgeDualToG4InComponents}
    }
    \\
    &\;=\;f
    \tfrac{1}{2\cdot 7!}
    \,
    \epsilon_{
      a_1 \cdots a_7 
      b_1 \cdots b_4
    } 
    \big(\,
      \overline{\psi}
      \,
      \Gamma^{[b_1 b_2}
      \,
      \rho^{b_3 b_4]}
    \big)
    &
    \proofstep{
      by \eqref{OddCovariantDerivativeOfFluxDensity}
    }
    \\
   & \;=\;
    \underbrace{
    \tfrac{84}{2\cdot 7!}
    }_{ \color{gray} 1/5! }
    \,
    \big(\,
      \overline{\psi}
      \,
      \Gamma_{[a_1 \cdots a_5}
      \,
      \rho_{a_6 a_7]}
    \big)
    &
    \proofstep{
      by \eqref{TheAlgebraicImplicationsOfGravitinoEquation}\,.
    }
  \end{array}
$$

\vspace{-5mm}
\end{proof}

\medskip 
With Lem. \ref{SuperBianchiIdentityForG4InComponents},\,
\ref{SuperBianchiIdentityForG7InComponents},\, \ref{SuperFluxAndGravitinoBianchiEquivalentToRaritaSchwinger},
\,\ref{DerivingTheEinsteinEquation},
and \ref{RheonomyForSuperCFieldFLux}
the proof of Thm. \ref{11dSugraEoMFromSuperFluxBianchiIdentity} is now complete.

\medskip

In concluding, we highlight how this relates back to fields on the bosonic underlying spacetime $\bosonic{X}$:

\begin{corollary}[\bf 11d SuGra on bosonic spacetime from quantizable super-flux]
\label{11dSugraOnBosonicSpacetimeFromQuantizableSuperFlux}
  Given an ordinary 11-dimensional smooth manifold $\bosonic{X}$ equipped with geometric $\mathrm{Spin}(1,10)$-structure $P \xrightarrow{\;} \bosonic{X}$, there is an isomorphism of smooth super-sets (Def. \ref{SuperSmoothSets}) between 
  \begin{itemize}[leftmargin=.7cm]
  \item[\bf (i)] the on-shell field space (as in \S\ref{SuperSmoothFieldSpaces}) of 11d SuGra on $\bosonic{X}$,

  \item[\bf (ii)] super-flux densities $\big(G_4^s,, G^s_7\big)$ of the form \eqref{SuperFluxDensitiesInIntroduction} on super-spacetime structures on the extending super-manifold $X := \bosonic{X}\big\vert \mathbf{32} \underset{\mathclap{\scalebox{.5}{$\mathrm{Spin}(1,10)$}}}{\times} P$ \eqref{ExtendingSpinManifoldToSupermanifold} which are closed as $\mathfrak{l}S^4$-valued differential forms (Ex. \ref{ClosedlS4ValuedDifferentialForms}).
  \end{itemize}
\end{corollary}
\begin{proof}
  Thm. \ref{11dSugraEoMFromSuperFluxBianchiIdentity} shows that the restriction from $X$ to $\bosonic{X}$ is plotwise injective, and with rheonomy as in \cite[\S 4]{Tsimpis04} (cf. Lem. \ref{RheonomyForSuperCFieldFLux}) it follows that it is plotwise surjective.
\end{proof}

\smallskip

\noindent
{\bf Conclusion.}
While Thm. \ref{11dSugraEoMFromSuperFluxBianchiIdentity} --- apart from some mild but consequential changes of perspective, cf. Rem. \ref{FormOfTheSuperTorsionConstraint} --- is essentially the claim of \cite[\S III.8.5]{CDF91}, which in turn is essentially the claim originating with \cite{CF80}\cite{BrinkHowe80}, the proof seems to have never been recorded, and the necessity of proving also the converse direction (namely that the $\psi^0$- and the $\psi^3$-components of the gravitino Bianchi imply no further conditions besides the Rarita-Schwinger equation, which ends up being the bulk of the work) may not have received attention before and seems out of reach without computer algebra such as \cite{Gr01}.

At any rate, this derivation of on-shell 11d SuGra from just the demand of quantizable super C-field flux is remarkable in view of the (UV-)completion of the theory, as discussed in \S\ref{SuperFluxQuantization}. We discuss further implications in \cite{GSS-M5Brane}\cite{GSS-Exceptional}.

\end{document}